\documentclass[12pt]{amsart}
\usepackage{amsmath,amsthm, amscd}
\usepackage{graphicx,amsmath,amsfonts}

\newtheorem{thm}{Theorem}[section]
\newtheorem{cor}[thm]{Corollary}
\newtheorem{lem}[thm]{Lemma}
\newtheorem{prop}[thm]{Proposition}

\newtheorem{ex}[thm]{Example} 

\newtheorem{rem}[thm]{Remark}
\newtheorem{claim}[thm]{Claim}

\numberwithin{equation}{section}

\newcommand{\eproof}{\hfill\rule{2.2mm}{3.0mm}}

\newcommand{\R}{{\mathbb R}}
\newcommand{\Z}{{\mathbb Z}}
\newcommand{\T}{{\mathbb T}}

\newcommand{\Q}{{\mathbb Q}}
\newcommand{\N}{{\mathbb N}}

\renewcommand{\eqref}[1]{(\ref{#1})}

\newcommand{\beq}[1]{\begin{equation} \label{#1}}
\newcommand{\eeq}{\end{equation}}

\begin{document}
\title{The $abc$-problem for Gabor systems}
\author{Xin-Rong Dai}

\address{School of Mathematical and Computational Science\\
Sun Yat-Sen University, Guangzhou, 510275, P. R. China.}
\email{daixr@mail.sysu.edu.cn}

\author{Qiyu Sun}
\address{Department of Mathematics, University of Central Florida, Orlando, FL 32816, USA}
\email{qiyu.sun@ucf.edu}



\begin{abstract}
 A  Gabor system
generated by a window function $\phi$  and a rectangular lattice $a \Z\times \Z/b$
is given by
$$
{\mathcal G}(\phi, a \Z\times \Z/b):=\{e^{-2\pi i n t/b} \phi(t- m a):\ (m,  n)\in \Z\times \Z\}.$$
One of fundamental problems in  Gabor analysis is to identify  window functions $\phi$  and time-frequency shift lattices
$a \Z\times \Z/b$
such that the corresponding Gabor system ${\mathcal G}(\phi, a \Z\times \Z/b)$ is a Gabor frame for
$L^2(\R)$, the space  of all square-integrable functions on the  real line $\R$.
The range 
 of density parameters $a$ and $b$ such that  the  Gabor system
${\mathcal G}(\phi, a\Z\times \Z/b)$ is a frame for $L^2(\R)$ 
 is
fully known surprisingly only  for  few 
window functions, including
 the Gaussian window  and  totally positive windows.
Janssen's tie suggests that the range of density parameters could be  arbitrarily complicated for window functions, especially outside Feichtinger algebra.
 An eye-catching example of such a window function
is the ideal window function on an interval.
In this paper,
we provide a  full classification of   triples
$(a,b,c)$   for which
the  Gabor system ${\mathcal G}(\chi_I, a \Z\times \Z/b)$ generated by
the ideal window function $\chi_I$ on an interval $I$ of length $c$ is a Gabor frame for $L^2(\R)$.
For the  classification of such triples $(a, b, c)$ 
(i.e., the $abc$-problem for Gabor systems), we  introduce  maximal invariant sets of some piecewise linear transformations
and establish the equivalence between Gabor frame property and triviality of
maximal invariant sets. We then study
 dynamic system associated with the  piecewise linear transformations
and explore
 various 
  properties of their  maximal invariant sets. By  performing
  holes-removal surgery for  maximal invariant sets to shrink and augmentation operation for a line with marks to expand, we finally parameterize
 those triples $(a, b, c)$ for which  maximal invariant sets are trivial.
The   novel techniques
involving 
non-ergodicity of dynamical systems associated with some novel non-contractive and non-measure-preserving transformations
lead to
 our arduous answer to the $abc$-problem for Gabor systems. 



\end{abstract}

\maketitle

\vspace{-2ex}
\tableofcontents



\section{Introduction}

Let $L^2:=L^2(\R)$ be the space of all square-integrable functions on the real line $\R$ with
  the   inner product and norm  on $L^2$   denoted by $\langle \cdot, \cdot\rangle$ and
$\|\cdot\|_2$ respectively.
 In this paper,
 a  {\em Gabor system}  
generated by a window function $\phi\in L^2$  and a rectangular lattice $a \Z\times \Z/b$
is given by
\begin{equation} \label{gabor.def}
{\mathcal G}(\phi, a \Z\times \Z/b):=\{e^{-2\pi i n t/b} \phi(t- m a):\ (m,  n)\in \Z\times \Z\}\end{equation}
(\cite{BGL10,  grochenigbook, grocheniglinert03, HW01});  a {\em frame} for $L^2$  is a collection
${\mathcal F}$ of  functions in $L^2$ satisfying
\begin{equation}\label{frame.def}
0< 
\inf_{\|f\|_2=1} \Big(\sum_{\phi\in {\mathcal F}} |\langle f, \phi\rangle|^2\Big)^{1/2}\le
\sup_{\|f\|_2=1} \Big(\sum_{\phi\in {\mathcal F}} |\langle f, \phi\rangle|^2\Big)^{1/2} 
<\infty
\end{equation}
(\cite{casazzasurvey, cerdaseip02, christensenbook,  duffintams52}); and
 a {\em Gabor frame} is a  Gabor system
that forms a frame for   $L^2$ (\cite{heilsurvey, janssen01}).
Here we use $1/b$ to label the frequency spacing instead of $b$ in  other sources on Gabor theory, due to the convenience to state our
full classification of triples
$(a,b,c)$   for which
the  Gabor system ${\mathcal G}(\chi_I, a \Z\times \Z/b)$ generated by
the ideal window function $\chi_I$ on an interval $I$ of length $c$ is a Gabor frame for $L^2$.

\smallskip

 Gabor frames have  links to operator algebra and complex analysis, and they
  have been shown to be suitable for lots of applications involving time-dependent frequency. The history of the Gabor theory  could
 date back to  the completeness  claim in 1932 by von Neumann 
 \cite[p. 406]{neumannbook} and
the expansion conjecture in 1946 by Gabor \cite[Eq. 1.29]{gabor46}, with both  confirmed to be correct in later 1970.
Gabor frames
become widely studied after the landmark paper in 1986 by Daubechies, Grossmann  and Meyer,
where they proved that for any positive density parameters $a, b$ satisfying $a/b<1$
 there exists a compactly supported smooth function $\phi$ such that ${\mathcal G}(\phi, a\Z\times \Z/b)$
is a  Gabor frame \cite{dgm86}.

\smallskip

One of fundamental problems in  Gabor analysis is to identify  window functions  and time-frequency shift sets
such that the corresponding Gabor system  is a Gabor frame.
Given  a window function $\phi\in L^2$  and a rectangular lattice $a \Z\times \Z/b$,
 a well-known  necessary  condition for the  Gabor system ${\mathcal G}(\phi, a\Z\times \Z/b)$
to be  a  Gabor frame,  obtained via Banach algebra technique, is that the density parameters $a$ and $b$ satisfy
$a/b\le 1$
\cite{baggett90, daubechiesIEEE, janssen94,   landau93, Rieffel81}.
But that necessary condition on density parameters is far from providing an answer to the above fundamental problem.

\smallskip

The range of density parameters $a, b$ such that  the  Gabor system
${\mathcal G}(\phi, a\Z\times \Z/b)$ is a frame for $L^2$   is
fully known  stunningly only  for  small number of window functions $\phi$ 
\cite{grochenigstockler, janssenjfa03, janssenstrohmer,  lyubarskii92, seip92, seipwallsten92}.
Among those, the Gaussian window $\sqrt[4]{2} \exp(-\pi t^2)$ has received special attention
 \cite{gabor46, neumannbook}. It is conjectured by  Daubechies and Grossmann \cite{dg88} and later
proved independently  by Lyubarskii \cite{lyubarskii92}
and by Seip and Wallsten \cite{seip92, seipwallsten92}  via complex analysis technique
 that
the  Gabor system ${\mathcal G}(\sqrt[4]{2} \exp(-\pi t^2), a\Z\times \Z/b)$
associated with the Gaussian window is a  Gabor frame  if and only if $a/b<1$.
A significant advance on the range of density parameters was recently made by Gr\"ochenig and St\"ockler
that for a totally positive function  $\phi$ of finite type,
${\mathcal G}(\phi, a\Z\times \Z/b)$ is a Gabor frame if and only if $a/b<1$
\cite{grochenigstockler}.

\smallskip

  For a window function $\phi$ in Feichtinger's algebra, an important result proved by  Feichtinger and Kaiblinger  \cite{feichtinger04}
   states that the range of density parameters $a, b$ such that  the  Gabor system
${\mathcal G}(\phi, a\Z\times \Z/b)$ is a frame for $L^2$ is  an open domain on the plane. But for window functions outside the  Feichtinger algebra, the range of density parameters
could be  arbitrarily complicated, c.f. the  famous Janssen's tie \cite{janssen03}. A striking example of such a window function  is the ideal window function  on an interval  \cite{ckpams02, hangu, janssen03}.

\smallskip

Denote by $\chi_E$  the characteristic function on a set $E$.
Recall that  given an interval $I$,
 ${\mathcal G}(\chi_{I}, a\Z\times \Z/b)$
 is a Gabor frame if and only if ${\mathcal G}(\chi_{I+d}, a\Z\times \Z/b)$ is a Gabor frame
  for every $d\in \R$.
 By the above shift-invariance, the interval $I$ can be assumed to be
half-open  and have zero as its left endpoint, i.e., $I=[0, c)$ for some $c>0$. Thus the  range problem
for the ideal window on an interval reduces
to the so-called  $abc$-problem for Gabor  systems:
 {\em the  classification of all  triples $(a, b, c)$ of positive numbers such that ${\mathcal G}(\chi_{[0,c)}, a\Z\times \Z/b)$
is a   Gabor frame} \cite{ckpams02}.
In this paper, we   provide a complete answer to the above $abc$-problem for Gabor systems, see Theorems \ref{newmaintheorem1}--\ref{newmaintheorem5}.

\smallskip
By  dilation-invariance,
 ${\mathcal G}(\chi_{[0,c)}, a\Z\times \Z/b)$ is  a Gabor frame  if and only if
 $ {\mathcal G}(\chi_{[0,dc)}, da)\Z\times \Z/(db))$ is for any $d>0$. 
Thus  the $abc$-problem for Gabor systems can be reduced to
finding out all pairs $(a, b)$ of time-frequency spacing parameters such that ${\mathcal G}(\chi_{[0,1)}, a\Z\times \Z/b)$
are Gabor frames \cite{hangu, janssen03}, or  all pairs $(a, c)$ of time-spacing and window-size parameters
such that ${\mathcal G}(\chi_{[0,c)}, a\Z\times \Z)$ are Gabor frames.
In this paper, we 
 state our results without any normalization
  on any one of the time-spacing, frequency-spacing  and
window-size parameters
as it does not  help us much.

 \smallskip

 Our answer to the $abc$-problem for Gabor system is illustrated in Figure \ref{classification.fig}, where
 on the left subfigure we normalize the  frequency-spacing parameter $b$,  while on the right subfigure we normalize the window-size parameter $c$ and
 use the conventional frequency-spacing parameter $1/b$  as the $y$-axis,  c.f. Janssen's tie in \cite{janssen03}.
\begin{figure}[hbt]
\centering
\begin{tabular}{cc}
  \includegraphics [width=62mm, height=78mm]{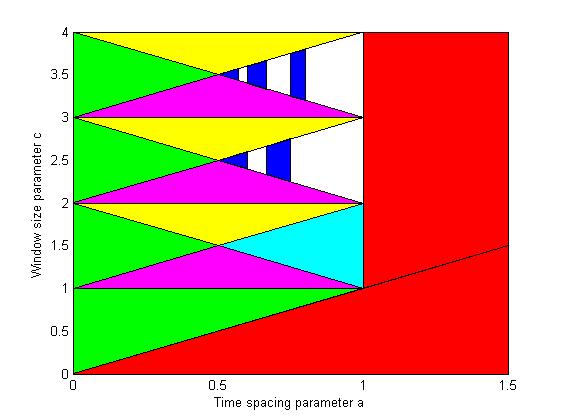} & \includegraphics [width=62mm, height=78mm]{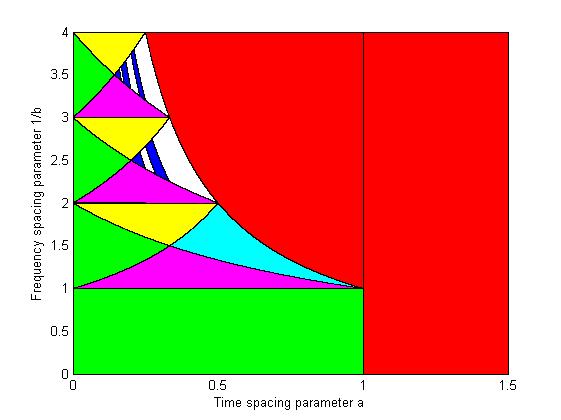} 
   \end{tabular}
   \caption{\small Left: Classification of   pairs  $(a,  c)$  such that ${\mathcal G}(\chi_{[0,c)}, a\Z\times \Z)$ are Gabor frames.
Right: Classification of pairs $(a, 1/b)$ such that ${\mathcal G}(\chi_{[0,1)}, a\Z\times \Z/b)$ are Gabor frames. 
}
\label{classification.fig}
\end{figure}
For pairs $(a, c)$ in the  red region of the left subfigure of Figure \ref{classification.fig},
${\mathcal G}(\chi_{[0,c)}, a\Z\times \Z)$
are not Gabor frames by 
Theorem \ref{newmaintheorem1},
while for pairs $(a, c)$   in  the  green, blue and dark blue regions ${\mathcal G}(\chi_{[0,c)}, a\Z\times \Z)$
are Gabor frames by  
 Theorems \ref{newmaintheorem1}--\ref{newmaintheorem3}.
In the  yellow  region, 
it follows from Conclusion (VI) of Theorem \ref{newmaintheorem2} that
the set of pairs $(a, c)$ such that ${\mathcal G}(\chi_{[0,c)}, a\Z\times \Z)$
 are not  Gabor frames contains needles (line segments) of  lengths  ${\rm gcd}( \lfloor c\rfloor+1, p)/q-\{0, 1/q\}$ hanging vertically from the ceiling $\lfloor c\rfloor+1$
at every rational time shift location $a=p/q$.
In the  purple region, 
  we obtain from  Conclusion (VII) of Theorem \ref{newmaintheorem2} that
the set of pairs $(a, c)$ such that ${\mathcal G}(\chi_{[0,c)}, a\Z\times \Z)$
are not  Gabor frames contains   floors $\lfloor c\rfloor\ge 2$ and also needles (line segments) of  lengths  ${\rm gcd}( \lfloor c\rfloor, p)/q-\{0, 1/q\}$
growing vertically from  floors $\lfloor c\rfloor$
at every rational time shift location $a=p/q$.
It has rather complicated geometry for the set of pairs $(a, c)$ in the white region such that ${\mathcal G}(\chi_{[0,c)}, a\Z\times \Z)$
are not  Gabor frames. That set contains some needles (line segments) on the vertical line  growing from
 rational time shift locations  and  few needle holes (points) on the vertical line located at
 irrational time shifts  by  
 Theorems \ref{newmaintheorem4} and \ref{newmaintheorem5}.
From the above observations, we see that the range of density parameters $a, b$ such that
 the Gabor system ${\mathcal G}(\chi_{[0, c)}, a\Z\times \Z/b)$
 are Gabor frames is neither open nor connected, and it
 has very puzzling   structure, c.f. the openness of the range of density parameters
 for   a window function in Feichtinger's algebra
\cite{feichtinger04}.

\smallskip

The paper is organized as follows.
In Section \ref{maintheorem.section}, we state our main theorems of this paper.
 The first two main theorems (Theorems \ref{newmaintheorem2} and \ref{newmaintheorem3}) are proved in
 Sections \ref{gabornullspace1.section} and \ref{gabornullspace2.section}
 respectively.
After studying various  properties of
the dynamic system associated with   
 some non-contractive and non-measure-preserving transformations
in  Sections \ref{gaborirrational.section},
we parameterize those triples $(a, b, c)$ of positive numbers such that
${\mathcal G}(\chi_{[0,c)}, a\Z\times \Z/b)$ are Gabor frames,
and finally establish our most challenging results (Theorems \ref{newmaintheorem4} and \ref{newmaintheorem5})
in  Sections \ref{gaborrational1.section} and \ref{rational.section}.
%
In Appendix \ref{algorithm.appendix},  we provide a finite-step algorithm  to
 verify whether the  Gabor system  ${\mathcal G}(\chi_{[0,c)}, a\Z\times \Z/b)$
is a  Gabor frame  for any given triple of $(a, b, c)$ of positive numbers.
 In Appendix \ref{sampling.appendix}, we apply our results on Gabor systems to
 identify all intervals $I$ and time-sampling spacing lattices $b\Z\times a\Z$ such that
 signals $f$ in the shift-invariant  space
$$V_2(\chi_I, b\Z)=\Big\{\sum_{\lambda\in b\Z}  d(\lambda) \chi_I(\cdot-\lambda): \ \sum_{\lambda\in b\Z} |d(\lambda)|^2<\infty\Big\}$$
can be stably recovered from their equally-spaced samples $f(t_0+\mu), \mu\in a\Z$, for any initial sampling position $t_0\in \R$.
In Appendix \ref{ergodic.appendix}, we discuss non-ergodicity of
 a new non-contractive and non-measure-preserving transformation. 


\smallskip

In this paper, we will use the following notation.
We denote  the set of rational numbers by $\Q$;
  the integral part of a
real number $t$ by $\lfloor t\rfloor$; the sign of a real number $t$ by ${\rm sgn}(t)$;
the greatest common divisor between
 $p$ and $q$ in a lattice $r\Z$ with $r>0$ by ${\rm gcd}(p,q)$;
the Lebesgue measure of a measurable set $E$  by $|E|$;
the transpose of a matrix (vector) ${\bf A}$ by ${\bf A}^T$;
 the null space of a matrix ${\bf A}$ by $N({\bf A})$;
the column  vector whose entries take value $r\in \R$ by
${\bf r}:=(\cdots, r, r, r, \cdots)^T$;
and
 the space of all
square-summable vectors ${\bf z}:=({\bf z}(\lambda))_{\lambda\in \Lambda}$ on a given index set $\Lambda$ by $\ell^2:=\ell^2(\Lambda)$,
 with its standard norm and inner product denoted by
$\|\cdot\|_{2}:=\|\cdot\|_{\ell^2(\Lambda)}$ and $\langle\cdot, \cdot\rangle:=\langle\cdot,\cdot\rangle_{\ell^2(\Lambda)}$
respectively.
Also for $b>0$, we let ${\mathcal B}_b:=\{ ({\bf x}(\lambda))_{\lambda\in b\Z}  \ |\ {\bf x}(\lambda)\in \{0,1\}\
{\rm  for\ all} \ \lambda\in b\Z\}$
consist of all {\em binary column vectors} whose components taking values  either zero or one, and
${\mathcal B}_b^0:=\{ ({\bf x}(\lambda))_{\lambda\in b\Z}\in {\mathcal B}_b\ |  \ {\bf x}(0)=1\}$
contain all binary vectors taking value one at the origin.

\section{Main Theorems}
\label{maintheorem.section}

In this section,  we present our answer to the  $abc$-problem for Gabor systems,
and  a confirmation to the conjecture
in \cite[Section 3.3.5]{janssen03}.



\smallskip

Let us start from recalling  some known classification  of triples $(a, b , c)$ of positive numbers, see for instance \cite{dgm86, hangu, janssen03}.

\begin{thm}\label{newmaintheorem1} 
Let $(a,b,c)$  be a triple of positive numbers, and let ${\mathcal G}(\chi_{[0,c)}, a\Z\times \Z/b)$ be the Gabor system in \eqref{gabor.def} generated by
the characteristic function on the interval $[0, c)$.  Then the following
statements hold.

\begin{itemize}
\item [{(I)}]
If $ a>c$, then ${\mathcal G}(\chi_{[0,c)}, a\Z\times \Z/b)$ is not a  Gabor frame.
\item[{(II)}] If $a=c$, then  ${\mathcal G}(\chi_{[0,c)}, a\Z\times \Z/b)$ is   a Gabor frame  if and only if  $a\le b$.

 \item [{(III)}] If $a<c$ and $b\le  a$, then
 ${\mathcal G}(\chi_{[0,c)}, a\Z\times \Z/b)$ is not a  Gabor  frame.

\item[{(IV)}] If $a<c$ and $b\ge c$, then ${\mathcal G}(\chi_{[0,c)}, a\Z\times \Z/b)$ is  a  Gabor  frame.


 \end{itemize}
 \end{thm}


The conclusions in Theorem \ref{newmaintheorem1} are illustrated in the  red and low right-triangle green regions of
Figure \ref{classification.fig}.

\smallskip

By Theorem \ref{newmaintheorem1}, it remains to classify triples $(a, b, c)$ of positive numbers such that $a<b<c$.
To do so,  for any given triple $(a, b, c)$ of positive numbers, we let
\begin{equation}\label{infinitematrix.def}
{\bf M}_{a, b, c}(t):=\big(\chi_{[0,c)}(t-\mu+\lambda)\big)_{\mu\in a\Z, \lambda\in b\Z}, \ t\in \R,
\end{equation}
 and we introduce a periodic set
 \begin{equation}\label{dabc2.def}
{\mathcal D}_{a, b, c}:=\big\{ t\in \R\big|\  {\bf M}_{a, b, c}(t) {\bf x}={\bf 2}\ {\rm for \ some\ binary\ vectors} \
{\bf x}\in {\mathcal B}_b^0\big\}\end{equation}
 that contains all
$t$ on the real line such that there exists a binary vector solution ${\bf x}\in {\mathcal B}_b^0$
 to the infinite-dimensional linear system
\begin{equation}\label{infinitelinearsystem.eq1}
 {\bf M}_{a, b, c}(t) {\bf x}={\bf 2}.\end{equation}
The uniform stability of  infinite matrices ${\bf M}_{a, b, c}(t), t\in \R$,
 in \eqref{infinitematrix.def},
\begin{equation}\label{framestabilitymatrix.tm.neweq1} 0<\inf_{t\in \R} \inf_{\|z\|_2=1} \|{\bf M}_{a, b, c}(t) {\bf z}\|_2\le
\sup_{t\in \R} \sup_{\|z\|_2=1} \|{\bf M}_{a, b, c}(t) {\bf z}\|_2<\infty,
\end{equation}
was used by Ron and Shen  \cite {RS97} to characterize
the  frame property  for the Gabor system ${\mathcal G}(\chi_{[0, c)}, a\Z\times \Z/b)$, see
 Lemma \ref{framestabilitymatrix.lem} in Section  \ref{gabornullspace1.section}.

\smallskip

We observe that infinite matrices ${\bf M}_{a, b, c}(t), t\in \R$,  in \eqref{infinitematrix.def}
have their rows containing
 $\lfloor c/b\rfloor+\{0, 1\}$  consecutive ones, and
their rows are  obtained by shifting one (or zero) unit of the previous row with possible reduction or expansion by one unit.
In the case that  $a/b$ is rational,  infinite matrices in \eqref{infinitematrix.def} have certain shift-invariance in the sense that
 their $(\mu+qa)$-th row  can be obtained by shifting $p$-units of the $\mu$-th row where $p$ and $q$
 are coprime integers satisfying $a/b=p/q$, c.f. \cite[Eq. 3.3.68]{janssen03}.
 The above observations could be illustrated from  examples below:
\begin{equation}\label{irrationalexample1.infinitematrix.eq1}
{\small {\bf M}_{a, b, c}(0)=\left(\begin{array}{ccccccccccccccccccc}
 \ddots& \vdots& \vdots& \vdots & \vdots & \vdots & \vdots & \vdots  \\
 &  0 & 1 & 1& 1 & 1  & 1& 1 & 0 \\
 & & 0 & 1 & 1& 1&  1& 1 & 1& 0 \\
 & & & 0 & 1 & 1& 1&  1 & 1 & 1& 0 \\
 &  & & & 0 & 1 & 1& 1& 1& 1& 0\\
 & & &&   0& 1 & 1& 1& 1& 1 & 1& 0 & \\
 & & & &  & 0 & 1 & 1  & 1& 1 & 1& 1& 0 & \\
 & & & && & 0 & 1 & 1& 1& 1& 1 & 1  & 0 &   \\
 & & & && & &0 & 1 & 1& 1& 1& 1  & 1&0 &  \\
 & & & && & & & 0 & 1 & 1& 1& 1  & 1& 0 &  \\
 & & & && & & & 0 & 1 & 1& 1& 1 & 1 & 1& 0 &  \\
 & & & && &  & &  &  \vdots& \vdots& \vdots& \vdots& \vdots & \vdots & \vdots&\ddots
\end{array}\right) } \end{equation}
for the triple $(a, b, c)= (\pi/4, 1, 23-11\pi/2)$ with irrational ratio between $a$ and $b$; and
\begin{equation}
\label{rationalexample1.infinitematrix.eq1} {\tiny  {\bf M}_{a, b, c}(0)=\left(\begin{array}{ccccccccccccccccccc}
 \ddots& \vdots& \vdots& \vdots & \vdots & \vdots & \vdots & \vdots \\
 & 0& 1 & 1& 1   & 1& 1&  0  &  \\
& &  0 & 1 & 1& 1  & 1 & 1& 0 \\
& & & 0 & 1 & 1& 1& 1& 1& 0 \\
& & & & 0 & 1 & 1& 1 &1 & 0 \\
& &  & & & 0 & 1 & 1&  1& 1& 0\\
& & & & &0&  1 & 1& 1& 1&  1& 0 & \\
& & & &&  &  0& 1 & 1& 1 & 1& 1& 0 & \\
& & & & &  & & 0 & 1 & 1  &1& 1& 0 & \\
& & & & && & &  0 & 1 & 1& 1& 1  & 0 &   \\
& & & & && & &0 & 1 & 1& 1&  1  & 1 & 0 &  \\
& & & & && & & & 0 & 1 & 1&  1 & 1& 1 &  0 &  \\
& & & & && & & & & 0 & 1 & 1&  1 & 1 &  0 &  \\
& & & & && & &  & & & 0 & 1 & 1&  1 & 1 &  0 &  \\
& & & & && & &  & & & 0 & 1 & 1&  1 & 1 & 1& 0 &  \\
& & & & &&  & & &  & &  &  \vdots& \vdots& \vdots& \vdots& \vdots & \vdots & \ddots
\end{array}\right) } \end{equation}
for the triple $(a, b, c)=(13/17, 1,  77/17)$ with
rational ratio between $a$ and $b$.
%
Due to the above special structure of the binary  infinite matrices
  in \eqref{infinitematrix.def}, 
 we may reduce their uniform stability \eqref{framestabilitymatrix.tm.neweq1}
 to the non-existence of trinary vectors in their null spaces, 
 and even further to the non-existence of
  binary vector solutions   of 
  the infinite-dimensional  linear systems \eqref{infinitelinearsystem.eq1}:
%
%
\begin{equation}\label{gabordabcequivalence}
{\mathcal G}(\chi_{[0,c)}, a\Z\times \Z/b) \ { is \ a \ Gabor\ frame\ if \ and \ only \ if} \ {\mathcal D}_{a, b, c}
 =\emptyset,
 \end{equation}
 see
  Theorem \ref{framenullspace1.tm} in Section \ref{gabornullspace1.section}.
 The necessity of the above equivalence 
 has been  applied implicitly in \cite{hangu, janssen03} for their  classifications.
   We apply the above equivalence  \eqref{gabordabcequivalence} to take one step forward in the  way  to solve the $abc$-problem for Gabor systems.


\begin{thm}\label{newmaintheorem2}
Let $(a,b,c)$  be a triple of positive numbers with $a<b<c$,
 and let ${\mathcal G}(\chi_{[0,c)}, a\Z\times \Z/b)$ be the Gabor system in \eqref{gabor.def} generated by
the characteristic function on the interval $[0, c)$. Set
$$c_0=c-\lfloor c/b\rfloor b.$$
 Then the following
statements hold.

\begin{itemize}

\item[{(V)}] {\rm (\cite{janssen03})} \ If $c_0 \ge a$ and
 $ c_0\le b-a$, then ${\mathcal G}(\chi_{[0,c)}, a\Z\times \Z/b)$ is  a Gabor frame.

 \item[{(VI)}] If
 $c_0 \ge a$ and  $c_0> b-a$,  then ${\mathcal G}(\chi_{[0,c)}, a\Z\times \Z/b)$ is not a  Gabor frame
 if and only if $a/b=p/q$ for some coprime integers, and either
\begin{itemize}
 \item [{1)}] $c_0>b-{\rm gcd}( \lfloor c/b\rfloor+1, p)b/q$ and ${\rm gcd}
 ( \lfloor c/b\rfloor+1, p)\ne \lfloor c/b\rfloor+1$, or
\item [{2)}]  $c_0>b-{\rm gcd}( \lfloor c/b\rfloor+1, p)b/q+b/q$ and
 ${\rm gcd}( \lfloor c/b\rfloor+1, p)= \lfloor c/b\rfloor+1$.
\end{itemize}

\item[{(VII)}] If
 $c_0< a$ and  $c_0\le  b-a$,  then ${\mathcal G}(\chi_{[0,c)}, a\Z\times \Z/b)$ is not a  Gabor frame
 if and only if either
\begin{itemize}
\item[{3)}]  $c_0=0$; or
 \item [{4)}] $a/b=p/q$ for some coprime integers $p$ and $q$,
   $0<c_0< {\rm gcd}(\lfloor c/b\rfloor, p)b/q$  and ${\rm gcd}(\lfloor c/b\rfloor, p)\ne \lfloor c/b\rfloor$; or

 \item [{5)}] $a/b=p/q$ for some coprime integers $p$ and $q$,
 $0<c_0< {\rm gcd}(\lfloor c/b\rfloor, p) b/q-b/q$  and ${\rm gcd}(\lfloor c/b\rfloor, p)= \lfloor c/b\rfloor$.

\end{itemize}

\end{itemize}

\end{thm}

The conclusions in Theorem \ref{newmaintheorem2} are illustrated in the green, yellow and purple regions of
Figure \ref{classification.fig}.

\smallskip

 By Theorems \ref{newmaintheorem1} and \ref{newmaintheorem2},  we now classify those triples $(a, b, c)$ satisfying
 $a<b<c$ and $b-a<c_0:=c-\lfloor c/b\rfloor b<a$, which turns out to be very complicated.
Our approach is not to solve the  linear system \eqref{infinitelinearsystem.eq1} directly, and instead
to find  binary vector solutions ${\bf x}\in {\mathcal B}_b^0$
 of  a ``similar" infinite-dimensional
 linear system
 \begin{equation}\label{infinitelinearsystem.eq2}
{\bf M}_{a, b, c}(t) {\bf x}={\bf 1}. 
\end{equation}
Denote by ${\mathcal S}_{a, b, c}$ the periodic set of all $t\in \R$ such that
there exists a binary vector solution to the 
linear system
\eqref{infinitelinearsystem.eq2},
\begin{equation}\label{dabc1.def}
{\mathcal S}_{a, b, c}:=\big\{ t\in \R\big|\ {\bf M}_{a, b, c}(t) {\bf x}={\bf 1}\ {\rm for \ some\ vector} \
{\bf x}\in {\mathcal B}_b^0\big\}.
\end{equation}
The set ${\mathcal S}_{a, b, c}$  just introduced  is a supset of   ${\mathcal D}_{a, b, c}$
in \eqref{dabc2.def}, and the set ${\mathcal D}_{a, b, c}$ can be obtained from ${\mathcal S}_{a, b, c}$ by some set operations,
\begin{eqnarray*} {\mathcal D}_{a, b, c}&= &  \big({\mathcal S}_{a, b, c}\cap ([0, c_0+a-b)+a\Z)\cap ( {\mathcal S}_{a, b, c}-\lfloor c/b\rfloor b)\big)\nonumber\\
& & \cup \big({\mathcal S}_{a, b, c}\cap (\cup_{\lambda\in [b, (\lfloor c/b\rfloor -1) b]\cap b\Z} ({\mathcal S}_{a, b, c}-\lambda))\big),
\end{eqnarray*}
see  \eqref{dabc2subsetdabc1} and Theorem \ref{dabcsabc.thm}.
Hence  the classification of
triples $(a, b, c)$ of positive numbers such that
the Gabor system ${\mathcal G}(\chi_{[0,c)}, a\Z\times \Z/b)$ is a Gabor frame
reduces further to
characterizing
\begin{equation}\label{sabcemptynonempty1} {\rm  i)}\quad  {\mathcal S}_{a, b, c}=\emptyset,\end{equation}
and
\begin{equation}\label{sabcemptynonempty2}  \ {\rm   ii)} \quad
{\mathcal S}_{a, b, c}\ne\emptyset \ {\rm  but} \ {\mathcal D}_{a, b, c}=\emptyset.
\end{equation}

\smallskip

An  advantage of the above approach
 is that
  there is a {\em unique} solution ${\bf x}_t\in {\mathcal B}_b^0$ to the infinite-dimensional linear system
${\bf M}_{a,b,c}(t)
 {\bf x}_t={\bf 1}$ for $t\in {\mathcal S}_{a, b, c}$, while
 multiple binary vector solutions could exist for the  linear system
\eqref{infinitelinearsystem.eq1}
 for $t\in {\mathcal D}_{a, b, c}$, see Proposition \ref{uniqueness.cor}.
Denote by  $\lambda_{a,b,c}(t)$ the smallest positive index  in
$a\Z$ such that ${\bf x}_t(\lambda_{a,b,c}(t))=1$. This yields the
following one-to-one map on the set ${\mathcal S}_{a, b, c}$:
\begin{equation}\label{dabc1onetoonemap}
 {\mathcal S}_{a, b, c} \ni t\longleftrightarrow t+\lambda_{a,b,c}(t)\in {\mathcal S}_{a, b, c},
\end{equation}
because $\tau_{\lambda_{a,b,c}(t)} {\bf x}_t\in {\mathcal B}_b^0$ and
$${\bf M}_{a, b,c}
(t+\lambda_{a,b,c}(t))\tau_{\lambda_{a,b,c}(t)} {\bf x}_t
={\bf M}_{a, b,c} (t) {\bf x}_t={\bf 1},$$
where $\tau_{\lambda'}{\bf z}:= ({\bf z}(\lambda+\lambda'))_{\lambda\in b\Z}$
for $  {\bf z}:=({\bf z}(\lambda))_{\lambda\in b\Z}$.
Further inspection shows that the maps $t\rightarrow t+\lambda_{a,b,c}(t)$ and
 $ t+\lambda_{a,b,c}(t)\rightarrow t$ on ${\mathcal S}_{a, b, c}$
can be extended to  piecewise linear transformations  $R_{a, b,c}$ and $\tilde R_{a, b, c}$ on the real line $\R$.
Here for a given triple $(a,b,c)$ of positive numbers satisfying $a<b<c$ and
$b-a<c_0:=c-\lfloor c/b\rfloor b<a$, we define  piecewise linear transformations
$R_{a,b,c}$ and $\tilde R_{a,b,c}$ on the real line $\R$  by
\begin{equation}\label{rabcnewplus.def}
R_{a,b,c}(t):=\left\{\begin{array} {ll}  t+ \lfloor c/b\rfloor b +b &
{\rm if} \ t\in [0, c_0+a-b)+a\Z\\
t &
{\rm if} \ t\in [c_0+a-b, c_0)+a\Z\\
t+\lfloor c/b\rfloor b &
{\rm if} \ t\in [c_0, a)+a\Z,\end{array}\right.
\end{equation}
and
\begin{equation}\label{tilderabcnewplus.def}
\tilde R_{a,b,c}(t):=\left\{\begin{array} {ll}t-\lfloor c/b\rfloor b &
{\rm if} \ t\in [c-a, c-c_0)+a\Z\\
t &
{\rm if} \ t\in [c- c_0, c+b-c_0-a)+a\Z\\
  t-\lfloor c/b\rfloor b -b &
{\rm if} \ t\in [c+b-c_0-a, c)+a\Z.
\end{array}\right.
\end{equation}
 Our extremely important observation is that  ${\mathcal S}_{a, b, c}$  is  the {\em   maximal invariant set}
  under the piecewise linear transformations $R_{a, b, c}$
   and $\tilde R_{a, b, c}$,
   $$ 
 R_{a, b, c} {\mathcal S}_{a, b, c}\subset {\mathcal S}_{a, b, c} \quad {\rm and}\quad
 \tilde R_{a, b, c}  {\mathcal S}_{a, b, c}\subset {\mathcal S}_{a, b, c},$$
 that  has empty intersection with their black holes
  $[c_0+a-b, c_0)+a\Z$ and $ [c- c_0, c-c_0+b-a)+a\Z$, see 
 Theorem \ref{dabc1maximal.cor.thm}.

\smallskip

  Applying the above  maximal invariant set property for  ${\mathcal S}_{a, b, c}$, we  take another step forward in the direction to solve the $abc$-problem for Gabor systems.

\begin{thm}\label{newmaintheorem3}
Let $(a,b,c)$  be a triple of positive numbers with $a<b<c$ and $b-a<c_0:=c-\lfloor c/b\rfloor b<a$,
 and let ${\mathcal G}(\chi_{[0,c)}, a\Z\times \Z/b)$ be the Gabor system in \eqref{gabor.def} generated by
the characteristic function on the interval $[0, c)$. Set
$$ c_1=\lfloor c/b\rfloor b
-\lfloor (\lfloor c/b\rfloor b/a)\rfloor a.$$ Then the following
statements hold.

\begin{itemize}

\item[{(VIII)}] {\em (\cite{hangu, janssen03})} \ If $\lfloor c/b\rfloor=1$,  then
${\mathcal G}(\chi_{[0,c)}, a\Z\times \Z/b)$ is   a Gabor frame.

\item[{(IX)}] If   $\lfloor c/b\rfloor\ge 2$ and
 $c_1>2a-b$, then ${\mathcal G}(\chi_{[0,c)}, a\Z\times \Z/b)$ is a Gabor frame.

\item[{(X)}] If  $\lfloor c/b\rfloor\ge 2$ and
$c_1=2a-b$, then ${\mathcal G}(\chi_{[0,c)}, a\Z\times \Z/b)$ is  a Gabor frame if and only if
$a/b=p/q$ for some coprime integers $p$ and $q$,
$c_0\le b-a+b/q$ 
and $\lfloor c/b\rfloor+1=p$.

\item[{(XI)}] If $\lfloor c/b\rfloor\ge 2$ and
$c_1=0$, then ${\mathcal G}(\chi_{[0,c)}, a\Z\times \Z/b)$ is a Gabor frame if and only if $a/b=p/q$ for some coprime integers $p$ and $q$,
$c_0\ge a-b/q$ 
and
$\lfloor c/b\rfloor=p$.

\end{itemize}

\end{thm}

The conclusions in Theorem \ref{newmaintheorem3} are illustrated in the blue and dark blue regions of
Figure \ref{classification.fig}.

By Theorems \ref{newmaintheorem1}, \ref{newmaintheorem2} and \ref{newmaintheorem3}, it remains to classify
all triples $(a, b, c)$ of positive numbers satisfying
$a<b<c, \lfloor c/b\rfloor\ge 2, b-a<c_0:=c-\lfloor c/b\rfloor b<a$ and $0<c_1:=\lfloor c/b\rfloor b
-\lfloor (\lfloor c/b\rfloor b/a)\rfloor a<2b-a$. 
For that purpose, we  need explicit construction of the maximal invariant set ${\mathcal S}_{a, b, c}$
  if it is nonempty.
  We observe that  Hutchinson's remarkable construction
\cite{hutchinson81} does not apply as  
piecewise linear transformations $R_{a, b, c}$ and $\tilde R_{a,b,c}$ are {\em non-contractive}.
On the other hand,  from its maximal invariance, we obtain that
for the triple $(a, b, c)=(\pi/4, 1, 23-11\pi/2)$ with irrational time-frequency-spacing ratio $b/a$,
the maximal invariant set
$${\mathcal S}_{a, b, c}= \big[ 18-\frac{23\pi}{4},11-\frac{7\pi}{2}\big) \cup \big[ 12-\frac{15\pi}{4} ,5-
\frac{3\pi}{2} \big) \cup \big[ 6-\frac{7\pi}{4},  17-\frac{21\pi}{4}\big)  +\frac{\pi}{4} \Bbb{Z}
$$
 has its complement consisting of three holes of size $1-\pi/4$ on one period,
and that  for the triple $(a, b, c)=(13/17, 1, 77/17)$ with rational time-frequency-spacing ratio $b/a$,
$$ \mathcal{S}_{a,b,c}  =  \big[\frac{2}{17}, \frac{3}{17}\big)\cup\big [\frac{9}{17}, \frac{10}{17}\big)\cup\big [\frac{12}{17}, \frac{13}{17}\big)+\frac{13}{17}\Z$$
is composed of three intervals of same length $1/17$ on the period $[0, 13/17)$, see Examples \ref{irrationalexample1} and \ref{rationalexample1}.
A breakthrough of this paper is to show that if ${\mathcal S}_{a, b, c}\ne \emptyset$ then the black hole $[c_0+a-b, c_0)+a\Z$  of the transformation $R_{a, b, c}$
attracts the black hole $[c-c_0, c-c_0+b-a)+a\Z$ of the other transformation $\tilde R_{a,b,c}$
when applying $R_{a, b, c}$ {\em finitely many times}, i.e.,
$$(R_{a, b, c})^L ([c-c_0, c-c_0+b-a)+a\Z)=[c_0+a-b, c_0)+a\Z
$$
for some nonnegative integer $L$, and hence
\begin{equation}\label{sabc.complement}
{\mathcal S}_{a, b, c}=\R\backslash \big(\cup_{n=0}^L (R_{a, b, c})^n([c-c_0, c-c_0+b-a)+a\Z\big),
\end{equation}
see Theorems \ref{dabc1holes.tm} and  \ref{dabc1discreteholes.tm}, and Examples \ref{irrationalexample1}, \ref{rationalexample1} and \ref{rationalexample2}.
The above construction of  the maximal invariant set ${\mathcal S}_{a, b, c}$
 leads to 
 the following characterization of \eqref{sabcemptynonempty2}:
{\em ${\mathcal S}_{a, b, c}\ne\emptyset$  but  ${\mathcal D}_{a, b, c}=\emptyset$ if and only if
$ (\lfloor b/c\rfloor+1)|{\mathcal S}_{a, b, c}\cap [0, c_0+a-b)|+
\lfloor c/b\rfloor |{\mathcal S}_{a, b, c}\cap [c_0, a)|=a$,} see Theorem \ref{sabcstar.tm}.

\smallskip

Now it remains to discuss most challenging characterization of \eqref{sabcemptynonempty1}: classifying all triples $(a, b, c)$ such that ${\mathcal S}_{a, b, c}\ne \emptyset$.
By \eqref{sabc.complement},  the maximal invariant set ${\mathcal S}_{a, b, c}$
has its complement composed by finitely many holes on a period.
So we may squeeze out those holes on the line  and then reconnect their  endpoints.
The above  holes-removal
 surgery  could be described by the map
  \begin{equation}\label{xabc.def}
Y_{a,b,c}(t):= {\rm sgn}(t) |[\min(0,t), \max(0, t))\cap {\mathcal S}_{a, b, c}|
\end{equation}
on the line  in the sense that
%
it is an isomorphism from the set ${\mathcal S}_{a, b, c}$ to the {\em line
with marks} (image of the holes).
In Figure \ref{holesremoval.fig} below, we illustrate
the performance of
the holes-removal surgery
via
$$a\T\ni  a \exp(2\pi i t/a)\longmapsto Y_{a, b, c}(a)
\exp\big(-2\pi i Y_{a,b,c}(t)/Y_{a,b,c}(a)\big)\in Y_{a, b, c}(a)\T,$$
where 
 $\big(\frac{\pi}{4}, 1, 23-\frac{11\pi}{2}\big)$, 
 $\big(\frac{6}{7}, 1, \frac{23}{7}\big)$,
  $\big(\frac{13}{17}, 1, \frac{77}{17}\big)$ and $\big(\frac{13}{17}, 1, \frac{75}{17}\big)$
 are used as  triples $(a, b, c)$ in the four subfigures respectively, c.f. Examples \ref{irrationalexample1}, \ref{rationalexample1} and \ref{rationalexample2}.
In  that figure, 
the  blue arcs  in the big circle  are contained in $a\exp(2\pi i {\mathcal S}_{a, b, c}/a)$,
 the red dashed arcs in the big circle belong to $a\exp(2\pi i (\R\backslash {\mathcal S}_{a, b, c})/a)$,
 and the
circled marks in the small circle are   $Y_{a, b, c}(a)
\exp\big(2\pi i K_{a, b, c}/Y_{a,b,c}(a)\big)$, where $K_{a, b, c}$ is the set of all marks on the line.
\begin{figure}[hbt]
\centering
\begin{tabular}{cc}
  \includegraphics [width=60mm, height=60mm]{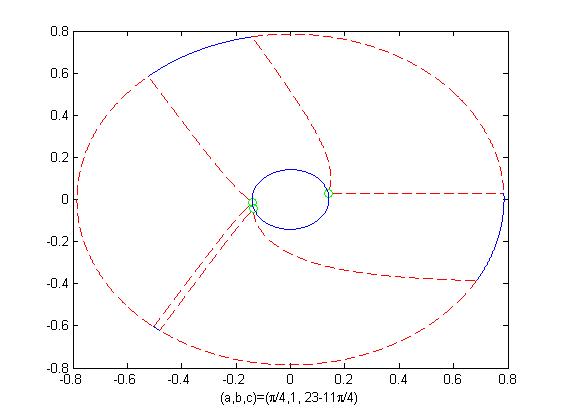} &
    \includegraphics [width=60mm, height=60mm]{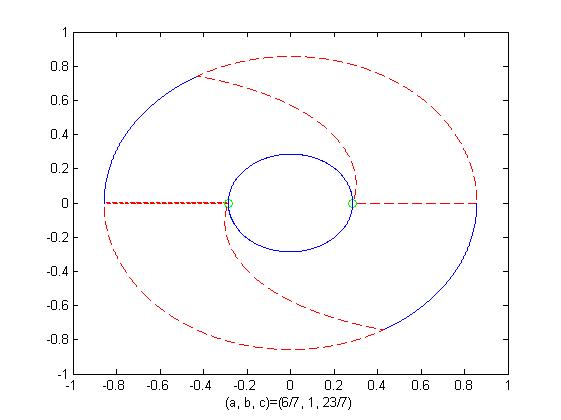}\\
  \includegraphics [width=60mm, height=60mm]{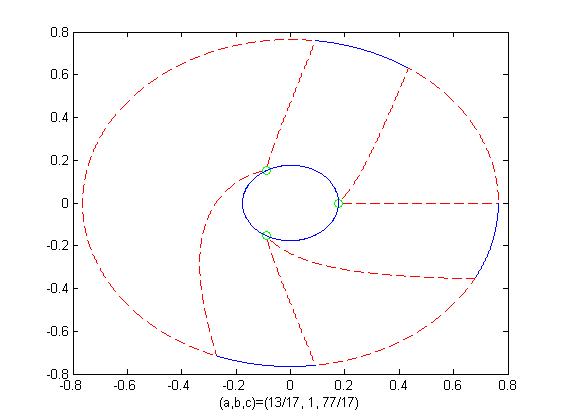} &
    \includegraphics [width=60mm, height=60mm]{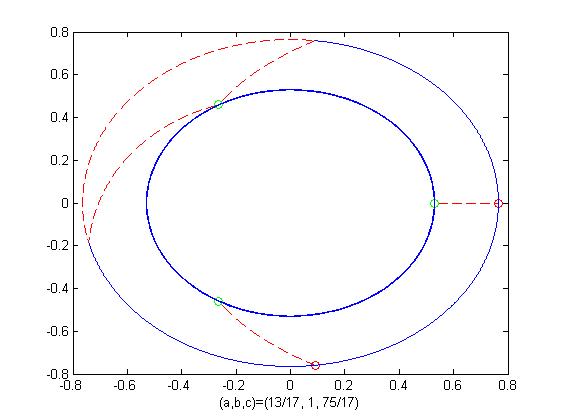}    \\
   \end{tabular}
\caption{\small  Holes-removal surgery for the maximal invariant set ${\mathcal S}_{a, b, c}$, see Examples \ref{irrationalexample1},
\ref{rationalexample1} and \ref{rationalexample2}. 
}
\label{holesremoval.fig}
\end{figure}

\smallskip

Having  the maximal invariant set ${\mathcal S}_{a, b, c}$  constructed in \eqref{sabc.complement} and holes-removal surgery described by
the map in \eqref{xabc.def}, we next consider
dynamic system of piecewise linear transformations $R_{a, b, c}$ and $\tilde R_{a,b,c}$.
We observe that those piecewise linear transforms are not {\em measure-preserving} on the
real line,
but their restrictions on  the maximal invariant set ${\mathcal S}_{a, b, c}$ are measure-preserving and one is the inverse of the other.
More importantly, we show that there is a rotation $S(\theta_{a, b, c})$ of the circle  $\R/(Y_{a, b, c}(a)\Z)$,
\begin{equation}\label{sthetaabc.def}
S(\theta_{a, b,c}) (z+Y_{a, b, c}(a)\Z)=\theta_{a, b,c}+
z+Y_{a, b, c}(a)\Z,\quad  z\in \R/(Y_{a, b, c}(a)\Z),\end{equation}
 such that 
 the following diagram commutes,
\begin{equation}\label{diagramcommutes} \begin{CD}
{\mathcal S}_{a, b, c} @>R_{a, b, c} >> {\mathcal S}_{a, b, c} \\
@V{Y_{a, b,c}}VV @VV{ Y_{a,b,c}}V\\
\R/(Y_{a, b, c}(a)\Z) @>>S(\theta_{a, b,c})>  \R/(Y_{a, b, c}(a)\Z)
\end{CD}
\end{equation}
see Theorem \ref{dabc1algebra.tm}.  In other words, the restriction of
piecewise linear transformations $R_{a, b, c}$ and $\tilde R_{a,b,c}$ on the maximal invariant set ${\mathcal S}_{a, b, c}$ can be thought as
{\em rotations on the circle},
see Appendix \ref{ergodic.appendix} for
the non-ergodicity  of  the piecewise linear transformation $R_{a, b, c}$
on the real line.

\smallskip
Having ``rotation" property of the piecewise linear transformation $R_{a, b, c}$,
we are almost ready to find  all triples $(a, b, c)$ such that ${\mathcal S}_{a, b, c}\ne\emptyset$.
We observe that the hole-removal surgery is {\em reversible}, that is,
the maximal invariant set ${\mathcal S}_{a, b, c}$ can be obtained from the real line by putting marks at appropriate positions
 and then inserting holes of appropriate sizes at
marked positions. But that augmentation operation  is much more delicate and complicated than the hole-removal surgery.
%

\smallskip

For the  irrational time-frequency lattice case (i.e., $a/b\not\in \Q$),
 holes to be inserted  should have the same size $b-a$ (c.f. the upper-left subfigure of Figure \ref{holesremoval.fig}) and
the location of marks could be parameterized by 
 the numbers of holes contained in the intervals $[0, c_0+a-b)$ and $[c_0, a)$ respectively,
  see  Theorem \ref{dabc1characterization.tm}.
 This  leads to a parametric characterization of the statement  ${\mathcal S}_{a, b, c}\ne \emptyset$, and
 the following classification of triples $(a, b, c)$  the  irrational time-frequency lattice case. 



\begin{thm}\label{newmaintheorem4}
Let $(a,b,c)$  be a triple of positive numbers
such that $a<b<c,   b-a<c_0:=c-\lfloor c/b\rfloor b <a, \lfloor c/b\rfloor\ge 2$,
$0<c_1:=\lfloor c/b\rfloor b
-\lfloor (\lfloor c/b\rfloor b/a)\rfloor a<2a-b$. Let ${\mathcal G}(\chi_{[0,c)}, a\Z\times \Z/b)$ be the Gabor system in \eqref{gabor.def} generated by
the characteristic function on the interval $[0, c)$.
\begin{itemize}

\item[{(XII)}] If   $a/b\not\in \Q$, then ${\mathcal G}(\chi_{[0,c)}, a\Z\times \Z/b)$ is not a Gabor frame if and only if there exist
nonnegative integers $d_1$ and $d_2$ such that (a)\ $ a\ne
c-(d_1+1)(\lfloor c/b\rfloor+1) (b-a)-(d_2+1) \lfloor c/b\rfloor
(b-a)\in a\Z$; (b)\ $\lfloor c/b\rfloor b+(d_1+1)(b-a)<c<\lfloor
c/b\rfloor b+b- (d_2+1) (b-a)$; and (c)\ $\# E_{a, b, c}=d_1$,
where
 $ m= ((d_1+d_2+1) c_1- c_0+(d_1+1)(b-a))/a $ and
\begin{eqnarray} \label{eabc.def} E_{a, b, c} & = & \big\{ n\in [1, d_1+d_2+1]\big|
 \ n (c_1 -m(b-a))\nonumber\\
&& \quad \in [0, c_0-(d_1+1)(b-a)) +(a-(d_1+d_2+1)(b-a))\Z\big\}.\end{eqnarray}

\end{itemize}

\end{thm}

The conclusion of  Theorem \ref{newmaintheorem4} is illustrated in the white region of Figure \ref{classification.fig}.
In the above theorem, we insert  $d_1$ and $d_2$ holes contained
in the  intervals $[0, c_0+a-b)$ and $[c_0, a)$ respectively,
and put marks at
$\cup_{n=1}^{d_1+d_2+1} (n (c_1-m(b-a))+ (a-(d_1+d_2+1)(b-a))\Z)$.

\smallskip

For the  rational time-frequency lattice case (i.e., $a/b\in \Q$),
 we write
$a/b=p/q$ for some coprime integers $p$ and $q$. Recall that
for $c\not\in b\Z/q$, ${\mathcal G}(\chi_{[0,c)}, a\Z\times \Z/b)$ is a  Gabor frame
 if and only if both ${\mathcal G}(\chi_{[0,\lfloor qc/b\rfloor b/q)}, a\Z\times \Z/b)$ and
${\mathcal G}(\chi_{[0,\lfloor qc/b+1\rfloor b/q)}, a\Z\times \Z/b) $ are  Gabor frames
\cite[Section 3.3.6.1]{janssen03}. Then it suffices to consider
$a/b=p/q$ and $c/b\in \Z/q$ for some coprime integers $p$ and $q$.
In that case, we show that  marks on the line are equally spaced and gaps to be inserted  have two different sizes,
which could be parameterized by
 the numbers of gaps of large and small sizes  contained in intervals $[0, c_0+a-b)$ and $[c_0, a)$, see Theorem \ref{dabc1discretecharacterization.tm}
 and c.f. Figure \ref{holesremoval.fig}.
   Applying the characterization  for ${\mathcal S}_{a, b, c}\ne \emptyset$
   in Theorem \ref{dabc1discretecharacterization.tm}, we
reach the last step to solve the $abc$-problem for Gabor systems.

\begin{thm}\label{newmaintheorem5}
Let $(a,b,c)$  be a triple of positive numbers
such that $a<b<c,   b-a<c_0 <a, \lfloor c/b\rfloor\ge 2$ and
$0<c_1<2a-b$, where  $c_0=c-\lfloor c/b\rfloor b \ \ {\rm and} \ \  c_1=\lfloor c/b\rfloor b
-\lfloor (\lfloor c/b\rfloor b/a)\rfloor a$. Let ${\mathcal G}(\chi_{[0,c)}, a\Z\times \Z/b)$ be the Gabor system in \eqref{gabor.def} generated by
the characteristic function on the interval $[0, c)$.

\begin{itemize}

\item[{(XIII)}] If    $a/b=p/q$ for some coprime
integers $p$ and $q$, and $c\in b\Z/q$,
 then ${\mathcal G}(\chi_{[0,c)}, a\Z\times \Z/b)$ is not a  Gabor frame if and only if
the triple $(a, b ,c)$ satisfies one of the following three conditions:
\begin{itemize}
\item[{6)}] $c_0<{\rm gcd}(a, c_1)$ and $\lfloor c/b\rfloor ({\rm gcd}(a, c_1)-c_0)\ne  {\rm gcd}(a, c_1)$.

\item [{7)}] $b-c_0< {\rm gcd}(a, c_1+b)$ and
$ (\lfloor c/b\rfloor +1) ({\rm gcd}(a, c_1+b)+c_0-b) \ne {\rm gcd}(a, c_1+b)$.

\item[{8)}]
 There exist nonnegative integers $d_1, d_2, d_3, d_4$ such that
  (a) $0< a-(d_1+d_2+1)(b-a)\in N b\Z/q$; (b)
$Nc_1+(d_1+d_3+1)(b-a)\in a \Z$;  (c)
$ (d_1+d_2+1)(N c_1+(d_1+d_3+1)(b-a))-(d_1+d_3+1)a\in N a\Z$; (d)
${\rm gcd}(Nc_1+(d_1+d_3+1)(b-a), Na)=a$; (e)
$\#E_{a,b,c}^d= d_1$; (f)
$c_0=(d_1+1)(b-a)+(d_1+d_3+1)(a-(d_1+d_2+1)(b-a))/N+\delta$
for some $-\min(a-c_0, (a-(d_1+d_2+1)(b-a))/N)<\delta<\min(c_0+a-b, (a-(d_1+d_2+1)(b-a))/N)$;
and  (g) $|\delta|+ a/(N\lfloor c/b\rfloor+(d_1+d_3+1))\ne (a-(d_1+d_2+1)(a-b))/N$,
where $N:=d_1+d_2+d_3+d_4+2$
 and $E_{a, b, c}^d$ is defined by
\begin{eqnarray} \label{eabcd.def} E_{a, b, c}^d & = & \big\{ n\in [1, d_1+d_2+1]\big|
 \ n (N c_1+(d_1+d_3+1)(b-a))\nonumber\\
&& \quad \in (0, (d_1+d_3+1)a) +Na \Z\big\}.\end{eqnarray}
\end{itemize}

\item[{(XIV)}] {\rm (\cite{janssen03})}\ If $a<b<c,   b-a<c_0 <a, \lfloor c/b\rfloor\ge 2$,
$0< c_1< 2a-b$, $a/b=p/q$  for some coprime
integers $p$ and $q$, and $c\not\in b\Z/q$,
then  ${\mathcal G}(\chi_{[0,c)}, a\Z\times \Z/b)$ is a  Gabor frame
 if and only if both ${\mathcal G}(\chi_{[0,\lfloor qc/b\rfloor b/q)}, a\Z\times \Z/b)$ and
${\mathcal G}(\chi_{[0,\lfloor qc/b+1\rfloor b/q)}, a\Z\times \Z/b) $ are  Gabor frames.

\end{itemize}

\end{thm}

The conclusions in the above theorem is illustrated in the white region of Figure \ref{classification.fig}.
In Case 6) of Conclusion (XIII) in Theorem \ref{newmaintheorem5},
the set $K_{a, b, c}$ of marks is $({\rm gcd}(a, c_1)-c_0)\Z$  
and   gaps inserted at  marked positions have same length $c_0$. In Case 7)
of Conclusion (XIII) in Theorem \ref{newmaintheorem5},
 $K_{a, b, c}=({\rm gcd}(a, c_1+b)+c_0-b)\Z$
   and
  gaps  inserted at  marks in  $K_{a, b, c}$ are of size $b-c_0$.
  In Case 8) of Conclusion (XIII) in Theorem \ref{newmaintheorem5},
  $K_{a, b, c}= h\Z, Y_{a, b, c}(a)=Nh$
  and   gaps inserted at marked positions $lmh+Nh\Z, 1\le l\le N$, have size  $|b-a|+|\delta|$
 for $1\le l\le d_1+d_2+1$ and $|\delta|$ for $d_1+d_2+2\le l\le N$,
    where $N=d_1+d_2+d_3+d_4+2, h=(a-(d_1+d_2+1)(b-a))/N-|\delta|, m=(Nc_1+(d_1+d_3+1)(b-a))/a$ and
    $\delta=c_0-(d_1+1)(b-a)-(d_1+d_3+1)(a-(d_1+d_2+1)(b-a))/N$.

\smallskip
Combining Theorems \ref{newmaintheorem1}--\ref{newmaintheorem5} gives a complete answer to the   $abc$-problem for Gabor systems.
The classification diagram of triples $(a, b, c)$  in Theorems \ref{newmaintheorem1}--\ref{newmaintheorem5} is presented below:
$$ {\tiny
\begin{picture}(280, 180) 
\put(0, 150){\line(0, -1){50}}
\put(0, 150){\line(1,0){5}}
\put(0, 100){\line(1,0){5}}
\put(0, 125){\line(1,0){5}}
\put(5, 150){\boxed{{\rm (I)}\ a>c}}
\put(5, 125){\boxed{a<c}}
\put(5, 100){\boxed{{\rm (II)}\ a=c}}
\put(32, 125){\line(1,0){23}}
\put(55, 100){\line(0,1){50}}
\put(55, 150){\line(1,0){5}}
\put(55, 100){\line(1,0){5}}
\put(55, 125){\line(1,0){5}}
\put(60, 150){\boxed{{\rm (III)}\ b\le a}}
\put(60, 100){\boxed{{\rm (IV)}\ b\ge c}}
\put(60, 125){\boxed{a<b<c}}
\put(101, 125){\line(1,0){9}}
\put(110, 90){\line(0,1){75}}
\put(110, 145){\line(1,0){5}}
\put(110, 115){\line(1,0){5}}
\put(110, 90){\line(1,0){5}}
\put(110, 165){\line(1,0){5}}
\put(115, 165){\boxed{{\rm (V)}\  c_0\ge a, c_0\le b-a }}
\put(115, 140){\boxed{{\rm (VI)}\ c_0\ge a, c_0>b-a}}
\put(115, 115){\boxed{c_0<a, c_0>b-a}}
\put(115, 90){\boxed{{\rm (VII)}\  c_0<a, c_0\le b-a}}
\put(187, 115){\line(1,0){23}}
\put(210, 100){\line(0,1){35}}
\put(210, 100){\line(1,0){5}}
\put(210, 135){\line(1,0){5}}
\put(215, 135){\boxed{ {\rm (VIII)}\ \lfloor c/b\rfloor=1}}
\put(215, 100){\boxed{  \lfloor c/b\rfloor\ge 2}}
\put(3, 55){\boxed{{\rm (XIII)}\ a/b\in\Z/q, c/b\in \Z/q}}
\put(3, 5){\boxed{{\rm (XIV)} \ a/b\in\Z/q, c/b\not\in \Z/q}}
\put(111, 5){\line(1,0){4}}
\put(110, 55){\line(1,0){5}}
\put(115, 5){\line(0,1){50}}
\put(115, 5){\line(0,1){50}}
\put(115, 25){\line(1,0){23}}
\put(138, 25){\boxed{a/b\in \Q}}
\put(120, 60){\boxed{ {\rm (XII)} \ a/b\not\in \Q}}
\put(174, 25){\line(1,0){6}}
\put(177, 60){\line(1,0){3}}
\put(180, 25){\line(0,1){35}}
\put(180, 45){\line(1,0){5}}
\put(185, 45){\boxed{ 0<c_1<2a-b}}
\put(203, 65){\boxed{ {\rm (XI)}\ c_1=0}}
\put(186, 25){\boxed{ {\rm (X)}\ c_1=2a-b}}
\put(184, 5){\boxed{ {\rm (IX)}\ c_1>2a-b}}
\put(249, 45){\line(1,0){16}}
\put(250, 65){\line(1,0){15}}
\put(249, 25){\line(1,0){16}}
\put(250, 5){\line(1,0){15}}
\put(265, 5){\line(0,1){95}}
\put(256, 100){\line(1,0){9}}
\end{picture}
} $$

\smallskip

From Classifications  (V)--(IX) and (XII) in  Theorems \ref{newmaintheorem2}--\ref{newmaintheorem4}, it confirms a conjecture in \cite[Section 3.3.5]{janssen03}:
{\em
If $a<b<c, a/b\not\in \Q$ and $c\not\in a\Q+b\Q$, then  ${\mathcal G}(\chi_{[0,c)}, a\Z\times \Z/b)$
is a Gabor frame for $L^2$}.
This, together with Classification (IV) in Theorem \ref{newmaintheorem1}  and the shift-invariance, implies that
 the range of density parameters $a, b$
 such that  
${\mathcal G}(\chi_I, a\Z\times \Z/b)$ is a Gabor  frame  is 
 a dense subset of
the open 
region ${\mathcal U}_c:=\{(a, b): 0<a< \max(b,c)\}$, where $c$ is the length of the interval $I$.

\section{Gabor Frames and Trivial
Null Spaces of Infinite Matrices}
\label{gabornullspace1.section}

In this section, we prove Theorem \ref{newmaintheorem2}.

\smallskip

To prove Theorem \ref{newmaintheorem2}, we  start from  recalling some algebraic properties for infinite matrices 
 in
\eqref{infinitematrix.def}:
\begin{equation}\label{sabcinvariant.eq1}
{\bf M}_{a, b, c} (t-\lambda') {\bf z}={\bf M}_{a, b, c} (t) \tau_{\lambda'}
{\bf z} \ \  {\rm for \ all} \ \lambda'\in b\Z \end{equation}
 and
\begin{equation}\label{sabcinvariant.eq2}
  {\bf M}_{a, b, c}(t-\mu') {\bf z} =\big(({\bf M}_{a, b, c} (t) {\bf z})
  (\mu+\mu')\big)_{\mu\in a\Z}  \  \  {\rm for \ all} \ \mu'\in a\Z,
  \end{equation}
  where $t\in \R,  {\bf z}:=({\bf z}(\lambda))_{\lambda\in b\Z}$,
  and $\tau_{\lambda'}{\bf z}:= ({\bf z}(\lambda+\lambda'))_{\lambda\in b\Z}$.
By the shift property \eqref{sabcinvariant.eq2} for infinite matrices ${\bf M}_{a, b, c}(t), t\in \R$,
the  sets ${\mathcal D}_{a, b, c}$ and  ${\mathcal S}_{a, b,c}$  in \eqref{dabc2.def} and \eqref{dabc1.def}  respectively are periodic sets with period $a$,
\begin{equation}\label{dabcperiodic.section2} {\mathcal D}_{a, b, c}= {\mathcal D}_{a, b, c}+a\Z 
\quad {\rm and} \quad
 {\mathcal S}_{a, b, c}={\mathcal S}_{a, b, c}+a\Z.\end{equation}

\smallskip

Next we recall the equivalence between  frame property  for the Gabor system ${\mathcal G}(\chi_{[0, c)}, a\Z\times \Z/b)$
 and uniform stability of  infinite matrices ${\bf M}_{a, b, c}(t), t\in \R$,
 in \eqref{infinitematrix.def}, which was established by Ron and Shen.


\begin{lem}\label{framestabilitymatrix.lem} {\rm (\cite{RS97})}\
 Let $(a, b, c)$ be a triple of positive numbers satisfying $\max(a, b)<c$,
 ${\mathcal G}(\chi_{[0,c)}, a\Z\times \Z/b)$  be the Gabor system in \eqref{gabor.def}, and let
  ${\bf M}_{a, b, c}(t), t\in \R$, be the infinite matrices in
  \eqref{infinitematrix.def}.
Then  ${\mathcal G}(\chi_{[0,c)}, a\Z\times \Z/b)$
 is a Gabor frame if and only if
 there exist positive constants $A$ and $B$ such that
\begin{equation}\label{framestabilitymatrix.tm.eq1}
A \|{\bf z}\|_2 \le  \|{\bf M}_{a, b, c}(t) {\bf z}\|_2 \le B \|{\bf z}\|_2\quad {\rm for \ all} \ {\bf z}\in \ell^2\ {\rm and}\ t\in \R.
\end{equation}
 \end{lem}
Infinite matrices ${\bf M}_{a, b, c}(t), t\in \R$,  in \eqref{infinitematrix.def}
have their rows containing
 $\lfloor c/b\rfloor+\{0, 1\}$  consecutive ones, and
their rows are  obtained by shifting one (or zero) unit of the previous row with possible reduction or expansion by one unit, c.f.
\eqref{irrationalexample1.infinitematrix.eq1} and \eqref{rationalexample1.infinitematrix.eq1}.
Due to the above special structures of  infinite matrices 
in \eqref{infinitematrix.def},
 we can reduce their uniform 
 stability property
 \eqref{framestabilitymatrix.tm.neweq1}
 to the trivial intersection property between
 their null spaces $N({\bf M}_{a, b, c} (t))$ and the set ${\mathcal B}_b-{\mathcal B}_b$ containing
 trinary vectors whose components take values in $\{-1, 0, 1\}$, and even further to  the  empty-set property for
 ${\mathcal D}_{a, b, c}$.

  \begin{thm}\label{framenullspace1.tm}
 Let $(a, b, c)$ be a triple of positive numbers satisfying $a<b<c$, ${\mathcal G}(\chi_{[0,c)}, a\Z\times \Z/b)$
 be the Gabor system  in \eqref{gabor.def}, and let
  ${\bf M}_{a, b, c} (t), t\in \R$,  be   infinite matrices in \eqref{infinitematrix.def}.
     Then the following statements are equivalent to each other.
 \begin{itemize}
 \item[{(i)}] The Gabor system ${\mathcal G}(\chi_{[0,c)}, a\Z\times \Z/b)$
 is  a Gabor frame.


\item[{(ii)}] For every $t\in \R$,
only zero vector ${\bf 0}$ is contained in the intersection between
  ${\mathcal B}_b-{\mathcal B}_b$ and the null space of the infinite matrix ${\bf M}_{a, b, c}(t)$; i.e.,
  \begin{equation*} N({\bf M}_{a, b, c}(t))\cap ({\mathcal B}_b-{\mathcal B}_b)
  =\{{\bf 0}\}\quad\ {\rm for \ every} \  t\in \R.
  \end{equation*}


 \item [{(iii)}] ${\mathcal D}_{a, b, c}=\emptyset$.

\end{itemize}

\end{thm}

In the next theorem, it is shown that it could further reduce the  empty-set property for
${\mathcal D}_{a, b, c}$
 in  Theorem \ref{framenullspace1.tm}
to verifying whether  it contains some particular points.

\begin{thm}\label{framenullspace2.tm}
Let $(a,b,c)$ be a triple of positive numbers satisfying $a<b<c$,  the set
 ${\mathcal D}_{a, b, c}$
 be as in \eqref{dabc2.def}, and define
\begin{equation}
\label{zab.def}
{\mathcal Z}_{a,b}:=\left\{\begin{array} {ll} \{0\} & {\rm if} \ a/b\not\in \Q,\\
\{0, b/q, \ldots, b(p-1)/q\}& {\rm if} \ a/b=p/q \ {\rm  for\ some}\\
&  {\rm coprime\ integers}\  p \ {\rm  and} \  q.
\end{array}\right.\end{equation}
 Then  ${\mathcal D}_{a, b, c}=\emptyset$ if and only if
${\mathcal D}_{a, b, c}\cap ({\mathcal Z}_{a, b}\cup
(c-{\mathcal Z}_{a, b}))=\emptyset$.
\end{thm}

We remark that the implication (i)$\Longrightarrow$(iii) in  Theorem \ref{framenullspace1.tm}
 has been  applied implicitly in \cite{hangu, janssen03} for their  classification.

%
%
%

\smallskip

In the  next three subsections, we  prove Theorem \ref{framenullspace1.tm},
 \ref{newmaintheorem2}  
 and  \ref{framenullspace2.tm} respectively.


\subsection{Trivial null spaces of infinite matrices} 
\label{trivialnull.subsection}
%
In this subsection, we prove Theorem \ref{framenullspace1.tm}.
To do so, particularly the implication
(iii)$\Longrightarrow$(i), we need an equivalence between empty set property for the set
${\mathcal D}_{a, b, c}$ and uniform boundedness of
 maximal lengths of consecutive twos in
the vector ${\bf M}_{a, b, c}(t) {\bf x}$
for any $t\in \R$ and  ${\bf x}\in {\mathcal B}_b$.
Precisely, we define
\begin{equation}\label{qabc.def}
Q_{a,b,c}:= \sup_{t\in \R} Q_{a, b, c} (t):=\sup_{t\in
\R}\big(\sup_{{\bf x}\in {\mathcal B}_b} Q_{a, b, c} (t,{\bf x})\big),
\end{equation}
where for any $t\in \R$ and ${\bf x}\in {\mathcal B}_b$,
$$K(t,{\bf x}):=\big\{\mu\in a\Z\big|\ {\bf M}_{a, b, c}(t) {\bf x}(\mu)=2\big\},$$
and
$$Q_{a, b, c}(t,{\bf x}):=\left\{\begin{array}{ll} 0 & {\rm if}\ K(t,{\bf x})=\emptyset\\
 \sup\big\{n\in \N\big| \  [\mu, \mu+na)\cap a\Z& \\
  \qquad \quad \subset K(t,{\bf x}) \ {\rm for \ some}\ \mu\in a\Z\big\} & {\rm otherwise}\end{array}\right.$$
is the maximal length of consecutive twos in the vector ${\bf M}_{a, b, c}(t) {\bf x}$.
We show in Lemma
\ref{dabc2toqabc.lem} below 
that
$
{\mathcal D}_{a, b, c}=\emptyset$ if and only if the
index $Q_{a,b,c}$  associated with infinite matrices ${\bf M}_{a, b,c}(t), t\in \R$,
is finite.

  \begin{lem}\label{dabc2toqabc.lem}  Let $(a, b, c)$ be a triple of positive numbers satisfying $a<b<c$, and let
  ${\bf M}_{a, b, c}(t), t\in \R$,  be  the infinite matrices in \eqref{infinitematrix.def}.
 Then  ${\mathcal D}_{a, b, c}=\emptyset$ if and only if
   \begin{equation}\label{dabc2toqabc.lem.eq5}
 Q_{a,b,c} <+\infty.
\end{equation}
  \end{lem}

The   crucial and most technical step in the proof of  Theorem \ref{framenullspace1.tm}
is to establish the following stability inequality:
\begin{equation} \label{framematrix.lem.eq1}
\sum_{0\le \mu\le a Q_{a,b,c}+2a+b+c} |({\bf M}_{a, b, c}(t) {\bf
z})(\mu)|\ge \frac{b}{2c}  |{\bf z}(0)|
\end{equation}
for all $t\in [0,b)$ and vectors
 ${\bf z}=({\bf z}(\lambda))_{\lambda\in b\Z}$. 

\begin{lem}\label{infinitematrixstability.lem}
Let $(a,b,c)$ be a triple of positive numbers satisfying $a<b<c$, and let $Q_{a,b,c}$ be as in \eqref{qabc.def}.  If
$ Q_{a,b,c}<+\infty$, then \eqref{framematrix.lem.eq1} holds
for any $t\in [0,b)$ and
 ${\bf z}=({\bf z}(\lambda))_{\lambda\in b\Z}$.
\end{lem}

Now we start the proof of Theorem \ref{framenullspace1.tm} by assuming that Lemmas \ref{dabc2toqabc.lem} and \ref{infinitematrixstability.lem}
hold.

\begin{proof}[Proof of Theorem \ref{framenullspace1.tm}]
 (i)$\Longrightarrow$(ii): \quad
  Suppose, on the contrary, that
 there  exist  $t_0\in \R$,
 and a nonzero vector ${\bf z}^*=({\bf z}^*(\lambda))_{\lambda\in b\Z}$
 such that
 \begin{equation} \label{framenullspace1.tm.pf.eq1}  {\bf M}_{a, b, c}(t_0){\bf z}^*=0, \ {\rm and} \
  {\bf z}^*(\lambda)\in \{-1, 0, 1\}\ {\rm for \ all} \ \lambda\in b\Z.\end{equation}
 By Lemma \ref{framestabilitymatrix.lem} and the assumption (i),
  ${\bf M}_{a, b, c}(t_0)$ has the $\ell^2$-stability; i.e.,
  \begin{equation}\label{framenullspace1.tm.pf.eq1+}
0<  \inf_{\|{\bf z}\|_2=1} \|{\bf M}_{a, b, c}(t_0){\bf z}\|_2\le \sup_{\|{\bf z}\|_2=1} \|{\bf M}_{a, b, c}(t_0){\bf z}\|_2<\infty.\end{equation}
This together with \eqref{framenullspace1.tm.pf.eq1} implies that
  \begin{equation} \label{framenullspace1.tm.pf.eq2} {\bf z}^*\not\in \ell^2(b\Z).\end{equation}
Set ${\bf z}_N^*:= ({\bf z}^*(\lambda)\chi_{[-N, N]}(\lambda) )_{\lambda\in b\Z}$,
 $N\ge 2$.
Then
we obtain from \eqref{infinitematrix.def},  \eqref{framenullspace1.tm.pf.eq1} and \eqref{framenullspace1.tm.pf.eq2} that
$$ \left\{\begin{array}{l} \lim_{N\to \infty} \|{\bf z}_N^*\|_{2}=\infty,\\
\|{\bf M}_{a, b, c}(t_0){\bf z}^*_N\|_{\infty}\le \|{\bf M}_{a, b, c}(t_0){\bf 1}\|_{\infty}\le
 c/b+1, {\rm and}\ \\
 ({\bf M}_{a, b, c}(t_0) {\bf z}_N^*)(\mu)=0\quad {\rm  for\ all}\quad \mu-t_0\not\in [N-c, N]\cup [-N-c, -N].\end{array}\right. $$
 Therefore
 $\lim_{N\to \infty} {\|{\bf M}_{a, b, c}(t_0) {\bf z}_N^*\|_2} /{\|{\bf z}_N^*\|_2}=0$,
 which contradicts to the $\ell^2$-stability \eqref{framenullspace1.tm.pf.eq1+}.

\smallskip

 (ii)$\Longrightarrow$(iii):\quad  
 Suppose, on the contrary, that there exist  $t_0\in \R$ and
 a vector ${\bf x}\in {\mathcal B}_b^0$ such that $ {\bf M}_{a, b, c}(t_0) {\bf x}={\bf 2}$.  Denote by
$K$  the support of the vector ${\bf x}$; i.e., the set of all $\lambda\in b\Z$ with ${\bf x}(\lambda)=1$.
Then from $ {\bf M}_{a, b, c}(t_0) {\bf x}={\bf 2}$ it follows that
\begin{equation}\label{frammatrix.tm.pf.eq3}
 \#(K\cap (-t_0+\mu+[0,c))=2\quad {\rm  for\ all} \  \mu\in a\Z.\end{equation}
  Thus $\#(K)=+\infty$
and $K=\{\lambda_j: j\in \Z\}$ for some strictly increasing sequence $\{\lambda_j\}_{j=-\infty}^\infty$ in $b\Z$.
For any $\mu\in a\Z$, we obtain from \eqref{frammatrix.tm.pf.eq3} that
  $ K\cap (-t_0+\mu+[0,c))$ is either
$\{\lambda_{2j}, \lambda_{2j+1}\}$ or $\{\lambda_{2j-1}, \lambda_{2j}\}$ for some  $j\in \Z$.
Thus
\begin{equation*}
 \#(K_i\cap (-t_0+\mu+[0,c))=1
 \quad {\rm  for\ all} \  \mu\in a\Z \ {\rm and} \ i=0,1,\end{equation*}
where $K_i=\{\lambda_{i+2j}: \ j\in \Z\}, i=0,1$.
Define ${\bf x}_i:=({\bf x}_i(\lambda))_{\lambda\in a\Z}, i=0,1$, by
${\bf x}_i(\lambda)=1$ if $\lambda\in K_i$ and ${\bf x}_i(\lambda)=0$ otherwise. Then   ${\bf x}_i, i=0,1$,
have one and only one of them in ${\mathcal B}_b^0$ while the other one in  ${\mathcal B}_b\backslash {\mathcal B}_b^0$,   and they  satisfy
\begin{equation}\label{frammatrix.tm.pf.eq4}
{\bf x}={\bf x}_0+{\bf x}_1, \ {\bf x}_0\ne {\bf x}_1\ \ {\rm and} \ \ {\bf M}_{a, b, c}(t_0) {\bf x}_0={\bf M}_{a, b, c}(t_0) {\bf x}_1={\bf 1}.\end{equation}
Thus ${\bf y}={\bf x}_0-{\bf x}_1$ is a nonzero vector in ${\mathcal B}_b-{\mathcal B}_b$  contained in the null space of the infinite matrix
${\mathbf M}_{a, b, c}(t_0)$, which contradicts to the assumption (ii).

%



\smallskip
(iii)$\Longrightarrow$(i):\quad
Let $Q_{a,b,c}$ be as in \eqref{qabc.def}. Then $Q_{a, b,c}<\infty$ by Lemma \ref{dabc2toqabc.lem}.
For any $f\in L^2$, 
\begin{eqnarray*}
 & & \Big(Q_{a,b,c}+\frac{3a+b+c}{a}\Big)^2\sum_{\phi\in {\mathcal G}(\chi_{[0,c)}, a\Z\times \Z/b)} |\langle f, \phi\rangle|^2\nonumber \\
& \ge & \Big(Q_{a,b,c}+\frac{3a+b+c}{a}\Big) \sum_{\mu\in a\Z}\  \sum_{0\le \mu'\le a Q_{a,b,c}+2a+b+c} \
\sum_{n\in \Z} \nonumber\\
& & \qquad \Big|\int_{0}^b \Big(\sum_{\lambda\in b\Z}  \chi_{[0,c)}(t-\mu'+\lambda)f(t+\mu+\lambda)\Big) e^{-2\pi in t/b}dt\Big|^2\nonumber\\
& = & b \sum_{\mu\in a\Z}
\int_{0}^b
\Big( \big(Q_{a,b,c}+\frac{3a+b+c}{a}\big)  \sum_{0\le \mu'\le a Q_{a,b,c}+2a+b+c}\nonumber\\
& & \qquad
\Big |\sum_{\lambda\in b\Z} \chi_{[0,c)}(t-\mu'+\lambda)f(t+\mu+\lambda) \Big|^2\Big) dt \nonumber \\
& \ge  & b \sum_{\mu\in a\Z}
\int_{0}^b
\Big(  \sum_{0\le \mu'\le a Q_{a,b,c}+2a+b+c}
\Big |\sum_{\lambda\in b\Z}  \chi_{[0,c)}(t-\mu'+\lambda)f(t+\mu+\lambda)\Big|\Big)^2 dt \nonumber \\
& \ge & \frac{b^3}{4c^2} \sum_{\mu\in a\Z}
\int_{0}^b
 |f(t+\mu)|^2 dt \ge 
 \frac{b^3\lfloor b/a\rfloor}{4c^2}  \|f\|_2^2,
\end{eqnarray*}
where 
the third inequality follows from Lemma
\ref{infinitematrixstability.lem}, and the last inequality is true as
$\sum_{\mu\in a\Z} \chi_{[0, b)}(t-\mu)\ge \lfloor b/a\rfloor$ for all $t\in \R$.
Therefore ${\mathcal G}(\chi_{[0,c)}, a\Z\times \Z/b)$ is a Gabor frame and the implication
(iii)$\Longrightarrow$(i) is established.
\end{proof}

%

We finish this subsection by the proofs of Lemmas \ref{dabc2toqabc.lem} and \ref{infinitematrixstability.lem}.

\begin{proof}[Proof of Lemma \ref{dabc2toqabc.lem}] ($\Longleftarrow$)\quad
Suppose on the contrary that ${\mathcal D}_{a, b, c}\ne \emptyset$. Take $t_0\in {\mathcal D}_{a, b, c}$
and a vector ${\bf x}_0\in {\mathcal B}_b^0$ satisfying ${\bf M}_{a, b, c}(t_0) {\bf x}_0={\bf 2}$.
Then $Q_{a, b, c}(t,{\bf x}_0)=+\infty$, which implies that
$Q_{a, b, c}(t) = +\infty$, a contradiction.

($\Longrightarrow$)\quad
  Suppose, on the contrary, that
  $Q_{a, b,c}=+\infty$. Then  for all $n\ge 1$  there exist $t_n\in \R, \mu_n\in a\Z$ and ${\bf x}_n\in {\mathcal B}_b$ such that
  ${\bf M}_{a, b, c}(t_n) {\bf x}_n(\mu)=2$ for all $\mu_n\le \mu\le \mu_n+2na$.
 By \eqref{sabcinvariant.eq1} and
 \eqref{sabcinvariant.eq2}, without loss of generality, we may assume that $t_n\in [0,b)$ and
  \begin{equation} \label{infinitematrixfinitedegree.lem.pf.eq3}
  ({\bf M}_{a, b, c}(t_n) {\bf x}_n)(\mu)=2\quad {\rm for \ all} \ \mu\in [-na, na]\cap a\Z,
  \end{equation}
 otherwise
replacing $t_n$ by the unique number $t_n^\prime\in [0,b)$ satisfying $t_n-\mu_n-na-t_n^\prime\in b\Z$ and ${\bf x}_n$
by $\tau_{t_n^\prime-t_n+\mu_n+na}{\bf x}_n$. Moreover, we can additionally assume that
 ${\bf x}_n:=({\bf x}_n(\mu))_{\mu\in b\Z}\in {\mathcal B}_b^0, n\ge 1$, satisfy
\begin{equation}\label{infinitematrixfinitedegree.lem.pf.eq4}
{\bf x}_{n'}(\lambda)= {\bf x}_n(\lambda) \quad {\rm for\ all} \ \lambda\in [-n b, n b]\cap b\Z\ {\rm and}\ n'\ge n,
\end{equation}
and
\begin{equation}\label{infinitematrixfinitedegree.lem.pf.eq5}
\{t_n\}_{n=1}^\infty\ \  {\rm is \ a \ monotone\ sequence},
\end{equation}
otherwise, replace  $\{{\bf x}_n\}_{n=1}^\infty$ and $\{t_n\}_{n=1}^\infty$ by their subsequences
 satisfying \eqref{infinitematrixfinitedegree.lem.pf.eq4} and \eqref{infinitematrixfinitedegree.lem.pf.eq5}.
 Denote by $t_\infty$ and ${\bf x}_\infty$ the limit of the sequence $\{t_n\}_{n=1}^\infty$ of numbers in $[0, b)$ and
 the sequence  $\{{\bf x}_n\}_{n=1}^\infty$ of vectors in ${\mathcal B}_b^0$ respectively. Clearly ${\bf x}_\infty\in {\mathcal B}_b^0$.

If there exists $n_0$ such that $t_n=t_\infty$ for all $n\ge n_0$, then
  for any given $\mu\in a\Z$,
 \begin{equation*}\label{infinitematrixfinitedegree.lem.pf.eq6}
\big( {\bf M}_{a, b, c}(t_\infty) {\bf x}_\infty\big)(\mu)=\big({\bf M}_{a, b, c} (t_n) {\bf x}_n\big) (\mu)=2\end{equation*}
  for sufficiently large $n$ by \eqref{infinitematrixfinitedegree.lem.pf.eq4}.
Thus $ {\bf M}_{a, b,c}(t_\infty) {\bf x}_\infty={\bf 2}$ and $t_\infty\in {\mathcal D}_{a, b, c}$,
 which contradicts to the assumption that ${\mathcal D}_{a, b, c}=\emptyset$.

 If $\{t_n\}_{n=1}^\infty$ is a strictly decreasing sequence,
 then for any given $\lambda\in b\Z$ and $\mu\in a\Z$,
   $\chi_{[0,c)}(t_\infty-\mu+\lambda)=
 \chi_{[0,c)}(t_n-\mu+\lambda)$ for sufficiently large $n$.
 This together with  \eqref{infinitematrixfinitedegree.lem.pf.eq3} and
 \eqref{infinitematrixfinitedegree.lem.pf.eq4}
 implies that
 \begin{equation*} 
\big( {\bf M}_{a, b, c}(t_\infty) {\bf x}_\infty\big)(\mu)=\lim_{n\to \infty} \big({\bf M}_{a, b, c} (t_n) {\bf x}_n\big) (\mu)=2,\end{equation*}
which contradicts  to
the assumption that ${\mathcal D}_{a, b, c}=\emptyset$.

 If $\{t_n\}_{n=1}^\infty$ is a strictly increasing sequence, then
 for any given $\lambda\in a\Z$ and $\mu\in a\Z$,
    $\chi_{(0,c]}(t_\infty-\mu+\lambda)=
 \chi_{[0,c)}(t_n-\mu+\lambda)$ for sufficiently large $n$.
This together with  \eqref{infinitematrixfinitedegree.lem.pf.eq3} and
 \eqref{infinitematrixfinitedegree.lem.pf.eq4}
 yields  that
$$ \sum_{\lambda\in  b\Z}  \chi_{(0, c]}(t_\infty-\mu+\lambda) {\bf x}_\infty(\lambda)=2\quad {\rm for \ all} \ \mu\in a\Z,$$
or equivalently
$$ \sum_{\lambda\in  b\Z}  \chi_{[0, c)}(c-t_\infty-\mu+\lambda) {\bf x}_\infty(-\lambda)=2\quad {\rm for \ all} \ \mu\in a\Z.$$
Thus $c-t_\infty\in {\mathcal D}_{a, b, c}$, which is a contradiction.
 \end{proof}


\begin{proof} [Proof of Lemma \ref{infinitematrixstability.lem}]
For $t\in [0,b)$, let $\lambda_0=0, \mu_0=\lfloor t/a\rfloor a$ and
let $\delta_0\ge 0$ be the unique element in $[c+\mu_0-t-b, c+\mu_0-t)\cap b\Z$.
If $\delta_0=0$, then \eqref{framematrix.lem.eq1} holds as
$|({\bf M}_{a, b, c}(t){\bf z})(\mu_0)|=|{\bf z}(0)|$ and $\mu_0\le t\le a Q_{a,b,c}+2a+b+c$.

Now we prove \eqref{framematrix.lem.eq1} in the case that $\delta_0\ge b$. To do so,
 we introduce families of triples
$(\lambda_k^l, \mu_k^l, \delta_k^l)\in b\Z\times a\Z\times b\Z, 0\le k\le M_l, 1\le l\le \delta_0/b$, iteratively.
For $i=0,1$, we let $\lambda_i^l =ilb, \mu_i^l=\lfloor (t+\lambda_i^l)/a\rfloor a$
and  $\delta_i^l$ be the unique element in $[c+\mu_i^l-t-b, c+\mu_i^l-t)\cap b\Z$.
Suppose that we have defined the triple $(\lambda_k^l, \mu_k^l, \delta_k^l)$,
we set $M_l=k$ if $\delta_k^l\ge c+\mu_k^l-t+a-b$, and otherwise we
define $(\lambda_{k+2}^l, \mu_{k+2}^l, \delta_{k+2}^l)$ by $\lambda_{k+2}^l=\delta_k^l+b, \mu_{k+2}^l=\lfloor (t+\lambda_{k+2}^l)/a\rfloor a$ and
$\delta_{k+2}^l\in [c+\mu_{k+2}^l-t-b, c+\mu_{k+2}^l-t)\cap b\Z$.
We remark that $\delta_k^l\ge c+\mu_k^l-t+a-b$ if  the $(\mu_k^l+1)$-th row of the matrix ${\mathbf M}_{a, b, c}(t)$ is obtained by shift one unit to the left of the $\mu_k$-th row with reduction by one unit, c.f. the fourth and ninth rows in \eqref{irrationalexample1.infinitematrix.eq1}.
Here is the visual interpretation of
indices $\lambda_k^l$ and $\mu_k^l, 1\le k\le M_l$,   in the matrix ${\bf M}_{a, b, c}(t)$:

{\tiny
$$
\begin{array} {c}
 \\    \\ \\
 \mu_0  \rightarrow  \\
 \\
\\
 \\
 \mu_1  \rightarrow\\
 \\
 \\
 \mu_2  \rightarrow\\
 \\
\vdots\\   \\
\mu_{M_l}\rightarrow\\
 \\
 \end{array}
  \left(
 \begin{array}{cccccccccccccccccccccccc}  \lambda_0^l & \!\! &\! &\! \lambda_1^l &\! &\! &\!  &\! \lambda_2^l &\! &\! &\! &\! \lambda_3^l
  &\!\! &\! &\! \lambda_4^l &\! \cdots &\! \lambda_{M_l}^l\\
  \downarrow &\! &\! &\! \downarrow &\! &\! &\! &\! \downarrow &\! &\! &\! &\! \downarrow&\! &\! &\! \downarrow &\!  &\!  \downarrow\\
  &\!  &\! &\!  &\! &\! &\!\\
   1 & \!\! 1 &\! \cdots &\! 1 &\! 1 &\! \cdots &\! 1 &\!   &\!  &\!\\
    & \!\! 1 &\! \cdots &\! 1 &\! 1 &\! \cdots &\! 1 &\! 1  &\!  &\!\\
&\!  &\! &\!  &\!  &\! \ddots &\!  &\! &\! &\!\\
 &\!  &\!  &\! 1  &\! 1 &\! \cdots   &\! 1 &\! 1  &\! \cdots &\!\cdots &\! 1 \\
 &\!  &\!  &\!  &\! 1  &\! \cdots &\! 1   &\! 1 &\! \cdots  &\!\cdots &\! 1 &\! 1 \\
  &\! &\!  &\!  &\!  &\!  &\! &\! &\!  \ddots\\
 &\! &\! &\! &\!&\!  &\!  &\! 1 &\! 1 &\! \cdots &\! 1 &\! 1 &\! \cdots &\! 1 &\!   &\!  &\!\\
 &\! &\!&\!&\!&\!&\!  &\!  &\!  1 &\! \cdots &\! 1 &\! 1 &\! \cdots &\! 1 &\! 1  &\!  &\!\\
  &\! &\!  &\!&\!  &\!  &\!  &\! &\! &\! &\! &\! &\!   &\! &\! &\! \ddots\\
 &\! &\! &\!  &\! &\! &\! &\!  &\!&\!  &\!  &\! &\! &\! &\!  &\! &\! 1 &\! 1 &\! \cdots &\! 1   \\
 &\! &\! &\! &\!&\!&\!  &\! &\! &\! &\! &\!  &\!  &\! &\!  &\! &\!  &\! 1 &\! \cdots    &\! 1
   \end{array} \right)   .
$$
}

From the above construction, we see  triples $(\lambda_k^l,
\mu_k^l, \delta^l_k), 0\le k\le M_l$, have the following
properties:
\begin{equation}\label{framematrixstability.lem.pf.eq1}
\left\{\begin{array}{ll}
\lambda_k^l\in [\mu_k^l-t, \mu_k^l-t+a) & {\rm if} \  0\le k\le M_l \\
\lambda_{k+2}^l\in [c+\mu_k^l-t, c+\mu_k^l-t+a) & {\rm if}  \  0\le k\le M_l-2,
\end{array}\right.
\end{equation}
 \begin{equation}\label{framematrixstability.lem.pf.eq2} [\mu_{M_l}^l-t+c, \mu_{M_l}^l-t+c+a)\cap b{\Z}=\emptyset\quad {\rm if}\ M_l< \infty,
\end{equation}
 and
 \begin{equation} \label{framematrixstability.lem.pf.eq3}
 \{\lambda_k^l\}_{k=0}^{M_l} \ {\rm and}\  \{\mu_k^l\}_{k=0}^{M_l}\ {\rm  are\ strictly\ increasing\ sequences}.
 \end{equation}
Define ${\bf x}_l:=({\bf x}_l(\lambda))_{\lambda\in b\Z}$ by
${\bf x}_l(\lambda)=1$ if $\lambda=\lambda_k^l$ for some $0\le k\le M_l$, and ${\bf x}_l(\lambda)=0$ otherwise.
Then  ${\bf x}_l\in {\mathcal B}_b^0$  by \eqref{framematrixstability.lem.pf.eq3},
and
for $\mu_0^l\le \mu\le \mu_{M_{l-1}}^l$,
\begin{eqnarray*}
({\bf M}_{a, b, c}(t) {\bf x}_l) (\mu) & = & \Big(\sum_{0\le k\le M_l, k \ {\rm even} }+
\sum_{0\le k\le M_l, k \ {\rm odd} }\Big)
\chi_{[0,c)}(t-\mu+\lambda_{k}^l)
\nonumber \\
& = & \chi_{[0,\mu_0^l]}(\mu)+\sum_{2\le k\le M_l, k \ {\rm even}} \chi_{(\mu_{k-2}^l, \mu_k^l]}(\mu)\nonumber\\
& & +\chi_{(t+\lambda_1^l-c,\mu_1^l]}(\mu)
+\sum_{3\le k\le M_l, k \ {\rm odd }} \chi_{(\mu_{k-2}^l, \mu_k^l]}(\mu)
=2,
\end{eqnarray*}
where the second equation follows from
$$[\mu_{k-2}^l+a, \mu_k^l]\subset (t+\lambda_k^l-c, t+\lambda_k^l]\subset (\mu_{k-2}^l, \mu_k^l+a), 2\le k\le M_l$$
which  holds by \eqref{framematrixstability.lem.pf.eq1}. This leads to the following crucial estimate:
 \begin{equation} \label{framematrixstability.lem.pf.eq4}
 \mu_{M_{l-1}}^l-\mu_0^l\le a Q_{a,b,c}.
 \end{equation}
 By \eqref{framematrixstability.lem.pf.eq1} and \eqref{framematrixstability.lem.pf.eq3}, we have that
 \begin{equation}\label{framematrixstability.lem.pf.eq5}
 \mu_{M_l}-\mu_{M_{l-1}}\le \mu_{M_l}-\mu_{M_{l-2}}\le
 \lambda_{M_l}^l+t-(\lambda_{M_l}^l-c+t-a)\le a+c.
 \end{equation}
 Combining \eqref{framematrixstability.lem.pf.eq4} and \eqref{framematrixstability.lem.pf.eq5}
 and recalling $\mu_0\le t<b$, we obtain
 \begin{equation}\label{framematrixstability.lem.pf.eq6}
\mu_{M_l}\le a Q_{a,b,c}+a+b+c. \end{equation}

By \eqref{framematrixstability.lem.pf.eq6}, $M_l<\infty$ for all $1\le l\le \delta_0/b$.
Now we establish \eqref{framematrix.lem.eq1} if $M_{l_0}$ is  an even integer
for some $1\le l_0\le \delta_0/b$. In this case, applying
\eqref{framematrixstability.lem.pf.eq1} and
\eqref{framematrixstability.lem.pf.eq2}  with $l=l_0$, we obtain
that
\begin{eqnarray}\label{framematrixstability.lem.pf.eq7}
 & &  |({\bf M}_{a, b, c}(t) {\bf z})(\mu_{2k}^{l_0})|+|({\bf M}_{a, b, c}(t){\bf z})(\mu_{2k}^{l_0}+a)|\nonumber \\
  & \ge  &\!  |({\bf M}_{a, b, c}(t) {\bf z})(\mu_{2k}^{l_0})- ({\bf M}_{a, b, c}(t){\bf z})(\mu_{2k}^{l_0}+a)|\nonumber\\
&= &|{\bf z}(\lambda_{2k+2}^{l_0})-{\bf z}(\lambda_{2k}^{l_0})|
\end{eqnarray}
for all integers $0\le k\le  M_{l_0}/2-1$, and
\begin{eqnarray}\label{framematrixstability.lem.pf.eq8}
|({\bf M}_{a, b, c}(t) {\bf z})(\mu_{2k}^{l_0})|+|({\bf M}_{a, b, c}(t){\bf z})(\mu_{2k}^{l_0}+a)|
\ge |{\bf z}(\lambda_{M_{l_0}}^{l_0})|
\end{eqnarray}
for $k=M_{l_0}/2$.
Combining
\eqref{framematrixstability.lem.pf.eq6}, \eqref{framematrixstability.lem.pf.eq7} and \eqref{framematrixstability.lem.pf.eq8}, we get that
\begin{eqnarray*}
 & &  2 \sum_{0\le \mu\le  a Q_{a,b,c}+2a+b+c} |({\bf M}_{a, b, c}(t){\bf z})(\mu)|\nonumber\\
& \ge &  \sum_{k=0}^{M_{l_0}/2} |({\bf M}_{a, b, c}(t) {\bf z})(\mu_{2k}^{l_0})|
+|({\bf M}_{a, b, c}(t){\bf z})(\mu_{2k}^{l_0}+a)|\nonumber\\
& \ge &
\sum_{k=0}^{M_{l_0}/2-1}
|{\bf z}(\lambda_{2k+2}^{l_0})-{\bf z}(\lambda_{2k}^{l_0})| + |{\bf z}(\lambda_{M_{l_0}}^{l_0})|
\ge |{\bf z}(\lambda_0^{l_0})|=|{\bf z}(0)|.
\end{eqnarray*}
Hence
\eqref{framematrix.lem.eq1}  follows  in the case that $M_{l_0}$ is even for some $1\le l_0\le \delta_0/b$.

Finally we prove \eqref{framematrix.lem.eq1}  in the case that
$M_{l},1\le l\le \delta_0/b$, are all odd integers. In this case,
mimicking the argument used  to establish
\eqref{framematrixstability.lem.pf.eq7} and
\eqref{framematrixstability.lem.pf.eq8}, we obtain that
\begin{equation}\label{framematrixstability.lem.pf.eq7+}
  |({\bf M}_{a, b, c}(t) {\bf z})(\mu_{2k+1}^{l})|+|({\bf M}_{a, b, c}(t){\bf z})(\mu_{2k+1}^{l}+a)|
\ge |{\bf z}(\lambda_{2k+3}^{l})-{\bf z}(\lambda_{2k+1}^{l})|
\end{equation}
for all integers $0\le k\le (M_{l}-3)/2$, and
\begin{equation}\label{framematrixstability.lem.pf.eq8+}
|({\bf M}_{a, b, c}(t) {\bf z})(\mu_{2k+1}^{l})|+|({\bf M}_{a, b, c}(t){\bf z})(\mu_{2k+1}^{l}+a)|
\ge |{\bf z}(\lambda_{M_{l}}^{l})|
\end{equation}
for $k=(M_{l}-1)/2$, where $1\le l\le \delta_0/b$.
Therefore
\begin{eqnarray*}
& & \frac{2 \delta_0}{b} \sum_{0\le \mu\le a Q_{a,b,c}+2a+b+c}
|({\bf M}_{a, b, c}(t) {\bf z})(\mu)|\nonumber \\
 &  \ge &   |({\bf M}_{a, b, c}(t){\bf z})(\mu_0)| +\sum_{1\le l\le \delta_0/b}\ \sum_{k=0}^{(M_l-1)/2}
 \Big( |({\bf M}_{a, b, c}(t) {\bf z})(\mu_{2k+1}^{l})|\nonumber\\
 & & \quad +|({\bf M}_{a, b, c}(t){\bf z})(\mu_{2k+1}^{l}+a)|\Big)\nonumber \\
 & \ge & \Big|\sum_{0\le \lambda\le \delta_0} {\bf z}(\lambda)\Big|+
 \sum_{1\le l\le \delta_0/b} |{\bf z}(\lambda_1^l)|\nonumber\\
 & =& \Big|\sum_{0\le \lambda\le \delta_0} {\bf z}(\lambda)\Big|+
 \sum_{1\le \lambda\le \delta_0} |{\bf z}(\lambda)|
 \ge |{\bf z}(0)| \end{eqnarray*}
by   \eqref{framematrixstability.lem.pf.eq7+} and \eqref{framematrixstability.lem.pf.eq8+}.
This together with $\delta_0\le c+\mu_0-t\le c$ proves
\eqref{framematrix.lem.eq1} in the case that $M_{l}, 1\le l\le \delta_0/b$, are all odd integers.
This completes the proof of Lemma \ref{infinitematrixstability.lem}.
\end{proof}

\subsection{Proof of  Theorem \ref{newmaintheorem2}}\label{gaborstabilitymatrix.section.subsection2}
In this subsection, we apply Theorem \ref{framenullspace1.tm} to
prove Theorem \ref{newmaintheorem2} by verifying whether
${\mathcal D}_{a, b, c}$ is an empty set or not.
To do so,  we  notice
 that  ${\mathcal S}_{a, b, c}$  is a supset of ${\mathcal D}_{a, b, c}$,
   \begin{equation} \label{dabc2subsetdabc1}
{\mathcal D}_{a, b, c}\subset {\mathcal S}_{a, b, c},
\end{equation}
  because any vector ${\bf x}\in {\mathcal B}_b^0$ satisfying $ {\bf M}_{a, b, c}(t) {\bf x}={\bf 2}$
can be written as the sum of two vectors ${\bf x}_1, {\bf x}_2\in
{\mathcal B}_b$ such that ${\bf M}_{a, b, c}(t) {\bf x}_1= {\bf M}_{a,b,c}(t) {\bf x}_2={\bf 1}$, c.f. \eqref{frammatrix.tm.pf.eq4}
in the proof of Theorem \ref{framenullspace1.tm}.
We also  we 
need some elementary properties for
the supset ${\mathcal S}_{a, b, c}$,
including 
empty intersection property \eqref{dabcbasic1.lem.eq1} with the  black holes.

\begin{prop}\label{dabcbasic1.prop}
Let $(a,b,c)$ be a triple of positive numbers with $a<b<c$. Set $c_0=c-\lfloor c/b\rfloor b$,
 $(c_0+a-b)_+=\max(c_0+a-b,0)$ and  $c_0\wedge a=\min(c_0, a)$.
Define
\begin{equation}
\label{lambdaabc.def}
\lambda_{a,b,c}(t)=\left\{\begin{array}{ll}
 \lfloor c/b\rfloor b +b & {\rm if} \ t\in [0, (c_0+a-b)_+)+a\Z\\
 0 & {\rm if} \ t\in [(c_0+a-b)_+, c_0\wedge a)+a\Z\\
  \lfloor c/b\rfloor b   & {\rm if} \ t\in [c_0\wedge a, a)+a\Z,
  \end{array}\right. \end{equation}
and
\begin{equation}
\label{tildelambdaabc.def}
\tilde \lambda_{a,b,c}(t)=\left\{\begin{array}{ll}
 -\lfloor c/b\rfloor b -b & {\rm if} \ t\in [c-(c_0+a-b)_+, c)+a\Z\\
 0 & {\rm if} \ t\in [c-c_0\wedge a, c-(c_0+a-b)_+)+a\Z\\
 - \lfloor c/b\rfloor b   & {\rm if} \ t\in [c-a, c-c_0\wedge a)+a\Z.
  \end{array}\right. \end{equation}
  Then
\begin{equation}\label{dabcbasic1.lem.eq1}
{\mathcal S}_{a, b, c}\cap \big([(c_0+a-b)_+, c_0\wedge a)\cup [c-c_0\wedge a, c-(c_0+a-b)_+)+a\Z\big)=
\emptyset,\end{equation}
and
 \begin{equation}
 \label{dabcbasic1.lem.eq3}
 {\bf x}(\lambda)=\left\{\begin{array}{ll} 0 & {\rm  if} \  \tilde \lambda_{a,b,c}(t)<\lambda<\lambda_{a,b,c}(t)\ {\rm and }\ \lambda\ne 0\\
1 & {\rm  if} \ \lambda=\lambda_{a,b,c}(t), 0,  \tilde \lambda_{a,b,c}(t)\end{array}\right.
\end{equation}
for any $t\in {\mathcal S}_{a, b, c}$ and  ${\bf x}\in {\mathcal B}_b^0$ satisfying ${\bf M}_{a, b, c}(t){\bf x}={\bf 1}$.
 \end{prop}

Now assuming that Proposition \ref{dabcbasic1.prop} holds, we start to prove Theorem \ref{newmaintheorem2}.
The main ideas behind our proof of   Theorem  \ref{newmaintheorem2} are as  follows.
To prove Conclusion (V) of Theorem \ref{newmaintheorem2}, we use Proposition \ref{dabcbasic1.prop} to
 verify that ${\mathcal S}_{a, b, c}=\emptyset$, which together with \eqref{dabc2subsetdabc1} and Theorem \ref{framenullspace1.tm}
  leads to the desired frame property for the Gabor system ${\mathcal G}(\chi_{[0,c)} a\Z\times\Z/b)$.
The crucial step in the proof of Conclusion (VI) is that
for any $t_0\in {\mathcal D}_{a, b, c}$, the binary vector ${\bf x}\in {\mathcal B}_b^0$
satisfying ${\mathcal M}_{a, b, c}(t_0) {\bf x}={\bf 2}$ is supported on
$(\lfloor c/b\rfloor+1) b\Z\cup (\lambda_1+(\lfloor c/b\rfloor+1) b\Z)$
for some $\lambda_1\in [b, \lfloor c/b\rfloor b]\cap b\Z$, which implies that
$$\{t_0, t_0+\lambda\}+ (\lfloor c/b\rfloor+1) b\Z +a\Z\subset {\mathcal D}_{a, b, c},$$
see \eqref{maintheorem.pf.nullspace1.eq3} and \eqref{maintheorem.pf.nullspace1.eq4}.
Define
\begin{equation}\label{tildeinfinitematrix.def}
\tilde {\bf M}_{a, b, c}(t)= (\chi_{(0, c]}(t-\mu+\lambda))_{\mu\in a\Z, \lambda\in b\Z},\ t\in \R.\end{equation}
We use similar argument to prove Conclusion (VII), except that
for any $t_0\in \tilde {\mathcal D}_{a, b, c}:=c-{\mathcal D}_{a, b, c}$, the binary vector ${\bf x}\in {\mathcal B}_b^0$
satisfying $\tilde {\mathcal M}_{a, b, c}(t_0) {\bf x}={\bf 2}$ is supported on
$\lfloor c/b\rfloor  b\Z\cup (\lambda_1+ \lfloor c/b\rfloor b\Z)$
for some $\lambda_1\in [b, \lfloor c/b\rfloor b-1]\cap b\Z$,
see \eqref{maintheorem.pf.eq4+} and \eqref{maintheorem.pf.eq5+}.

\begin{proof}[Proof of Theorem \ref{newmaintheorem2}]

(V): \quad
By 
\eqref{dabc2subsetdabc1} and Theorem \ref{framenullspace1.tm},  it suffices to prove
${\mathcal S}_{a, b, c}=\emptyset$, which follows from Proposition \ref{dabcbasic1.prop}
as $[(c_0+a-b)_+, c_0\wedge a)+a\Z=[0, a)+a\Z=\R$  in this case.
This  conclusion was also established  in
\cite[Section 3.3.3.2]{janssen03}, we include the proof as the conclusion ${\mathcal S}_{a, b, c}=\emptyset$
will be used later to prove Theorem \ref{framenullspace2.tm}.

\smallskip

(VI):\quad   ($\Longrightarrow$)\
By 
Theorem \ref{framenullspace1.tm},
 ${\mathcal D}_{a, b, c}\ne \emptyset$.
Then by \eqref{dabcperiodic.section2} there exist $t_0\in [0,a)$ and ${\bf x}\in {\mathcal B}_b^0$   such that
 ${\bf M}_{a, b, c}(t_0) {\bf x}={\bf 2}$.
 By \eqref{frammatrix.tm.pf.eq4}, we can write ${\bf x}={\bf x}_0+{\bf x}_1$ such that ${\bf x}_0\in {\mathcal B}_b^0, {\bf x}_1\in {\mathcal B}_b$
 and ${\bf M}_{a, b, c}(t_0) {\bf x}_0={\bf M}_{a, b, c}(t_0) {\bf x}_1={\bf 1}$.
 Let $\lambda_i, i=1,2, 3$, be the first three positive indices in
 $b\Z$ with ${\bf x}(\lambda_i)=1, i=1, 2,3$.
 Then $\lambda_2$ is the first positive index in $b\Z$ with ${\bf x}_0(\lambda_2)=1$ by \eqref{frammatrix.tm.pf.eq4}.
  By Proposition \ref{dabcbasic1.prop}, we obtain that
 $\lambda_2=\lambda_3-\lambda_1=\lfloor c/b\rfloor b+b$. Further inspection  also shows that  $t_0\in [0, c_0+a-b),
 t_0+\lambda_1\in [0, c_0+a-b)+a\Z$,
 and   the unique number $t_1\in [0,a)\cap (t_0+\lambda_2+a\Z)$
 belongs to ${\mathcal D}_{a, b, c}$, where $0\le c_0:=c-\lfloor c/b\rfloor b<b$. Inductively for any $n\ge 1$,
the unique number $t_n$ in $[0,a)\cap (t_{n-1}+(\lfloor c/b\rfloor +1)b+a\Z)$
belongs to ${\mathcal D}_{a, b, c}\cap [0, c_0+a-b)$ and $t_n+\lambda_1\in {\mathcal D}_{a, b, c}\cap ([0, c_0+a-b)+a\Z)$.
  From the above construction, we see that $t_n-t_0-n(\lfloor c/b\rfloor +1)b\in a\Z$.
  Hence $a/b$ is rational, as otherwise $\{t_n: \ n\ge 0\}$ is dense in $[0,a)$, which is a contradiction as $c_0+a-b<a$ and $t_n+\lambda_1\in [0, c_0+a-b)+a\Z$
  for all $n\ge 0$.
  Now we  write $a/b=p/q$ for some coprime integers $p$ and $q$, set
  $m= {\rm gcd}(\lfloor c/b\rfloor+1, p)$,
   and let $t_0^\prime$ be the unique number in $[0,  m b/q)$
  such that $t_0^\prime-t_0\in m b\Z/q$.
 Then
 \begin{eqnarray}\label{maintheorem.pf.nullspace1.eq3}
\{t_n: n\ge 0\}+a\Z  & = &  \{ t_0+ n(\lfloor c/b\rfloor+1)b: n\ge 0\}+a\Z\nonumber\\
 &=&
 t_0+(\lfloor c/b\rfloor +1)b\Z+a\Z = t_0+  m b\Z/q\nonumber\\
 & = & \{ t_0^\prime+n m b/q: 0\le n\le p/m-1 \}+a\Z,
 \end{eqnarray}
  where the first equality follows from the definition of $t_n$, the second one holds as $p(\lfloor c/b\rfloor+1)b\in a\Z$,
 and the last one is true by $m={\rm gcd}((\lfloor c/b\rfloor+1)q, p)$.
 Similarly, we have that
  \begin{equation}\label{maintheorem.pf.nullspace1.eq4}
\{t_n+\lambda_1: n\ge 0\}+a\Z   =  \{ t_0^{\prime\prime}+n m b/q: 0\le n\le p/m-1 \}+a\Z,
 \end{equation}
 where $t^{\prime\prime}_0$ is the unique number  in $[0,  m b/q)$
  such that $t_0^{\prime\prime}-t_0-\lambda_1\in m b\Z/q$.
    Hence from \eqref{maintheorem.pf.nullspace1.eq3}, \eqref{maintheorem.pf.nullspace1.eq4} and
  the property that $t_n, t_n+\lambda_1\in   [0, c_0+a-b)+a\Z, n\ge 1$, we obtain that
  $\{ t_0^\prime+n m b/q: 0\le n\le p/m-1 \}\subset
 [0, c_0+a-b)$ and   $\{ t_0^{\prime\prime}+n m b/q: 0\le n\le p/m-1 \}\subset
 [0, c_0+a-b)$. Observe that $t^{\prime\prime}_0-t_0^\prime\in \lambda_1+mb\Z/q\subset b \Z/q$, and
 that  $t^{\prime\prime}_0-t_0^\prime\ne 0$ if $m=\lfloor c/b\rfloor+1$
as otherwise $(\lambda_1/b) q= (\lfloor c/b\rfloor+1)r$ for some $r\in \Z$. This together with $1={\rm gcd}(q, p)\ge {\rm gcd}(q, m)$ implies that
$\lambda_1\in (\lfloor c/b\rfloor+1) b\Z$, which is a contradiction as $0<\lambda_1<(\lfloor c/b\rfloor+1) b$.
 Therefore
 $c_0+a-b>(p/m-1)mb/q$ if $m\ne \lfloor c/b\rfloor+1$
and
 $c_0+a-b>(p/m-1)mb/q+b/q$ if $m= \lfloor c/b\rfloor+1$.
%
This completes the proof of the necessity.

 ($\Longleftarrow$)\  In the case that ${\rm gcd}(\lfloor c/b\rfloor+1, p)\ne \lfloor c/b\rfloor+1$, we define
${\bf x}(\lambda)=1$ if $\lambda\in \{0, {\rm gcd}(\lfloor c/b\rfloor+1, p)b\}+(\lfloor c/b\rfloor+1)b\Z$
and ${\bf x}(\lambda)=0$ otherwise.
Then ${\bf x}\in {\mathcal B}_b^0$ and
\begin{eqnarray}\label{maintheorem.pf.nullspace1.eq5}
{\bf M}_{a,b, c} (0) {\bf  x}(\mu) & = &  \sum_{k\in \Z} \chi_{[0,c)}\big(k(\lfloor c/b\rfloor+1)b-\mu\big)\nonumber\\
& & +
\sum_{k\in \Z} \chi_{[0,c)}\big(k(\lfloor c/b\rfloor+1)b+{\rm gcd}(\lfloor c/b\rfloor+1, p)b-\mu\big)\nonumber\\
& = &
\sum_{i=1}^2 \sum_{k\in \Z} \chi_{[0,c)}\big(k(\lfloor c/b\rfloor+1)b+s_i {\rm gcd}(\lfloor c/b\rfloor+1, p) b/q\big)\nonumber \\
& = & \sum_{i=1}^2 \chi_{[0,c)}\big(s_i {\rm gcd}(\lfloor c/b\rfloor+1, p) b/q\big)=2, \ \mu\in a\Z,
\end{eqnarray}
where $0\le s_1, s_2\le  q (\lfloor c/b\rfloor+1)/{\rm gcd}(\lfloor c/b\rfloor+1, p) -1$ are so chosen that
 $-\mu-s_1  {\rm gcd}(\lfloor c/b\rfloor+1, p)b/q\in (\lfloor c/b\rfloor+1)b\Z$
and $-\mu+{\rm gcd}(\lfloor c/b\rfloor+1, p)b -s_2  {\rm gcd}(\lfloor c/b\rfloor+1, p)b/q\in (\lfloor c/b\rfloor+1)b\Z$.
Therefore ${\mathcal D}_{a, b, c}\ne \emptyset$ and hence  the corresponding  Gabor system
${\mathcal G}(\chi_{[0,c)}, a\Z\times \Z/b)$  is not a Gabor frame by Theorem \ref{framenullspace1.tm}.

In the case that ${\rm gcd}(\lfloor c/b\rfloor+1, p)=\lfloor c/b\rfloor+1$, we have that $p\ge 2$ and we define
${\bf x}(\lambda)=1$ if $\lambda\in \{0, \lambda_1\}+(\lfloor c/b\rfloor+1)b\Z$
and ${\bf x}(\lambda)=0$ otherwise, where $\lambda_1\in b\Z$ is chosen so that $b\le \lambda_1\le (p-1) b$ and $q\lambda_1/b-1\in p\Z$.
Then
\begin{eqnarray}\label{maintheorem.pf.nullspace1.eq6}
{\bf M}_{a,b, c}(0) {\bf x}(\mu) & = &  \sum_{k\in \Z} \chi_{[0,c)}\big(k(\lfloor c/b\rfloor+1)b-\mu\big)\nonumber\\
& & +
\sum_{k\in \Z} \chi_{[0,c)}\big(k(\lfloor c/b\rfloor+1)b+\lambda_1 -\mu\big)\nonumber\\
& = &
\sum_{k\in \Z} \chi_{[0,c)}\big(k(\lfloor c/b\rfloor+1)b+s_3 (\lfloor c/b\rfloor+1) b/q\big)\nonumber\\
& &  +
\sum_{k\in \Z} \chi_{[0,c)}\big(k(\lfloor c/b\rfloor+1)b+s_4 (\lfloor c/b\rfloor+1) b/q+ b/q\big) \nonumber \\
& = & \chi_{[0,c)}(s_3 (\lfloor c/b\rfloor+1) b/q)\nonumber\\
& & +
\chi_{[0,c)}(s_4 (\lfloor c/b\rfloor+1)b/q+ b/q)=2
\end{eqnarray}
where
 $0\le s_3, s_4\le  q  -1$ are so chosen that
 $-\mu-s_3  (\lfloor c/b\rfloor+1)b/q\in (\lfloor c/b\rfloor+1)b\Z$
 and $-\mu+\lambda_1-b/q-s_4  (\lfloor c/b\rfloor+1)b/q\in (\lfloor c/b\rfloor+1)b\Z$,
 and the last equality follows from the assumption on $c_0$.
 Therefore ${\mathcal D}_{a, b, c}\ne \emptyset $  in
 the case that ${\rm gcd}(\lfloor c/b\rfloor+1, p)=\lfloor c/b\rfloor+1$
 and hence the Gabor system ${\mathcal G}(\chi_{[0,c)}, a\Z\times \Z/b)$  is not a Gabor frame by Theorem \ref{framenullspace1.tm}.

  \smallskip

  (VII):\quad  Let $\tilde {\bf M}_{a, b, c}(t), t\in \R$, be as in \eqref{tildeinfinitematrix.def} and define
\begin{equation}\label{tildedabc2.def} \tilde {\mathcal D}_{a, b, c}=c-{\mathcal D}_{a, b, c}.\end{equation}
Then
\begin{equation}\label{tildedabc2.def2} \tilde {\mathcal D}_{a, b, c}=\big\{t\in \R:\  \tilde {\bf M}_{a, b, c}(t) {\bf x}={\bf 2}
\ {\rm for \ some} \ {\bf x}\in {\mathcal B}_b^0\big\}.
\end{equation}
 We will apply an argument  similar to the one used in the proof of the conclusion (VI)  essentially replacing
 ${\mathcal M}_{a, b, c}(t)$, ${\mathcal D}_{a, b, c}$  and $\lfloor c/b\rfloor+1$ by $\tilde {\mathcal M}_{a, b, c}(t)$,
  $\tilde {\mathcal D}_{a, b, c}$ and $\lfloor c/b\rfloor$
respectively.

  ($\Longrightarrow$)\ Let $t_0\in \tilde {\mathcal D}_{a, b, c}$ for some $t_0\in (0, a]$. The existence of such a  number $t_0$ follows from
 \eqref{sabcinvariant.eq1}, \eqref{tildedabc2.def} and Theorem \ref{framenullspace1.tm}.
Then by \eqref{tildeinfinitematrix.def} and \eqref{tildedabc2.def2},
 $\tilde {\bf M}_{a, b, c}(t_0) {\bf x}={\bf 2}$ for some  ${\bf x}\in {\mathcal B}_b^0$.
Following the argument to prove \eqref{frammatrix.tm.pf.eq4}, we can write ${\bf x}={\bf x}_0+{\bf x}_1$ for some
${\bf x}_0\in {\mathcal B}_b^0$ and ${\bf x}_1\in {\mathcal B}_b$
satisfying $\tilde {\bf M}_{a, b, c}(t_0) {\bf x}_0=\tilde {\bf M}_{a, b, c}(t_0) {\bf x}_1={\bf 1}$, and let
 $\lambda_i, i=1,2, 3$, be the first three positive indices in
 $b\Z$ with ${\bf x}(\lambda_i)=1, i=1, 2,3$.
 Then $\lambda_2$ is the first positive index in $b\Z$ with ${\bf x}_0(\lambda_2)=1$,
     and  $\lambda_2=\lambda_3-\lambda_1=\lfloor c/b\rfloor b$
by following the argument used in the proof of Proposition \ref{dabcbasic1.prop}.
Furthermore  $t_0\in (c_0, a],
 t_0+\lambda_1\in (c_0, a]+a\Z$,
 and   the unique number $t_1\in (0,a]$ with $t_1-t_0-\lambda_2\in a\Z$
 belongs to $\tilde {\mathcal D}_{a, b, c}$. Inductively for any $n\ge 1$, the unique number $t_n\in (0,a]$ satisfying
  $t_n-t_{n-1}- \lfloor c/b\rfloor b
 \in a\Z$  belongs to $\tilde {\mathcal D}_{a, b, c}\cap (c_0, a]$ and $t_n+\lambda_1\in \tilde {\mathcal D}_{a, b, c}\cap ((c_0, a]+a\Z)$.
  From the above construction,  $t_n-t_0-n\lfloor c/b\rfloor b\in a\Z$.
  Hence $a/b\in \Q$ if $c_0\ne 0$, as otherwise $\{t_n: \ n\ge 0\}$ is dense in $(0,a]$, which is a contradiction.
  Now consider $c_0\ne 0$ and write $a/b=p/q$ for some coprime integers $p$ and $q$, set
  $m= {\rm gcd}(\lfloor c/b\rfloor, p)$.
Let $t_0^\prime$ and $t^{\prime\prime}_0$ be the unique number in $(0,  m b/q]$
  such that $t_0^\prime-t_0\in m b\Z/q$  and  $t_0^{\prime\prime}-t_0-\lambda_1\in m b\Z/q$.
Then
 \begin{equation}\label{maintheorem.pf.eq4+}
\{t_n: n\ge 0\}+a\Z 
  =  \{ t_0^\prime+n m b/q: 0\le n\le p/m-1 \}+a\Z,
 \end{equation}
and
  \begin{equation}\label{maintheorem.pf.eq5+}
\{t_n+\lambda_1: n\ge 0\}+a\Z   =  \{ t_0^{\prime\prime}+n m b/q: 0\le n\le p/m-1 \}+a\Z.
 \end{equation}
    Hence from \eqref{maintheorem.pf.eq4+}, \eqref{maintheorem.pf.eq5+} and
  the property that $t_n, t_n+\lambda_1\in   (c_0, a]+a\Z, n\ge 1$, it follows that
  $\{ t_0^\prime+n m b/q: 0\le n\le p/m-1 \}\subset
 (c_0, a]$ and   $\{ t_0^{\prime\prime}+n m b/q: 0\le n\le p/m-1 \}\subset
 (c_0,a]$. Observe that $t^{\prime\prime}_0-t_0^\prime\in \lambda_1+mb\Z/q\subset b\Z/q$ and
 that  $t^{\prime\prime}_0-t_0^\prime\ne 0$ if $m=\lfloor c/b\rfloor$ as otherwise $\lambda_1\in mb\Z$, which is a contradiction.
 Therefore
 $0<c_0<mb/q$ if $m\ne \lfloor c/b\rfloor$
and
 $0<c_0<mb/q-b/q$ if $m= \lfloor c/b\rfloor$.
This completes the proof of the necessity.

 ($\Longleftarrow$)\  In the case that $c_0=0$,  we define
 ${\bf x}(\lambda)=1$ if $\lambda\in \{0, b\}+\lfloor c/b\rfloor b\Z$ and ${\bf x}(\lambda)=0$  otherwise.
 Then ${\bf x}\in {\mathcal B}_b^0$ and
\begin{equation*}\label{maintheorem.pf.nullspace1.eq9}
 {\bf M}_{a, b,c}(0) {\bf x}(\mu) =  \sum_{k\in \Z} \chi_{[0,c)}(-\mu+kc)+\chi_{[0,c)}(-\mu+b+kc)=2\end{equation*}
for all $\mu\in a\Z$. Then ${\mathcal D}_{a, b, c}\ne \emptyset$ and ${\mathcal G}(\chi_{[0,c)}, a\Z\times \Z/b)$ is not a Gabor frame by Theorem
\ref{framenullspace1.tm}.

In the case that $a/b=p/q$ for some coprime integers $p$ and $q$, $0<c_0< {\rm gcd}(\lfloor c/b\rfloor, p)b/q$
and ${\rm gcd}(\lfloor c/b\rfloor, p)\ne \lfloor c/b\rfloor$, we define
${\bf x}(\lambda)=1$ if $\lambda\in \{0 , {\rm gcd}(\lfloor c/b\rfloor, p) b \}+\lfloor c/b\rfloor b\Z $ and ${\bf x}(\lambda)=0$  otherwise.
Then ${\bf x}\in {\mathcal B}_b^0$ and
\begin{eqnarray} \label{maintheorem.pf.nullspace1.eq10}
\tilde {\bf M}_{a, b, c}(0) {\bf x}(\mu)
 &= & \sum_{k\in\Z} \chi_{(0,c]}(s_1 {\rm gcd}(\lfloor c/b\rfloor, p) b/q+ k \lfloor c/b\rfloor b)\nonumber\\
 & &  + \sum_{k\in\Z} \chi_{(0,c]}(s_2 {\rm gcd}(\lfloor c/b\rfloor, p) b/q + k \lfloor c/b\rfloor b)=2, 
\end{eqnarray}
where $1\le s_1, s_2\le q \lfloor c/b\rfloor/{\rm gcd}(\lfloor c/b\rfloor, p)$ are so chosen that
$-\mu-s_1 {\rm gcd}(\lfloor c/b\rfloor, p) b/q\in \lfloor c/b\rfloor b \Z$
and $-\mu+{\rm gcd}(\lfloor c/b\rfloor, p)b  -s_2 {\rm gcd}(\lfloor c/b\rfloor, p) b/q\in \lfloor c/b\rfloor b \Z$.
Hence $0\in \tilde {\mathcal D}_{a, b, c}$ (or equivalently $c\in {\mathcal D}_{a, b, c}$ by \eqref{tildedabc2.def2}) and
${\mathcal G}(\chi_{[0,c)}, a\Z\times \Z/b)$ is not a Gabor frame by Theorem
\ref{framenullspace1.tm}.

In the case that $a/b=p/q$ for some coprime integers $p$ and $q$, $0<c_0< {\rm gcd}(\lfloor c/b\rfloor, p)b/q-b/q$
and ${\rm gcd}(\lfloor c/b\rfloor, p)=\lfloor c/b\rfloor$, we have that $\lfloor c/b\rfloor\ge 2$
and define
${\bf x}(\lambda)=1$ if $\lambda\in \{0, \lambda_1\}+\lfloor c/b\rfloor b\Z$ and ${\bf x}(\lambda)=0$ otherwise, where
$b\le \lambda_1\le (\lfloor c/b\rfloor -1)b$ is so chosen that $q\lambda_1/b +1\in \lfloor c/b\rfloor \Z$.
Then ${\bf x}\in {\mathcal B}_b^0$ and
\begin{eqnarray}  \label{maintheorem.pf.nullspace1.eq11}
\tilde {\bf M}_{a, b, c}(0) {\bf x}(\mu)
 &= & \sum_{k\in\Z} \chi_{(0,c]}(s_3 {\rm gcd}(\lfloor c/b\rfloor, p) b/q+ k \lfloor c/b\rfloor b)\nonumber\\
 & &  + \sum_{k\in\Z} \chi_{(0,c]}(s_4 {\rm gcd}(\lfloor c/b\rfloor, p) b/q-b/q + k \lfloor c/b\rfloor b)=2,
\end{eqnarray}
where $1\le s_3, s_4\le q \lfloor c/b\rfloor/{\rm gcd}(\lfloor c/b\rfloor, p)$ are so chosen that
$-\mu-s_3 {\rm gcd}(\lfloor c/b\rfloor, p) b/q\in \lfloor c/b\rfloor b \Z$
and $-\mu+\lambda_1-b/q  -s_4 {\rm gcd}(\lfloor c/b\rfloor, p) b/q\in \lfloor c/b\rfloor b \Z$.
This proves that $0\in \tilde {\mathcal D}_{a, b, c}$ and hence
$c\in {\mathcal D}_{a, b, c}$ by \eqref{tildedabc2.def2}. Thus
 ${\mathcal G}(\chi_{[0,c)}, a\Z\times \Z/b)$ is not a Gabor frame by Theorem
\ref{framenullspace1.tm}.
\end{proof}

We finish this subsection with the proof of Proposition \ref{dabcbasic1.prop}.

\begin{proof} [Proof of Proposition \ref{dabcbasic1.prop}]
   Let $t\in {\mathcal S}_{a, b, c}$  and ${\bf M}_{a, b, c}(t) {\bf x}={\bf 1}$ for some ${\bf x}\in {\mathcal B}_b^0$.
 By \eqref{sabcinvariant.eq2}, we may assume that
$t\in [0,a)$. Let $\lambda_1\in b\Z$ be the smallest positive index such that ${\bf x}(\lambda_1)=1$.
 Then
$ \lambda_1\ge \lfloor c/b\rfloor b$
 because
 \begin{equation}\label{dabcbasic1.lem.pf.eq1}
 1=\chi_{[0,c)}(t)\le \chi_{[0,c)}(t)+ \chi_{[0,c)}(t+\lambda_1)\le \sum_{\lambda\in b\Z} \chi_{[0,c)}(t+\lambda) {\bf x}(\lambda)=1,\end{equation}
 and
$ \lambda_1\le \lfloor c/b\rfloor b+b$
 since otherwise
  $\sum_{\lambda\in b\Z} \chi_{[0,c)}(t-a+\lambda) {\bf x}(\lambda)=0$.
Thus  $\lambda_1$ is either $\lfloor c/b\rfloor b$ or $\lfloor c/b\rfloor b+b$. If $\lambda_1=\lfloor c/b\rfloor b$,
 then $t\ge c_0$ (and hence $ \min(c_0, a)\le t<a$)
  by \eqref{dabcbasic1.lem.pf.eq1}.
 If $\lambda_1=\lfloor c/b\rfloor b+b$, then  $ t<c_0+a-b$ as $1=\sum_{\lambda\in b\Z} \chi_{[0,c)}(t-a+\lambda){\bf x}(\lambda)
=\chi_{[0,c)}(t-a+\lfloor c/b\rfloor b+b)$.
 This proves that  $t\not\in [(c_0+a-b)_+, \min(c_0, a))$
 and $\lambda_1=\lambda_{a,b,c}(t)$.
 Hence
\begin{equation}\label{dabcbasic1.lem.pf.eq2}
 {\mathcal S}_{a, b, c}\cap \big([(c_0+a-b)_+, c_0\wedge a)+a\Z\big)=\emptyset
 \end{equation}
 and
  \begin{equation} \label{dabcbasic1.lem.pf.eq3}
 {\bf x}(\lambda)=\left\{\begin{array}{ll} 0 & {\rm  if} \  0<\lambda<\lambda_{a,b,c}(t),\\
1 & {\rm  if} \ \lambda=\lambda_{a,b,c}(t).\end{array}\right.
\end{equation}

For the above vector ${\bf x}$, we have that
$\tilde {\bf M}_{a, b, c} (c-t) \tilde {\bf x}={\bf 1}$,
where  $\tilde {\bf M}_{a, b, c}(t)$ is given in \eqref{tildeinfinitematrix.def}
and
$\tilde {\bf x}=({\bf x}(-\lambda))_{\lambda\in b\Z}\in {\mathcal B}_b^0$. Then mimicking the above argument
to establish \eqref{dabcbasic1.lem.pf.eq2} and \eqref{dabcbasic1.lem.pf.eq3}, we obtain that
\begin{equation}\label{dabcbasic1.lem.pf.eq4}
 {\mathcal S}_{a, b, c}\cap \big([c-c_0\wedge a, c-(c_0+a-b)_+)+a\Z\big)=\emptyset,
 \end{equation}
 and
  \begin{equation} \label{dabcbasic1.lem.pf.eq5}
 \tilde {\bf x}(\lambda)=\left\{\begin{array}{ll} 0 & {\rm  if} \  0<\lambda<-\tilde\lambda_{a,b,c}(t)\\
1 & {\rm  if} \ \lambda=-\tilde \lambda_{a,b,c}(t).\end{array}\right.
\end{equation}
Combining \eqref{dabcbasic1.lem.pf.eq2}--\eqref{dabcbasic1.lem.pf.eq5}, we obtain \eqref{dabcbasic1.lem.eq1} and \eqref{dabcbasic1.lem.eq3}.
\end{proof}

\subsection{Binary solutions of infinite-dimensional linear systems}
In this subsection, we prove Theorem \ref{framenullspace2.tm}.
To do so, we first recall some basic properties concerning
black holes, ranges, invertibility of   piecewise linear transformations
$R_{a,b,c}$ and $\tilde R_{a,b,c}$ on the real line. 

\begin{prop}\label{rabcinvertibility.prop}
 Let $(a,b,c)$ be a triple of positive numbers satisfying $a<b<c$ and $b-a<c_0:=c-\lfloor c/b\rfloor b< a$,
 and let $R_{a,b,c}$ and $\tilde R_{a, b,c}$ be infinite matrices in  \eqref{rabcnewplus.def} and \eqref{tilderabcnewplus.def}.
 Then
 \begin{equation}\label{blackholesoftransformations}
 \left\{\begin{array}{l}
 R_{a,b,c}(t)= t  \quad  {\rm if} \ t\in [c_0+a-b, c_0)+a\Z\\
  \tilde R_{a,b,c}(t)= t  \quad  {\rm if} \ t\in [c-c_0, c+b-c_0-a)+a\Z,
 \end{array}
 \right.
 \end{equation}
  \begin{equation}\label{rangeoftransformations}
   \left\{\begin{array}{l}
 R_{a,b,c}\big(\R\backslash ([c_0+a-b, c_0)+a\Z)\big) 
 =\R\backslash([c-c_0, c+b-c_0-a)+a\Z)  \\
  \tilde R_{a,b,c}\big(\R\backslash([c-c_0, c+b-c_0-a)+a\Z)\big) 
 =\R\backslash([c_0+a-b, c_0)+a\Z), 
 \end{array}
 \right.
 \end{equation}
and
   \begin{equation}\label{invertibilityoftransformations}
   \left\{\begin{array}{l} \tilde R_{a,b,c}(R_{a,b,c}(t))=t \quad {\rm for \ all} \ \ t\not\in [c_0+a-b, c_0)+a\Z\\
    R_{a,b,c}(\tilde R_{a,b,c}(t))=t \quad  {\rm for \ all } \ \ t\not\in [c- c_0, c+b-c_0-a)+a\Z. 
   \end{array}\right.
   \end{equation}
%
\end{prop}

 The black hole property \eqref{blackholesoftransformations}, the range property \eqref{rangeoftransformations},
and the  invertibility property \eqref{invertibilityoftransformations}
outside  black holes follow  directly from   \eqref{rabcnewplus.def} and \eqref{tilderabcnewplus.def}. We omit the detailed arguments here.
%

\smallskip

To prove Theorem \ref{framenullspace2.tm},
we  then show that
the  piecewise linear transformations $R_{a, b,c}$ and
$\tilde R_{a, b,c}$ have their restrictions on  the sets ${\mathcal D}_{a, b, c}$ and  ${\mathcal S}_{a, b, c}$
being well-defined, invariant,  invertible and
measure-preserving.

\begin{prop}
\label{rabcbasic1.thm}
 Let $(a,b,c)$ be a triple of positive numbers satisfying $a<b<c$ and $b-a<c_0:=c-\lfloor c/b\rfloor b< a$, the sets
 ${\mathcal D}_{a, b, c}$ and ${\mathcal S}_{a, b, c}$ be as in
 \eqref{dabc2.def} and \eqref{dabc1.def}, and  let  transformations $R_{a,b,c}$ and
 $\tilde R_{a, b,c}$ be as in  \eqref{rabcnewplus.def}
and \eqref{tilderabcnewplus.def}.
 Then the following statements hold.
\begin{itemize}
\item[{(i)}] The sets ${\mathcal D}_{a,b,c}$ and ${\mathcal S}_{a, b, c}$ are invariant under the
 transformations $R_{a,b,c}$ and $\tilde R_{a, b,c}$;
i.e., \begin{equation}\label{rabcbasic1.thm.eq1}
R_{a,b,c}  {\mathcal D}_{a,b,c}=\tilde  R_{a,b,c}  {\mathcal D}_{a,b,c}= {\mathcal D}_{a,b,c}\ {\rm and} \
R_{a,b,c} {\mathcal S}_{a, b, c}=\tilde R_{a,b,c}  {\mathcal S}_{a,b,c}= {\mathcal S}_{a,b,c}.\end{equation}

\item[{(ii)}] The restriction of  $\tilde {R}_{a,b,c}$ onto  ${\mathcal S}_{a,b,c}$ (resp. ${\mathcal D}_{a,b,c}$)
is the inverse of the restriction of  $R_{a, b,c}$  onto  ${\mathcal S}_{a,b,c}$ (resp. ${\mathcal D}_{a,b,c}$), and vice versa; i.e.,
\begin{equation}\label{rabcbasic1.thm.eq1b}
{R}_{a, b, c} (\tilde {R}_{a, b, c}  (t))=
\tilde {R}_{a, b, c} ( {R}_{a, b, c} (t))=t \ {\rm for\ all} \  t\in {\mathcal S}_{a,b,c}\  ({\rm resp}. \ {\mathcal D}_{a,b,c}).
\end{equation}

\item [{(iii)}] The restriction of  the piecewise linear transformations $R_{a, b,c}$ and $\tilde {R}_{a,b,c}$
 onto  ${\mathcal S}_{a,b,c}$ (resp.  ${\mathcal D}_{a,b,c}$) are measure-preserving transformations; i.e.,
\begin{equation}\label{rabcbasic1.thm.eq1c} |R_{a, b, c}(E)|= |\tilde R_{a, b, c}(E)|=|E|
\end{equation}
for any measurable set $E\subset {\mathcal S}_{a, b, c}$ (resp. ${\mathcal D}_{a,b,c}$).

\end{itemize}
\end{prop}

We will use Proposition \ref{rabcbasic1.thm}  to
establish  the uniqueness of solution ${\bf x}\in {\mathcal B}_b^0$ for
the infinite-dimensional linear system ${\bf M}_{a, b, c} (t) {\bf
x}={\bf 1}$, which
is a pivotal observation and our starting point  to  explore various properties of the sets ${\mathcal D}_{a, b, c}$ and ${\mathcal S}_{a, b, c}$ deeply.

 \begin{prop}\label{uniqueness.cor}
 Let $(a,b,c)$ be a triple of positive numbers satisfying $a<b<c$ and $b-a<c_0:=c-\lfloor c/b\rfloor b< a$, and let the set
 ${\mathcal S}_{a, b, c}$ be as in \eqref{dabc1.def}.  Then for any $t\in {\mathcal S}_{a, b, c}$ there exists
 a unique vector ${\bf x}\in {\mathcal B}_b^0$ such that
 ${\bf M}_{a, b, c} (t) {\bf x}={\bf 1}$.
\end{prop}

We will also apply  Proposition \ref{rabcbasic1.thm} to show that the black
hole of the piecewise linear transformation $\tilde R_{a, b,c}$ has empty intersection with the invariant set ${\mathcal S}_{a, b, c}$,
which plays important roles in the explicit construction of the maximal invariant set ${\mathcal S}_{a, b, c}$ in Theorems \ref{dabc1holes.tm} and \ref{dabc1discreteholes.tm}.

\begin{prop}\label{blackholestwo.prop}
 Let $(a,b,c)$ be a triple of positive numbers satisfying $a<b<c$ and $b-a<c_0:=c-\lfloor c/b\rfloor b< a$,  the set
 ${\mathcal S}_{a, b, c}$ be as in \eqref{dabc1.def}, and
 let the transformation $R_{a,b,c}$  be as in  \eqref{rabcnewplus.def}.
  Then
  \begin{equation} \label{blackholestwo.eq1}
(R_{a,b,c})^n([c-c_0, c-c_0+b-a)+a\Z)\cap {\mathcal S}_{a, b, c}= \emptyset
\quad
{\rm  for \ all} \ n\ge 0.
\end{equation}
\end{prop}

\smallskip

To prove Theorem \ref{framenullspace2.tm},
we  finally need the following density property of the sets ${\mathcal S}_{a, b, c}$
and ${\mathcal D}_{a, b, c}$
 around the origin, 
c.f. Lemma \ref{dabc1rationalmax.prop} for similar density property of the set ${\mathcal S}_{a, b, c}$ when the ratio between $a$ and $b$ is rational.

\begin{lem}\label{rabcbasic2.lem}
 Let $(a,b,c)$ be a triple of positive numbers satisfying $a<b<c, b-a<c_0:=c-\lfloor c/b\rfloor b< a$ and $a/b\not\in \Q$,
  the piecewise linear transformation $R_{a,b,c}$  be as in  \eqref{rabcnewplus.def}, and let
  ${\mathcal D}_{a, b, c}$ and
 ${\mathcal S}_{a, b, c}$ be as in \eqref{dabc2.def} and
\eqref{dabc1.def} respectively.
Then the following conclusions hold.
\begin{itemize}
\item[{(i)}] If ${\mathcal S}_{a, b, c}\ne
\emptyset$, then
\begin{equation}\label{rabcbasic2.lem.eq1}
(0,\epsilon)\cap {\mathcal S}_{a,b,c}\ne \emptyset \ {\rm and} \ (-\epsilon, 0)\cap {\mathcal S}_{a, b, c}\ne \emptyset
\end{equation}
for any $\epsilon>0$.

\item [{(ii)}] If ${\mathcal D}_{a, b, c}\ne
\emptyset$, then \begin{equation}\label{rabcbasic2.lem.eq1d}
(0,\epsilon)\cap {\mathcal D}_{a,b,c}\ne \emptyset \ {\rm and} \ (-\epsilon, 0)\cap {\mathcal D}_{a, b, c}\ne \emptyset
\end{equation}
for any $\epsilon>0$.
\end{itemize}

   \end{lem}

We prove  Theorem \ref{framenullspace2.tm} by assuming that Proposition \ref{rabcbasic1.thm} and Lemma \ref{rabcbasic2.lem} hold.

\begin{proof}[Proof of Theorem \ref{framenullspace2.tm}]  The necessity is obvious.

Now the sufficiency. By \eqref{maintheorem.pf.nullspace1.eq9}, \eqref{maintheorem.pf.nullspace1.eq5}, \eqref{maintheorem.pf.nullspace1.eq6},
 \eqref{maintheorem.pf.nullspace1.eq10},
\eqref{maintheorem.pf.nullspace1.eq11}, and   Theorems \ref{newmaintheorem1} and \ref{newmaintheorem2},
it  remains to prove the sufficiency for  the cases that  $b-a<c_0<a$ and $a/b\not\in \Q$ and that  $b-a<c_0<a$ and $a/b\in\Q$.

{\bf Case 1}: $b-a<c_0<a$ and $a/b\not\in \Q$.

 Suppose on the contrary that
  ${\mathcal D}_{a, b, c}\ne \emptyset$. Then  by
  \eqref{sabcinvariant.eq2}, Lemma \ref{rabcbasic2.lem}, and the assumption that ${\mathcal D}_{a, b, c}\cap {\mathcal Z}_{a, b}
=\emptyset$, there exist $t_n\in {\mathcal D}_{a, b, c}\cap (0, a), n\ge 1$ such that $\{t_n\}_{n=1}^\infty$ is a decreasing sequence that converges to zero.
Following the argument used in the proof of Lemma \ref{dabc2toqabc.lem},
we obtain that
$0\in {\mathcal D}_{a, b, c}$. This leads to the contradiction.


 {\bf Case 2}: $b-a<c_0<a$ and $a/b\in \Q$.

 Write $a/b=p/q$ for some coprime integers $p$ and $q$.
Suppose on the contrary that ${\mathcal D}_{a, b, c}\ne \emptyset$; i.e.,
 $ {\bf M}_{a,b, c}(t) {\bf x}={\bf 2}$ for some $t\in \R$ and
 ${\bf x}\in {\mathcal B}_b^0$.
 By \eqref{sabcinvariant.eq2} and the assumption that ${\mathcal D}_{a, b, c}\cap {\mathcal Z}_{a,b}=\emptyset$,
we may assume that $t\in [0, pb/q)\backslash \{0, b/q, \ldots, b(p-1)/q\}$ without loss of generality.
If $0\ne t-\lfloor qt/b\rfloor b/q<
 c-\lfloor qc/b\rfloor b/q$, then
 $\chi_{[0,c)}(\lfloor qt/b\rfloor b/q-\mu+\lambda)=\chi_{[0,c)}(t-\mu+\lambda)$ for all $\mu\in a\Z$ and $\lambda\in b\Z$. This  implies that
$$
 {\bf M}_{a,b, c}(\lfloor qt/b\rfloor b/q) {\bf x}
= {\bf M}_{a,b, c}(t) {\bf x}={\bf 2},
$$
and hence $\lfloor qt/b\rfloor b/q\in {\mathcal D}_{a, b,c}\cap {\mathcal Z}_{a, b}$, which is a contradiction.

 If $0\ne t-\lfloor qt/b\rfloor b/q\ge
 c-\lfloor qc/b\rfloor b/q$, then
   $\chi_{(0,c]}((\lfloor qt/b\rfloor+1) b/q-\mu+\lambda)=\chi_{[0, c)}(t-\mu+\lambda)$
for  all $\mu\in a\Z$ and $\lambda\in b\Z$. Hence
 ${\bf M}_{a,b, c}\big(c-(\lfloor qt/b\rfloor+1) b/q\big) \tilde {\bf x}={\bf 2}$, where
 $\tilde {\bf x}=({\bf x}(-\lambda))_{\lambda\in b\Z}\in {\mathcal B}_b^0$.
This contradicts to the assumption that ${\mathcal D}_{a, b, c}\cap (c-{\mathcal Z}_{a, b})=\emptyset$.
%
%
%
%
%
%
%
%
%
 \end{proof}

We finish this subsection with proofs of Proposition \ref{rabcbasic1.thm}, Lemma \ref{rabcbasic2.lem}, and Propositions \ref{uniqueness.cor}
and \ref{blackholestwo.prop}.
%
%
%

\begin{proof} [Proof of Proposition \ref{rabcbasic1.thm}]
(i)\quad By \eqref{sabcinvariant.eq1} and Proposition
\ref{dabcbasic1.prop}, we have that
\begin{equation}\label{rabcbasic1.lem.pf.eq1} R_{a,b,c} {\mathcal S}_{a, b, c}\subset {\mathcal S}_{a, b, c}\end{equation}
and
\begin{equation} \label{rabcbasic1.lem.pf.eq2} \tilde R_{a,b,c}  {\mathcal S}_{a, b, c}\subset {\mathcal S}_{a, b, c}.\end{equation}
Observe that
 \begin{equation}\label{rabcbasic1.lem.pf.eq3}
 \tilde R_{a,b,c} (R_{a,b,c}(t))= R_{a,b,c} (\tilde R_{a,b,c}(t))= t\quad {\rm for \ all} \ t\in  {\mathcal S}_{a, b, c} \end{equation}
 by Propositions \ref{dabcbasic1.prop} and \ref{rabcinvertibility.prop}.
Hence \begin{equation}\label{rabcbasic1.thm.eq1part1}
R_{a,b,c} {\mathcal S}_{a, b, c}=\tilde R_{a,b,c} {\mathcal S}_{a, b, c}={\mathcal S}_{a, b, c}\end{equation}
by
\eqref{rabcbasic1.lem.pf.eq1}, \eqref{rabcbasic1.lem.pf.eq2} and \eqref{rabcbasic1.lem.pf.eq3}.

Let $t\in {\mathcal D}_{a, b, c}$ and
$ {\bf M}_{a, b, c} (t) {\bf x}={\bf 2}$
for some ${\bf x}\in {\mathcal B}_b^0$. By \eqref{frammatrix.tm.pf.eq4}, there is a decomposition
${\bf x}={\bf x}_0+{\bf x}_1$ such that ${\bf x}_0\in {\mathcal B}_b^0,  {\bf x}_1\in {\mathcal B}_b$
and ${\bf M}_{a, b, c}(t){\bf x}_0=  {\bf M}_{a, b, c}(t) {\bf x}_1={\bf 1}$. Therefore
$\tau_{\lambda_{a,b,c}(t)} {\bf x}\in {\mathcal B}_b^0$ by Proposition \ref{dabcbasic1.prop}
and ${\bf M}_{a, b, c} (R_{a,b,c}(t)) \tau_{\lambda_{a,b,c}(t)}{\bf x}=
 {\bf M}_{a, b, c} (t) {\bf x}=
{\bf 2}$. This proves that
\begin{equation}\label{rabcbasic1.lem.pf.eq4}
R_{a,b,c} {\mathcal D}_{a, b, c}\subset {\mathcal D}_{a, b, c}.
\end{equation}
Similarly we have that
\begin{equation}\label{rabcbasic1.lem.pf.eq5}
\tilde R_{a,b,c}  {\mathcal D}_{a, b, c}\subset {\mathcal D}_{a, b, c}.
\end{equation}
Recalling that ${\mathcal D}_{a, b, c}\subset {\mathcal S}_{a, b, c}$ and combining \eqref{rabcbasic1.lem.pf.eq3}, \eqref{rabcbasic1.lem.pf.eq4}
and \eqref{rabcbasic1.lem.pf.eq5},  we
obtain
\begin{equation}\label{rabcbasic1.thm.eq1part2}
R_{a,b,c} {\mathcal D}_{a, b, c}=\tilde R_{a,b,c} {\mathcal D}_{a, b, c}={\mathcal D}_{a, b, c}.\end{equation}
Therefore
\eqref{rabcbasic1.thm.eq1} follows from  \eqref{rabcbasic1.thm.eq1part1} and \eqref{rabcbasic1.thm.eq1part2}.

 (ii)\quad  The invertibility of the transformations
$R_{a, b,c}$ and $\tilde R_{a, b, c}$ in the second conclusion follows from
\eqref{dabcbasic1.lem.eq1} and \eqref{invertibilityoftransformations}.

(iii)\quad By \eqref{rabcnewplus.def},
\begin{equation}\label{rabcbasic1.thm.pf.eq-1}|R_{a,b,c}(E)|\le |E|
\end{equation}
for any measurable set $E$, and
the above inequality becomes an equality,
\begin{equation}\label{rabcbasic1.thm.pf.eq0}| R_{a,b,c}(E)|=|E|,
\end{equation}
whenever $E$ has empty intersection with the black hole $[c_0+a-b, c_0)+a\Z$ of the transformation $R_{a,b,c}$.
This, together with \eqref{dabc2subsetdabc1}, \eqref{dabcbasic1.lem.eq1} and the first conclusion, proves that
$R_{a, b,c}$ is a measure-preserving transformation on the sets ${\mathcal D}_{a, b, c}$ and ${\mathcal S}_{a, b, c}$.

The measure-preserving property for the transformation $\tilde R_{a,b,c}$ can be established by applying similar argument.
\end{proof}

   \begin{proof} [Proof of Lemma \ref{rabcbasic2.lem}]\ (i):\quad Let  $t_n:=(R_{a, b, c})^n(t_0), n\ge 0$, be the orbit of
the transformation $R_{a,b,c}$ with  initial $t_0\in {\mathcal S}_{a, b,c}$,
 and set
   $\tilde t_n:=t_n-\lfloor t_n/a\rfloor a, n\ge 0$.
   Without loss of generality, we assume that $\tilde t_n\ne 0$ for all $n\ge 0$, because
   at most one in the sequence  $\tilde t_n, n\ge 0$, could be zero by the assumption $a/b\not\in \Q$
   and in that case we replace the initial $t_0$ by $t_{n_0}$ for some sufficiently large $n_0$.
Clearly  the proof of  the first non-empty-set conclusion in  \eqref{rabcbasic2.lem.eq1} reduces to proving that
   for sufficiently small $\epsilon>0$ there exists an index $n$ with the property that
   \begin{equation}\label{rabcbasic2.lem.pf.eq1}
   \tilde t_n\in (0,\epsilon).\end{equation}
  From \eqref{dabcperiodic.section2},
   Propositions \ref{dabcbasic1.prop} and  \ref{rabcbasic1.thm}, it follows that
   $\tilde t_n\in
   {\mathcal S}_{a, b, c}\cap ([0,c_0+a-b)\cap [c_0, a)), n\in \Z_+$.
   Recall that for $n\ne m$, $ t_n-t_{m}=kb$ for some $0\ne k\in \Z$ by \eqref{rabcnewplus.def}, which together with $a/b\not\in \Q$ implies that
   \begin{equation}\label{rabcbasic2.lem.pf.eq2}
   \tilde t_n-\tilde t_m\ne 0\ {\rm whenever} \ n\ne m.\end{equation}
As $\tilde t_n, n\ge 0$, lie in the bounded set $(0,a)$, there exist integers $n_1<n_2$ such that $|\tilde t_{n_1}-\tilde t_{n_2}|<\epsilon$.
   Without loss of generality, we assume that $\tilde t_{n_1}\not \in (0,\epsilon)$, otherwise the conclusion
   \eqref{rabcbasic2.lem.pf.eq1} holds by letting $n=n_1$.
   Recall that $\tilde t_n\in [0,c_0+a-b)\cap [c_0, a)$ for all $n\ge 0$,  we have that
 either $\tilde t_{n_1}, \tilde t_{n_2}\in [0, c_0+a-b)$ or
$\tilde t_{n_1}, \tilde t_{n_2}\in [c_0, a)$. This  implies that $\lambda_{a,b,c}(t_{n_1})=\lambda_{a,b,c}(t_{n_2})$,
where $\lambda_{a, b, c}(t)$ is defined in \eqref{lambdaabc.def}.
 If  $\lambda_{a,b,c}(t_{n_1+k})=\lambda_{a,b,c}(t_{n_2+k})$ for all positive integers $k$,
 then $t_{n_1+k}-t_{n_2+k}=t_{n_1}-t_{n_2}$ for all positive integers $k$. Applying the above equality with $k=(n_2-n_1)l, l\in \N$, we obtain that
 $t_{n_1+l(n_2-n_1)}=t_{n_1}+ l(t_{n_2}-t_{n_1})$ for all $l\in \N$.
 Therefore
$\tilde t_{n_1+l_0(n_2-n_1)}= \tilde t_{n_1}+l_0(\tilde t_{n_2}-\tilde t_{n_1})\in (0,\epsilon)$
for $l_0= \lfloor \tilde t_{n_1}/ (\tilde t_{n_1}-\tilde t_{n_2})\rfloor $ if
 $
 \tilde t_{n_2}-\tilde t_{n_1}<0$,  and
 $\tilde t_{n_1+(l_1+1)(n_2-n_1)}= \tilde t_{n_1}+(l_1+1)(\tilde t_{n_2}-\tilde t_{n_1})-a\in (0,\epsilon)$
for $l_1= \lfloor (a-\tilde t_{n_1})/ (\tilde t_{n_2}-\tilde t_{n_1})\rfloor $ if
 $\tilde t_{n_2}-\tilde t_{n_1}>0$. This together with \eqref{rabcbasic2.lem.pf.eq2} proves that  \eqref{rabcbasic2.lem.pf.eq1}
  holds when $\lambda_{a,b,c}(t_{n_1+k})=\lambda_{a,b,c}(t_{n_2+k})$ for all positive integers $k$.
 Now we consider the case that there exists a positive integer $m$ such that
 $\lambda_{a,b,c}(t_{n_1+k})=\lambda_{a,b,c}(t_{n_2+k})$ for all nonnegative  integer $k\in [0,m-1]$ and
 $\lambda_{a,b,c}(t_{n_1+m})\ne \lambda_{a,b,c}(t_{n_2+m})$.
 Then we can prove by induction that $t_{n_1+k}-t_{n_2+k}=t_{n_1}-t_{n_2}$ for all $0\le k\le m-1$. This implies that
  either $|\tilde t_{n_1+m-1}-\tilde t_{n_2+m-1}|= |\tilde t_{n_1}-\tilde t_{n_2}|$
 or  $|\tilde t_{n_1+m-1}-\tilde t_{n_2+m-1}|+|\tilde t_{n_1}-\tilde t_{n_2}|=a$.
 On the other hand,  it follows from  $\lambda_{a,b,c}(t_{n_1+m})\ne \lambda_{a,b,c}(t_{n_2+m})$ that
  $\tilde t_{n_1+m-1}, \tilde t_{n_2+m-1}$ should  lie in the different intervals $[0, c_0+a-b)$ and $[c_0, a)$.
 Therefore  $|\tilde t_{n_1+m-1}-\tilde t_{n_2+m-1}|+|\tilde t_{n_1}-\tilde t_{n_2}|=a$, which implies
 that one and only one of $\tilde t_{n_1+(m-1)}, \tilde t_{n_2+(m-1)}$ belongs to $(0,\epsilon)$, while the other one of
 $\tilde t_{n_1+(m-1)}, \tilde t_{n_2+(m-1)}$ belongs to $[a-\epsilon, a)$ . Hence
  \eqref{rabcbasic2.lem.pf.eq1}
  holds when $\lambda_{a,b,c}(t_{n_1+m})=\lambda_{a,b,c}(t_{n_2+m})$ for some positive integers $m$
  and we complete the proof of the conclusion  \eqref{rabcbasic2.lem.pf.eq1}.

The second conclusion in \eqref{rabcbasic2.lem.eq1} can proved by
using similar argument with  $l_0$ replaced by $l_0+1$, $l_1+1$
by $l_1$.

\smallskip

(ii):\quad  We can follow the argument of the first conclusion line by line except with
the set ${\mathcal S}_{a,b,c}$ replaced by its subset ${\mathcal D}_{a, b, c}$. We omit
the detailed arguments here.
%
%
%
   \end{proof}

\begin{proof} [Proof of Proposition \ref{uniqueness.cor}] \ Suppose on the contrary that
 ${\bf M}_{a, b, c} (t) {\bf x}_0={\bf M}_{a, b, c}(t) {\bf x}_1={\bf 1}$ for two
 distinct vectors ${\bf x}_0, {\bf x}_1\in {\mathcal B}_b^0$.
 Then there exists  $0\ne \lambda_0\in  b\Z$ such that
 ${\bf x}_0 (\lambda_0)\ne {\bf x}_1(\lambda_0)$ while ${\bf x}_0(\lambda)= {\bf x}_1(\lambda)$
 for all $|\lambda|<|\lambda_0|$. Without loss of generality, we assume that ${\bf x}_0(\lambda_0)=1$ and ${\bf x}_1(\lambda_0)=0$.
  Let $\lambda_1\in b\Z$ be so chosen that  $\lambda_1\lambda_0\ge 0$, $|\lambda_1|<|\lambda_0|$ and
$|\lambda_1-\lambda_0|= \min_{\lambda\in (-|\lambda_0|, |\lambda_0|)\cap b\Z, {\bf x}_0(\lambda)={\bf x}_1(\lambda)=1 }
 |\lambda-\lambda_0|$.
Thus  ${\bf M}_{a, b, c} (t-\lambda_1) \tau_{\lambda_1 } {\bf x}_0={\bf M}_{a, b, c} (t-\lambda_1) \tau_{\lambda_1} {\bf x}_1={\bf 1}$ by
  \eqref{sabcinvariant.eq2}, and   both $\tau_{\lambda_1}{\bf x}_0$ and $\tau_{\lambda_1}{\bf x}_1$
 belong to ${\mathcal B}_b^0$ by the selection of the index $\lambda_1$.
 As $\lambda_0-\lambda_1$ is the closest positive (or negative) index  to zero such that  $\tau_{\lambda_1}{\bf x}_0(\lambda_0-\lambda_1)=1$.
 Then it follows from Proposition \ref{dabcbasic1.prop} that
 $\tau_{\lambda_1}{\bf x}_1(\lambda_0-\lambda_1)=1$, which contradicts to ${\bf x}_1 (\lambda_0)=0$.
\end{proof}

\begin{proof}[Proof of Proposition \ref{blackholestwo.prop}]\
Suppose, on the contrary, that \eqref{blackholestwo.eq1} does not
hold.  Then there exists a nonnegative integer $m$ such that
$(R_{a,b,c})^m([c-c_0, c-c_0+b-a)+a\Z)\cap {\mathcal S}_{a,b,c}\ne \emptyset$ and $(R_{a,b,c})^n([c-c_0,
c-c_0+b-a)+a\Z)\cap {\mathcal S}_{a, b, c}= \emptyset$
 for all $0\le n<m$.
We observe that $m\ne 0$ (or equivalently $m$ is a positive integer)  by Proposition \ref{dabcbasic1.prop}.
Take $t\in (R_{a,b,c})^m([c-c_0, c-c_0+b-a)+a\Z)\cap {\mathcal S}_{a, b, c}$.
Then $t=R_{a,b,c}(s)$ for some $s\in (R_{a,b,c})^{m-1}([c-c_0, c-c_0+b-a)+a\Z)$. If $s\in [c_0+a-b, c_0)+a\Z$, then $t=s$
by \eqref{rabcnewplus.def}, which is a contradiction as  $([c_0+a-b, c_0)+a\Z)\cap {\mathcal S}_{a, b, c}=\emptyset$ by
Proposition \ref{dabcbasic1.prop}. If $s\not \in [c_0+a-b, c_0)+a\Z$, then $s=\tilde R_{a,b,c}(t)\in {\mathcal S}_{a, b, c}$
by \eqref{invertibilityoftransformations} 
and Proposition \ref{rabcbasic1.thm}. Hence $s\in (R_{a,b,c})^{m-1}([c-c_0, c-c_0+b-a)+a\Z) \cap {\mathcal S}_{a, b, c}$,
which is a contradiction.
\end{proof}

\section{
Maximal Invariant Sets of Piecewise Linear Transformations}
\label{gabornullspace2.section}

In this section, we prove Theorem \ref{newmaintheorem3}.
To do so, we need the most crucial observation of this paper  that
 ${\mathcal S}_{a, b, c}$  is
 the maximal set such that
\begin{equation} \label{rabc1invariant.tm.newweq2}
R_{a, b, c} {\mathcal S}_{a, b, c}\subset {\mathcal S}_{a, b, c} \quad {\rm and} \quad  \tilde R_{a, b, c} {\mathcal S}_{a, b, c}\subset {\mathcal S}_{a, b, c},
\end{equation}
and
\begin{equation}\label{rabc1invariant.tm.eq2}
{\mathcal S}_{a, b, c}\cap ([c- c_0, c-c_0+b-a)+a\Z)={\mathcal S}_{a, b, c}\cap ([c_0+a-b, c_0)+a\Z)=\emptyset,
\end{equation}
c.f. Lemma \ref{dabc0rationalmax.prop} 
for the case that $a/b\in \Q$.

\begin{thm}  \label{dabc1maximal.cor.thm}
 Let $(a,b,c)$ be a triple of positive numbers satisfying $a<b<c$ and $b-a<c_0:=c-\lfloor c/b\rfloor b< a$, the set
 ${\mathcal S}_{a, b, c}$ be as in
 \eqref{dabc1.def}, and  let  transformations $R_{a,b,c}$ and
 $\tilde R_{a, b,c}$ be as in  \eqref{rabcnewplus.def}
and \eqref{tilderabcnewplus.def}.
Then
\begin{itemize}
\item [{(i)}]  ${\mathcal S}_{a, b, c}$  is the  maximal  set
that is invariant under the piecewise linear transformations $R_{a, b, c}$
   and $\tilde R_{a, b, c}$, and that  has empty intersection with their black holes.

   \item [{(ii)}] The complement of  ${\mathcal S}_{a, b, c}$  is the minimal  set
that is invariant under the piecewise linear transformations $R_{a, b, c}$
   and $\tilde R_{a, b, c}$, and that  contains their black holes.

   \end{itemize}
\end{thm}

To prove Theorem \ref{newmaintheorem3} (and  the characterization of \eqref{sabcemptynonempty2} in Theorem  \ref{sabcstar.tm}),
 we need
 the following deep connection between  the invariant sets ${\mathcal D}_{a,b,c}$ and ${\mathcal S}_{a, b, c}$.

\begin{thm}\label{dabcsabc.thm} Let $(a, b, c)$ be a triple of positive numbers with $a<b<c$ and  $b-a<c_0:=c-\lfloor c/b\rfloor b <a$,
and let the sets
 ${\mathcal D}_{a, b, c}$ and ${\mathcal S}_{a, b, c}$ be as in
 \eqref{dabc2.def} and \eqref{dabc1.def}.
Then
\begin{eqnarray}\label{dabcsabc.thm.eq1}
{\mathcal D}_{a, b, c} & = &
\big({\mathcal S}_{a, b, c}\cap ([0, c_0+a-b)+a\Z)\cap ( {\mathcal S}_{a, b, c}-\lfloor c/b\rfloor b)\big)\nonumber\\
& & \cup \big({\mathcal S}_{a, b, c}\cap (\cup_{\lambda\in [b, (\lfloor c/b\rfloor -1) b]\cap b\Z} ({\mathcal S}_{a, b, c}-\lambda))\big).
\end{eqnarray}
\end{thm}

In next two subsections, we prove Theorems \ref{dabc1maximal.cor.thm} and \ref{dabcsabc.thm}, and
apply them to prove  Theorem \ref{newmaintheorem3}.

\subsection{Maximal invariant sets}
\label{subsection2.2}

%
%
%
%

To establish   maximality
in Theorem \ref{dabc1maximal.cor.thm}, we
 characterize whether a particular point
belongs to the set ${\mathcal S}_{a, b, c}$, which will also be used in the proofs of Theorems \ref{dabc1holes.tm} and \ref{dabc1discreteholes.tm}.

\begin{prop} \label{dabc1pointcharacterization.prop} Let $(a,b,c)$ be a triple of positive numbers satisfying $a<b<c$ and
 $b-a<c_0:=c-\lfloor c/b\rfloor b< a$,
 and let $R_{a,b,c}$ and $\tilde R_{a, b,c}$ be as in  \eqref{rabcnewplus.def} and \eqref{tilderabcnewplus.def}.
Then   $t\in {\mathcal S}_{a, b, c}$ if and only if  $(R_{a,b,c})^n (t)$ and $( \tilde R_{a, b,c})^n(t), n\ge 0$, do not belong to
 the black holes of the piecewise linear transformations $R_{a, b, c}$ and $\tilde R_{a, b, c}$.
 \end{prop}

Now assuming that  Proposition \ref{dabc1pointcharacterization.prop}
holds,   we
start to prove Theorem \ref{dabc1maximal.cor.thm}.

\begin{proof}[Proof of Theorem \ref{dabc1maximal.cor.thm}]
(i)\quad
By \eqref{dabcbasic1.lem.eq1} in Proposition  \ref{dabcbasic1.prop},
and
\eqref{rabcbasic1.thm.eq1} in Proposition \ref{rabcbasic1.thm}, 
the set ${\mathcal S}_{a, b, c}$ is an invariant set under
 piecewise linear transformations
$R_{a, b, c}$ and $\tilde R_{a, b, c}$ which has empty intersection with their black holes.
Then it suffices to  show the maximality. Let $E$ be an invariant set under
 piecewise linear transformations
$R_{a, b, c}$ and $\tilde R_{a, b, c}$ that has empty intersection with their black holes.
Take $t\in E$. Then it follows from the invariance of the set $E$ that
\begin{equation} (R_{a, b, c})^n (t)\in E \quad {\rm  and} \quad (\tilde R_{a, b, c})^n (t)\in E
\end{equation}
for all nonnegative integers $n$.
 This, together with Proposition \ref{dabc1pointcharacterization.prop} and  the empty intersection assumption between the set $E$ and
the black holes of the piecewise linear transformations
$R_{a, b, c}$ and $\tilde R_{a, b, c}$, implies that $t\in {\mathcal S}_{a, b, c}$. Thus $E\subset {\mathcal S}_{a, b, c}$ and the maximality
of the set ${\mathcal S}_{a, b, c}$ follows.

(ii)\quad
By \eqref{dabcbasic1.lem.eq1} in Proposition  \ref{dabcbasic1.prop}
and
\eqref{rabcbasic1.thm.eq1} in Proposition \ref{rabcbasic1.thm}, 
the complement of  ${\mathcal S}_{a, b, c}$ is an invariant set under
the piecewise linear transformations
$R_{a, b, c}$ and $\tilde R_{a, b, c}$, and it contains their black holes.
For any set $E$  that is  invariant under the transformations $R_{a, b, c}$ and $\tilde R_{a, b, c}$
 that contains their black holes, 
we obtain from  \eqref{invertibilityoftransformations} that $R_{a, b, c}(t), \tilde R_{a, b, c}(t)\not\in E$ for any $t\not \in E$.
This, together with  Proposition \ref{dabc1pointcharacterization.prop}, proves the desired minimality for  the complement of  ${\mathcal S}_{a, b, c}$.
\end{proof}


Next we prove Proposition \ref{dabc1pointcharacterization.prop}.

\begin{proof} [Proof of Proposition \ref{dabc1pointcharacterization.prop}]
($\Longrightarrow$)\ Take $t\in {\mathcal S}_{a, b, c}$. Then
$(R_{a,b,c})^n (t)\in {\mathcal S}_{a, b, c}$ and $( \tilde
R_{a, b,c})^n(t)\in {\mathcal S}_{a, b, c}$ for all $n\ge 0$ by
Proposition \ref{rabcbasic1.thm}. This together with Proposition
\ref{dabcbasic1.prop} proves the desired empty intersection
property with black holes.

\smallskip
($\Longleftarrow$)\  Take any real number $t$ such that $(R_{a,b,c})^n (t)$ and $( \tilde R_{a, b,c})^n(t), n\ge 0$, do not belong to
 the black holes of the transformations $R_{a, b, c}$ and $\tilde R_{a, b, c}$.
Define
 $$t_n=\left\{\begin{array}{ll} (R_{a, b,c})^n (t) & {\rm if} \ n\ge 1\\
t & {\rm if} \ n=0\\
(\tilde R_{a, b,c})^{-n} (t)  & {\rm if} \ n<0,
\end{array}\right.$$
and $\lambda_n=t_n-t, n\in \Z$.
Then
\begin{equation}\label{dabc1pointcharacterization.cor.pf.eq1}  t_{n+m}=(R_{a,b,c})^m (t_n)\quad {\rm for\ all} \ n\in \Z \ {\rm and} \ 0\le m\in \Z
\end{equation}
and
\begin{equation}\label{dabc1pointcharacterization.cor.pf.eq2}  \lambda_n\in b\Z\  {\rm and} \  \lambda_{n+1}-\lambda_n\in
 \{\lfloor c/b\rfloor b+b, \lfloor c/b\rfloor b\}\quad {\rm for \ all} \ n\in \Z,
\end{equation}
by \eqref{invertibilityoftransformations}, and the assumption that
$t_n\not\in [c_0+a-b, c_0)+a\Z$ for all $n\ge 0$ and $ t_n\not\in [c-c_0, c-c_0+b-a)+a\Z$ for all $n\le 0$.
Define $ {\bf x}_t(\lambda)=1$ if $\lambda=\lambda_n$ for some $n\in \Z$ and ${\bf x}_t(\lambda)=0$ otherwise, and set
${\bf x}_t=({\bf x}_t(\lambda))_{\lambda\in b\Z}$. Then
${\bf x}_t$ belong to ${\mathcal B}_b^0$.
Let $\mu_n\in a\Z$ be so chosen that $\tilde t_n:=t_n-\mu_n\in [0, a)$. Then
$ \{\mu_n\}_{n\in \Z}$ is a strictly increasing sequence with
\begin{equation} \label{dabc1pointcharacterization.cor.pf.eq3}
\lim_{n\to +\infty} \mu_n=+\infty\ {\rm  and} \ \lim_{n\to -\infty} \mu_n=-\infty\end{equation}
by \eqref{dabc1pointcharacterization.cor.pf.eq2},
and
\begin{eqnarray} \label{dabc1pointcharacterization.cor.pf.eq4}
 & & \sum_{\lambda\in b\Z} \chi_{[0,c)}(t-\mu_n+\lambda) {\bf x}_t(\lambda)  =
\sum_{m\in \Z} \chi_{[0,c)}(t-\mu_n+\lambda_m)\nonumber\\
& = & \sum_{m\in \Z} \chi_{[0,c)}(t_m-\mu_n)=\chi_{[0,c)}(t_n-\mu_n)=1\quad {\rm for \ all }\  n\in \Z,
\end{eqnarray}
where the first equation follows from the definition of the vector ${\bf x}_t$  and the third one holds as
  $t_m-\mu_n\le t_n-\mu_n-b<0$ for all $m<n$ and
$t_m-\mu_n\ge (t_{n+1}-t_n)+(t_n-\mu_n)=(\lambda_{n+1}-\lambda_n)+(t_n-\mu_n)\ge c$ for all $m>n$.
Similarly   for any $\mu\in a\Z$ with $\mu_n<\mu< \mu_{n+1}$,
\begin{eqnarray} \label{dabc1pointcharacterization.cor.pf.eq5} 
 & & \sum_{\lambda\in b\Z} \chi_{[0,c)}(t-\mu+\lambda) {\bf x}_t(\lambda)=
\sum_{m\in \Z} \chi_{[0,c)}(t_m-\mu)=1 
\end{eqnarray}
as $t_m-\mu\le t_n-\mu<\mu_n+a-\mu\le 0$ for $m\le n$,
$0\le t_{n+1}-\mu_{n+1}<t_m-\mu\le t_{n+1}-\mu_n-a<c$ for  $m=n+1$, and
$t_m-\mu\ge t_{n+2}-\mu_{n+1}+a\ge c$ for $m\ge n+2$.
Combining \eqref{dabc1pointcharacterization.cor.pf.eq4} and \eqref{dabc1pointcharacterization.cor.pf.eq5} proves
${\bf M}_{a, b, c}(t) {\bf x}_t={\bf 1}$ and hence $t\in {\mathcal S}_{a, b, c}$.
%
\end{proof}

We conclude this subsection by the proof of Theorem \ref{dabcsabc.thm}.

\begin{proof} [Proof of  Theorem  \ref{dabcsabc.thm}]  
We use the double inclusion method.
Take $t\in {\mathcal D}_{a, b, c}$. Let ${\bf x}\in {\mathcal B}_b^0$ be so chosen that ${\mathbf M}_{a, b, c}(t) {\bf x}={\bf 2}$
and denote by $K$ the support of the vector ${\bf x}$; i.e., the set of all $\lambda\in b\Z$ with ${\bf x}(\lambda)=1$.
Write $K=\{\lambda_j: j\in \Z\}$ for some strictly increasing sequence $\{\lambda_j\}_{j=-\infty}^\infty$ in $b\Z$ with $\lambda_0=0$,
and define ${\bf x}_i:=({\bf x}_i(\lambda))_{\lambda\in a\Z}, i=0,1$, by
${\bf x}_i(\lambda)=1$ if $\lambda\in K_i$ and ${\bf x}_i(\lambda)=0$ otherwise, where $K_i=\{\lambda_{i+2j}: \ j\in \Z\}, i=0,1$.
Then following the argument in the proof of the implication (ii)$\Longrightarrow$(iii)  of Theorem \ref{framenullspace1.tm},
we have that
\begin{equation}\label{dabcsabc.tm.pf.eq1}
 {\bf x}_0, \tau_{\lambda_1} {\bf x}_1\in {\mathcal B}_b^0 \ \ {\rm and} \ \
{\bf M}_{a, b, c}(t) {\bf x}_0={\bf M}_{a, b, c}(t+\lambda_1) \tau_{\lambda_1}{\bf x}_1={\bf 1}.\end{equation}

Notice that ${\mathcal D}_{a, b, c}\cap ([c_0+a-b, c_0)+a\Z)=\emptyset$
by the supset property ${\mathcal D}_{a, b, c}\subset {\mathcal S}_{a, b, c}$ in \eqref{dabc2subsetdabc1} and
the trivial intersection property ${\mathcal S}_{a, b, c}\cap ([c_0+a-b, c_0)+a\Z)=\emptyset$  in \eqref{dabcbasic1.lem.eq1}.
Then
either $t\in [0, c_0+a-b)+a\Z$ or $[c_0, a)+a\Z$.
For the first case that  $t\in [0, c_0+a-b)+a\Z$,  $\lambda_2= \lfloor c/b\rfloor b+b$ by \eqref{lambdaabc.def} and \eqref{dabcbasic1.lem.eq3} in
Proposition \ref{dabcbasic1.prop}. Hence $\lambda_1\in [b, \lfloor c/b\rfloor b]\cap b\Z$ and $t+\lambda_1\in {\mathcal S}_{a, b, c}$ by
\eqref{dabcsabc.tm.pf.eq1}. Thus
\begin{equation} \label{dabcsabc.tm.pf.eq2}
t\in \big({\mathcal S}_{a, b, c}\cap ([0, c_0+a-b)+a\Z)\big)\cap \big(\cup_{\lambda_1\in [b, \lfloor c/b\rfloor b]\cap b\Z}
({\mathcal S}_{a, b, c}-\lambda_1)\big)\end{equation}
for the first case. Similarly for the second case that  $t\in [c_0+a-b, a)+a\Z$,
 $\lambda_2= \lfloor c/b\rfloor b$ by \eqref{lambdaabc.def} and \eqref{dabcbasic1.lem.eq3} in
Proposition \ref{dabcbasic1.prop}, which together with  \eqref{dabcsabc.tm.pf.eq1} implies that
\begin{equation} \label{dabcsabc.tm.pf.eq3}
t\in \big({\mathcal S}_{a, b, c}\cap ([ c_0, a)+a\Z)\big)\cap \big(\cup_{\lambda_1\in [b, (\lfloor c/b\rfloor -1) b]\cap b\Z}
({\mathcal S}_{a, b, c}-\lambda_1)\big)\end{equation}
for the second case. Combining \eqref{dabcsabc.tm.pf.eq2} and \eqref{dabcsabc.tm.pf.eq3} proves the first inclusion
\begin{eqnarray} \label{dabcsabc.tm.pf.eq4}
{\mathcal D}_{a, b, c} & \subset &
\big({\mathcal S}_{a, b, c}\cap ([0, c_0+a-b)+a\Z)\cap ( {\mathcal S}_{a, b, c}-\lfloor c/b\rfloor b)\big)\nonumber\\
& & \cup \big({\mathcal S}_{a, b, c}\cap (\cup_{\lambda\in [b, (\lfloor c/b\rfloor -1) b]\cap b\Z} ({\mathcal S}_{a, b, c}-\lambda))\big).
\end{eqnarray}

Take $t\in {\mathcal S}_{a, b, c}\cap ([0, c_0+a-b)+a\Z)\cap ( {\mathcal S}_{a, b, c}-\lfloor c/b\rfloor b)$. Then
there exist ${\bf x}_0, {\bf x}_1\in {\mathcal B}_b^0$ such that
\begin{equation}\label{dabcsabc.tm.pf.eq5}
{\mathbf M}_{a, b, c}(t) {\bf x}_0= {\mathbf M}_{a, b, c}(t+\lfloor c/b\rfloor b ) {\bf x}_1={\bf 1}.\end{equation}
Define  ${\bf x}={\bf x}_0+\tau_{-\lfloor c/b\rfloor b}{\bf x}_1$.
By \eqref{sabcinvariant.eq1} and \eqref{dabcsabc.tm.pf.eq5}, we have that
\begin{equation}\label{dabcsabc.tm.pf.eq6}
 {\mathbf M}_{a, b, c}(t) {\bf x}= {\mathbf M}_{a, b, c}(t) {\bf x}_0+  {\mathbf M}_{a, b, c}(t+\lfloor c/b\rfloor b ) {\bf x}_1={\bf 2}.
\end{equation}
Now let us verify that ${\bf x}\in {\mathcal B}_b^0$.
Write ${\bf x}=({\bf x}(\lambda))_{\lambda\in b\Z}$. Observe that ${\bf x}(\lambda)\in \{0, 1, 2\}$ for all $\lambda\in b\Z$
and ${\bf x}(0)\ge {\bf x}_0(0)\ge 1$. Then it suffices to prove that ${\bf x}(\lambda)\ne 2$ for all $\lambda\in b\Z$. Suppose, to the contrary, that
${\bf x}(\lambda_0)=2$ for some $\lambda_0\in b\Z$. Then ${\bf x}_0(\lambda_0)=1$ and $\tau_{-\lfloor c/b\rfloor b}{\bf x}_1(\lambda_0)=1$. Hence
$\tau_{\lambda_0} {\bf x}_0, \tau_{\lambda_0-\lfloor c/b\rfloor b}{\bf x}_1\in {\mathcal B}_b^0$ and
$${\mathbf M}_{a, b, c}(t+\lambda_0) \tau_{\lambda_0} {\bf x}_0={\mathbf M}_{a, b, c}(t) {\bf x}_0 ={\bf 1}\ {\rm and}\
{\mathbf M}_{a, b, c}(t+\lambda_0) \tau_{\lambda_0-\lfloor c/b\rfloor b}{\bf x}_1={\bf 1}$$
by \eqref{sabcinvariant.eq1} and \eqref{dabcsabc.tm.pf.eq5}.  Thus
$\tau_{\lambda_0} {\bf x}_0=\tau_{\lambda_0-\lfloor c/b\rfloor b}{\bf x}_1$ by Proposition \ref{uniqueness.cor},
which is a contradiction as
$\tau_{-\lfloor c/b\rfloor b}{\bf x}_1(\lfloor c/b\rfloor b)={\bf x}_1(0)=1$  by the assumption that ${\bf x}_1\in {\mathcal B}_b^0$
and
${\bf x}_0(\lfloor c/b\rfloor b)=0$ by \eqref{dabcbasic1.lem.eq3} and the assumption that $t\in [0, c_0+b-a)\cap {\mathcal S}_{a, b, c}$.
Therefore ${\bf x}$ is a binary vector in ${\mathcal B}_b^0$. This together with
\eqref{dabcsabc.tm.pf.eq6} proves that
\begin{equation}\label{dabcsabc.tm.pf.eq7}
{\mathcal S}_{a, b, c}\cap ([0, c_0+a-b)+a\Z)\cap ( {\mathcal S}_{a, b, c}-\lfloor c/b\rfloor b)\subset {\mathcal D}_{a, b, c}.
\end{equation}
Applying similar argument, we can prove that
\begin{equation}\label{dabcsabc.tm.pf.eq8}
{\mathcal S}_{a, b, c}\cap ([0, c_0+a-b)+a\Z)\cap ( {\mathcal S}_{a, b, c}-\lambda)\subset {\mathcal D}_{a, b, c}
\end{equation}
for all $\lambda\in [b, (\lfloor c/b\rfloor-1) b]\cap b\Z$, and
\begin{equation}\label{dabcsabc.tm.pf.eq8*}
{\mathcal S}_{a, b, c}\cap ([c_0, a)+a\Z)\cap ( {\mathcal S}_{a, b, c}-\lambda)\subset {\mathcal D}_{a, b, c}
\end{equation}
for all $\lambda\in [b, (\lfloor c/b\rfloor-1) b]\cap b\Z$.
The desired equality \eqref{dabcsabc.thm.eq1} then follows from
\eqref{dabcsabc.tm.pf.eq4}, \eqref{dabcsabc.tm.pf.eq7}, \eqref{dabcsabc.tm.pf.eq8} and \eqref{dabcsabc.tm.pf.eq8*}.
\end{proof}

\subsection{Proof of  Theorem \ref{newmaintheorem3}}
%
 The main ideas behind the proof of Theorem \ref{newmaintheorem3}  are as follows.
  The conclusion (VIII) follows from the results in \cite[Section 3.3.3.5, 3.3.3.6 and 3.3.4.3]{janssen03}.
We include a different proof by showing that the only binary vector ${\bf x}$
satisfying ${\bf M}_{a, b, c}(t) {\bf x}={\bf 2}$ is the  vector ${\bf 1}$.
We prove  Conclusion (IX) by showing that any point not in the black hole $[c_0+a-b, c_0)+a\Z$ of the transformation $R_{a, b,c}$
is contained in the  hole $(R_{a,b,c})^n [c-c_0, c-c_0+b-a)+a\Z$ for some $n\ge 1$, see \eqref{conclusionix.pf.eq1}.
The crucial step in the proof of Conclusion (X)  (resp. Conclusion (XI))
is to show that  ${\mathcal S}_{a, b, c}=[0, c_0+a-b)+a\Z$
in \eqref{conclusionx.pf.eq1} (resp.  ${\mathcal S}_{a, b, c}=[c_0, a)+a\Z$ in \eqref{conclusionxi.pf.eq2}).

\begin{proof}[Proof of Theorem \ref{newmaintheorem3}]
(VIII):\quad   Suppose on the contrary that ${\mathcal G}(\chi_{[0,c)}, a\Z\times \Z/b)$ is not a Gabor frame.
 Then by Theorem \ref{framenullspace1.tm} there exist $t\in \R$ and  $({\bf x}(\lambda))_{\lambda\in b\Z}\in {\mathcal B}_b^0$
 such that
 \begin{equation}\label{maintheorem.pf.nullspace1.eqviii12} \sum_{\lambda\in b\Z} \chi_{[0,c)}(t-\mu+\lambda) {\bf x}(\lambda)=2 \ {\rm for  \ all} \ \mu\in a\Z.\end{equation}
By the assumption $\lfloor c/b\rfloor=1$ and $b<c$,
given any $t\in \R$ and $\mu\in a\Z$, the equality $\chi_{[0,c)}(t-\mu+\lambda)=1$ holds for  at most two distinct $\lambda\in b\Z$.
This together with \eqref{maintheorem.pf.nullspace1.eqviii12}  that ${\bf x}(\lambda)=1$ for all $\lambda\in b\Z$, and
also that
\begin{equation} \label{maintheorem.pf.nullspace1.eqviii13}
t-\mu\not\in [c, 2b)+b\Z\ {\rm  for\ all} \ \mu\in a\Z.
\end{equation}

If $a/b\not\in \Q$,   then there exists  $\mu_0\in a\Z$
by the density of the set  $a\Z+b\Z$  in $\R$ such  that $t-\mu_0\in [c, 2b)+b\Z$, which contradicts to
\eqref{maintheorem.pf.nullspace1.eqviii13}.

If $a/b\in \Q$, then  $a/b=p/q$ for some positive coprime integers $p$ and $q$. Hence
\begin{equation*} \label{maintheorem.pf.nullspace1.eqviii14}
t\not\in [c, 2b)+b\Z/q=\R,
\end{equation*}
where the first conclusion follows from \eqref{maintheorem.pf.nullspace1.eqviii13} and the equality
holds as $2b-c=b-c_0>b-a\ge b/q$ by the assumption $0<b-a<c_0<a$.
This is a contradiction.

(IX):\quad
 By the supset property \eqref{dabc2subsetdabc1} and Theorem
\ref{framenullspace1.tm}, it suffices to prove the following result.

\begin{prop}\label{conclusionix.prop}
 Let $(a, b, c)$  be a triple of positive numbers that satisfies $a<b<c,  b-a<c_0:=c-\lfloor c/b\rfloor b<a, \lfloor c/b\rfloor\ge 2$,
$c_1=c-c_0-\lfloor (c-c_0)/a\rfloor a>2a-b$, and let ${\mathcal S}_{a, b, c}$ be as in \eqref{dabc1.def}.
Then
${\mathcal S}_{a, b, c}=\emptyset$.
\end{prop}

\begin{proof}
By Propositions \ref{dabcbasic1.prop} and
\ref{blackholestwo.prop},
it suffices to
prove
\begin{equation}\label{conclusionix.pf.eq1}  [c_0-a, c_0+a-b)+a\Z\subset
\cup_{n=0}^L (R_{a, b, c})^n([c-c_0, c-c_0+b-a)+a\Z),\end{equation}
  where $L=\max (\lfloor
(c_0+a-b)/(c_1+b-2a)\rfloor, \lfloor (a-c_0)/(a-c_1)\rfloor)$.

For any $t\in [0, c_0+a-b)$, write $t= l(c_1+b-2a)+t'$ for some
$t'\in [0, \min(c_1+b-2a, c_0+a-b))$ and $0\le l\le L$. Then
\begin{eqnarray} \label{conclusionix.pf.eq2}
 t & \in &  (R_{a, b, c})^l(t')+a\Z\subset (R_{a, b, c})^l([0,
c_1+b-2a)+a\Z)\nonumber\\
& \subset &  \cup_{n=0}^L (R_{a, b, c})^n([c-c_0,
c-c_0+b-a)+a\Z)\end{eqnarray}
  for all $t\in [0, c_0+a-b)$.

Similarly for any $s\in [c_0-a, 0)$, let $s= l'(c_1-a)+s'$ for
some $s'\in [\max(c_1-a, c_0-a), 0)$ and $0\le l'\le L$. Then
\begin{eqnarray} \label{conclusionix.pf.eq3}
 s & \in &  (R_{a, b, c})^{l'}(s')+a\Z\subset (R_{a, b, c})^{l'}([c_1-a, 0)+a\Z)\nonumber\\
 & \subset &  \cup_{n=0}^L (R_{a, b, c})^n([c-c_0,
c-c_0+b-a)+a\Z)\end{eqnarray}
  for all $s\in [c_1-a, 0)$.
Combining \eqref{conclusionix.pf.eq2} and
\eqref{conclusionix.pf.eq3} and applying the periodic property
\eqref{dabcperiodic.section2} proves \eqref{conclusionix.pf.eq1}.
\end{proof}

\smallskip

(X):\quad  Mimicking the argument used to prove \eqref{conclusionix.pf.eq3}, we can show that
$$\cup_{n=0}^\infty (R_{a, b, c})^n([c-c_0, c-c_0+b-a)+a\Z)=[c_0+a-b, a)+a\Z,$$ which implies that
the set ${\mathcal S}_{a, b, c}$  in \eqref{dabc1.def} is given by
\begin{equation} \label{conclusionx.pf.eq1}
{\mathcal S}_{a, b, c}=[0, c_0+a-b)+a\Z.\end{equation}
From the assumption on $c_1$, the ratio $a/b$ is rational. We write $a/b=p/q$ for some coprime integers $p$ and $q$.
Clearly $p\ge 2$ as $b-a<c_0<a$.
By the assumption that $c_1=2a-b$, we have that  $\lfloor c/b\rfloor+1\in p\Z$, which implies that
\begin{equation*}
R_{a, b, c}(t)-t\in a\Z \quad {\rm for \ all}\ t\in {\mathcal S}_{a, b, c}=[0, c_0+a-b)+a\Z.
\end{equation*}
This together with  Theorem \ref{framenullspace1.tm} and Proposition \ref{dabcsabc.thm}  implies that
the Gabor system ${\mathcal G}(\chi_{[0,c)}, a\Z\times \Z/b)$ is a frame of $L^2$ if and only if
${\mathcal D}_{a, b, c}$ in \eqref{dabc2.def} is an empty set if and only if
$([0, c_0+a-b)+a\Z)\cap ([0, c_0+a-b)+\lambda+a\Z)=\emptyset$ for all $\lambda\in [b, \lfloor c/b\rfloor b]\cap b\Z$.
Observe that  $([0, c_0+a-b)+a\Z)\cap ([0, c_0+a-b)+\lambda +a\Z)= [0, c_0+a-b)+a\Z\ne \emptyset$ for $\lambda=pb\in [b, \lfloor c/b\rfloor b]\cap b\Z$
provided that $\lfloor c/b\rfloor \ge p$, and also
that $([0, c_0+a-b)+a\Z)\cap ([0, c_0+a-b)+ \lambda +a\Z)= [b/q, c_0+a-b)+a\Z\ne \emptyset$
for $\lambda=kb \in [b, \lfloor c/b\rfloor b]\cap b\Z$ where $1\le k\le p-1$ is the unique integer such that $ qk-1\in p\Z$,
 provided that $\lfloor c/b\rfloor +1 = p$ and $c_0+a-b>b/q$.
Therefore the assumptions that $\lfloor c/b\rfloor +1=p$ and $c_0+a-b\le b/q$ are necessary  for
  the Gabor system ${\mathcal G}(\chi_{[0,c)}, a\Z\times \Z/b)$ being a frame of $L^2$.
On the other hand,  if $\lfloor c/b\rfloor +1=p$ and $c_0+a-b\le b/q$,  one may verify that
$([0, c_0+a-b)+a\Z)\cap ([0, c_0+a-b)+\lambda +a\Z)=([0, c_0+a-b)+a\Z)\cap ([0, c_0+a-b)+k(\lambda) b/q +a\Z)=\emptyset$
for all $\lambda\in [1, \lfloor c/b\rfloor ]b\cap b\Z$, where $k(\lambda)$ is the unique integer in
$[1, p-1]$ such that $k(\lambda)b/q-\lambda\in a\Z$.
Therefore the assumptions that $\lfloor c/b\rfloor +1=p$ and $c_0+a-b\le b/q$
is also sufficient  for the Gabor system ${\mathcal G}(\chi_{[0,c)}, a\Z\times \Z/b)$ being a frame of $L^2$.

(XI)\quad Mimicking the argument used to prove \eqref{conclusionix.pf.eq2}, we may show that
\begin{equation} \label{conclusionxi.pf.eq2}
{\mathcal S}_{a, b, c}=[c_0, a)+a\Z.\end{equation}
Now we can apply similar argument used in the proof of the conclusion (X) of this theorem.
From the assumption that $c_1=0$, it follows  $a/b=p/q$ for some coprime integers $p$ and $q$ with $p\ge 2$
and $\lfloor c/b\rfloor\in p\Z$.
By \eqref{conclusionxi.pf.eq2} and Theorem \ref{framenullspace1.tm},
we can show  that ${\mathcal G}(\chi_{[0,c)}, a\Z\times \Z/b)$ is a frame of $L^2$ if and only if
$([c_0, a)+a\Z)\cap ([c_0, a)+\lambda+a\Z)=\emptyset$ for all $\lambda\in [b, \lfloor c/b\rfloor b-b]\cap b\Z$.
Then the desired  necessary condition for
  the Gabor system ${\mathcal G}(\chi_{[0,c)}, a\Z\times \Z/b)$ being a frame of $L^2$ follows from the observation that
 $([c_0, a)+a\Z)\cap ([c_0, a)+p b +a\Z)= [c_0, a)+a\Z\ne \emptyset$ if $\lfloor c/b\rfloor \ge p+1$, and
that $([c_0, a)+a\Z)\cap ([c_0, a)+ kb +a\Z)= [c_0, a-b/q)+a\Z\ne \emptyset$ if $\lfloor c/b\rfloor = p$ and $a-c_0>b/q$
where $1\le k\le p-1$ is the unique integer such that $ qk+1\in p\Z$.
The sufficiency for the  conditions that $\lfloor c/b\rfloor =p$ and $a-c_0\le b/q$
holds as
$([c_0, a)+a\Z)\cap ([c_0, a)+\lambda +a\Z)=([c_0, a)+a\Z)\cap ([c_0, a)-k(\lambda) b/q +a\Z)=\emptyset$
for all $\lambda\in [1, \lfloor c/b\rfloor ]b\cap b\Z$, where $k(\lambda)$ is the unique integer in
$[1, p-1]$ such that $k(\lambda)b/q+\lambda\in a\Z$.
\end{proof}

%


\section{
Properties of Maximal Invariant Sets}
\label{gaborirrational.section}

To prove Theorems \ref{newmaintheorem4} and \ref{newmaintheorem5}, we need some deep properties
of the maximal invariant sets ${\mathcal S}_{a, b, c}$.
Let us start from some examples of  maximal invariant sets ${\mathcal S}_{a, b, c}$.

\begin{ex} \label{irrationalexample1}
{\rm Take $a=\pi/4, 
b=1$ and
 $c=23-11\pi/2$. 
The black holes of the corresponding transformations $R_{a, b, c}$ and  $\tilde R_{a, b,c}$ are
$[17-21\pi/4, 18-11\pi/2)+\pi\Z/4$  and $[5-3\pi/2,6-7\pi/4)+\pi\Z/4$  respectively,
which can be transformed
back and forth via the middle hole $[11-7\pi/2, 12-15\pi/4)+\pi\Z/4$; i.e.,
$$\left\{\begin{array}{l}
R_{a, b, c} ([5-3\pi/2,6-7\pi/4)+\pi\Z/4)=[11-7\pi/2, 12-15\pi/4)+\pi\Z/4\\
(R_{a, b, c})^2 ([5-3\pi/2,6-7\pi/4)+\pi\Z/4)=[17-21\pi/4, 18-11\pi/2)+\pi\Z/4.
\end{array}\right.
$$
So the maximal invariant  set
\begin{eqnarray*}
{\mathcal S}_{a,b,c}  &=&[ 18-23\pi/4,11-7\pi/2 ) \cup [ 12-15\pi/4 ,5-
3\pi/2 )\\
 & & \cup [ 6-7\pi/4,  17-21\pi/4)  +\pi \Bbb{Z}/4  \\
\  &\approx &[ -0.0642,0.0044) \cup [ 0.2190,0.2876) \cup
[ 0.5022,0.5066) +0.7864 \Bbb{Z}
\end{eqnarray*}
consists of intervals of  different lengths on one period and contains a small neighborhood of the lattice $\pi \Bbb{Z}/4$,
 c.f. Figure \ref{holesremoval.fig}.
}\end{ex}

\begin{ex} \label{rationalexample1}
{\rm For the triple $(a, b, c)=(13/17, 1, 77/17)$, 
 the black holes of the corresponding transformations $R_{a, b, c}$ and $\tilde R_{a, b, c}$ are
$[5/17, 9/17)+13\Z/17$ and $[3/17, 7/17)+13\Z/17$.  Applying the transformation $R_{a, b, c}$ to
 the black hole of the transformation $\tilde R_{a, b, c}$,
we obtain that
\begin{equation}\label{rationalexample1.eq1}\left\{\begin{array}{l}
R_{a, b, c}([3, 7)/17+13\Z/17)=([5, 7)\cup [10, 12))/17+13\Z/17,\\
(R_{a, b, c})^2([3, 7)/17+13\Z/17) =([5, 7)\cup [0, 2))/17+13\Z/17,\\
(R_{a, b, c})^3([3, 7)/17+13\Z/17) =[5, 9)/17+13\Z/17.
\end{array}\right.
\end{equation}
Thus  the maximal invariant set
\begin{eqnarray*} \mathcal{S}_{a,b,c} & = & ([2, 3)\cup [9, 10)\cup [12, 13))/17+13\Z/17\\
& \approx & [0.1176, 0.1764)\cup [0.5294, 0.5882)\cup [0.7059, 0.7647)+0.7647\Z
\end{eqnarray*}
consists of intervals of same length $1/17$ on the period $[0, 13/17)$ and contains  small left  neighborhood of  the lattice $13\Z/17$,
and its complement $\R\backslash \mathcal{S}_{a,b,c}=([0,2)\cup[3, 9)\cup [10, 12))/17+13\Z/17$
contains one  big gap of size $6/17$, two small gaps of size $2/17$ on the period $[0, 13/17)$,
and a small gap attached to the right-hand side of the lattice $13\Z/17$.

For the triple $(a, b, c)=(13/17, 1, 73/17)$,
   the maximal invariant set
$\mathcal{S}_{a,b,c}=([0, 1)\cup [7, 8)\cup [10, 11))/17+13\Z/17$
 contains a small right  neighborhood of  the lattice $13\Z/17$,
while its complement $\R\backslash \mathcal{S}_{a,b,c}=([1,7)\cup[8, 10)\cup [11, 13))/17+13\Z/17$
contains a small gap attached to the left-hand side of the lattice $13/17$.

For  the triple  $(a, b, c)=(13/17, 1, 75/17)$,
 the maximal invariant set
\begin{eqnarray*}
{\mathcal S}_{a, b, c}&=&([0, 3)\cup [7, 10)\cup[10, 13))/17+13\Z/17\\
& = & [0, 0.1765)\cup [0.4118, 0.5882)\cup [0.5882, 0.7647)+0.7647 \Z\end{eqnarray*}
consists of intervals of ``same" length $3/17$ and contains small left and right neighborhoods of  the lattice $13\Z/17$. On the other hand,
its  complement $\R\backslash \mathcal{S}_{a,b,c}=[3, 7)/17+13\Z/17$
contains one  big gap of size $4/17$ and two small gaps of size
``zero" at $\{0, 10/17\}$ on the period $[0, 13/17)$, c.f. Figure \ref{holesremoval.fig}.
 }\end{ex}

 \begin{ex}
\label{rationalexample2}
{\rm For the triple $(a, b, c)=(6/7, 1, 23/7)$,
  black holes of the corresponding  transformations $R_{a, b, c}$ and $\tilde R_{a, b,c}$
are $[1,2)/7+6\Z/7$ and $[3, 4)/7+6\Z/7$ respectively.
Observe that
$$\left\{\begin{array}{l}
R_{a, b, c}([3, 4)/7+6\Z/7)=[0, 1)/7+6\Z/7\\
(R_{a, b, c})^2([3, 4)/7+6\Z/7)=[4, 5)/7+6\Z/7\\
(R_{a, b, c})^3([3, 4)/7+6\Z/7)=[1, 2)/7+6\Z/7,
\end{array}\right.
$$
which also implies that
$
R_{a, b, c}([3, 5)/7+6\Z/7)=[0, 2)/7+6\Z/7$.
Therefore the maximal invariant set
\begin{eqnarray*} {\mathcal S}_{a, b, c} & = & [2,3)/7\cup [5, 6)/7+6\Z/7\\
&= &[0.2857, 0.4286)\cup [0.7143, 0.8571)+0.8571\Z\end{eqnarray*} consists of intervals of length $1/7$,
while its complement $\R\backslash {\mathcal S}_{a, b, c}=([0,2)\cup [3, 5))/7+6\Z/7$
 consists of gaps of length $2/7$, c.f. Figure \ref{holesremoval.fig}. 
}\end{ex}

In the above examples, we see that
the black hole $[c_0+a-b, c_0)+a\Z$  of the transformation $R_{a, b, c}$
attracts the black hole $[c-c_0, c-c_0+b-a)+a\Z$ of the other transformation $\tilde R_{a,b,c}$
when applying ${\mathcal R}_{a, b, c}$  {\em finitely many times}, i.e.,
$$(R_{a, b, c})^L ([c-c_0, c-c_0+b-a)+a\Z)=[c_0+a-b, c_0)+a\Z
$$
for some nonnegative integer $L$,
and that
 holes
$$(R_{a, b, c})^l ([c-c_0, c-c_0+b-a)+a\Z)=[c_0+a-b, c_0)+a\Z, 0\le l\le L-1,
$$
may or may not overlap with
the black hole $[c_0+a-b, c_0)+a\Z$  of the transformation $R_{a, b, c}$, which depends on
the
ratio $a/b$ 
being irrational or not.



\smallskip

For the first case that the ratio  between time-spacing parameter $a$ and  frequency-spacing parameter $b$
 is irrational (i.e. $a/b\not\in \Q$),  we  show that
if  ${\mathcal S}_{a, b, c}\ne \emptyset$ then
 the  black hole $[c_0+a-b, c_0)+a\Z$ of the transformation
 $R_{a, b,c}$ and  the   black hole $[c-c_0, c-c_0+b-a)+a\Z$ of the transformation $\tilde R_{a,b,c}$
  are inter-transformable  through
 mutually disjoint  periodic holes $(R_{a, b,c})^n ([c- c_0, c-c_0+b-a)+a\Z)=(\tilde R_{a, b,c})^{D-n}
     ([c_0+a-b, c_0)+a\Z), 0\le n\le D$,
  where $D$ is a nonnegative integer.
Furthermore the complement of the set ${\mathcal S}_{a, b, c}$ is
    the union of  mutually disjoint   holes of {\em same size}, but  the set ${\mathcal S}_{a, b, c}$ is
    the union of disjoint intervals of {\em ``different'' sizes}. 

 \begin{thm} \label{dabc1holes.tm}
 Let $(a,b,c)$ be a triple of positive numbers satisfying
$a<b<c,   b-a<c_0:=c-\lfloor c/b\rfloor b<a, \lfloor c/b\rfloor\ge 2$ and $a/b\not\in \Q$.
Assume that
 ${\mathcal S}_{a, b, c}\ne \emptyset$. Then there exists a nonnegative integer $D\le \lfloor a/(b-a)\rfloor-1$ such that
  \begin{equation}\label{dabc1holes.tm.eq-1}
 (\tilde R_{a, b,c})^{n} (c_0+a-b)+a\Z
    =
   (R_{a, b,c})^{D-n} (c- c_0)+a\Z
   \end{equation}
   for all $0\le n\le D$,
 \begin{equation}\label{dabc1holes.tm.eq0}
 \big((R_{a, b,c})^n (c- c_0)+[0, b-a]+a\Z \big)
 \cap \big((R_{a, b,c})^{n'} (c- c_0)+[0, b-a]+a\Z\big)=\emptyset
 \end{equation}
 for all $0\le n\ne n'\le D$. Moreover
 \begin{eqnarray}\label{dabc1holes.tm.eq1}
\R \backslash {\mathcal S}_{a, b, c}
& = &
\cup_{n=0}^{D} \big( (R_{a, b,c})^n (c- c_0)+[0, b-a)+a\Z\big)\nonumber\\
& = &
\cup_{n=0}^{D} (R_{a, b,c})^n ([c- c_0, c-c_0+b-a)+a\Z)\nonumber\\
& = &
\cup_{n=0}^{D} (\tilde R_{a, b,c})^n ( [c_0+a-b, c_0)+a\Z)\nonumber\\
& = &
\cup_{n=0}^{D} \big((\tilde R_{a, b,c})^n (c_0+a-b)+[0, b-a)+a\Z\big).
\end{eqnarray}
\end{thm}

For the second case that
 $a/b$ is rational, we write
$a/b=p/q$ for some coprime integers $p$ and $q$. We restrict ourselves to consider $c\in b\Z/q$
because
for $c\not\in b\Z/q$, ${\mathcal G}(\chi_{[0,c)}, a\Z\times \Z/b)$ is a  Gabor frame
 if and only if both ${\mathcal G}(\chi_{[0,\lfloor qc/b\rfloor b/q)}, a\Z\times \Z/b)$ and
${\mathcal G}(\chi_{[0,\lfloor qc/b+1\rfloor b/q)}, a\Z\times \Z/b) $ are  Gabor frames, see
 \cite[Section 3.3.6.1]{janssen03} and the conclusion (XIV) of Theorem \ref{newmaintheorem5}.
%
Observe that for $a/b=p/q$ and $c\in b\Z/q$,
\begin{equation}\label{sabcrational.eq}
 {\bf M}_{a, b, c}(t)={\bf M}_{a, b, c}(\lfloor qt/b\rfloor b/q),\ t\in \R,\end{equation}
which implies that
\begin{equation}\label{dabcicontinuoustodiscrete.eq1}
{\mathcal D}_{a, b, c}= {\mathcal D}_{a, b, c}\cap
b\Z/q+[0, b/q)\ \ {\rm and} \  \ {\mathcal S}_{a, b, c}= {\mathcal S}_{a, b, c}\cap
b\Z/q+[0, b/q).
\end{equation}
Even further,
the sets ${\mathcal D}_{a, b, c}$ and ${\mathcal S}_{a, b, c}$ are essentially {\em finite} sets, as
they  are
completely determined by their restrictions  to the finite set $ \{0, b/q, \ldots, (p-1)b/q\}$,
\begin{equation}\label{dabcicontinuoustodiscrete.eq1b}
\left\{\begin{array}{l} {\mathcal D}_{a, b, c}= {\mathcal D}_{a, b, c}\cap \{0, b/q, \ldots, (p-1)b/q\}
+pb\Z/q+[0, b/q)\\
{\mathcal S}_{a, b, c}= {\mathcal S}_{a, b, c}\cap \{0, b/q, \ldots, (p-1)b/q\}+pb\Z/q+[0, b/q),
\end{array}\right.
\end{equation}
by \eqref{dabcicontinuoustodiscrete.eq1} and  the periodic property
and \eqref{dabcperiodic.section2}. 
In the next theorem,  we show that the  maximal invariant set ${\mathcal S}_{a, b, c}$
is the union of half-open intervals of {\em same size}
 and its
 complement   $\R\backslash {\mathcal S}_{a, b, c}$
  is the union of mutually disjoints gaps of {\em ``two different'' sizes}. 

\begin{thm} \label{dabc1discreteholes.tm}
 Let  $(a,b,c)$ be a triple of positive numbers satisfying
$a<b<c, b-a<c_0:=c-\lfloor c/b\rfloor b<a, \lfloor c/b\rfloor\ge
2, a/b=p/q$  for some coprime integers $p$ and $q$, and  $c/b\in
\Z/q$. Assume that ${\mathcal S}_{a, b, c}\ne \emptyset$. Then there
are $\delta\in [0, c_0+a-b]\cap b\Z/q,\delta^\prime\in [c_0-a,
0]\cap b\Z/q$,
 and  nonnegative integers $N_1$ and $N_2$ with the following properties:
  \begin{itemize}
  \item [{(i)}] At least one of $\delta$ and $\delta'$ is equal to zero; i.e.,
 \begin{equation}\label {dabc1discreteholes.tm.eq1}\delta\delta^\prime=0.\end{equation}
 \item [{(ii)}]  The periodic gaps $(R_{a, b, c})^n(c-c_0+[\delta', b-a+\delta))+a\Z, 0\le n\le
 N_1$,
  have length $b-a+\delta-\delta'$, and
 the periodic  gap $(R_{a, b, c})^{N_1}(c-c_0+[\delta', b-a+\delta))+a\Z$
 coincides
 with  $[c_0+a-b-\delta, c_0-\delta')+a\Z$ that  contains the black hole
 of the piecewise linear transformation $R_{a, b, c}$; i.e.,
  \begin{eqnarray}\label {dabc1discreteholes.tm.eq2}  & & (R_{a, b, c})^n\big(c-c_0+[\delta', b-a+\delta)\big)+a\Z\nonumber\\
 &   =  &
 (R_{a, b, c})^n ( c-c_0+b-a+\delta) +
 [a-b+\delta'-\delta, 0)+a\Z
 \end{eqnarray}
 for all $0\le n\le N_1$, and
 \begin{equation}\label {dabc1discreteholes.tm.eq3} (R_{a, b, c})^{N_1}( c-c_0+b-a+\delta)+a\Z=
 c_0-\delta'+a\Z.
 \end{equation}

 \item[{(iii)}] The periodic gaps $(R_{a, b, c})^m \big([c_0+a-b-\delta, c_0-\delta')\backslash [c_0+a-b, c_0)\big)+a\Z, 1\le m\le N_2$,
 have length $\delta-\delta'$ and the periodic gap $(R_{a, b, c})^m \big([c_0+a-b-\delta, c_0-\delta')\backslash [c_0+a-b, c_0)\big)+a\Z$ with
 $m=N_2$ is the same as $[\delta', \delta)+a\Z$, provided that $\delta-\delta'\ne 0$; i.e.,
   \begin{eqnarray}\label {dabc1discreteholes.tm.eq4} & &  (R_{a, b, c})^{m}\big([c_0+a-b-\delta, c_0-\delta')
   \backslash [c_0+a-b, c_0)\big)+a\Z\nonumber\\
  &= & (R_{a, b, c})^{m}(c_0-\delta')+ [\delta'-\delta, 0)+a\Z\  \end{eqnarray}
  for all $1\le m\le N_2$, and
\begin{equation}\label{dabc1discreteholes.tm.eq5} (R_{a, b, c})^{N_2}([c_0+a-b-\delta, c_0-\delta')
\backslash [c_0+a-b, c_0))+a\Z=[\delta', \delta)+a\Z.
\end{equation}

 \item[{(iii)${}^\prime$}] $(R_{a, b, c})^{N_2}(c_0)
\in a\Z$ provided that  $\delta=\delta'=0$.

 \item[{(iv)}]  The periodic gaps $(R_{a, b, c})^n(c-c_0+[\delta', b-a+\delta))+a\Z, 0\le n\le N_1$,
 of length $b-a+\delta-\delta'$, and
 $(R_{a, b, c})^m \big([c_0+a-b-\delta, c_0-\delta')\backslash [c_0+a-b, c_0)\big)+a\Z,
 1\le m\le N_2$, of length $\delta-\delta'$ together with neighboring intervals of length $b/(2q)$ at each side  are mutually
 disjoint, provided that $\delta-\delta'\ne 0$.

 \item[{(iv)${}^\prime$}]  The periodic gaps $(R_{a, b, c})^n([c-c_0, c-c_0+b-a))+a\Z, 0\le n\le N_1$,
 of length $b-a$ associated with neighboring intervals of length $b/(2q)$
  at each side, and
 $(R_{a, b, c})^m (c_0)+[-b/(2q), b/(2q))+a\Z,
 1\le m\le N_2$,  are mutually
 disjoint, provided that $\delta=\delta'=0$.

\item [{(v)}] The complement of the set ${\mathcal S}_{a, b, c}$ is the union of periodic gaps $(R_{a, b,
c})^n(c-c_0+[\delta', b-a+\delta))+a\Z, 0\le n\le N_1$, and
 $(R_{a, b, c})^m \big([c_0+a-b-\delta, c_0-\delta')\big)+a\Z, 1\le m\le N_2$; i.e.,
 \begin{eqnarray}\label {dabc1discreteholes.tm.eq7} \R\backslash {\mathcal S}_{a, b, c} & = &
\big(\cup_{n=0}^{N_1} (R_{a, b, c})^n(c-c_0+[\delta', b-a+\delta))+a\Z\big)\nonumber\\
& & \cup  \big(\cup_{m=1}^{N_2} (R_{a, b, c})^m [c_0+a-b-\delta,
c_0-\delta')+a\Z\big)\nonumber\\
& = & \cup_{n=0}^{N_1+N_2} (R_{a, b, c})^n\big(c-c_0+[\delta', b-a+\delta)+a\Z\big).
\end{eqnarray}

\item [{(vi)}] The set ${\mathcal S}_{a, b, c}$ is composed of
intervals of same length,
 \begin{equation}\label {dabc1discreteholes.tm.eq8}  {\mathcal S}_{a, b, c}  =
\cup_{n=0}^{N_1+N_2}  (R_{a, b, c})^n ( c-c_0+b-a+\delta)+[0, h)+a\Z 
\end{equation}
where
\begin{equation}\label{dabc1discreteholes.tm.eq9}
(N_1+N_2+1)(h+\delta-\delta')+(N_1+1)(b-a)=a.\end{equation}
\end{itemize}
\end{thm}

%
%

We remark that for triples $(13/17, 1, 77/17), (13/17, 1, 73/17)$ and $(13/17, 1, 75/77)$
in Example \ref{rationalexample1} and $(6/7, 1, 23/7)$ in Example \ref{rationalexample2},
the corresponding  pair $(\delta, \delta')$ in Theorem \ref{dabc1discreteholes.tm}
is given by $(2/17, 0), (0, -2/17)$, $(0, 0)$, $(1/7, 0)$ respectively. 

\smallskip

Listed below are some comparisons of  the maximal invariant set ${\mathcal S}_{a, b, c}$ for the first case that $a/b\not\in \Q$ and
for the second case that  $a/b=p/q$ for some coprime integers $p$ and $q$ and $c\in b\Z/q$, where $\delta, \delta'$ are as in Theorem
\ref{dabc1discreteholes.tm}, c.f. Examples \ref{irrationalexample1}, \ref{rationalexample1} and \ref{rationalexample2}, and Figure \ref{holesremoval.fig}:

\begin{itemize}
\item
For the second case, the restriction of
  the maximal invariant periodic set ${\mathcal S}_{a, b, c}$  on the period $[0, a)$
is the union of  finitely many half-open intervals of  same size if $\delta-\delta'>0$, while
for  the first case
it  contains the union of  finitely many half-open intervals of  different  sizes.

\item For the second case,
the restrictions of the  periodic complement $\R\backslash {\mathcal S}_{a, b, c}$ 
on the period $[0, a)$
consists of
 finitely many mutually disjoint periodic gaps of lengths $b-a+\delta-\delta'$ and $\delta-\delta'$
  if $\delta-\delta'>0$, while
for  the first case
it includes
 finitely many mutually disjoint periodic  holes of   same size $b-a$.

\item For the second case, either $[0, \delta)$ or $[\delta', 0)$ is a gap contained in
the  periodic complement $\R\backslash {\mathcal S}_{a, b, c}$
if $\delta-\delta'>0$, while for the first case  no  left or right holes at the origin
are contained in the  periodic complement $\R\backslash {\mathcal S}_{a, b, c}$, see Lemma \ref{rabcbasic2.lem}.

\item  The complement of  the  maximal invariant set ${\mathcal S}_{a, b, c}$  for the second
case with $\delta=\delta'=0$ looks ``similar" to the one for the first case, since their restrictions on one period
both consist of finitely many mutually disjoint half-open intervals of length $b-a$.
On the other hand, they are ``different" if we  adhere  intervals of  sufficiently small length $\epsilon$   to each side of
 those half-open intervals.
For the second case, those gaps with neighboring intervals
has the following ``loop" structure via the piecewise linear transformation $R_{a, b, c}$
 with shrinking and expanding at its black hole and the origin:
$$
\begin{array}{ccc}
\boxed{c-c_0+[-\epsilon, \epsilon+b-a)+a\Z} &  
 \longleftarrow & \boxed{ [-\epsilon, \epsilon)+a\Z}\\
\downarrow
 &  & \uparrow\\
\vdots & & \vdots\\
\downarrow &  & \uparrow\\
\boxed{[c_0+a-b, c_0)+[-\epsilon, \epsilon)+a\Z} & 
\longrightarrow &
\boxed{R_{a, b, c}(c_0)+ [-\epsilon, \epsilon)+a\Z}
\end{array}
$$
Hence the piecewise linear transformation $R_{a, b, c}$ has ``finite" order,  as
 there exists a positive integer $L$ such that $(R_{a, b, c})^L$ is an ``identity" operator $I$
in the sense that
 $(R_{a, b, c})^L(t)+a\Z=t+a\Z$ for all $t\in {\mathcal S}_{a, b,c}$, see Theorem \ref{dabc1algebra.tm}.
But for the first case, applying the transform $R_{a, b, c}$
finitely many time,  the black hole  of the piecewise linear transformation $\tilde R_{a, b,c}$ is ``totally" attracted by the black hole of
the piecewise linear transformation $R_{a, b, c}$, which indicates that the piecewise linear transformation $R_{a, b, c}$ has ``infinite" order.

\item The complement of  the  maximal invariant set ${\mathcal S}_{a, b, c}$  for
the case that $N_2=0$, c.f. Example \ref{rationalexample2},
 looks also  ``similar" to the one for the first case, in the sense that their restrictions on one period
both consist of finitely many mutually disjoint half-open intervals of same length, but
 the lengths are different for those two cases. Also for the case that $N_2=0$,
 the  maximal invariant set ${\mathcal S}_{a, b, c}$ is the union of  finitely many half-open intervals of same size.

\item  For both cases,  the black hole of
the piecewise linear transformation $R_{a, b, c}$
attracts the black hole  of the piecewise linear transformation $\tilde R_{a, b,c}$ when applying the piecewise linear transformation $R_{a, b, c}$
 finitely many times,
$$[c_1, c_1+b-a)+a\Z\overset{R_{a, b, c}}\longmapsto\cdots \overset{R_{a, b, c}}\longmapsto [c_0+a-b, c_0)+a\Z 
 $$
see \eqref{dabc1holes.tm.pf.eq20},
 \eqref{fromblackholetoblackhole1} and
\eqref{dabc1discreteholes.tm.pf.eq6**}. Moreover,
\begin{eqnarray} \label{sabcblackholes} {\mathcal S}_{a,  b, c} &=&\R\backslash \big( \cup_{n=0}^\infty  (R_{a, b, c})^n ([c_1, c_1+b-a)+a\Z)\big)\nonumber\\
&=&\R\backslash \big( \cup_{n=0}^L (R_{a, b, c})^n ([c_1, c_1+b-a)+a\Z)\big)\end{eqnarray}
for some nonnegative integer $L\ge 0$ when ${\mathcal S}_{a, b, c}\ne \emptyset$.

\end{itemize}

\bigskip

Having  explicit construction of  the maximal invariant set ${\mathcal S}_{a, b, c}$ in Theorems \ref{dabc1holes.tm} and \ref{dabc1discreteholes.tm},
we next consider covering property of the maximal invariant set ${\mathcal S}_{a, b, c}$ and
then its application to
characterize \eqref{sabcemptynonempty2}.


\begin{thm}\label{covering.tm2}
  Let  $(a,b,c)$ be a triple of positive numbers satisfying
$a<b<c, \lfloor c/b\rfloor\ge 2, b-a<c_0:=c-\lfloor c/b\rfloor b<a$ and $0\le c_1:= c-c_0-\lfloor (c-c_0)/a\rfloor a\le 2a-b$,
and let ${\mathcal S}_{a, b, c}$
 be as in \eqref{dabc1.def}.
If
 ${\mathcal S}_{a, b, c}\ne \emptyset$ and either $a/b\not\in \Q$ or $a/b=p/q$ and $c\in b\Z/q$ for some coprime integers $p$ and $q$,
then
\begin{equation} \label{covering.tm2.eq1}
\big({\mathcal S}_{a, b, c}\cap ([0, c_0+a-b)+a\Z)+ \lfloor c/b\rfloor b\big)
\cup \big( \cup_{k=0}^{\lfloor c/b\rfloor -1}
({\mathcal S}_{a, b, c}+ kb)\big) =\R.
\end{equation}
\end{thm}

As a direct application of  the above theorem, we have that
$$
 \cup_{k=0}^{\lfloor c/b\rfloor -1}
({\mathcal S}_{a, b, c}+ kb) \subset \R \subset \cup_{k=0}^{\lfloor c/b\rfloor }
({\mathcal S}_{a, b, c}+ kb).
$$
So roughly speaking,   the maximal invariant  periodic set  ${\mathcal S}_{a, b, c}$ is either an empty set or  its $(\lfloor c/b\rfloor+1)$ copies
would cover the whole line.
Recall that the set ${\mathcal D}_{a, b, c}$ can be obtained from  the maximal invariant set ${\mathcal S}_{a, b, c}$
by some set operations (Theorem  \ref{dabcsabc.thm}). This together with  Theorem  \ref{covering.tm2} leads to the
following equivalence between the empty set property for ${\mathcal D}_{a, b, c}$ and
 an equality about the Lebesgue measure of the maximal invariant set ${\mathcal S}_{a, b, c}$.

\begin{thm}\label{sabcstar.tm}
 Let  $(a,b,c)$ be a triple of positive numbers satisfying
$a<b<c, b-a<c_0:=c-\lfloor c/b\rfloor b<a, 0<c_1:=\lfloor c/b\rfloor b-\lfloor (\lfloor c/b\rfloor b/a)\rfloor a<2a-b$
and $\lfloor c/b\rfloor\ge 2$. Assume that  ${\mathcal S}_{a, b, c}\ne \emptyset$ and
either $ a/b\not\in \Q$ or $a/b=p/q$ and $c/b\in \Z/q$ for some coprime integers $p$ and $q$.
Then
  ${\mathcal D}_{a, b, c}=\emptyset$  if and only if
\begin{equation}\label{sabcstar.tm.eq1} (\lfloor c/b\rfloor+1) |{\mathcal S}_{a, b, c}\cap [0,
c_0+a-b)|+
  \lfloor c/b\rfloor |{\mathcal S}_{a, b, c}\cap [c_0, a)|= a.\end{equation}
\end{thm}

For the triple $(a, b, c)=(\pi/4, 1, 23-11\pi/2)$ in Example \ref{irrationalexample1}, the equality
\eqref{sabcstar.tm.eq1} holds, and hence ${\mathcal G}(\chi_{[0, 23-11\pi/2)}, \pi\Z/4\times \Z)$
is a Gabor frame even though the maximal invariant set ${\mathcal S}_{a,b, c}$ is not an empty set.
For the triple $(a, b, c)=(\sqrt{3}/2, 1, 15\sqrt{3}/2)$ with irrational ratio between $a$ and $b$,
${\mathcal G}(\chi_{[0, 15 \sqrt{3}/2)}, \sqrt{3}\Z/2\times \Z)$ is not a Gabor frame as
\eqref{sabcstar.tm.eq1} does not hold. In fact,  in this case
\begin{eqnarray*}
{\mathcal D}_{a,b, c}={\mathcal S}_{a,b,c}
&=& \big([12-7\sqrt{3}, 7-4\sqrt{3}) \cup [ 8-9\sqrt{3}/2,3-3\sqrt{3}/2) \\
& & \cup [ 4-2\sqrt{3},11-6\sqrt{3})\big)
+\sqrt{3}\Z/2.
\end{eqnarray*}
For the triple $(a, b, c)=(6/7, 1, 23/7)$ in Example \ref{rationalexample2}
and the triples $(a, b, c)=(13, 17, 77)/17$ and $(13, 17, 73)/17$
in Example \ref{irrationalexample1},
the equality \eqref{sabcstar.tm.eq1} holds, but for the triple $(13, 17, 75)/17$ in Example \ref{irrationalexample1}
the equality \eqref{sabcstar.tm.eq1} does not hold. Thus  Gabor systems ${\mathcal G}(\chi_{[0, 23/7)}, 6\Z/7\times \Z)$,
${\mathcal G}(\chi_{[0, 77/17)}, 13\Z/17\times \Z)$ and ${\mathcal G}(\chi_{[0, 73/17)}, 13\Z/17\times \Z)$
are Gabor frames even though the maximal invariant sets are
nontrivial, but the Gabor system ${\mathcal G}(\chi_{[0, 75/17)}, 13\Z/17\times \Z)$ is not a Gabor frame.

\bigskip

In this section, we finally consider the restriction of transformations $R_{a, b, c}$ and $\tilde R_{a, b, c}$ on
 the maximal invariant set ${\mathcal S}_{a, b, c}$.
The
finite-interval property for the maximal invariant set ${\mathcal S}_{a, b, c}$ in Theorems
\ref{dabc1holes.tm} and  \ref{dabc1discreteholes.tm}
can also be interpreted
as  its complement  consists of finitely many  holes
on a period. So we may shrink those holes into points, which
maps the maximal invariant set ${\mathcal S}_{a, b, c}$ into the real line with marks.
More importantly, after performing the above holes-removal surgery,
the  set ${\mathcal K}_{a, b, c}$ of marks on the line is a {\em finite cyclic
group} if  $a/b=p/q$ and $c\in b\Z/q$ for some coprime integers $p$ and $q$,
and the set ${\mathcal K}_{a, b, c}$ of marks on the line is a  finite subset of infinite cyclic group if $a/b\not\in \Q$, and
  the  nonlinear application of the  transformation
$R_{a,b, c}$ on the maximal invariant set ${\mathcal S}_{a, b, c}$  becomes a  {\em rotation} on the circle
$\R/Y_{a,b,c}(a)\Z$ with marks ${\mathcal K}_{a, b, c}$,
see Appendix \ref{ergodic.appendix} for non-ergodicity of the piecewise linear transformation $R_{a, b, c}$.

  \begin{thm}\label{dabc1algebra.tm}
  Let  $(a,b,c)$ be a triple of positive numbers satisfying
$a<b<c, \lfloor c/b\rfloor\ge 2, b-a<c_0:=c-\lfloor c/b\rfloor b<a$ and $0\le c_1:= c-c_0-\lfloor (c-c_0)/a\rfloor a\le 2a-b$,  and let ${\mathcal S}_{a, b, c}$ and $Y_{a,b,c}$
 be as in \eqref{dabc1.def} and \eqref{xabc.def} respectively. Then the following statements hold.
 \begin{itemize}
\item[{(i)}] If
 ${\mathcal S}_{a, b, c}\ne \emptyset$ and either $a/b\not\in \Q$ or $a/b=p/q$ and $c\in b\Z/q$ for some coprime integers $p$ and $q$,
then under the isomorphism $Y_{a, b, c}$ from ${\mathcal S}_{a, b, c}$ to the line with marks, the action to apply the
piecewise linear transformation $R_{a, b,c}$ on ${\mathcal S}_{a, b, c}$
  becomes a shift on the line with marks; i.e.,
\begin{equation}\label{dabc1algebra.tm.eq1}
Y_{a, b, c}(R_{a, b, c}(t)+a\Z)= Y_{a, b, c}(t)+ Y_{a, b, c}
(c_1+b-a)+Y_{a,b,c}(a) \Z\quad {\rm for \ all}\ t\in
{\mathcal S}_{a, b, c}.
\end{equation}

\item[{(ii)}] If
 ${\mathcal S}_{a, b, c}\ne \emptyset$,  $a/b=p/q$ and $c\in b\Z/q$ for some coprime integers $p$ and $q$,
then marks on the line (i.e., images of  the gaps in the
complement of the set ${\mathcal S}_{a, b, c}$
 under
the isomorphism $Y_{a, b, c}$) form a finite cyclic group
generated by $Y_{a, b, c} (c_1+b-a)+Y_{a,b,c}(a) \Z$,
\begin{equation*}
{\mathcal K}_{a, b, c}= Y_{a, b, c} (c_1+b-a)\Z +Y_{a, b, c}(a)\Z
={\rm gcd}\big(Y_{a, b, c} (c_1+b-a), Y_{a, b, c} (a)\big)\Z,
\end{equation*}
 where ${\mathcal K}_{a, b, c}$ is
the set of marks on the line.

\item[{(iii)}] If
 ${\mathcal S}_{a, b, c}\ne \emptyset$ and  $a/b\not\in \Q$,
 the set ${\mathcal K}_{a, b, c}$ of marks on the line is given by
 \begin{equation*}
 {\mathcal K}_{a, b, c}=\cup_{n=1}^M  (n Y_{a, b, c} (c_1+b-a)+Y_{a, b, c}(a)\Z)
 \end{equation*}
 where $M$ is the unique positive integer such that
 $ MY_{a, b, c}(c_1+b-a)-Y_{a, b, c}(c_0)\in Y_{a, b, c}(a)$.

\end{itemize}
  \end{thm}


In the next four subsections, we prove  Theorems  \ref{dabc1holes.tm}, \ref{dabc1discreteholes.tm}, \ref{covering.tm2} and \ref{sabcstar.tm}, and
\ref{dabc1algebra.tm} respectively.


%

%

\subsection{Maximal invariant sets with irrational  time-frequency lattices}
\label{subsection2.3}

To prove Theorem  \ref{dabc1holes.tm}, we need a  characterization
of non-empty-set property for the maximal invariant set ${\mathcal S}_{a, b, c}$.

 \begin{lem} \label{dabc1holes.lem}
 Let $(a,b,c)$ be a triple of positive numbers satisfying
$a<b<c,   b-a<c_0:=c-\lfloor c/b\rfloor b<a, \lfloor c/b\rfloor\ge
2$ and $a/b\not\in \Q$. Then   ${\mathcal S}_{a, b, c}\ne
\emptyset$ if and only if there exists a nonnegative integer $D\le
\lfloor a/(b-a)\rfloor-1$ such that  $(R_{a, b,c})^{D} ([c- c_0,
c+b-c_0-a)+a\Z)=[c_0+a-b, c_0)+a\Z$ and
that  $(R_{a, b,c})^n ([c- c_0, c+b-c_0-a)+a\Z), 0\le n\le D-1$, has their closures being mutually disjoint
and  contained
 in $(0, c_0+a-b)\cup (c_0, a)+a\Z$.
\end{lem}

The sufficiency of the above lemma  follows from the invariance of the set
$\cup_{n=0}^{D}(R_{a,b,c})^n ([c-c_0, c-c_0+b-a)+a\Z)$ under transformations $R_{a, b, c}$ and $\tilde R_{a, b, c}$,
and the minimality of the set  $\R\backslash {\mathcal S}_{a, b, c}$ in Theorem \ref{dabc1maximal.cor.thm}.
To prove the necessity in Lemma \ref{dabc1holes.lem}, we let $D$ be the minimal nonnegative integer such that
$(R_{a, b,c})^{D} ([c- c_0,
c+b-c_0-a)+a\Z)\cap ([c_0+a-b, c_0)+a\Z)\ne \emptyset$. The existence of such an integer $D$
follows from the mutually disjointness of holes $(R_{a, b,c})^{n} ([c- c_0,
c+b-c_0-a)+a\Z), 0\le n\le D-1$ (and their closures), see \eqref{dabc1holes.tm.pf.eq13}.
We then prove that $(R_{a, b,c})^{D} ([c- c_0,
c+b-c_0-a)+a\Z)=[c_0+a-b, c_0)+a\Z$ by applying Lemma \ref{rabcbasic2.lem} and Propositions \ref{rabcbasic1.thm} and \ref{blackholestwo.prop}, see \eqref{dabc1holes.tm.pf.eq20}.
Now assuming that
Lemma \ref{dabc1holes.lem} holds, we  start our proof of Theorem \ref{dabc1holes.tm}.

\begin{proof}[Proof of Theorem \ref{dabc1holes.tm}]
Assume that ${\mathcal S}_{a, b, c}\ne \emptyset$.
By Lemma \ref{dabc1holes.lem}, there exists a nonnegative integer $D\le \lfloor a/(b-a)\rfloor-1$
 such that $(R_{a, b,c})^{D} ([c- c_0,
c+b-c_0-a)+a\Z)=[c_0+a-b, c_0)+a\Z$, and
$(R_{a, b,c})^n ([c- c_0, c+b-c_0-a)+a\Z), 0\le n\le D-1$, have their closures being mutually disjoint
and contained in $(0, c_0+a-b)\cup (c_0, a)+a\Z$.
Then
$\cup_{n=0}^{D} (R_{a,b,c})^n
([c-c_0, c-c_0+b-a)+a\Z)$ is minimal set that contains black holes of transformations $R_{a, b, c}$ and $\tilde R_{a, b, c}$,
and that is invariant under those transformations $R_{a, b, c}$ and $\tilde R_{a, b, c}$.
Here the invariance holds as
  \begin{equation} \label{dabc1holes.tm.pff.eq3} (R_{a,b,c})^n
([c-c_0, c-c_0+b-a)+a\Z)= (\tilde R_{a,b,c})^{D-n}
([c_0+a-b, c_0)+a\Z)\end{equation}
for all $0\le n\le D$ by Proposition \ref{rabcinvertibility.prop} and the
property that
\begin{equation}  \label{dabc1holes.tm.pff.eq2} (R_{a, b,c})^n ([c- c_0, c+b-c_0-a]+a\Z)\subset (0, c_0+a-b)\cup (c_0, a)+a\Z
\end{equation}
for all $0\le n\le D-1$.
This together with Theorem \ref{dabc1maximal.cor.thm} proves that
\begin{equation}\label{dabc1holes.tm.pff.eq1}
{\mathcal S}_{a, b, c}  =  \R\backslash \cup_{n=0}^{D} (R_{a,b,c})^n
([c-c_0, c-c_0+b-a)+a\Z).
\end{equation}
Hence the conclusions \eqref{dabc1holes.tm.eq-1}, \eqref{dabc1holes.tm.eq0}
and \eqref{dabc1holes.tm.eq1} follow from  \eqref{rabcnewplus.def}, \eqref{tilderabcnewplus.def}
and
\eqref{dabc1holes.tm.pff.eq3}--\eqref{dabc1holes.tm.pff.eq1}.
\end{proof}

We finish this subsection with the proof of Lemma \ref{dabc1holes.lem}.

\begin{proof} [Proof of Lemma \ref{dabc1holes.lem}]\  ($\Longleftarrow$)\
By \eqref{rabcbasic1.thm.pf.eq-1},
$\big|(R_{a,b,c})^n ([c-c_0, c-c_0+b-a)+a\Z)\cap [0,a)\big|\le b-a$
for all $0\le n\le D$, which together with the periodicity \eqref{rabcnewplus.def} of the transformation $R_{a, b, c}$ implies that
$$|A_{D}\cap [0,a)|\le (D+1)(b-a)<a,$$
where  $A_{D}=\cup_{n=0}^{D}(R_{a,b,c})^n ([c-c_0, c-c_0+b-a)+a\Z)$.
Hence
\begin{equation}\label{dabc1holes.tm.pf.eq3}
A_{D}\ne \R.\end{equation}
By \eqref{dabc1holes.tm.pf.eq3},  the sufficiency reduces to proving that
 \begin{equation}\label{dabc1holes.tm.pf.eq4} \R\backslash A_{D}\subset {\mathcal S}_{a, b, c}.\end{equation}
(In particular,  $\R\backslash A_{D}= {\mathcal S}_{a, b, c}$
by  Theorem \ref{dabc1maximal.cor.thm} and Proposition \ref{dabc1pointcharacterization.prop}). Recall that the black hole
$[c-c_0, c-c_0+b-a)+a\Z$ of the transformation  $\tilde R_{a,b,c}$
is contained in $A_{D}$, and that the black hole $[c_0+a-b,
c_0)+a\Z$ of the transformation $R_{a,b,c}$ is also contained in
$A_{D}$ as it is same as $(R_{a,b,c})^{D} ([c-c_0, c-c_0+b-a)+a\Z)
$ by the assumption.
Then we obtain from   \eqref{blackholesoftransformations} and \eqref{invertibilityoftransformations}  that
\begin{equation} \label{dabc1holes.tm.pf.eq5}
\tilde R_{a,b,c} (\R\backslash A_{D})\subset \R\backslash A_{D}.
\end{equation}
Similarly we get from \eqref{rabcnewplus.def}  and \eqref{invertibilityoftransformations}  that
$R_{a,b,c} (t) \not\in [c-c_0, c-c_0+b-a)+a\Z$ and that $R_{a,b,c} (t)  \not\in \cup_{n=1}^{D}(R_{a,b,c})^n ([c-c_0, c-c_0+b-a)+a\Z)$
for any $t\not\in A_{D}$. This implies that
 \begin{equation} \label{dabc1holes.tm.pf.eq5b}
  R_{a, b,c} (\R\backslash A_{D})\subset \R\backslash A_{D}.
 \end{equation}
 Take any $t\in \R\backslash A_{D}$, then $(R_{a, b, c})^n(t)$ and $(\tilde R_{a, b, c})^n(t), n\ge 0$,
belong to $\R\backslash A_{D}$ by \eqref{dabc1holes.tm.pf.eq5} and \eqref{dabc1holes.tm.pf.eq5b}
and hence they do not fall in the black holes of the transformations $R_{a, b, c}$ and $\tilde R_{a, b,c}$. Therefore
$t\in {\mathcal S}_{a, b, c}$ by Proposition
\ref{dabc1pointcharacterization.prop}, and   hence \eqref{dabc1holes.tm.pf.eq4} is established.

($\Longrightarrow$)\quad First we need the following claim:

\begin{claim}\label{dabc1holes.lem.pf.claim1}
{ There exists a nonnegative integer $n\le \lfloor a/(b-a)\rfloor-1$ such that
\begin{equation}\label{dabc1holes.tm.pf.eq11}
\big((R_{a,b,c})^n([c-c_0, c-c_0+b-a)+a\Z)\big)\cap ([c_0+a-b, c_0)+a\Z)\ne \emptyset.
\end{equation}
}
\end{claim}

\begin{proof} Suppose, on the contrary, that
\begin{equation}\label{dabc1holes.tm.pf.eq12}
\big((R_{a,b,c})^n([c-c_0, c-c_0+b-a)+a\Z)\big)\cap ([c_0+a-b, c_0)+a\Z)= \emptyset\end{equation}
for all $0\le n\le \lfloor a/(b-a)\rfloor-1$.
We first need to prove  that
\begin{equation}\label{dabc1holes.tm.pf.eq13}
(R_{a,b,c})^n ([c-c_0, c-c_0+b-a)+a\Z)\big)\cap
(R_{a,b,c})^{n'} ([c-c_0, c-c_0+b-a)+a\Z)\big)=\emptyset
\end{equation}
for all $0\le n\ne n'\le \lfloor a/(b-a)\rfloor-1$.
Suppose on the contrary that \eqref{dabc1holes.tm.pf.eq13} does not hold.
Then
 there exists $t\in \big((R_{a,b,c})^n ([c-c_0, c-c_0+b-a)+a\Z)\big)\cap
\big((R_{a,b,c})^{n'} ([c-c_0, c-c_0+b-a)+a\Z)\big)$. By \eqref{invertibilityoftransformations} and \eqref{dabc1holes.tm.pf.eq12},
$(\tilde R_{a,b,c})^{\max(n,n')-1}(t)\in  \big(R_{a,b,c} ([c-c_0, c-c_0+b-a)+a\Z)\big)\cap \big([c-c_0, c-c_0+b-a)+a\Z\big)$,
which is a contradiction as $\big(R_{a,b,c} ([c-c_0, c-c_0+b-a)+a\Z)\big)\cap \big([c-c_0, c-c_0+b-a)+a\Z\big)
\subset ([c-c_0+b, c-c_0+2a)+a\Z)\cap  \big([c-c_0, c-c_0+b-a)+a\Z\big)=\emptyset$ by \eqref{rabcnewplus.def} and \eqref{dabc1holes.tm.pf.eq12}.
Hence \eqref{dabc1holes.tm.pf.eq13} is established.

By \eqref{rabcbasic1.thm.pf.eq0} and \eqref{dabc1holes.tm.pf.eq12}, we have that
\begin{equation}\label{dabc1holes.tm.pf.eq14}
\big|\big((R_{a,b,c})^n([c-c_0, c-c_0+b-a)+a\Z)\big)\cap [0,a)\big|=b-a
\end{equation}
for all $0\le n\le \lfloor a/(b-a)\rfloor-1$.
This together with  \eqref{dabc1holes.tm.pf.eq13}
implies that
\begin{eqnarray*}
a & \ge &  \big|\cup_{n=0}^{\lfloor a/(b-a)\rfloor-1}\big((R_{a,b,c})^n([c-c_0, c-c_0+b-a)+a\Z)\big)\cap [0,a)\big|\nonumber\\
& & +|[c_0+a-b, c_0)|\nonumber\\
 & \ge &  \sum_{n=0}^{\lfloor a/(b-a)\rfloor-1}\big|\big((R_{a,b,c})^n([c-c_0, c-c_0+b-a)+a\Z)\big)\cap [0,a)\big|\nonumber\\
  & &  +(b-a)\nonumber\\
 & = &  (\lfloor a/(b-a)\rfloor+1)(b-a)>a,
\end{eqnarray*}
which is a contradiction. 
This completes the proof of Claim \ref{dabc1holes.lem.pf.claim1}.
\end{proof}

We return to  work on the proof of the sufficiency. Let $D$ be
the smallest nonnegative integer such that
\begin{equation} \label{dabc1holes.tm.pf.eq16} \big((R_{a,b,c})^{D}([c-c_0, c-c_0+b-a)+a\Z)\big)\cap ([c_0+a-b, c_0)+a\Z)\ne \emptyset.
\end{equation}
Then $D\le \lfloor a/(b-a)\rfloor-1$ by
\eqref{dabc1holes.tm.pf.eq11}. By the definition of the integer
$D$, we have that
\begin{equation} \label{dabc1holes.tm.pf.eq17} \big((R_{a,b,c})^{n}([c-c_0, c-c_0+b-a)+a\Z)\big)\cap ([c_0+a-b, c_0)+a\Z)= \emptyset
\end{equation}
for all $0\le n<D-1$. Following the argument used to establish \eqref{dabc1holes.tm.pf.eq13}, we have
that
\begin{equation} \label{dabc1holes.tm.pf.eq18-}
\Big((R_{a,b,c})^{n}\big([c-c_0, c-c_0+b-a)+a\Z)\big)\Big) \cap \Big((R_{a,b,c})^{n'}\big([c-c_0, c-c_0+b-a)+a\Z)\big)\Big)
=\emptyset
\end{equation}
for all $0\le n\ne n'\le D-1$.
We claim that for all $0\le n<D-1$,
\begin{equation} \label{dabc1holes.tm.pf.eq18}
(R_{a,b,c})^n([c-c_0, c-c_0+b-a)+a\Z)=[\tilde u_n, u_n)+a\Z
\end{equation}
with
\begin{equation}\label{dabc1holes.tm.pf.eq19} u_n=\tilde u_n+b-a\quad {\rm and}\quad
[\tilde u_n, u_n]\subset (0, c_0+a-b) \ {\rm or} \ (c_0, a);\end{equation}
that is, $(R_{a, b, c})^n ([c-c_0, c-c_0+b-a)+a\Z$ are periodic holes of length $b-a$ `strictly' contained in $(0, c_0+a-b)\cup (c_0, a)+a\Z$.
Let $u_0$ be the unique number in $(0,a]$ such that $ c-c_0+b-a-u_0\in a\Z$, and
$\tilde u_0<u_0$ be  the largest number such that $c-c_0-\tilde u_0\in a\Z$. Then
$[c-c_0, c-c_0+b-a)+a\Z=[\tilde u_0, u_0)+a\Z$. Furthermore  $\tilde u_0>0$ and $u_0<a$ as otherwise either
 $((0, \epsilon)+a\Z)\cap {\mathcal S}_{a, b, c}=\emptyset$
or $((-\epsilon, 0)+a\Z)\cap {\mathcal S}_{a, b, c}=\emptyset$  for some small $\epsilon>0$
by  Proposition \ref{blackholestwo.prop}, which contradicts to the conclusion \eqref{rabcbasic2.lem.eq1} in Lemma \ref{rabcbasic2.lem}.
This together with \eqref{dabc1holes.tm.pf.eq17} and the definitions of $\tilde u_0$ and $u_0$ implies  that  $u_0=\tilde u_0+b-a$ and
$[\tilde u_0, u_0]\subset (0, c_0+a-b] \ {\rm or} \ [c_0, a)$.
Therefore it suffices to  prove $u_0\ne c_0+a-b$ and $\tilde u_0\ne c_0$.
Suppose on the contrary that $u_0=c_0+a-b$, then
$(R_{a,b,c}([0,\epsilon)+a\Z))\cap {\mathcal S}_{a, b, c} =
([ u_0, u_0+\epsilon)+a\Z)\cap {\mathcal S}_{a, b, c}\subset
([c_0+a-b, c_0)+a\Z)\cap{\mathcal S}_{a, b, c}=\emptyset$
for sufficiently small $\epsilon\in (0, c_0+a-b)$ by \eqref{blackholestwo.eq1}.
Thus $([0,\epsilon)+a\Z)\cap {\mathcal S}_{a, b, c}=\emptyset$ by \eqref{tildelambdaabc.def}
and Proposition \ref{rabcbasic1.thm}, which contradicts to \eqref{rabcbasic2.lem.eq1} in Lemma \ref{rabcbasic2.lem}.
This proves that  $u_0\ne c_0+a-b$. Similarly we can prove that $\tilde u_0\ne c_0$.
Therefore the conclusions \eqref{dabc1holes.tm.pf.eq18} and \eqref{dabc1holes.tm.pf.eq19} hold for $n=0$.
Inductively, we assume that
$(R_{a,b,c})^n([c-c_0, c-c_0+b-a)+a\Z)=[\tilde u_n, u_n)+a\Z$
with $u_n=\tilde u_n+b-a$ and $[\tilde u_n, u_n]\subset (0, c_0+a-b) \ {\rm or} \ (c_0, a)$.
Then we see that
$(R_{a,b,c})^{n+1}([c-c_0, c-c_0+b-a)+a\Z)=[R_{a,b,c}\tilde u_n, R_{a,b,c} u_n)+a\Z$
with $R_{a,b,c}u_n-R_{a,b,c} \tilde u_n=b-a$. Let
$u_{n+1}$ be the unique number in $(0,a]$ such that $R_{a,b,c} u_n-u_{n+1}\in a\Z$, and
$\tilde u_{n+1}<u_{n+1}$ be  the largest number such that $R_{a,b,c} \tilde u_n-\tilde u_{n+1}\in a\Z$. Then
$(R_{a,b,c})^{n+1}([c-c_0, c-c_0+b-a)+a\Z)=[\tilde u_{n+1}, u_{n+1})+a\Z$. Similarly we can show that
$0<\tilde u_{n+1}<u_{n+1}<a$ and $[\tilde u_{n+1}, u_{n+1})\cap [c_0+a-b, c_0)=\emptyset$, and hence
$ u_{n+1}=\tilde u_{n+1}+b-a$ and $[\tilde u_{n+1}, u_{n+1}]\subset (0, c_0+a-b]$ or  $[\tilde u_{n+1}, u_{n+1}]\subset [c_0, a)$.
Now we prove that $u_{n+1}\ne c_0+a-b$ and $\tilde u_{n+1}\ne c_0$.
Suppose on the contrary that $u_{n+1}=c_0+a-b$. From the inductive hypothesis, we have that for $0\le k\le n$,
$(R_{a,b,c})^{k+1}([0, \epsilon)+a\Z)=[u_k, u_k+\epsilon)+a\Z$  where $\epsilon>0$ is  sufficiently small.
This together with \eqref{invertibilityoftransformations}, the inductive hypothesis and Proposition \ref{rabcbasic1.thm}
implies that $([0, \epsilon)+a\Z)\cap {\mathcal S}_{a, b, c}=\emptyset$, which contradicts to \eqref{rabcbasic2.lem.eq1}
in Lemma \ref{rabcbasic2.lem}.
This proves that  $u_{n+1}\ne c_0+a-b$. Similarly we can prove that $\tilde u_{n+1}\ne c_0$.
Hence we can proceed  the inductive proof of the conclusion \eqref{dabc1holes.tm.pf.eq17}.

Next we prove that
\begin{equation}\label{dabc1holes.tm.pf.eq20}
(R_{a,b,c})^{D}([c-c_0, c-c_0+b-a)+a\Z)=[c_0+a-b, c_0)+a\Z.
\end{equation}
From \eqref{rabcnewplus.def}, \eqref{dabc1holes.tm.pf.eq18} and \eqref{dabc1holes.tm.pf.eq19} it follows that
\begin{equation}\label{dabc1holes.tm.pf.eq21}
(R_{a,b,c})^{D}([c-c_0, c-c_0+b-a)+a\Z)=[\tilde u_{D}, u_{D})+a\Z
\end{equation}
for some $u_{D}\in (0,a]$ and $u_{D}-\tilde u_{D}=b-a$. 
By \eqref{dabc1holes.tm.pf.eq16},   $[\tilde u_{D}, u_{D})\cap [c_0+a-b, c_0)\ne \emptyset$, which implies that
$u_{D}>c_0+a-b$.  Suppose that $c_0+a-b<u_{D}<c_0$, then
$(R_{a,b,c})^{D+1} ([0, \epsilon)+a\Z)=[\tilde u_{D}, \tilde u_{D}+\epsilon)+a\Z$
for sufficiently small $\epsilon$ by  \eqref{dabc1holes.tm.pf.eq18} and \eqref{dabc1holes.tm.pf.eq19}.
This together with \eqref{invertibilityoftransformations}, \eqref{dabc1holes.tm.pf.eq18} and \eqref{dabc1holes.tm.pf.eq19}
 and Proposition \ref{rabcbasic1.thm} implies that $([0, \epsilon)+a\Z)\cap {\mathcal S}_{a, b, c}=\emptyset$,
 which contradicts to \eqref{rabcbasic2.lem.eq1} in Lemma \ref{rabcbasic2.lem}.
Therefore $u_D\ge c_0$.
 Similarly, we can prove that $\tilde u_{D}\le c_0+a-b$.  Hence \eqref{dabc1holes.tm.pf.eq20}
 follows as $u_{D}-\tilde u_{D}=b-a$.

Finally  we prove that  $(R_{a, b, c})^n ([c-c_0, c-c_0+b-a)+a\Z), 0\le n\le D$,
have their closures being mutually disjoint. By Proposition \ref{blackholestwo.prop}, we have that
$$(R_{a, b, c})^n ([c-c_0, c-c_0+b-a)+a\Z)\cap {\mathcal S}_{a, b, c}=\emptyset, \ 0\le n\le D.$$
From \eqref{dabc1holes.tm.pf.eq18}, \eqref{dabc1holes.tm.pf.eq19} and
Lemma \ref{rabcbasic2.lem}, we obtain that
$$(\tilde u_n-\epsilon, \tilde u_n)\cap {\mathcal S}_{a, b, c}\ne \emptyset \quad {\rm and}
\quad (u_n, u_n+\epsilon)\cap {\mathcal S}_{a, b, c}\ne \emptyset, \ 0\le n\le D $$
for any sufficiently small $\epsilon>0$. Combining the above two observations with
\eqref{dabc1holes.tm.pf.eq18-}--\eqref{dabc1holes.tm.pf.eq20} 
prove mutual disjointness for the closures of holes  $(R_{a, b, c})^n ([c-c_0, c-c_0+b-a)+a\Z), 0\le n\le D$.
%
%
%
%
\end{proof}

\subsection{Maximal invariant sets with rational  time-frequency  lattices} 

To prove  Theorem \ref{dabc1discreteholes.tm}, we need   the following result about the maximality of
the set ${\mathcal S}_{a, b, c}$ under the transformation $R_{a, b, c}$ in the case that $a/b\in \Q$, c.f.  Theorem \ref{dabc1maximal.cor.thm}
for the maximality of the set ${\mathcal S}_{a, b, c}$ under {\bf two} transformations $R_{a, b, c}$
and $\tilde R_{a, b, c}$ in the case that $a/b\not\in \Q$.

\begin{lem}\label{dabc0rationalmax.prop} Let $a<b<c,   b-a<c_0:=c-\lfloor c/b\rfloor b <a, \lfloor c/b\rfloor\ge 2$,  $a/b=p/q$  and $c\in b\Z/q$
for some coprime integers $p$ and $q$. Then
\begin{itemize}
\item [{(i)}]\
${\mathcal S}_{a, b, c}$ is the maximal set
that is invariant under the transformation $R_{a,  b, c}$ and that has empty intersection with
the black hole $[c_0+a-b, c_0)+a\Z$  of the transformation $R_{a, b, c}$.

\item [{(ii)}] $  \R\backslash {\mathcal S}_{a, b, c}$ is the minimal set
such that it is invariant under the piecewise linear transformation $R_{a, b, c}$ and that it contains the black holes
of  piecewise linear transformations $R_{a, b, c}$ and $\tilde R_{a, b, c}$.
\end{itemize}
\end{lem}

%
%

To prove  Theorem \ref{dabc1discreteholes.tm}, we also need
the following result about the dense property of
the  maximal invariant set ${\mathcal S}_{a, b, c}$ around the origin in the case that $a/b\in \Q$, c.f.  Lemma  \ref{rabcbasic2.lem}
for the dense property of the set ${\mathcal S}_{a, b, c}$   in the case that $a/b\not\in \Q$.

\begin{lem}\label{dabc1rationalmax.prop} Let $a<b<c,   b-a<c_0 <a, \lfloor c/b\rfloor\ge 2$,  $a/b=p/q$  and $c/b\in \Z/q$
for some coprime integers $p$ and $q$. Assume that ${\mathcal S}_{a, b, c}\ne \emptyset$. Then the following statements hold.
\begin{itemize}
\item[{(i)}]  At least one of two intervals $[0, b/q)$ and $c_0+[0, b/q)$ is contained in
${\mathcal S}_{a, b, c}$.

\item[{(ii)}] At least one of two intervals $[-b/q, 0)$ and $c_0+a-b+[-b/q, 0)$ is contained in
${\mathcal S}_{a, b, c}$.

\item[{(iii)}] At least one  of two intervals
$[0, b/q)$ and $[-b/q, 0)$ is contained in ${\mathcal S}_{a, b, c}$.
\end{itemize}
\end{lem}

Now supposing that Lemmas \ref{dabc0rationalmax.prop} and \ref{dabc1rationalmax.prop} hold,
we start the proof of Theorem \ref{dabc1discreteholes.tm}, c.f. Examples \ref{rationalexample1} and  \ref{rationalexample2} in Section \ref{gaborirrational.section}.

\begin{proof}[Proof of Theorem \ref{dabc1discreteholes.tm}]\
Let  $\delta\in [0, c_0+a-b] $ and $\delta'\in [c_0-a, 0]$ be so
chosen that
$[\delta', \delta)$ is the maximal interval contained in $\R\backslash {\mathcal S}_{a, b, c}$.
Then $\delta, \delta'\in b\Z/q$ by \eqref{dabcicontinuoustodiscrete.eq1},
and they satisfy \eqref{dabc1discreteholes.tm.eq1}
by the assumption  ${\mathcal S}_{a, b, c}\ne\emptyset$ and Lemma \ref{dabc1rationalmax.prop}, since at least one of two intervals
$[0, b/q)$ and $[-b/q, 0)$ is contained in ${\mathcal S}_{a, b, c}$.
Therefore the first conclusion (i) holds.

Now we
divide the following three cases:   (1) $\delta'=0$ and $\delta\ne 0$; (2)
$\delta=0$ and $\delta'\ne 0$;  and (3) $\delta=\delta'=0$, to verify the conclusions (ii)--(vi).

\smallskip

  \noindent
{\bf Case 1}\ {\em $\delta'=0$ and $\delta\ne 0$}

In this case,
\begin{equation}\label{dabc1discreteholes.tm.pf.case1eq1}
[-b/q, 0)\subset {\mathcal S}_{a, b, c} \quad{\rm and} \quad [0,\delta)\cap {\mathcal S}_{a, b, c}=\emptyset.
\end{equation}
Then
\begin{equation} \label{dabc1discreteholes.tm.pf.case1eq2} [c-c_0-b/q, c-c_0)=R_{a, b,c} [-b/q, 0)\subset
 R_{a, b, c}{\mathcal S}_{a, b, c}=
{\mathcal S}_{a, b, c}
\end{equation}
by  \eqref{rabcnewplus.def}, \eqref{dabc1discreteholes.tm.pf.case1eq1}, and Proposition \ref{rabcbasic1.thm};
\begin{eqnarray} \label{dabc1discreteholes.tm.pf.case1eq3} & & [c-c_0+b-a+\delta , c-c_0+b-a+\delta +b/q)\nonumber\\
& = &
\left\{\begin{array}{ll}
R_{a, b,c} [\delta , \delta +b/q)-a & {\rm if} \ 0<\delta <c_0+a-b \\
R_{a, b,c} [c_0, c_0+b/q) & {\rm if} \ \delta =c_0+a-b
\end{array}\right. \nonumber\\
& \subset &  R_{a, b, c} {\mathcal S}_{a, b, c}=
{\mathcal S}_{a, b, c}
\end{eqnarray}
by  \eqref{rabcnewplus.def}, \eqref{dabc1discreteholes.tm.pf.case1eq1}, Proposition \ref{rabcbasic1.thm},  the maximality
of the interval $[0, \delta)$
 in $\R\backslash{\mathcal S}_{a, b, c}$,
 and the first conclusion in Lemma \ref{dabc1rationalmax.prop};
and
\begin{eqnarray}\label{dabc1discreteholes.tm.pf.case1eq4}
 & & [c-c_0, c-c_0+b-a+\delta )\cap {\mathcal S}_{a, b, c}\nonumber\\
& = & \big (R_{a, b,c} [-a, \delta -a)  
\cup
 [c-c_0, c-c_0+b-a)\big)\cap  {\mathcal S}_{a, b, c}\nonumber\\
 & = & \ R_{a, b,c} ([-a, \delta -a) \cap   {\mathcal S}_{a, b, c})=\emptyset,
\end{eqnarray}
where the first equality follows from \eqref{rabcnewplus.def}, Proposition \ref{rabcbasic1.thm} and the assumption
$[0, \delta)\subset [0, c_0+a-b)$, and
the second equality holds by Propositions  \ref{dabcbasic1.prop} and \ref{rabcinvertibility.prop}.
Thus $[c-c_0, c-c_0+b-a+\delta)$
is a gap (i.e., an interval with empty set intersection with ${\mathcal S}_{a, b, c}$)
  with length $b-a+\delta$
  and boundary intervals of length $b/q$ at each side in the set ${\mathcal S}_{a, b, c}$.

\smallskip
{\em Proof of the conclusion (ii).} \quad Let $N_1$ be the smallest nonnegative integer such that
$(R_{a, b, c})^{N_1}([c-c_0, c+b-c_0-a+\delta)+a\Z) \cap ([c_0+a-b, c_0)+a\Z)\ne \emptyset$ if it exists, and ${N_1}=+\infty$ otherwise.
Mimicking the argument used in the proof of Claim \ref{dabc1holes.lem.pf.claim1}, we have that $N_1<\infty$.
We divide two cases to prove the  conclusion (ii).

\smallskip  {\bf Case 1a}:  $N_1=0$.

  In this case it follows from \eqref{dabc1discreteholes.tm.pf.case1eq2}, \eqref{dabc1discreteholes.tm.pf.case1eq3}
 \eqref{dabc1discreteholes.tm.pf.case1eq4} that
 $(R_{a, b, c})^{N_1}[c-c_0, c-c_0+b-a+\delta)=[c-c_0, c-c_0+b-a+\delta)$
is a gap   of length $b-a+\delta$
with boundary intervals of length $b/q$ at each side in the set ${\mathcal S}_{a, b, c}$.
This, together with $[c_0, c_0+b/q)\subset {\mathcal S}_{a, b, c}$ by Lemma \ref{dabc1rationalmax.prop}, and
the definition of the nonnegative integer $N_1$,
proves that
\begin{equation} \label{dabc1discreteholes.tm.pf.caseeq1eq5}
(R_{a, b, c})^{N_1}[c-c_0, c-c_0+b-a+\delta)+a\Z= [c_0+a-b-\delta, c_0)+a\Z
\end{equation}
and hence the conclusions \eqref{dabc1discreteholes.tm.eq2} and \eqref{dabc1discreteholes.tm.eq3} in the case that  $N_1=0$.

\smallskip  {\bf Case 1b}:  $1\le N_1<\infty$.

 In this case
 \begin{equation}\label{dabc1discreteholes.tm.pf.case1eq6}
  (R_{a, b, c})^n([c-c_0, c+b-c_0-a+\delta)+a\Z)\cap ([c_0+a-b, c_0)+a\Z)=\emptyset\end{equation}
  for all $0\le n\le {N_1}-1$.
 Let us verify that
 \begin{equation}\label{dabc1discreteholes.tm.pf.case1eq7} (R_{a, b, c})^n([c-c_0, c+b-c_0-a+\delta)+a\Z)=[b_n+a-b-\delta, b_n)+a\Z\end{equation}
  with
  \begin{equation}\label{dabc1discreteholes.tm.pf.case1eq8} [b_n+a-b-\delta-b/q, b_n+b/q)\subset [0, c_0+a-b)\cup [c_0, a),\end{equation}
     \begin{equation}\label{dabc1discreteholes.tm.pf.case1eq9-} ([b_n+a-b-\delta,b_n)+a\Z)\cap  {\mathcal S}_{a, b, c}=\emptyset\end{equation}
     and
      \begin{equation}\label{dabc1discreteholes.tm.pf.case1eq9} [b_n+a-b-\delta-b/q,b_n+a-b-\delta)+a\Z,
      [b_n, b_n+b/q)+a\Z\subset {\mathcal S}_{a, b, c}\end{equation}
   for all  $0\le n\le {N_1}-1$.
 For $n=0$, write $(R_{a, b, c})^n([c-c_0, c+b-c_0-a+\delta)+a\Z)=[c-c_0, c+b-c_0-a+\delta)+a\Z= [b_0+a-b-\delta,b_0)+a\Z$ with $b_0\in (0, a]$.
 Then the conclusions \eqref{dabc1discreteholes.tm.pf.case1eq7}, \eqref{dabc1discreteholes.tm.pf.case1eq8},
 \eqref{dabc1discreteholes.tm.pf.case1eq9-}
  and \eqref{dabc1discreteholes.tm.pf.case1eq9}
 for $n=0$
 follow from \eqref{dabcbasic1.lem.eq1}, \eqref{dabc1discreteholes.tm.pf.case1eq2}, \eqref{dabc1discreteholes.tm.pf.case1eq3},
 \eqref{dabc1discreteholes.tm.pf.case1eq4} and \eqref{dabc1discreteholes.tm.pf.case1eq6}.
 Inductively we assume that the conclusions \eqref{dabc1discreteholes.tm.pf.case1eq7}, \eqref{dabc1discreteholes.tm.pf.case1eq8},
 \eqref{dabc1discreteholes.tm.pf.case1eq9-} and \eqref{dabc1discreteholes.tm.pf.case1eq9} hold
 for all $0\le n\le k\le {N_1}-2$.
 Then for $n=k+1$,
\begin{eqnarray*}
  & & (R_{a, b, c})^{n}([c-c_0, c+b-c_0-a)+a\Z)\nonumber\\
  & = & R_{a, b, c} [b_k+a-b-\delta, b_k)+a\Z \quad {\rm (by \  \eqref{dabc1discreteholes.tm.pf.case1eq7} \ with} \ n=k)\nonumber\\
  & = & [R_{a, b, c}(b_{k}+a-b-\delta), R_{a, b, c}(b_{k}+a-b-\delta)+b-a+\delta)+a\Z\nonumber\\
   & & \qquad {\rm (by \  \eqref{dabc1discreteholes.tm.pf.case1eq8} \ with} \ n=k)\nonumber\\
 &  =: &  [b_{k+1}+a-b-\delta, b_{k+1})+a\Z
\end{eqnarray*}
 for some $b_{k+1}\in (0, a]$,
 \begin{eqnarray*}
  & & ([b_{k+1}+a-b-\delta, b_{k+1})+a\Z)\cap {\mathcal S}_{a, b, c}\nonumber\\
  & = & R_{a, b, c}\big( ([b_k+a-b-\delta, b_k)+a\Z)\cap {\mathcal S}_{a, b, c}\big)\nonumber\\
  & &   \quad {\rm (by \ \eqref{dabc1discreteholes.tm.pf.case1eq6},  \eqref{dabc1discreteholes.tm.pf.case1eq8}\ for} \ n=k,
\ {\rm and \ Proposition\ \ref{rabcinvertibility.prop}})\\
  & = & \emptyset,
\end{eqnarray*}
\begin{eqnarray*}  & & [b_{k+1}+a-b-\delta-b/q, b_{k+1}+a-b-\delta)+a\Z\nonumber\\
 & = &  [ R_{a, b, c}(b_{k}+a-b-\delta)-b/q,
R_{a, b, c}(b_{k}+a-b-\delta))+a\Z\nonumber\\
&= &  R_{a, b, c}[b_{k}+a-b-\delta-b/q, b_{k}+a-b-\delta)+a\Z\nonumber\\
 & & \qquad {\rm (by \  \eqref{dabc1discreteholes.tm.pf.case1eq9-} \ with} \ n=k)\nonumber\\
 &=& (R_{a, b, c})^{k+1} [c-c_0-b/q, c-c_0)+a\Z
\subset {\mathcal S}_{a, b, c} \quad {\rm  (by \ \eqref{dabc1discreteholes.tm.pf.case1eq2})}
  \end{eqnarray*}
  and similarly
\begin{eqnarray*}   & &  [b_{k+1}, b_{k+1}+b/q)+a\Z\nonumber\\
&= &  (R_{a, b, c})^{k+1} [c-c_0+b-a+\delta , c-c_0+b-a+\delta +b/q) +a\Z\subset {\mathcal S}_{a, b, c}
  \end{eqnarray*}
 by \eqref{dabc1discreteholes.tm.pf.case1eq3}, \eqref{dabc1discreteholes.tm.pf.case1eq7},
  \eqref{dabc1discreteholes.tm.pf.case1eq8} and \eqref{dabc1discreteholes.tm.pf.case1eq9} for $n=k$,  the definition
 \eqref{rabcnewplus.def} of the transformation
 $R_{a, b, c}$, and the invariance  \eqref{dabcbasic1.lem.eq1} of the set ${\mathcal S}_{a, b, c}$ under the transformation
 $R_{a, b, c}$.
 This together with \eqref{dabc1discreteholes.tm.pf.case1eq6} completes the inductive proof
of \eqref{dabc1discreteholes.tm.pf.case1eq7}, \eqref{dabc1discreteholes.tm.pf.case1eq8}, \eqref{dabc1discreteholes.tm.pf.case1eq9-},
 and  \eqref{dabc1discreteholes.tm.pf.case1eq9}.

Recall that $ (R_{a, b, c})^{N_1}[c-c_0, c-c_0+b-a+\delta)\cap [c_0+a-b, c_0)+a\Z\ne \emptyset$
and $[c_0, c_0+b/q)+a\Z\subset {\mathcal S}_{a, b, c}$ by the first conclusion in Lemma \ref{dabc1rationalmax.prop}.
Then applying \eqref{dabc1discreteholes.tm.pf.case1eq7} and \eqref{dabc1discreteholes.tm.pf.case1eq8}  with $n=N_1-1$, we obtain that
 \begin{equation}\label{dabc1discreteholes.tm.pf.case1eq10} (R_{a, b, c})^{N_1}[c-c_0, c-c_0+b-a+\delta)+a\Z
 =[c_0+a-b-\delta, c_0)+a\Z, \end{equation}
 because
  \begin{eqnarray}\label{dabc1discreteholes.tm.pf.case1eq10+}
  & & (R_{a, b, c})^{N_1}([c-c_0, c-c_0+b-a+\delta)+a\Z)\cap
   {\mathcal S}_{a, b, c}\\
& = & R_{a, b, c}\big( ([b_{N_1-1}+a-b-\delta, b_{N_1-1}+a-a)+a\Z)\cap {\mathcal S}_{a, b, c}\big)=\emptyset \nonumber
 \end{eqnarray}
 and
 \begin{eqnarray}\label{dabc1discreteholes.tm.pf.case1eq11} & &  [c_0+a-b-\delta-b/q,  c_0+a-b-\delta)+a\Z\\
 &= &
 R_{a, b, c} ([b_{N_1-1}-b+a-b/q, b_{N_1-1}-b+a)+a\Z)\subset {\mathcal S}_{a, b, c}.\nonumber
 \end{eqnarray}
Notice that
\begin{eqnarray}\label{dabc1discreteholes.tm.pf.case1eq11+}  & &  [b_n, b_n+b/q)+a\Z=R_{a, b, c} ([b_{n-1}, b_{n-1}+b/q)+a\Z\\
& = &\cdots=(R_{a, b, c})^n (c-c_0+b-a+\delta)+[0,b/q)+a\Z.\nonumber
\end{eqnarray}
 The conclusion (ii), particularly the equalities in
 \eqref{dabc1discreteholes.tm.eq2} and \eqref{dabc1discreteholes.tm.eq3}, in the case that $1\le N_1<\infty$
 then follows from \eqref{dabc1discreteholes.tm.pf.case1eq7}, \eqref{dabc1discreteholes.tm.pf.case1eq10}
 and \eqref{dabc1discreteholes.tm.pf.case1eq11+}.

%

  \bigskip

{\em Proof of the conclusion (iii).}\quad Let $N_2$ be the minimal nonnegative integer such that
$(R_{a, b, c})^{N_2}([c_0+a-b-\delta, c_0+a-b)+a\Z)\cap ([0, \delta)+a\Z)\ne \emptyset$
if it exists and $N_2=+\infty$ otherwise. Mimicking the argument used to  prove
Claim \ref{dabc1holes.lem.pf.claim1}, we have that $N_2<\infty$.
We divide three cases to verify \eqref{dabc1discreteholes.tm.eq4} and \eqref{dabc1discreteholes.tm.eq5}.

\smallskip  {\bf Case 1c}:  $N_2=0$.

In this case, it follows from
\eqref{dabc1discreteholes.tm.pf.case1eq1} and \eqref{dabc1discreteholes.tm.pf.case1eq11}
that
\begin{equation} \label{dabc1discreteholes.tm.pf.case1eq13}
 [c_0+a-b-\delta, c_0+a-b)+a\Z=[0, \delta)+a\Z,\end{equation}
 and hence \eqref{dabc1discreteholes.tm.eq4} and \eqref{dabc1discreteholes.tm.eq5} hold for $N_2=0$.

 \smallskip  {\bf Case 1d}:  $N_2=1$.

In this case, we have that
\begin{equation}\label{dabc1discreteholes.tm.pf.case1eq14}  R_{a, b, c} [c_0+a-b-\delta, c_0+a-b)+a\Z
=[\tilde b_1-\delta, \tilde b_1)+a\Z\end{equation}
for some $\tilde b_1\in (0, a]$ with $\tilde b_1-\delta-R_{a, b, c}(c_0+a-b-\delta)\in a\Z$.
 Recall that $[c_0+a-b-\delta-b/q,c_0+a-b-\delta)\subset {\mathcal S}_{a, b, c}$
and $[c_0, c_0+b/q)\subset {\mathcal S}_{a, b, c}$ by \eqref{dabc1discreteholes.tm.pf.case1eq1},
\eqref{dabc1discreteholes.tm.pf.case1eq11} and
Lemma \ref{dabc1rationalmax.prop},
we then obtain from \eqref{dabc1discreteholes.tm.pf.case1eq14} and Proposition \ref{rabcbasic1.thm} that
\begin{equation}\label{dabc1discreteholes.tm.pf.eq14} [\tilde b_1-\delta-b/q, \tilde b_1-\delta)+a\Z\subset {\mathcal S}_{a, b, c}
\  {\rm and} \ [\tilde b_1, \tilde b_1+b/q)+a\Z\subset {\mathcal S}_{a, b, c}.\end{equation}
Therefore $[\tilde b_1-\delta, \tilde b_1)$
is a  gap 
  of length $\delta$
  with boundary intervals of length $b/q$  at each side in the set ${\mathcal S}_{a, b, c}$.
  Thus
  \begin{equation*} [\tilde b_1-\delta, \tilde b_1)\cap ([c_0+a-b, c_0)+a\Z)=\emptyset
  \end{equation*}
as the gap containing $[c_0+a-b, c_0)$ is $(R_{a, b,c})^{N_1}[c-c_0,
c-c_0+b-a+\delta)+a\Z$ which has length $b-a+\delta$ and  boundary
intervals of length $b/q$  at each side in ${\mathcal S}_{a, b, c}$.
 By the definition of the nonnegative  integer $N_2$,
$([\tilde b_1-\delta, \tilde b_1)+a\Z)\cap ([0, \delta)+a\Z)\ne
\emptyset $. This together with
\eqref{dabc1discreteholes.tm.pf.eq14} and $[-b/q, 0)\in {\mathcal S}_{a, b, c}$ implies that $\tilde b_1=\delta$ and
\begin{equation}\label{dabc1discreteholes.tm.pf.eq15}
R_{a, b, c}[ c_0+a-b-\delta, c_0+a-b)+a\Z=[0, \delta)+a\Z.
\end{equation}
 The conclusion \eqref{dabc1discreteholes.tm.eq4} and \eqref{dabc1discreteholes.tm.eq5} for $N_2=1$ follow from
 \eqref{dabc1discreteholes.tm.pf.case1eq14}
 and \eqref{dabc1discreteholes.tm.pf.eq15}.

  \smallskip  {\bf Case 1e}:  $2\le N_2<+\infty$.

In this case, following the arguments for the first conclusion of this theorem and the arguments in the case that $N_2=1$, 
we may inductively show that
\begin{equation}\label{dabc1discreteholes.tm.pf.eq16}
(R_{a, b, c})^m ([c_0+a-b-\delta, c_0+a-b)+a\Z)=[\tilde b_m-\delta, \tilde b_m)+a\Z, 1\le m\le N_2-1,\end{equation}
for some $\tilde b_m\in (0, a]$ with $[\tilde b_m-\delta-b/q, \tilde b_m+b/q)\subset [0, c_0+a-b)\cup [c_0, a)$,
$[\tilde b_m-\delta-b/q, \tilde b_m-\delta)\subset {\mathcal S}_{a, b, c}$ and $[\tilde b_m, \tilde b_m+b/q)\subset {\mathcal S}_{a, b, c}$.
Applying \eqref{dabc1discreteholes.tm.pf.eq16} with $m=N_2-1$, and recalling the definition of the integer $N_2$
we obtain that
\begin{equation}\label{dabc1discreteholes.tm.pf.eq17} (R_{a, b, c})^{N_2}( [c_0+a-b-\delta, c_0+a-b)+a\Z)=[0, \delta)+a\Z
\end{equation}
and
\begin{equation}\label{dabc1discreteholes.tm.pf.eq18} [-b/q, 0), [\delta, \delta+b/q)\in {\mathcal S}_{a, b, c}.
\end{equation}
Therefore the conclusions \eqref{dabc1discreteholes.tm.eq4} and \eqref{dabc1discreteholes.tm.eq5} for $2\le N_2<\infty$ follow from
 \eqref{dabc1discreteholes.tm.pf.eq16},   \eqref{dabc1discreteholes.tm.pf.eq17}
 and \eqref{dabc1discreteholes.tm.pf.eq18}.



  \smallskip

  {\em Proof of the conclusion (iv).}\quad
First we prove  the mutually disjoint property for
$ (R_{a, b, c})^{n}([c-c_0, c+b-c_0-a+\delta)+a\Z), 0\le n\le N_1$, i.e.,
  \begin{equation}\label{dabc1discreteholes.tm.pf.case1eq12}
  (R_{a, b, c})^{n_1}([c-c_0, c+b-c_0-a+\delta)+a\Z)\cap (R_{a, b, c})^{n_2}([c-c_0, c+b-c_0-a+\delta)+a\Z)=\emptyset
  \end{equation}
  for all $0\le n_1\ne n_2\le N_1$.
 By \eqref{dabc1discreteholes.tm.pf.case1eq7}--\eqref{dabc1discreteholes.tm.pf.case1eq11}
 and Proposition \ref{blackholestwo.prop},
$ (R_{a, b, c})^{n}[c-c_0, c-c_0+b-a+\delta), 0\le n\le N_1$, are gaps
  with length $b-a+\delta$
  and boundary intervals of length $b/q$ at each side  contained in the set ${\mathcal S}_{a, b, c}$.
Then  for any $0\le n_1\ne n_2\le N_1$, either $(R_{a, b, c})^{n_1}([c-c_0, c+b-c_0-a+\delta)+a\Z)=
 (R_{a, b, c})^{n_2}([c-c_0, c+b-c_0-a+\delta)+a\Z)$
 or $(R_{a, b, c})^{n_1}([c-c_0, c+b-c_0-a+\delta)+a\Z)\cap (R_{a, b, c})^{n_2}([c-c_0, c+b-c_0-a+\delta)+a\Z)=\emptyset$.
 Suppose that $(R_{a, b, c})^{n_1}([c-c_0, c+b-c_0-a+\delta)+a\Z)= (R_{a, b, c})^{n_2}([c-c_0, c+b-c_0-a+\delta)+a\Z)$
 for some $0\le n_1<n_2\le N_1$. Then $n_2\le N_1-1$ by \eqref{dabc1discreteholes.tm.pf.case1eq6} and
  \eqref{dabc1discreteholes.tm.pf.case1eq10}.
  Thus by \eqref{dabc1discreteholes.tm.pf.case1eq6} and the one-to-one correspondence of
  the transformation $R_{a, b, c}$ on the complement of $[c_0+a-b, c_0)+a\Z$, we have that
  $(R_{a, b, c})^{n_2-n_1}[c-c_0, c-c_0+b-a+\delta)+a\Z =[c-c_0, c-c_0+b-a+\delta)+a\Z$,
  which contradicts to the range property \eqref{rangeoftransformations}
  as $ ((R_{a, b, c})^{n_2-n_1}[c-c_0, c-c_0+b-a+\delta)+a\Z\big)\cap ([c_0+a-b, c_0)+a\Z)=\emptyset$
  by \eqref{dabc1discreteholes.tm.pf.case1eq6}. This proves \eqref{dabc1discreteholes.tm.pf.case1eq12}.

  Next we verify that
 $(R_{a, b,c})^m ([c_0+a-b-\delta, c_0+a-b)+a\Z), 0\le m\le N_2$, are mutually disjoint when $1\le N_2<\infty$.
Recall that
  $(R_{a, b, c})^{m}( [c_0+a-b-\delta, c_0+a-b)+a\Z), 1\le m\le N_2$, are gaps
  of length $\delta$
    with boundary  intervals of length $b/q$ on each side
     contained in the set ${\mathcal S}_{a, b, c}$ by \eqref{dabc1discreteholes.tm.pf.eq15}, \eqref{dabc1discreteholes.tm.pf.eq16}
and \eqref{dabc1discreteholes.tm.pf.eq17}.
  Therefore  any two of  gaps
  $(R_{a, b, c})^{m}( [c_0+a-b-\delta, c_0+a-b)+a\Z), 1\le m\le N_2$,
  are either identical or has empty-set intersection.
  If there exist $1\le m_1< m_2\le N_2$ such that gaps
   $(R_{a, b, c})^{m_1}( [c_0+a-b-\delta, c_0+a-b)+a\Z)$ and $(R_{a, b, c})^{m_2}( [c_0+a-b-\delta, c_0+a-b)+a\Z)$ are identical,
   then the front parts of those gaps of length $\min(\delta, b-a)$ should be  identical too, i.e.,
    $(R_{a, b, c})^{m_1}( [c_0+a-b-\delta, c_0+a-b-\delta+\min(\delta, b-a))+a\Z=(R_{a, b, c})^{m_2}
    ( [c_0+a-b-\delta, c_0+a-b-\delta+\min(\delta, b-a))+a\Z$.
    Recall from  \eqref{dabc1discreteholes.tm.pf.caseeq1eq5}
    and \eqref{dabc1discreteholes.tm.pf.case1eq10} that
    $(R_{a, b, c})^{m}( [c_0+a-b-\delta, c_0+a-b-\delta+\min(\delta, b-a)+a\Z)= (R_{a, b, c})^{m+N_1}
    ([c-c_0, c-c_0+\min(\delta, b-a))+a\Z), 0\le m\le N_2$
    and
    $(R_{a, b, c})^{n}
    ([c-c_0, c-c_0+\min(\delta, b-a))+a\Z)\cap ([c_0+a-b, c_0)+a\Z)=\emptyset$ for all $0\le n\le N_1+N_2$.
    Therefore
   by  the one-to-one correspondence of the transformation
 $R_{a, b, c}$ on the complement of its black hole  in Proposition \ref{rabcinvertibility.prop},
 $(R_{a, b, c})^{m_2-m_1}[c-c_0, c-c_0+\min(\delta, b-a) )+a\Z =[c-c_0, c-c_0+\min(\delta, b-a))+a\Z$, which contradicts to
 the range property \eqref{rangeoftransformations} of the transformation $R_{a, b, c}$.
  This proves that
  $(R_{a, b, c})^{m}( [c_0+a-b-\delta, c_0+a-b)+a\Z), 1\le m\le N_2$, are mutually disjoint.

 From the above argument, we see  
that $(R_{a, b,c})^n ([c-c_0, c-c_0+b-a+\delta)+a\Z), 0\le n\le N_1$, are mutually disjoint,
and  that  $(R_{a, b, c})^{m}( [c_0+a-b-\delta, c_0+a-b)+a\Z), 1\le m\le N_2$, are mutually disjoint.
Then verification of the mutually disjoint property in the conclusion (iv)
reduces to showing that   any gap of the form $(R_{a, b,c})^n ([c-c_0, c-c_0+b-a+\delta)+a\Z), 0\le n\le N_1$,
has empty intersection with any gap of the form
 $(R_{a, b, c})^{m}( [c_0+a-b-\delta, c_0+a-b)+a\Z), 1\le m\le N_2$. This is true
because
  $(R_{a, b, c})^{n}[c-c_0, c-c_0+b-a+\delta), 0\le n\le N_1-1$, are gaps of length $b-a+\delta$
 with boundary  intervals of length $b/q$ at each side contained in the set ${\mathcal S}_{a, b, c}$ by
 \eqref{dabc1discreteholes.tm.pf.case1eq7}--\eqref{dabc1discreteholes.tm.pf.case1eq11}
 and
  $(R_{a, b, c})^{m}( [c_0+a-b-\delta, c_0+a-b)+a\Z), 1\le m\le N_2$, are gaps
  of length $\delta$
    with boundary  intervals of length $b/q$ on each side
     contained in the set ${\mathcal S}_{a, b, c}$ by \eqref{dabc1discreteholes.tm.pf.eq15}, \eqref{dabc1discreteholes.tm.pf.eq16}
and \eqref{dabc1discreteholes.tm.pf.eq17}.

  \smallskip {\em Proof of the conclusion (v).}\quad
Write $\delta=l(b-a)+\tilde \delta$ for some $0\le l\in \Z$ and $\tilde \delta\in (0, b-a]$.
From \eqref{dabc1discreteholes.tm.eq2}--\eqref{dabc1discreteholes.tm.eq5}, we obtain that
\begin{eqnarray*}
 & & (R_{a, b, c})^n([c-c_0, c-c_0+b-a)+a\Z)\nonumber\\
\qquad & = & \left\{ \begin{array}{l}
b_{n-\tilde l(N_1+N_2+1)}+\tilde l(b-a)-\delta+[a-b, 0)+a\Z \\
\quad  {\rm if} \ 0\le n- \tilde l(N_1+N_2+1)\le N_1 \ {\rm for \ some} \ 0\le \tilde l\le l, \\
\tilde b_{n- \tilde l(N_1+N_2+1)-N_1}+\tilde l (b-a)-\delta+[a-b, 0)+a\Z \\
\quad
{\rm if} \ 1\le n- \tilde l(N_1+N_2+1)-N_1\le N_2 \ {\rm for \ some} \ 0\le \tilde l\le l-1, \\
(({\tilde b}_{n-  l(N_1+N_2+1)-N_1}+[-\tilde \delta, 0))\cup [c_0+a-b, c_0-\tilde \delta))+a\Z \\
\quad
{\rm if} \ 1\le n-  l(N_1+N_2+1)-N_1\le N_2 ,\\
((b_{n-  (l+1)(N_1+N_2+1)}+[-\tilde \delta,0))\cup [c_0+a-b, c_0-\tilde \delta))+a\Z \\
\quad
{\rm if} \ 0\le n-  (l+1)(N_1+N_2+1)\le N_1,
\end{array}
\right.
\end{eqnarray*}
where $b_n=
(R_{a, b, c})^{n}(c-c_0+b-a+\delta),0\le n\le N_1$,
and $\tilde b_m=(R_{a, b, c})^{m}(c_0), 1\le m\le N_2$, c.f. Example \ref{rationalexample1} in Section \ref{gaborirrational.section}.
Therefore 
\begin{eqnarray}\label{blacktoblack.eq}
 & & (\cup_{n=0}^{N_1} ((R_{a, b, c})^n[c-c_0, c-c_0+b-a+\delta)+a\Z)\nonumber\\
 & &  \quad  \cup (\cup_{m=1}^{N_2}
(R_{a, b, c})^m[c_0+a-b-\delta, c_0+a-b)+a\Z)\nonumber\\
& & = \cup_{n=0}^{N_1+N_2} (R_{a, b, c})^n([c-c_0, c-c_0+b-a+\delta)+a\Z)\nonumber\\
& &  =
\cup_{n=0}^{(l+1)(N_1+N_2+1)+N_1} (R_{a, b, c})^n[c-c_0, c-c_0+b-a)+a\Z,\end{eqnarray}
and
\begin{equation}\label{fromblackholetoblackhole1}
(R_{a, b, c})^{(l+1)(N_1+N_2+1)+N_1}[c-c_0, c-c_0+b-a)+a\Z=[c_0+a-b, c_0)+a\Z.\end{equation}
Hence the union of  the  gaps
$((R_{a, b, c})^n[c-c_0, c-c_0+b-a+\delta)+a\Z), 0\le n\le N_1$,
and $(R_{a, b, c})^m[c_0+a-b-\delta, c_0+a-b)+a\Z), 1\le m\le N_2$, is invariant under the transformation $R_{a, b, c}$
and contains the black holes of the transformations $R_{a, b, c}$ and $\tilde R_{a, b, c}$.
It also indicates that any points not in that union will not be in that union under the transformation $R_{a, b, c}$.
 This together with Lemma \ref{dabc0rationalmax.prop} proves \eqref{dabc1discreteholes.tm.eq7} and hence the conclusion (v).

 \smallskip {\em Proof of the conclusion (vi).}\quad
From the arguments to prove the conclusions (ii) and (iii),  the boundary intervals of the mutually disjoint gaps $((R_{a, b, c})^n[c-c_0, c-c_0+b-a+\delta)+a\Z), 0\le n\le N_1$
and $(R_{a, b, c})^m[c_0+a-b-\delta, c_0+a-b)+a\Z), 1\le m\le N_2$,
of length $b/q$ at each side are contained in the set ${\mathcal S}_{a, b, c}$.
Therefore the set ${\mathcal S}_{a, b, c}$ is the union of intervals $[b_n, b_n+h_n)+a\Z, 0\le n\le N_1$
and $[\tilde b_m, \tilde b_m+\tilde h_m)+a\Z$, where $0<h_n\in b\Z/q, 0\le n\le N_1$, and $0<\tilde h_m\in a\Z/q, 1\le m\le N_2$,
 are chosen so that $ [b_n+h_n, b_n+h_n+b/q)+a\Z, 0\le n\le N_1$ and
$ [\tilde b_m+\tilde h_m, \tilde b_m+\tilde h_m+b/q)+a\Z, 1\le m\le N_2$,
are contained in gaps.
As $[0, \delta)+a\Z$ and $[c_0+a-b, c_0)+a\Z$ are contained in the union of the mutually disjoint gaps, each of  the intervals
$[b_n, b_n+h_n)+a\Z, 0\le n\le N_1$,
and $[\tilde b_m, \tilde b_m+\tilde h_m)+a\Z, 1\le m\le N_2$, is contained  either  in $[0, c_0+a-b)+a\Z$ or  in $[c_0, a)+a\Z$,
and its boundary interval of length $b/q$ at each side is not contained in the set ${\mathcal S}_{a, b, c}$.
Recall that $b_n-(R_{a, b,c})^n(c-c_0+b-a+\delta)\in a\Z, 0\le n\le N_1$,
and $\tilde b_m- (R_{a, b,c})^{m+N_1}(c-c_0+b-a+\delta)\in a\Z, 1\le m\le N_2$, from the second and third conclusions of this theorem.
Hence  the interval $[b_n, b_n+h_n)+a\Z= (R_{a, b,c})^n [b_0, b_0+h_0)+a\Z$
and $[\tilde b_m, \tilde b_m+\tilde h_m)+a\Z= (R_{a, b,c})^{m+N_1} [b_0, b_0+h_0)+a\Z$. This  together with the
measure-preserving property in Proposition \ref{rabcbasic1.thm} implies that
the length of intervals contained in the set ${\mathcal S}_{a, b, c}$ are the same, i.e.,
\begin{equation} h_n=\tilde h_m=h\ {\rm for \ all} \  0\le n\le N_1\ {\rm and} \ 1\le m\le N_2
\end{equation}
where $0<h\in a\Z/q$. Hence \eqref{dabc1discreteholes.tm.eq8} holds.
Finally we verify \eqref{dabc1discreteholes.tm.eq9}. Note that   the measure of the gaps
contained in $[0, a)$ is equal to $(N_1+1)(b-a+\delta)+ N_2 \delta$, while the measure of the intervals
contained in ${\mathcal S}_{a, b, c}\cap [0, a)$
is $(N_1+N_2+1) h$. Therefore
\begin{equation} (N_1+N_2+1) h+ (N_1+1) (b-a+\delta)+N_2\delta=a.
\end{equation}
Thus \eqref{dabc1discreteholes.tm.eq9} follows.
This completes the proof of the conclusion (ii)--(vi) for the case that $\delta'=0$ and $\delta\ne 0$.

\bigskip

  {\bf Case 2} \  {\em $\delta=0$ and $\delta'\ne 0$}.

We follow the argument used  for  Case 1 to prove the desired conclusion. We  omit the details  here.

\bigskip

{\bf Case 3}\ {\em $\delta=\delta'=0$.} 

In this case, we have that
\begin{equation}\label{dabc1discreteholes.tm.pf.case1}
[-b/q, b/q)\subset {\mathcal S}_{a, b, c}.
\end{equation}
Then
\begin{equation*} [c-c_0-b/q, c-c_0)\subset {\mathcal S}_{a, b, c} \ {\rm and}
\ [c+b-c_0-a, c+b-c_0-a+b/q)\subset {\mathcal S}_{a, b, c}
\end{equation*}
as  $[c-c_0-b/q, c-c_0)=R_{a, b, c} [-b/q, 0)$
and  $[c+b-c_0-a, c+b-c_0-a+b/q)=R_{a, b, c} [0, b/q)$ and
 the set ${\mathcal S}_{a, b, c}$  is invariant under the transformation $R_{a, b, c}$ by
  Proposition \ref{rabcbasic1.thm}.

\bigskip
  {\em Proof of the conclusion (ii).}\quad
Let $N_1$ be the smallest nonnegative integer such that
$(R_{a, b, c})^{N_1}([c-c_0, c+b-c_0-a)+a\Z) \cap ([c_0+a-b, c_0)+a\Z)\ne \emptyset$ if it exists, and ${N_1}=+\infty$ otherwise.
Following the argument in Case 1,  we can show that $N_1<+\infty$, and
\begin{equation}\label{dabc1discreteholes.tm.pf.eq4***}
(R_{a, b, c})^n([c-c_0, c+b-c_0-a)+a\Z), 0\le n\le N_1, {\rm are\ mutually\ disjoint},
\end{equation} 
 \begin{equation}\label{dabc1discreteholes.tm.pf.eq4**} (R_{a, b, c})^n([c-c_0, c+b-c_0-a)+a\Z)=[b_n-b+a, b_n)+a\Z\end{equation}
  for some $b_n\in (0, a)$ with  $b_n-(R_{a, b, c})^n(c+b-c_0-a)\in a\Z$,  $[b_n-b+a-b/q, b_n+b/q)\subset [0, c_0+a-b)\cup [c_0, a)$,
  \begin{equation}\label{dabc1discreteholes.tm.pf.eq5**} [b_n+a-b-b/q, b_n+a-b)+a\Z, [b_n, b_n+b/q)+a\Z
  \subset {\mathcal S}_{a, b, c}\end{equation}
  for $0\le n\le N_1$, and
  \begin{equation}\label{dabc1discreteholes.tm.pf.eq6**} [b_{N_1}+a-b, b_{N_1})=[c_0+a-b, c_0).\end{equation}
  Therefore the conclusion (ii) follows.

 {\em Proof of the conclusion (iii)${}^\prime$.}\quad
 Let $N_2$ be the smallest positive integer such that $(R_{a, b, c})^{N_2}(c_0)\in a\Z$ if it exists and $N_2=+\infty$ otherwise.
To prove (iii)${}^\prime$, it suffices to prove that
\begin{equation} \label{N2neinftydelta0} N_2<+\infty.\end{equation}
To prove \eqref{N2neinftydelta0}, we need the following mutually disjoint property  when $N_2<\infty$:
\begin{equation}\label{dabc1discreteholes.tm.pf.eq10**} (R_{a, b, c})^m(c_0)+[-b/(2q),
b/(2q))+a\Z, 0\le m\le N_2, \ {\rm are\ mutually\ disjoint}.\end{equation}

\begin{proof} Suppose, on the contrary, that the  mutual disjoint property in \eqref{dabc1discreteholes.tm.pf.eq10**} does not hold.
Then there exist $1\le m_1<m_2\le N_2$ such that
 \begin{equation}\label{dabc1discreteholes.tm.pf.eq11**} (R_{a, b, c})^{m_1}(c_0)+a\Z=(R_{a, b, c})^{m_2}(c_0)+a\Z.\end{equation}
This implies that $m_2<N_2$ by the definition of the integer $N_2$.
Recall  that $c_0\in {\mathcal S}_{a, b, c}$ by applying  \eqref{dabc1discreteholes.tm.pf.eq5**} and \eqref{dabc1discreteholes.tm.pf.eq6**} with $n=N_1$.
Then  $(R_{a, b, c})^{m}(c_0)\in {\mathcal S}_{a, b, c}$ by
 Proposition \ref{rabcbasic1.thm}. This, together with
the one-to-one correspondence for the transformation $R_{a, b, c}$ onto ${\mathcal S}_{a, b, c}$ given in
 Proposition \ref{rabcbasic1.thm}, leads to
  \begin{equation}\label{dabc1discreteholes.tm.pf.eq11***}c_0+a\Z=(R_{a, b, c})^{m_2-m_1}(c_0)+a\Z\end{equation}
   by \eqref{dabc1discreteholes.tm.pf.eq11**}.
Recall that $c_0\in (R_{a,b,c})^{N_1+1}(0)+a\Z$.
Then applying  the one-to-one correspondence of the transformation $R_{a, b, c}$ on the invariant set ${\mathcal S}_{a, b, c}$ again, we obtain from
\eqref{dabc1discreteholes.tm.pf.eq11***}  that
$(R_{a, b, c})^{m_2-m_1}(0)\in a\Z$.
This is a contradiction as $m_2-m_1\le N_1-1$ and
  $N_1+N_2+1$ is the
smallest positive integer $n$ such that $(R_{a,b,c})^{n}(0)\in a\Z$.
\end{proof}

Now we prove that $N_2\le p-1$. Suppose on the contrary that  $N_2\ge p$. Following  the argument in the proof of \eqref{dabc1discreteholes.tm.pf.eq10**}, we obtain
that
$(R_{a, b, c})^m(c_0)+[-b/(2q),
b/(2q))+a\Z, 0\le m\le p$ are mutually disjoint. This is a contradiction as there are at most $p$ elements in the set $[0, a)\cap b\Z/q$,
and hence proves the conclusion (iii)${}^\prime$.

\bigskip
 {\em Proof of the conclusion (iv)${}^\prime$.}\quad
By   \eqref{dabc1discreteholes.tm.pf.eq4***} and  \eqref{dabc1discreteholes.tm.pf.eq10**},
 verification of the mutually joint property in the conclusion (iv)${}^\prime$
reduces to showing that  $(R_{a, b, c})^m(c_0)\not\in \cup_{n=0}^{N_1}
(R_{a, b,c})^n ([c-c_0, c-c_0+b-a)+a\Z)$ for all $0\le m\le N_2$,  which is true
 as $(R_{a, b, c})^m(c_0)\in {\mathcal S}_{a, b, c}$ for all $m\ge 0$ and
 $(R_{a, b,c})^n ([c-c_0, c-c_0+b-a)+a\Z)\cap {\mathcal S}_{a, b, c}=\emptyset$
 for all $0\le n\le N_1$.

\bigskip
 {\em Proof of the conclusion (v).}\quad We can follow the argument in the first case $\delta>0$ and $\delta'=0$, as
 $\cup_{n=0}^{N_1} (R_{a, b,c})^n ([c-c_0, c-c_0+b-a)+a\Z)$ is invariant under the transformation $R_{a, b, c}$ and contains the
 black holes of the transformations $R_{a, b, c}$ and $\tilde R_{a, b, c}$. We omit the details here.

\bigskip
  {\em Proof of the conclusion (vi).}\quad By \eqref{dabc1discreteholes.tm.pf.eq4**}--\eqref{dabc1discreteholes.tm.pf.eq6**},
   gaps of length $b-a$ are $(R_{a, b, c})^n(c+b-c_0-a)+[a-b, 0)+a\Z, 0\le n\le N_1$,
  while  gaps of length zero are located at $(R_{a,b, c})^m(c_0)+a\Z=(R_{a, b, c})^{m+N_1}(c+b-c_0-a)+a\Z, 1\le m\le N_2$.
  By the conclusions (iii)${}^\prime$ and (v), we can divide the set ${\mathcal S}_{a, b, c}$ as the union of disjoint union of  intervals
  who are left-closed and right-open such that each interval has its left endpoint being the same of the right endpoint of a gap,
 each interval has its right endpoint being the same as  the left  endpoint of a gap, and
   each interval has its interior not containing the location of any gaps of length zero; that is
   \begin{equation}\label{delta0.pf.veq1}
   {\mathcal S}_{a, b, c}=\cup_{n=0}^{N_1+N_2}
   [b_n,
 b_n+h_n)+a\Z,\end{equation}
   where
   \begin{equation}\label{delta0.pf.veq2} b_n\in (R_{a, b, c})^n(c+b-c_0-a)+a\Z\end{equation}
   and
 \begin{equation}\label{delta0.pf.veq3} b_n+h_n   \in
\{b_{n_1}+a-b+a\Z\}_{n_1=0}^{N_1} \cup \{(R_{a, b,
c})^{m_1}(c_0)+a\Z\}_{m_1=1}^{N_2}\end{equation} for all $0\le
n\le N_1+N_2$, and
\begin{equation}\label{delta0.pf.veq4}
(R_{a,b, c})^m(c_0)\not\in \cup_{n=0}^{N_1+N_2} (b_n, b_n+h_n)+a\Z
\end{equation}
for all $1\le m\le N_2$.
Therefore the proof of \eqref{dabc1discreteholes.tm.eq8} reduces to  establishing
\begin{equation} \label{delta0.pf.veq5}
h_n=h_0, \ 0\le n\le N_1+N_2.
\end{equation}
By the measure-preserving property in Proposition \ref{rabcbasic1.thm} it  suffices to prove that
\begin{equation} \label{delta0.pf.veq6}
[b_{n+1}, b_{n+1}+h_{n+1})+a\Z\subset R_{a, b, c}([b_n, b_n+h_n)+a\Z), 0\le n\le N_1+N_2-1
\end{equation}
and
\begin{equation} \label{delta0.pf.veq7}
[b_{n}, b_{n}+h_{n})+a\Z\subset \tilde R_{a, b, c}([b_{n+1}, b_{n+1}+h_{n+1})+a\Z), 0\le n\le N_1+N_2-1.
\end{equation}
By \eqref{delta0.pf.veq2} and the measure-preserving property in Proposition \ref{rabcbasic1.thm},
we have that $R_{a, b, c}([b_n, b_n+h_n)+a\Z)=[b_{n+1}, b_{n+1}+h_n)+a\Z$.
Observe that for each $0\le n\le N_1+N_2$,  $b_{n+1}+h_n$
is the left endpoint of a gap because
$$b_{n+1}+h_n=\left\{\begin{array}{ll} R_{a, b, c}(b_n+h_n)  & {\rm if}\
b_{n}+h_{n}\not\in \{0, c_0+a-b\}+a\Z\\
R_{a, b, c}(c_0) & {\rm if}\ b_{n}+h_n \in
c_0+a-b+a\Z,\end{array}\right.$$ and $b_{n}+h_n\not\in
b_{N_1+N_2}+h_{N_1+N_2}+a\Z=a\Z$ for all $0\le n\le N_1+N_2-1$ by
the conclusion (iv)${}^\prime$. This together with
\eqref{delta0.pf.veq4} and the fact that $R_{a, b, c}([b_n,
b_n+h_n)+a\Z)\subset {\mathcal S}_{a, b, c}$ proves
\eqref{delta0.pf.veq6}.

Similarly, by \eqref{delta0.pf.veq2} and the measure-preserving
property in Proposition \ref{rabcbasic1.thm}, we have that $0\le n\le
N_1+N_2-1$,
 $\tilde R_{a, b, c}([b_{n+1}, b_{n+1}+h_{n+1})+a\Z)=[b_{n},
b_{n}+h_{n+1})+a\Z$ and  $b_{n}+h_{n+1}$ is the left endpoint of a gap
because
$$b_{n}+h_{n+1}=\left\{\begin{array}{ll} \tilde R_{a, b, c}(b_{n+1}+h_{n+1})  & {\rm if}\
b_{n+1}+h_{n+1}\not\in \{c, c-c_0 \}+a\Z\\
\tilde b_{N_1} & {\rm if}\  b_{n+1}+h_{n+1} \in
c+a\Z,\end{array}\right.$$ and $b_{n+1}+h_{n+1}\not\in c-c_0+
a\Z=b_0-b+a+a\Z$ for all $0\le n\le N_1+N_2-1$ by the conclusion
(iv)${}^\prime$. Thus  \eqref{delta0.pf.veq7} follows. This proves
\eqref{delta0.pf.veq5} and hence
\eqref{dabc1discreteholes.tm.eq8}.

The equation \eqref{dabc1discreteholes.tm.eq9} holds by \eqref{dabc1discreteholes.tm.eq8} and the conclusion (iv)
in this theorem.
\end{proof}

%
Now it remains to prove Lemmas \ref{dabc0rationalmax.prop} and \ref{dabc1rationalmax.prop}.

\begin{proof} [Proof of Lemma \ref{dabc0rationalmax.prop}]  (i)\ \ Clearly it suffices to prove
that 
\begin{equation}\label{dabc0rationalmax.prop.eq1}{\mathcal S}_{a, b, c}=(R_{a, b, c})^D\R\backslash ([c_0+a-b, c_0)+a\Z)\end{equation}
for some nonnegative integer $D$.
For $n\ge 0$, write
\begin{equation}\label{dabc0rationalmax.prop.pf.eq1}
(R_{a, b, c})^{n}\R \backslash ([c_0+a-b, c_0)+a\Z)= A_n+[0, b/q)+a\Z\end{equation}
where $A_n\subset \{0, b/q, \ldots, (p-1)b/q\}\backslash  [c_0+a-b, c_0)$. The existence of such a finite set $A_n$ follows from
 the assumption on the triple $(a, b, c)$ and the definition \eqref{rabcnewplus.def} of the transformation $R_{a, b,c}$.
Clearly, 
\begin{equation}\label {dabc0rationalmax.prop.pf.eq2}
A_{n+1}\subset  A_n, \ n\ge 0.
\end{equation}
Notice that there are at most $2p-q$ elements in the set $\{0, b/q, \ldots, (p-1)b/q\}\backslash  [c_0+a-b, c_0)$.
This together with  \eqref{dabc0rationalmax.prop.pf.eq2} implies that
\begin{equation}\label{dabc0rationalmax.prop.pf.eq3}
 A_{D+1}= A_{D}
\end{equation}
for some   nonnegative integer $D\le 2p-q$. Without loss of
generality, we assume that $D$ is the smallest nonnegative integer
such that \eqref{dabc0rationalmax.prop.pf.eq3} holds. By Propositions
\ref{rabcbasic1.thm} and  \ref{dabcbasic1.prop}, we
have that $ {\mathcal S}_{a, b, c}=(R_{a, b, c})^{D} {\mathcal S}_{a, b,c}\subset (R_{a, b, c})^{D} \R$ and
${\mathcal S}_{a, b,c}\cap ([c_0+a-b, c_0)+a\Z)=\emptyset$. This leads to the
following inclusion
\begin{equation}\label{dabc0rationalmax.prop.pf.eq4}
 {\mathcal S}_{a, b, c}\subset
A_{D}+[0, b/q)+a\Z .\end{equation}

Clearly the conclusion \eqref{dabc0rationalmax.prop.eq1} follows
in the case that $A_D=\emptyset$.
Now we consider the case that $A_D\ne \emptyset$ and  prove
\begin{equation}\label{dabc0rationalmax.prop.pf.eq5}
A_{D}+a\Z \subset {\mathcal S}_{a, b, c}.\end{equation}
For $n\ge 0$, it follows from \eqref{rabcnewplus.def} and \eqref{dabc0rationalmax.prop.pf.eq1}
 that
\begin{eqnarray*}
& & A_{n+1}+[0, b/q)+a\Z \\
& = &  R_{a, b, c}\big ( (R_{a, b, c})^n \R\backslash ([c_0+a-b, c_0)+a\Z)\big)\backslash
([c_0+a-b, c_0)+a\Z)\\
 & =&  R_{a, b, c} (A_n+[0, b/q)+a\Z)\backslash
([c_0+a-b, c_0)+a\Z)\\
 &= & R_{a, b, c}(A_n+a\Z) \backslash ([c_0+a-b, c_0)+a\Z)+[0, b/q).
\end{eqnarray*}
Therefore
\begin{equation}\label{dabc0rationalmax.prop.pf.eq6}
A_{n+1}+a\Z = R_{a, b, c}(A_n+a\Z)\backslash ([c_0+a-b, c_0)+a\Z), \ n\ge 0.
\end{equation}
Applying \eqref{dabc0rationalmax.prop.pf.eq6} with $n=D$  and using
\eqref{dabc0rationalmax.prop.pf.eq3},
 we obtain that
$$
A_{D}+a\Z = A_{D+1}+a\Z= R_{a, b, c}(A_D+a\Z) \backslash ([c_0+a-b, c_0)+a\Z).$$
Thus
\begin{equation}\label{dabc0rationalmax.prop.pf.eq7}
 R_{a, b, c}(A_D)+a\Z =
A_D +a\Z
\end{equation}
as the cardinality of the  sets $A_D$ and
$R_{a, b, c}(A_D+a\Z)\cap [0, a)$ are the same.
This together with \eqref{rangeoftransformations} and \eqref{dabc0rationalmax.prop.pf.eq6} implies that
\begin{equation} \label{dabc0rationalmax.prop.pf.eq8} (A_D+a\Z)\cap ([c-c_0, c+b-c_0+a)+a\Z)=\emptyset . \end{equation}
Hence
\begin{equation}\label{dabc0rationalmax.prop.pf.eq9}
 \tilde  R_{a, b, c}(A_D)+a\Z= A_D+a\Z
\end{equation}
by \eqref{dabc0rationalmax.prop.pf.eq1},  \eqref{dabc0rationalmax.prop.pf.eq7}, \eqref{dabc0rationalmax.prop.pf.eq8}
and Proposition \ref{rabcbasic1.thm}.

Take  $t\in A_D$. Then
$(R_{a, b, c})^n (t)$ and $(\tilde R_{a, b, c})^{n} (t) , n\ge 0$, belong to the set  $A_D+a\Z$ by
\eqref{dabc0rationalmax.prop.pf.eq7} and \eqref{dabc0rationalmax.prop.pf.eq9}, and hence they do not belong to
the black holes of the transformations $R_{a,b, c}$ and $\tilde R_{a, b,c}$. Therefore
 $t\in {\mathcal S}_{a, b, c}$ by Proposition \ref{dabc1pointcharacterization.prop}
and hence \eqref{dabc0rationalmax.prop.pf.eq5} is established.

Notice that
\begin{equation}\label{dabc0rationalmax.prop.pf.eq12}
A_{D}+[0, b/q)+a\Z \subset {\mathcal S}_{a, b, c}
\end{equation}
by  \eqref{sabcrational.eq} and
\eqref{dabc0rationalmax.prop.pf.eq5}. Then
in the case that $A_D\ne \emptyset$,  the conclusion
\eqref{dabc0rationalmax.prop.eq1} follows from \eqref{dabc0rationalmax.prop.pf.eq1},
\eqref{dabc0rationalmax.prop.pf.eq4} and
\eqref{dabc0rationalmax.prop.pf.eq12}.

\bigskip

(ii) The desired minimality follows from
\begin{equation}\label{dabc0rationalmax.prop.pf.eq12}
\R\backslash {\mathcal S}_{a, b, c}= \cup_{n=0}^\infty (R_{a, b, c})^n ([c-c_0, c+b-c_0-a)+a\Z).
\end{equation}
By Proposition \ref{blackholestwo.prop},
\begin{equation*} 
\cup_{n=0}^\infty (R_{a, b, c})^n ([c-c_0, c+b-c_0-a)+a\Z)\subset \R\backslash {\mathcal S}_{a, b, c}.
\end{equation*}
Then it suffices to prove
\begin{equation}\label{dabc0rationalmax.prop.pf.eq13}
 \R\backslash {\mathcal S}_{a, b, c}\subset\cup_{n=0}^L (R_{a, b, c})^n ([c-c_0, c+b-c_0-a)+a\Z)
\end{equation}
with
\begin{equation}\label{dabc0rationalmax.prop.pf.eq15}
(R_{a, b, c})^L([c-c_0, c+b-c_0-a)+a\Z)=[c_0+a-b, c_0)+a\Z\end{equation}
for some nonnegative integer $L\ge 0$, c.f. \eqref{dabc1holes.tm.eq-1} and \eqref{dabc1holes.tm.eq1} in Theorem  \ref{dabc1holes.tm}.
First we prove  that
\begin{equation}\label{dabc0rationalmax.prop.pf.eq14}
(R_{a, b, c})^L([c-c_0, c+b-c_0-a)+a\Z)\subset [c_0+a-b, c_0)+a\Z\end{equation}
for some nonnegative integer $L$.
By
\eqref{dabcicontinuoustodiscrete.eq1} and finite cardinality of  the set $b\Z/q\cap [0, a)$,
we only need  to verify that  for any $t\in [c-c_0, c+b-c_0-a)+a\Z)\cap b\Z/q$ there exists $L_0$ such that
$(R_{a, b, c})^{L_0}(t)\in [c_0+a-b, c_0)+a\Z$. Suppose on the contrary that
$(R_{a, b, c})^n(t)\not\in [c_0+a-b, c_0)+a\Z$ for all $n\ge 0$.
Then $(R_{a, b, c})^n (t)\not\in [c-c_0, c+b-c_0-a)+a\Z, n\ge 1$, by Proposition \ref{rabcinvertibility.prop},
which together with finite cardinality of  the set $b\Z/q\cap [0, a)$
implies  the existences of positive integers $1\le L_1<L_2$ such that $(R_{a, b, c})^{L_1}(t)-(R_{a, b, c})^{L_2}(t)\in a\Z$.
Set $t_0=(R_{a, b, c})^{L_1}(t)$. Then
$(R_{a, b, c})^n (t_0)$ and $(\tilde R_{a, b, c})^n (t_0), n\ge 1$, do not belong to  black holes of the transformations $R_{a, b, c}$
and $\tilde R_{a, b, c}$ as
$(R_{a, b, c})^n (t_0)-(R_{a, b, c})^{m+L_1}(t)\in a\Z$  and
 $(\tilde R_{a, b, c})^{\tilde n} (t_0)-(R_{a, b, c})^{L_2-\tilde m}(t)\in a\Z$
where  $m=n-\lfloor n/(L_2-L_1)\rfloor (L_2-L_1)$,
and $\tilde m=\tilde  n- \lfloor \tilde n/(L_2-L_1)\rfloor (L_2-L_1)$.
Thus $t_0\in {\mathcal S}_{a, b, c}$ by Proposition \ref{rabcbasic1.thm.eq1b}, which contradicts to \eqref{dabc0rationalmax.prop.pf.eq13}.

Next we prove \eqref{dabc0rationalmax.prop.pf.eq15}.
By \eqref{dabc0rationalmax.prop.pf.eq14}, it suffices to prove that
$$
(R_{a, b, c})^{L}(t_1)-(R_{a, b, c})^{L}(t_2)\not\in a\Z
$$
for any distinct $t_1, t_2\in ([c-c_0, c+b-c_0-a)+a\Z)\cap b\Z/q$.
Let $k_1, k_2$ be  minimal nonnegative integers such that
$$(R_{a, b, c})^{k_1}(t_1)=(R_{a, b, c})^{L}(t_1)\quad {\rm and} \quad  (R_{a, b, c})^{k_2}(t_2)=(R_{a, b, c})^{L}(t_2).$$
Without loss of generality, we assume that $k_1\le k_2$.
By one-to-one correspondence of the transformation $R_{a, b, c}$ given in Proposition \ref{rabcinvertibility.prop}
and the selection of integers $k_1$ and $k_2$,
$$t_1= ({\tilde R}_{a, b, c})^{k_1} \big((R_{a, b, c})^{k_1}(t_1)\big)\in  (R_{a, b, c})^{k_2-k_1}(t_2)+a\Z,$$
which is a contradiction by the assumption $t_1\not\in t_2+a\Z$ (if $k_2=k_1$)
and the range property of the transformation $R_{a, b, c}$ in Proposition
\ref{rabcinvertibility.prop} (if $k_2>k_1$). Hence \eqref{dabc0rationalmax.prop.pf.eq15} is proved.

By \eqref{dabc0rationalmax.prop.pf.eq15}, $\R\backslash \big(\cup_{n=0}^L (R_{a, b, c})^n ([c-c_0, c-c_0+b-a)+a\Z)\big)$
has empty intersection with $[c_0+a-b, c_0)+a\Z$, the black  hole  of the transformation $R_{a, b, c}$. Then by the first conclusion
of this lemma, the proof of \eqref{dabc0rationalmax.prop.pf.eq13} reduces to the invariance of the set
$\R\backslash \big(\cup_{n=0}^L (R_{a, b, c})^n ([c-c_0, c-c_0+b-a)+a\Z)\big)$ under the transformation $R_{a, b, c}$.
Suppose, on the contrary, that there exists $t\not\in \cup_{n=0}^L (R_{a, b, c})^n ([c-c_0, c-c_0+b-a)+a\Z)$
such that $R_{a, b, c}(t)\in \cup_{n=0}^L (R_{a, b, c})^n ([c-c_0, c-c_0+b-a)+a\Z)$. Then
$t\not\in [c_0+a-b, c_0)+a\Z$ by \eqref{dabc0rationalmax.prop.pf.eq15},
and $R_{a, b, c}(t)= (R_{a, b, c})^n (s)$ for some $0\le n\le L$ and $s\in [c-c_0, c-c_0+b-a)+a\Z$.
It follows that $n\ge 1$ from the range property of the transformation $R_{a, b, c}$ in Proposition \ref{rabcinvertibility.prop}.
If we further select the integer $n$ to be the smallest positive integer such that $R_{a, b, c}(t)= (R_{a, b, c})^n (s)$
for some $s\in [c-c_0, c-c_0+b-a)+a\Z$. Then
 $t=(R_{a, b, c})^{n-1}(s)$ by the one-to-one correspondence of the transformation $R_{a, b, c}$
given in  Proposition \ref{rabcinvertibility.prop}, which contradicts to the assumption that
$t\not\in \cup_{n=0}^L (R_{a, b, c})^n ([c-c_0, c-c_0+b-a)+a\Z)$.
\end{proof}

We finish this subsection by the proof of Lemma \ref{dabc1rationalmax.prop}.

\begin{proof}  [Proof of Lemma \ref{dabc1rationalmax.prop}]\ (i)\quad Suppose on the contrary that both $[0, b/q)$ and $[c_0, c_0+b/q)$ are not contained in
${\mathcal S}_{a, b, c}$. Then
$0, c_0\not\in {\mathcal S}_{a, b, c}$ by \eqref{dabcicontinuoustodiscrete.eq1}, which together with \eqref{dabcicontinuoustodiscrete.eq1}
and Proposition \ref{dabcbasic1.prop} implies that
\begin{equation}\label{dabc1rationalmax.prop.pf.eq1} {\mathcal S}_{a, b, c}\subset ([b/q, c_0+a-b)\cup [c_0+b/q, a))+a\Z.\end{equation}
Thus
\begin{equation}\label{dabc1rationalmax.prop.pf.eq2} R_{a, b, c}({\mathcal S}_{a, b, c}-b/q)\subset {\mathcal S}_{a, b, c}-b/q\end{equation}
as $R_{a, b, c}(t-b/q)=R_{a, b, c}(t)-b/q$ for all $t\in {\mathcal S}_{a, b, c}\subset ([b/q, c_0+a-b)\cup [c_0+b/q, a))+a\Z$.
 Thus both ${\mathcal S}_{a, b, c}$ and ${\mathcal S}_{a, b, c}-b/q$ are invariant under the transformation
 $R_{a, b,c}$  and have empty intersection with the black hole $[c_0+a-b, c_0)+a\Z$ of the transformation $R_{a, b,c}$ by
  \eqref{dabc1rationalmax.prop.pf.eq2} and Proposition \ref{rabcbasic1.thm}. Hence
 by the maximality of the set ${\mathcal S}_{a, b, c}$ in  Lemma \ref{dabc0rationalmax.prop},
$ {\mathcal S}_{a, b, c}-b/q\subset {\mathcal S}_{a, b, c}$,
which contradicts to \eqref{dabc1rationalmax.prop.pf.eq1} because $t_0-b/q\in {\mathcal S}_{a, b, c}-b/q$
but $t_0-b/q\not\in {\mathcal S}_{a, b, c}$ by \eqref{dabcicontinuoustodiscrete.eq1} and
\eqref{dabc1rationalmax.prop.pf.eq1} where $t_0$ is the smallest positive number in ${\mathcal S}_{a, b, c}\cap [0, a)$.

(ii)\quad Suppose on the contrary that both $[-b/q, 0)$ and
$[c_0+a-b-b/q, c_0+a-b)$ are not contained in ${\mathcal S}_{a, b, c}$. Then $a-b/q, c_0+a-b-b/q\not\in {\mathcal S}_{a, b, c}$ by \eqref{dabcicontinuoustodiscrete.eq1}. Following the
above argument in the proof of the first conclusion, $R_{a, b,
c}({\mathcal S}_{a, b, c}+b/q)= {\mathcal S}_{a, b, c}+b/q$  and $({\mathcal S}_{a, b, c}+b/q )\cap ([c_0+a-b,
c_0)+a\Z)=\emptyset$. Hence the set ${\mathcal S}_{a, b, c}+b/q$ is invariant under the transformation
 $R_{a, b,c}$  and has empty intersection with the black hole $[c_0+a-b, c_0)+a\Z$ of the transformation $R_{a, b,c}$,
 which contradicts to the maximality of the set ${\mathcal S}_{a, b, c}$ established in  Lemma \ref{dabc0rationalmax.prop}.

(iii)\quad
Suppose on the contrary that both
$[0, b/q)$ and $[-b/q, 0)$ are not contained in ${\mathcal S}_{a, b, c}$.
Then
\begin{equation}\label{dabc1rationalmax.prop.pf.eq5}
[c_0, c_0+b/q)\subset {\mathcal S}_{a, b, c}\ {\rm and }\   [c_0+a-b-b/q, c_0+a-b)\in {\mathcal S}_{a, b, c}
\end{equation}
by the first two conclusions of this lemma.
First we show that there exists a nonnegative integer $1\le D\le (2p-q)/(q-p)$
such that
\begin{equation}\label{dabc1rationalmax.prop.pf.eq6}
(R_{a, b, c})^D([c-c_0, c+b-c_0-a)+a\Z)\cap
([c_0+a-b, c_0)+a\Z)\ne \emptyset.\end{equation}
Suppose on the contrary that \eqref{dabc1rationalmax.prop.pf.eq6} does not hold.
Then  $(R_{a, b, c})^n([c-c_0, c+b-c_0-a)+a\Z)\cap
([c_0+a-b, c_0)+a\Z)=\emptyset$ for all $0\le n\le (2p-q)/(q-p)$.
Following the argument in  \eqref{dabc1holes.tm.pf.eq13},  we have that
$(R_{a, b, c})^n([c-c_0, c+b-c_0-a)+a\Z), 0\le n\le (2p-q)/(q-p)$, are mutually disjoint.
This together with
$(R_{a, b, c})^n([c-c_0, c+b-c_0-a)+a\Z)= (R_{a, b, c})^n([c-c_0, c+b-c_0-a)\cap b\Z/q)+[0, b/q)+a\Z, 0\le n\le (2p-q)/(q-p)$,
implies that
$|\cup_{0\le n\le (2p-q)/(q-p)} (R_{a, b, c})^n([c-c_0, c+b-c_0-a)+a\Z)\cap ([0, a)\backslash [c_0+a-b, c_0))|= \lfloor p/(q-p)\rfloor (q-p)/q>
|[0, a)\backslash [c_0+a-b, c_0)|$, which is a contradiction. This proves \eqref{dabc1rationalmax.prop.pf.eq6}.

By \eqref{dabc1rationalmax.prop.pf.eq6}, we may
assume that the nonnegative integer $D$  in \eqref{dabc1rationalmax.prop.pf.eq6} is the minimal integer such that
\eqref{dabc1rationalmax.prop.pf.eq6} holds.
Following the above argument, we  may conclude that
\begin{equation}\label{dabc1rationalmax.prop.pf.eq7}
(R_{a, b, c})^n([c-c_0, c+b-c_0-a)+a\Z), 0\le n\le D,\ {\rm  are\ mutually\ disjoint}.
\end{equation}
 Now  let us verify the following claim:
\begin{equation}\label{dabc1rationalmax.prop.pf.eq8}
(R_{a, b, c})^n([c-c_0, c+b-c_0-a)+a\Z)=[b_n+a-b, b_n)+ a\Z\end{equation}
for some  $b_n\in (0, a], 0\le n\le D$,
and
\begin{equation}\label{dabc1rationalmax.prop.pf.eq8b}
(R_{a, b, c})^D([c-c_0, c+b-c_0-a)+a\Z)=[c_0+a-b, c_0)+ a\Z.\end{equation}

\begin{proof}[Proof of Claims \eqref{dabc1rationalmax.prop.pf.eq8}
and \eqref{dabc1rationalmax.prop.pf.eq8b}]
 If $D=0$, then
\eqref{dabc1rationalmax.prop.pf.eq8} and
\eqref{dabc1rationalmax.prop.pf.eq8b} follow from
\eqref{dabc1rationalmax.prop.pf.eq5} and the definition of the
nonnegative integer $D$. Now we consider $D\ge 1$. Let
$T_0=[c-c_0, c+b-c_0-a)+a\Z$ and  define $T_n, 1\le n\le D$,
inductively by
\begin{equation}\label{dabc1rationalmax.prop.pf.eq9}
T_n=\left\{\begin{array}{ll}
R_{a, b,c}(T_{n-1})  &  {\rm if} \ 0\not\in T_{n-1},\\
R_{a, b,c}(T_{n-1})\cup ( [c-c_0, c+b-c_0-a)+a\Z) & {\rm if} \ 0\in  T_{n-1}.
\end{array}\right.
\end{equation}
Clearly
$T_0=[b_0+a-b, b_0)+a\Z$ for some $b_0\in (0, a]$. Inductively, we assume that
$T_n=[\tilde b_n, b_n)+a\Z$ for some $\tilde b_n, b_n$ with $b_n\in (0, a]$ and $b-a\le b_n-\tilde b_n<a$.
If $0\not\in T_n$, then either $[\tilde b_n, b_n)\subset [0, c_0+a-b)$ or $[\tilde b_n, b_n)\subset [c_0, a)$. This implies that
\begin{eqnarray}\label{dabc1rationalmax.prop.pf.eq10}
T_{n+1} & = & R_{a, b,c}(T_{n})=[R_{a, b,c}(\tilde b_n), R_{a, b,c}(\tilde b_n)+b_n-\tilde b_n)+a\Z\nonumber\\
& =: & [\tilde b_{n+1}, b_{n+1})+a\Z\end{eqnarray}
for some $\tilde b_{n+1}, b_{n+1}$ with $b_{n+1}\in (0, a]$ and $b_{n+1}-\tilde b_{n+1}=b_n-\tilde b_n$.
If $0\in T_n$, then $\tilde b_n\le 0$. Moreover $\tilde b_n\ge c_0-a$ and $b_n\le c_0+a-b$, as otherwise $T_n$
 has nonempty intersection with the black hole
$[c_0+a-b, c_0)+a\Z$
of the transformation $R_{a, b, c}$, which contradicts to \eqref{dabc1rationalmax.prop.pf.eq6} and the observation that
$T_n\subset \cup_{m=0}^n (R_{a, b, c})^m([c-c_0, c+b-c_0-a)+a\Z)$.
Therefore
\begin{eqnarray}\label{dabc1rationalmax.prop.pf.eq11}
T_{n+1} & = & R_{a, b,c}(T_{n})\cup ( [c-c_0, c+b-c_0-a)+a\Z)\nonumber\\
&  =  & [\tilde b_n+\lfloor c/b\rfloor b, b_n+\lfloor c/b\rfloor b+b-a)+a\Z\nonumber\\
& =: & [\tilde b_{n+1}, b_{n+1})+a\Z\end{eqnarray}
for some $\tilde b_{n+1}, b_{n+1}$ with $b_{n+1}\in (0, a]$ and $b_{n+1}-\tilde b_{n+1}=b_n-\tilde b_n+b-a$.
Combining \eqref{dabc1rationalmax.prop.pf.eq10} and \eqref{dabc1rationalmax.prop.pf.eq11} proceeds the inductive proof that
\begin{equation}\label{dabc1rationalmax.prop.pf.eq12}
T_n=[\tilde b_n, b_n)+a\Z\ {\rm for \ all} \  0\le n\le D,\end{equation}
 where $b_n\in (0, a]$ and $ b_n-\tilde b_n\in [b-a, a)$.
Observe that $ (R_{a, b,c})^D([c-c_0, c+b-c_0-a)+a\Z)\subset  T_D\subset
\cup_{n=0}^D(R_{a, b,c})^n([c-c_0, c+b-c_0-a)+a\Z)$. Then
$T_D$ has nonempty intersection with the black hole $[c_0+a-b, c_0)+a\Z$ of the transformation $R_{a, b, c}$ by
\eqref{dabc1rationalmax.prop.pf.eq7}.
This together with \eqref{dabc1rationalmax.prop.pf.eq12}
implies that either $[c_0+a-b-b/q, c_0+a-b)\subset T_D$, or $[c_0, c_0+b/q)\subset T_D$ or $T_D=[c_0+a-b, c_0)+a\Z$.
Recall that $T_D\cap {\mathcal S}_{a, b, c}=\emptyset$ by  Proposition \ref{blackholestwo.prop}. Then both
$[c_0+a-b-b/q, c_0+a-b)$ and $[c_0, c_0+b/q)$ have empty intersection with $T_D$ by \eqref{dabc1rationalmax.prop.pf.eq5}.
Thus
\begin{equation}\label{dabc1rationalmax.prop.pf.eq13} T_D=[c_0+a-b, c_0)+a\Z.\end{equation}
This together with
\eqref{dabc1rationalmax.prop.pf.eq9}, \eqref{dabc1rationalmax.prop.pf.eq10} and
 \eqref{dabc1rationalmax.prop.pf.eq11} implies that
 \begin{equation}\label{dabc1rationalmax.prop.pf.eq14-} \tilde b_n>0\ {\rm  and} \ b_n-\tilde b_n=b-a\ {\rm for\ all} \ 0\le n\le D.\end{equation}
The desired conclusions \eqref{dabc1rationalmax.prop.pf.eq8} and \eqref{dabc1rationalmax.prop.pf.eq8b} then follow.
\end{proof}

Let us return to the proof of the conclusion (iii).  By
\eqref{dabcicontinuoustodiscrete.eq1},
\eqref{dabc1rationalmax.prop.pf.eq5}, \eqref{dabc1rationalmax.prop.pf.eq7},
\eqref{dabc1rationalmax.prop.pf.eq8} and  Proposition
\ref{blackholestwo.prop},
 either $[b_n+a-b, b_n)\subset [b/q, c_0+a-b)$ or $[b_n+a-b, b_n)\subset [c_0+b/q, a)$.
This implies that
\begin{equation} \label{dabc1rationalmax.prop.pf.eq14}
R_{a, b, c}(b_n+a-b-b/(2q))+a\Z 
=b_{n+1}+a-b-b/(2q)+a\Z
\end{equation}
for all $0\le n\le D-1$. By \eqref{dabc1rationalmax.prop.pf.eq5},
\eqref{dabc1rationalmax.prop.pf.eq8},
\eqref{dabc1rationalmax.prop.pf.eq8b},
\eqref{dabc1rationalmax.prop.pf.eq14}, and Propositions
\ref{rabcinvertibility.prop} and \ref{rabcbasic1.thm}, we have
that
\begin{eqnarray*}
 & & (\tilde R_{a, b, c})^n (c_0+a-b-b/(2q))+a\Z\nonumber\\
&= & (\tilde R_{a, b, c})^n (b_D+a-b-b/(2q))+a\Z\nonumber\\
& = & (\tilde R_{a, b, c})^{n-1} (b_{D-1}+a-b-b/(2q))+a\Z=\cdots\nonumber\\
& = &
b_{D-n}+a-b-b/(2q)+a\Z\subset {\mathcal S}_{a, b, c},\quad  0\le n\le D.
\end{eqnarray*}
Hence
$-b/(2q)+a\Z= \tilde R_{a, b, c}(c-c_0-b/(2q))+a\Z=(\tilde R_{a, b, c})^{D+1} (c_0+a-b-b/(2q))+a\Z\in {\mathcal S}_{a, b, c}$,
which together with
\eqref{dabcicontinuoustodiscrete.eq1} implies that $[-b/q, 0)\in {\mathcal S}_{a, b, c}$. This is  a contradiction.
\end{proof}

\subsection{Covering property of maximal invariant sets} 
In this subsection, we prove Theorems \ref{covering.tm2} and
 \ref{sabcstar.tm}.

\begin{proof}[Proof of Theorem \ref{covering.tm2}] \ Set
 $A_{\lambda}:={\mathcal S}_{a, b, c}\cap [0, c_0+a-b)+\lambda+a\Z$
and
$B_{\lambda}:={\mathcal S}_{a, b, c}\cap [c_0, a)+\lambda+a\Z$ for $\lambda\in b\Z$.
 We divide the proof into  two cases.

{\bf Case 1}: \ {\em $a/b\not\in \Q$.}

Take $t_0\in {\mathcal S}_{a, b, c}$. Then  $(R_{a, b,c})^n(t_0)\in {\mathcal S}_{a, b, c}$ by
Proposition \ref{rabcbasic1.thm}. Write $(R_{a, b,c})^n(t_0)=t_0+k_n b$,
where the  strictly increasing  sequence $\{k_n\}_{n=0}^\infty$  of nonnegative integers is defined inductively by
$k_0=0$ and
\begin{equation} \label{dbac2covering.cor.pf.eq1}
k_{n+1}-k_{n}=\left\{ \begin{array}{ll}
\lfloor c/b\rfloor+1 & {\rm  if} \ t_0+k_{n}b\in [0, c_0+a-b)+a\Z\\
\lfloor c/b\rfloor & {\rm  if} \ t_0+k_{n}b\in [c_0,a)+a\Z\end{array}\right.
\end{equation}
for $n\ge 0$.  Then
for any nonnegative  integer $l$,
\begin{eqnarray}\label{dbac2covering.cor.pf.eq2}
t_0+lb & = &  t_0+k_nb+(l-k_n)b \nonumber\\
&\in&  \big( \cup_{\lambda_2\in [0, (\lfloor c/b\rfloor -1)b]\cap b\Z} B_{\lambda_2}
 \cup \big(\cup_{\lambda_1\in [0, \lfloor c/b\rfloor b]\cap b\Z} A_{\lambda_1}
\big)
\end{eqnarray}
by \eqref{dbac2covering.cor.pf.eq1},
where $k_n$ is so chosen that $k_n\le l<k_{n+1}$.
Therefore
\begin{eqnarray*}\label{dbac2covering.cor.pf.eq3}
 & & \big\{t_0+lb-\lfloor (t_0+lb)/a\rfloor a|\ 0\le l \in \Z\big\}\nonumber\\
 & \subset & \big( \cup_{\lambda_2\in [0, (\lfloor c/b\rfloor -1)b]\cap b\Z} B_{\lambda_2}\cap [0, a)\big)
 \cup \big(\cup_{\lambda_1\in [0, \lfloor c/b\rfloor b]\cap b\Z} A_{\lambda_1}\cap [0, a)
\big)
\end{eqnarray*}
by \eqref{dabcperiodic.section2} and \eqref{dbac2covering.cor.pf.eq2}. Notice that
the left hand side of the above inclusion is a dense subset of $[0, a)$ by the
assumption $a/b\not\in \Q$, while its right hand side is the  union of  finitely many  intervals that are right-open  and left-closed
by Theorem \ref{dabc1holes.tm}.
Thus
\begin{eqnarray}\label{dbac2covering.cor.pf.eq4} [0, a)&= &
\big( \cup_{k=0}^{\lfloor c/b\rfloor -1}
({\mathcal S}_{a, b ,c}+ kb)\cap [0, a)\big)\nonumber\\
& & \cup \big(({\mathcal S}_{a, b, c}\cap [0, c_0+a-b) +
\lfloor c/b\rfloor b)\cap [0, a)\big)\end{eqnarray} and the
conclusion \eqref{covering.tm2.eq1} follows.
%

\smallskip
{\bf Case 2}: \ {\em $a/b=p/q$ and $c/b\in \Z/q$ for some coprime integers $p$ and $q$.}

Take $t_0\in {\mathcal S}_{a, b, c}\cap b\Z/q$. The existence of such a point $t_0$ follows from
\eqref{dabcicontinuoustodiscrete.eq1} and
the assumption that ${\mathcal S}_{a, b, c}\ne \emptyset$. Following the argument in
\eqref{dbac2covering.cor.pf.eq3}, we have that
\begin{equation}\label{dabc1todabc2discrete.characterization.pf.eq4}
t_0+lb  \in   \big(\cup_{\lambda_1\in [0, \lfloor c/b\rfloor b]\cap b\Z} A_{\lambda_1}
 \big)\cup \big(\cup_{\lambda_2\in [0, (\lfloor c/b\rfloor-1)b]\cap b\Z} B_{\lambda_2}
\big)
\end{equation}
for all $0\le l\in \Z$. Observe that
$\{t_0, t_0+b, \ldots, t_0+(p-1)b\}+a\Z=b \Z/q$.
The above observation together with \eqref{dabc1todabc2discrete.characterization.pf.eq4} implies that
\begin{equation} \label{dabc1todabc2discrete.characterization.pf.eq5}
b \Z/q \subset \big(\cup_{\lambda_1\in [0, \lfloor c/b\rfloor b]\cap b\Z} A_{\lambda_1}
 \big)\cup \big(\cup_{\lambda_2\in [0, (\lfloor c/b\rfloor-1)b]\cap b\Z} B_{\lambda_2}
\big).
\end{equation}
Combining \eqref{dabcicontinuoustodiscrete.eq1} and \eqref{dabc1todabc2discrete.characterization.pf.eq5}
proves the desired covering property \eqref{covering.tm2.eq1}.
\end{proof}

\begin{proof}[Proof of Theorem \ref{sabcstar.tm}]
($\Longrightarrow$)\  By Proposition \ref{dabcsabc.thm} and  the assumption that ${\mathcal D}_{a, b, c}=\emptyset$, we have
that
\begin{equation}\label{dabc1todabc2.characterization.pf.eq1}
t+\lambda\not\in {\mathcal S}_{a, b, c}\ \ {\rm for \ all} \ t\in {\mathcal S}_{a, b, c}\cap [0, c_0+a-b)\ {\rm and} \ \lambda\in [b,
 \lfloor c/b\rfloor b]\cap b\Z;\end{equation}
and
 \begin{equation}\label{dabc1todabc2.characterization.pf.eq4}
t+\lambda\not\in {\mathcal S}_{a, b, c}\ \ {\rm for \ all} \
t\in {\mathcal S}_{a, b, c}\cap [c_0, a)\ {\rm  and} \ \lambda\in [b,
 \lfloor c/b\rfloor b-b]\cap b\Z.\end{equation}
Therefore the sets
${\mathcal S}_{a, b, c} \cap [0, c_0+a-b)+\lambda_1+a\Z, \lambda_1\in [0,
 \lfloor c/b\rfloor b]\cap b\Z$, and
${\mathcal S}_{a, b, c}\cap [c_0, a)+\lambda_2+a\Z, \lambda_2\in [0,
 \lfloor c/b\rfloor b-b]\cap b\Z$,
are mutually disjoint.
This together with the covering property  in Theorem \ref{covering.tm2}
and the periodic property \eqref{dabcperiodic.section2} for the set ${\mathcal S}_{a, b, c}$
implies that
\begin{eqnarray*}
a & = & \sum_{\lambda_1\in [0,
 \lfloor c/b\rfloor b]\cap b\Z} |({\mathcal S}_{a, b, c} \cap [0, c_0+a-b)+\lambda_1+a\Z)\cap [0, a)|\nonumber\\
 & & + \sum_{\lambda_2\in [0,
 \lfloor c/b\rfloor b-b]\cap b\Z} |( {\mathcal S}_{a, b, c}\cap [c_0, a)+\lambda_2+a\Z)\cap [0,a)|\nonumber\\
 & = & \sum_{\lambda_1\in [0,
 \lfloor c/b\rfloor b]\cap b\Z} |({\mathcal S}_{a, b, c} \cap [0, c_0+a-b)+\lambda_1+a\Z)\cap [\lambda_1,  a+\lambda_1)|\nonumber\\
 & & + \sum_{\lambda_2\in [0,
 \lfloor c/b\rfloor b-b]\cap b\Z} |( {\mathcal S}_{a, b, c}\cap [c_0, a)+\lambda_2+a\Z)\cap [\lambda_2,a+\lambda_2)|\nonumber\\
  & = & ( \lfloor c/b\rfloor +1)|{\mathcal S}_{a, b, c} \cap [0, c_0+a-b)|+
 \lfloor c/b\rfloor | {\mathcal S}_{a, b, c}\cap [c_0, a)|
 \end{eqnarray*}
 and hence \eqref{sabcstar.tm.eq1} follows.

$(\Longleftarrow)$ Set
$A_\lambda=({\mathcal S}_{a, b, c} \cap [0, c_0+a-b)+\lambda+a\Z)\cap [0, a)$ and
$B_\lambda=({\mathcal S}_{a, b, c} \cap [c_0, a)+\lambda+a\Z)\cap [0, a), \lambda\in b\Z$.
By Theorem \ref{covering.tm2}, the sets
 $A_{\lambda_1}, \lambda_1\in [0,
 \lfloor c/b\rfloor b]\cap b\Z$ and
$B_{\lambda_2}, \lambda_2\in [0,
 \lfloor c/b\rfloor b-b]\cap b\Z$  form a covering for the interval $[0, a)$.
 This together with the assumption \eqref{sabcstar.tm.eq1}  and  the periodic property
\eqref{dabcperiodic.section2} for the set ${\mathcal S}_{a, b, c}$ implies that
 \begin{eqnarray*} a & = & ( \lfloor c/b\rfloor +1)|{\mathcal S}_{a, b, c} \cap [0, c_0+a-b)|+
 \lfloor c/b\rfloor | {\mathcal S}_{a, b, c}\cap [c_0, a)|\\
  & = &  \sum_{\lambda_1\in [0,
 \lfloor c/b\rfloor b]\cap b\Z} |A_{\lambda_1}| + \sum_{\lambda_2\in [0,
 \lfloor c/b\rfloor b-b]\cap b\Z} |B_{\lambda_2}|\nonumber\\
 & \ge  & \big| \big(\cup_{\lambda_1\in [0,
 \lfloor c/b\rfloor b]\cap b\Z} A_{\lambda_1})\cup \big(\cup_{\lambda_2\in [0,
 \lfloor c/b\rfloor b-b]\cap b\Z} B_{\lambda_2} \big)| =a.
 \end{eqnarray*}
 This implies that  the intersection of any two of those sets $A_{\lambda_1}, \lambda_1\in [0,
 \lfloor c/b\rfloor b]\cap b\Z$ and
$B_{\lambda_2}, \lambda_2\in [0,
 \lfloor c/b\rfloor b-b]\cap b\Z$, has zero Lebesgue measure. Hence they have empty intersection as those sets are
 finite union of intervals that are left-closed and right-open by Theorems \ref{dabc1holes.tm} and \ref{dabc1discreteholes.tm}.
 This together with Proposition \ref{dabcsabc.thm} proves that ${\mathcal D}_{a, b, c}=\emptyset$.
\end{proof}

\subsection{Algebraic property of maximal invariant sets} 
In this subsection, we prove Theorem \ref{dabc1algebra.tm}.

\begin{proof} [Proof of  Theorem \ref{dabc1algebra.tm}]
 (i): \quad   By  \eqref{xabc.def} and \eqref{dabcperiodic.section2}, we have that
\begin{equation}\label{commutativity.lem.pf.eq1}Y_{a,b,c}(t+a)= Y_{a, b, c}(t)+ Y_{a,b,c}(a) \quad {\rm for \ all}\ t\in {\mathcal S}_{a, b, c}.
\end{equation}
Now we divide two cases to verify \eqref{dabc1algebra.tm.eq1}  for
$t\in [0, a)\cap {\mathcal S}_{a, b, c}$.

{\bf Case 1}:\ $ a/b\not\in \Q$.

 For $t\in  [0, c_0+a-b)\cap {\mathcal S}_{a, b, c}$, we obtain from  \eqref{rabcnewplus.def},
   \eqref{rabcbasic1.thm.eq1}, \eqref{xabc.def} and
 \eqref{rabcbasic1.thm.pf.eq0}
 that
 \begin{eqnarray}\label{commutativity.lem.pf.eq2}
 Y_{a,b,c}(R_{a, b, c}( t)) & = &   |[0, R_{a,b,c}(t))\cap {\mathcal S}_{a, b, c}|\nonumber\\
 & = &  Y_{a, b, c}(R_{a,b,c}(0))+ |[R_{a, b, c} (0), R_{a, b, c} (t))\cap {\mathcal S}_{a, b, c}|\nonumber \\
  & = &  Y_{a, b, c}(R_{a,b,c}(0))+ |R_{a, b, c} ([0,t)\cap {\mathcal S}_{a, b, c})|\nonumber\\
& = &  
 Y_{a, b, c}(R_{a,b,c}(0))+Y_{a, b,c}(t).
  \end{eqnarray}
  Similarly for $t\in  [c_0, a)\cap {\mathcal S}_{a, b, c}$, we get
 \begin{eqnarray}\label{commutativity.lem.pf.eq3}
   Y_{a,b,c}(R_{a, b, c}( t))  
  & = & |[R_{a, b, c} (c_0), R_{a, b, c} (t))\cap {\mathcal S}_{a, b, c}|+  Y_{a, b, c}(c_0+\lfloor c/b\rfloor b)\nonumber \\
  & = &  |R_{a, b, c} ([c_0,t)\cap {\mathcal S}_{a, b, c})|+ |[0, c_0+\lfloor c/b\rfloor b+a)\cap {\mathcal S}_{a, b, c}| \nonumber\\
  & & -
   |[c_0+\lfloor c/b\rfloor b, c_0+\lfloor c/b\rfloor b+a)\cap {\mathcal S}_{a, b, c}|
\nonumber\\
  & = & |[c_0,t)\cap {\mathcal S}_{a, b, c}|+  |[0, \lfloor c/b\rfloor b+b)\cap {\mathcal S}_{a, b, c}|\nonumber\\
  & & + |R_{a,b,c}([0, c_0+a-b))\cap {\mathcal S}_{a, b, c}|- Y_{a,b,c}(a)
\nonumber\\
& = &  
 Y_{a, b,c}(t)+Y_{a, b, c}(R_{a,b,c}(0))-Y_{a,b,c}(a).
 \end{eqnarray}
 Then \eqref{dabc1algebra.tm.eq1} follows from
 \eqref{commutativity.lem.pf.eq2} and \eqref{commutativity.lem.pf.eq3}.

\smallskip

{\bf Case 2}: \ {\em $a/b=p/q$ and $c/b\in \Z/q$ for some coprime integers $p$ and $q$.}

For every $t\in [0, c_0+a-b)\cap {\mathcal S}_{a, b, c}$,
 we obtain from
\eqref{rabcbasic1.thm.eq1} and \eqref{rabcbasic1.thm.eq1c} in Proposition
 \ref{rabcbasic1.thm},    \eqref{dabc1discreteholes.tm.eq8}
 in Theorem \ref{dabc1discreteholes.tm}, and \eqref{commutativity.lem.pf.eq1}
 that
 \begin{eqnarray}\label{dabc1discretealgebraic.tm.pf.eq2}
 Y_{a,b,c}(R_{a, b, c}( t)) 
&= & Y_{a, b, c}(R_{a,b,c}(0))+Y_{a, b,c}(t)\nonumber \\
&\in &
Y_{a, b, c}(c_1+b-a)+Y_{a, b,c}(t)+Y_{a, b, c}(a)\Z.
 \end{eqnarray}
  Similarly for  every $t\in  [c_0, a)\cap {\mathcal S}_{a, b, c}$, we
  have that
 \begin{eqnarray}\label{dabc1discretealgebraic.tm.pf.eq3}
   Y_{a,b,c}(R_{a, b, c}( t))  
& = &  
 Y_{a, b,c}(t)+Y_{a, b, c}(R_{a,b,c}(0))-Y_{a,b,c}(a)\nonumber \\
&\in &
Y_{a, b, c}(c_1+b-a)+Y_{a, b,c}(t)+Y_{a, b, c}(a)\Z.
 \end{eqnarray}
 Then \eqref{dabc1algebra.tm.eq1} follows from
\eqref{dabc1discretealgebraic.tm.pf.eq2} and \eqref{dabc1discretealgebraic.tm.pf.eq3}.

\smallskip

(ii): \quad Let $N_1, N_2, \delta, \delta'$ be as in Theorem
\ref{dabc1discreteholes.tm}.  By
\eqref{dabc1discreteholes.tm.eq8}, the set ${\mathcal K}_{a, b, c}$ of
marks is given by
\begin{equation*}
{\mathcal K}_{a, b, c}=\{Y_{a, b, c} ((R_{a, b, c})^n(c-c_0+b-a+\delta))|
0\le n\le N_1+N_2\}+Y_{a, b, c}(a)\Z.
\end{equation*}
This, together with the first conclusion of this theorem and the fact that $[0,
\delta)$ is contained in a gap, implies that
\begin{eqnarray}\label{dabc1discretealgebraic.tm.pf.eq4}
{\mathcal K}_{a, b, c} & = & Y_{a, b, c}( R_{a, b, c}(\delta)-a)\nonumber\\
& & + \{n Y_{a, b, c}(c_1+b-a)| 0\le n\le N_1+N_2\}+Y_{a,
b,
c}(a)\Z\nonumber\\
& = & \{n Y_{a, b, c}(c_1+b-a)|\ 1\le n\le
N_1+N_2+1\}+Y_{a, b, c}(a)\Z.
\end{eqnarray}
On the other hand, it follows from
\eqref{dabc1discreteholes.tm.eq8} and
\eqref{dabc1discreteholes.tm.eq9} that $(R_{a, b,
c})^{N_1+N_2}(c-c_0+b+\delta)\in \delta+a\Z$. This together with
the first conclusion of this theorem implies that
\begin{equation}\label{dabc1discretealgebraic.tm.pf.eq5}
(N_1+N_2+1) Y_{a, b, c}(c_1+b-a)\in  Y_{a, b, c}^d(a)\Z.
\end{equation}
Combining \eqref{dabc1discretealgebraic.tm.pf.eq4} and
\eqref{dabc1discretealgebraic.tm.pf.eq5} proves that ${\mathcal K}_{a, b, c}$
form a finite cyclic group generated by $ Y_{a, b, c}(c_1+b-a)+ Y_{a, b, c}(a)\Z$.

(iii) \quad By Theorem \ref{dabc1holes.tm},
there exists a nonnegative integer $D$ such that
$(R_{a, b, c})^n(c-c_0+b-a)+[a-b, 0]+a\Z, 0\le n\le D$, are mutually disjoint, and
\begin{equation}\label{dabc1discretealgebraic.tm.pf.eq6}
(R_{a, b, c})^D(c-c_0+b-a)+[a-b,0)+a\Z=[c_0+a-b, c_0)+a\Z.\end{equation}
Therefore $(R_{a, b, c})^n(c-c_0+b-a)\in {\mathcal S}_{a, b, c}$ for all $0\le n\le D$, and
\begin{eqnarray}\label{dabc1discretealgebraic.tm.pf.eq7}
{\mathcal K}_{a, b,c} & = & \cup_{n=0}^D \{ Y_{a, b, c} ((R_{a, b, c})^n(c-c_0+b-a))+ Y_{a, b, c}(a) \Z\}\nonumber\\
& = &  \cup_{n=0}^D
\{ (n+1) Y_{a, b, c} (c-c_0+b-a)+ Y_{a, b, c}(a) \Z\}\nonumber\\
&= &
\cup_{m=1}^{D+1}
\{ m Y_{a, b, c} (c_1+b-a)+ Y_{a, b, c}(a) \Z\}
\end{eqnarray}
where the second equality follows from
 the first conclusion of this theorem.
Moreover, it follows \eqref{dabc1discretealgebraic.tm.pf.eq6} that
\begin{equation}\label{dabc1discretealgebraic.tm.pf.eq8}
(D+1) Y_{a, b, c} (c_1+b-a)-Y_{a, b, c}(c_0)\in Y_{a, b, c}(a)\Z.\end{equation}
Also we notice that the nonnegative integer $D$ satisfying \eqref{dabc1discretealgebraic.tm.pf.eq8}
is unique as $(R_{a, b, c})^n(0)\not\in a\Z$ for all positive integers $n$ by the assumption $a/b\not\in \Q$.
This uniqueness, together with \eqref{dabc1discretealgebraic.tm.pf.eq7} and
 \eqref{dabc1discretealgebraic.tm.pf.eq8}, proves the conclusion (iii).
\end{proof}

\begin{rem}\label{dabc1discretealgebraic.tm.rem} {\rm
Let $N_1$ and $N_2$ be as in Theorem \ref{dabc1discreteholes.tm}.
Notice that the mutually disjoint properties (iv) and
(iv)${}^\prime$ in Theorem \ref{dabc1discreteholes.tm}, $n Y_{a,
b, c}(c_1+b-a)\not\in Y_{a, b, c}(a)\Z$ for all $0\le
n\le N_1+N_2$. Thus ${\mathcal K}_{a, b, c}$ is a cyclic group of order
$N_1+N_2+1$. }\end{rem}

\section{Maximal Invariant Sets with Irrational Time-Frequency Lattice}
\label{gaborrational1.section}

In this section, we prove Theorem \ref{newmaintheorem4}.
To  do so, we need
 characterize
  the non-triviality \eqref{sabcemptynonempty1} of the maximal invariant set ${\mathcal S}_{a, b, c}$.
By Theorems \ref{dabc1holes.tm} and \ref{dabc1discreteholes.tm},  after
 performing the   holes-removal surgery,  the maximal invariant  set ${\mathcal S}_{a, b, c}$
  becomes
  the real line with marks.
This suggests that  for the case that $a/b\not\in \Q$ we can
expand the line with marks by  inserting holes $[0, b-a)$
  at  every location of marks to recover the maximal invariant set ${\mathcal S}_{a, b, c}$.
Using  the equivalence between the application of the piecewise linear transformation $R_{a, b, c}$ on  the set ${\mathcal S}_{a, b, c}$
and a rotation on the  circle with marks given in  Theorem  \ref{dabc1algebra.tm}, we can characterize
 the non-triviality \eqref{sabcemptynonempty1}  of the maximal invariant set ${\mathcal S}_{a, b, c}$
 via two  nonnegative integer parameters $d_1$ and $d_2$
 for the case that $a/b\not\in \Q$.

\begin{thm}
\label{dabc1characterization.tm}
 Let $(a,b,c)$ be a triple of positive numbers satisfying
$a<b<c, \lfloor c/b\rfloor\ge
2,  b-a<c_0:=c-\lfloor c/b\rfloor b<a, 0< c_1:= c-c_0-\lfloor (c-c_0) /a\rfloor a < 2a-b$, and
$a/b\not\in \Q$. Then
 ${\mathcal S}_{a, b, c}\ne \emptyset$ if and only if
 there exist nonnegative integers $d_1$ and $d_2$ such that
   \begin{equation}\label{cnecessary.eq1}
c-(d_1+1)(\lfloor c/b\rfloor+1) (b-a)-(d_2+1) \lfloor c/b\rfloor (b-a)\in a\Z,
\end{equation}
\begin{equation}\label{cnecessary.eq2}
\lfloor c/b\rfloor b+(d_1+1)(b-a)<c<\lfloor c/b\rfloor b+b- (d_2+1) (b-a),
\end{equation}
and
\begin{equation}\label{cnecessary.eq3}
 \#E_{a,b,c}= d_1,
\end{equation}
where  $E_{a, b, c}$ is defined
as in \eqref{eabc.def}.
\end{thm}

 The nonnegative  integers $d_1$ and $d_2$ in  Theorem \ref{dabc1characterization.tm}  satisfy $(d_1+d_2+1)<a/(b-a)$
by \eqref{cnecessary.eq2}, and they are uniquely
determined by the triple $(a, b, c)$ of positive numbers
by \eqref{cnecessary.eq1} and the assumptions that $\lfloor c/b\rfloor\ge 2$ and $a/b\not\in \Q$. We also notice  that
the  nonnegative integer parameters $d_1$ and $d_2$ in  Theorem \ref{dabc1characterization.tm}
are indeed the numbers of holes contained in $[0, c_0+a-b)$ and $[c_0, a)$ respectively.

\smallskip

In next two subsections, we prove Theorem  \ref{dabc1characterization.tm},
and apply Theorems \ref{sabcstar.tm} and \ref{dabc1characterization.tm}
  to  prove  Theorem
 \ref{newmaintheorem4} respectively.

\subsection{Nontrivial maximal invariant sets with irrational time-frequency lattices}
In this subsection, we prove  Theorem \ref{dabc1characterization.tm}.

\begin{proof}[Proof of Theorem \ref{dabc1characterization.tm}]
($\Longrightarrow$)\ Assume that  ${\mathcal S}_{a, b, c}\ne
\emptyset$. Let $D$ be the nonnegative integer in Theorem
\ref{dabc1holes.tm}. Then
$(R_{a,b,c})^D([c-c_0, c+b-c_0-a)+a\Z)=[c_0+a-b,c_0)+a\Z$,
and the periodic holes
$ (R_{a,b,c})^n([c-c_0, c+b-c_0-a)+a\Z)=(R_{a,b,c})^n(c-c_0)+[0,
b-a)+a\Z, 0\le n\le D$,
are mutually disjoint. 
This implies that
\begin{equation}\label{dabc1characterization.tm.pf.eq3}
c_0+a-b-\big(c-c_0+ d_1(\lfloor c/b\rfloor+1) b+d_2 \lfloor c/b\rfloor b\big)\in a\Z,
\end{equation}
where $d_1, d_2$ are the numbers of the indices $n\in [0, D-1]$ such that $(R_{a, b,c})^n ([c- c_0, c+b-c_0-a)+a\Z)$
is contained  in $[0, c_0+a-b)+a\Z$  and in  $[c_0, a)+a\Z$  respectively.
Then   \eqref{cnecessary.eq1}  follows from \eqref{dabc1characterization.tm.pf.eq3}.

\smallskip
Observe that
$$D=d_1+d_2$$
as for every $0\le n\le D-1$ the periodic hole
$(R_{a,b,c})^n([c-c_0, c+b-c_0-a)+a\Z)$ is  contained either in $[0,
c_0+a-b)+a\Z$ or in $[c_0, a)+a\Z$ by  Theorem \ref{dabc1holes.tm}.

Notice that the periodic holes $(R_{a,b,c})^n([c-c_0,
c+b-c_0-a)+a\Z), 0\le n\le d_1+d_2$, are  mutually disjoint and
have length $b-a$ by Theorem \ref{dabc1holes.tm} and that
there are $d_1$  (resp. $d_2$) holes contained in $[0, c_0+a-b)$
(resp. $[c_0, a)$) by the definition of integer parameters $d_1$ and $d_2$.
Therefore $d_1(b-a)<c_0+a-b$ and $d_2(b-a)<a-c_0$, which proves
\eqref{cnecessary.eq2}.

\smallskip

Let $\theta_{a,b,c}:=Y_{a,b,c}(c-c_0+b)$ be as in Theorem \ref{dabc1algebra.tm},
and  $\tilde \theta_{a, b,c}=Y_{a,b,c}(c_1)$. Recall that there
are $d_1+d_2+1$ holes contained in the interval $[0,a)$ by Theorem
\ref{dabc1holes.tm}. This together with
\eqref{commutativity.lem.pf.eq1} implies that
\begin{equation}\label{dabc1characterization.tm.pf.eq5}
Y_{a,b,c}(a)=a-(d_1+d_2+1)(b-a)\end{equation}
and
\begin{equation} \label{dabc1characterization.tm.pf.eq6-}
\tilde\theta_{a, b,c}-Y_{a,b,c}(c-c_0)\in 
Y_{a, b,c}(a)\Z, 
\end{equation}
where the inclusion holds because 
 $[c-c_0, c+b-c_0-a)$ is a black hole of the transformation $\tilde R_{a, b,c}$, and
\begin{equation} \label{dabc1characterization.tm.pf.eq6+}
Y_{a,b,c}(c-c_0)-\theta_{a,b,c}=-Y_{a, b,c}(a).\end{equation}
From Theorem \ref{dabc1algebra.tm}, we see that the marks are located at  $Y_{a,b,c}(c-c_0)
+n\theta_{a, b,c}+(a-(d_1+d_2+1)(b-a))\Z, 0\le n\le d_1+d_2$.
This together with \eqref{dabc1characterization.tm.pf.eq5}, \eqref{dabc1characterization.tm.pf.eq6-}
and \eqref{dabc1characterization.tm.pf.eq6+} implies that
the locations of marks are
$n\tilde \theta_{a,b,c}+ (a-(d_1+d_2+1)(b-a))\Z, 1\le n\le d_1+d_2+1$.

Recall that  there are  $d_1$ holes contained in  $[0, c_0+a-b)$. Then $Y_{a,b,c}(c_0+a-b)= c_0+a-b-d_1(b-a)$, which implies
that
\begin{equation}\label{dabc1characterization.tm.pf.eq7}c_0-(d_1+1)(b-a)-(d_1+d_2+1) \tilde \theta_{a,b,c}\in (a-(d_1+d_2+1)(b-a))\Z.\end{equation}
Let $m$ be the number of holes
$(R_{a,b,c})^n([c-c_0, c+b-c_0-a)+a\Z)=(R_{a,b,c})^n([c_1, c_1+b-a)+a\Z), 0\le n\le d_1+d_2$, contained in
$[0, c_1)+a\Z$. Due to the one-to-one correspondence between holes  and marks, $m$ is also  the cardinality
of the set
$\big\{1\le n\le d_1+d_2+1\big| n\tilde \theta_{a, b, c}\in [0, \tilde\theta_{a, b,c})+(a-(d_1+d_2+1)(b-a))\Z\big\}$.
This, together with the observation that $[c_1, c_1+b-a)$
is the black hole of the transformation $\tilde R_{a, b, c}$, implies that
\begin{equation} \label{dabc1characterization.tm.pf.eq8} \tilde \theta_{a, b, c}=c_1-m(b-a).
\end{equation}
Let $\tilde m$ be the unique integer  such that
$(d_1+d_2+1)\tilde \theta_{a,b,c}\in \tilde m (a-(d_1+d_2+1)(b-a))+ [0, a-(d_1+d_2+1)(b-a))$. We want to prove that
\begin{equation}\label{dabc1characterization.tm.pf.eq9} \tilde m=m.
\end{equation}
For any $1\le l\le \tilde m$, there exists one and only one $1\le n_l\le d_1+d_2+1$ such that
$  n_l \tilde \theta_{a, b,c}\in l(a-(d_1+d_2+1)(b-a))+[0, \tilde \theta_{a,b,c})$, which implies that
$\tilde m\le m$. Now we prove that $m\le \tilde m$. Suppose on the contrary that $m>\tilde m$. Then there exists an integer $1\le n\le d_1+d_2+1$
such that $n\tilde \theta_{a,b,c}\in [0, \tilde \theta_{a, b,c})+ (a-(d_1+d_2+1)(b-a)) (\Z\backslash \{1, \ldots,\tilde m\})$.
This implies that $n\tilde \theta_{a, b,c}\ge (\tilde m+1)  (a-(d_1+d_2+1)(b-a))$, which is a contradiction as $\tilde\theta_{a,b,c}\le
n\tilde \theta_{a, b,c}\le (d_1+d_2+1) \tilde \theta_{a, b,c} <(\tilde m+1) (a-(d_1+d_2+1)(b-a))$ by the definition of the integer $\tilde m$,
and hence \eqref{dabc1characterization.tm.pf.eq9} is established.

From  \eqref{dabc1characterization.tm.pf.eq7}, \eqref{dabc1characterization.tm.pf.eq8}
and  \eqref{dabc1characterization.tm.pf.eq9},
\begin{eqnarray*}  & &  (d_1+d_2+1) (c_1-m(b-a)) =  (d_1+d_2+1) \tilde\theta_{a, b,c}\\
& = &  c_0-(d_1+1)(b-a)+m(a-(d_1+d_2+1)(b-a)), \end{eqnarray*}
which implies that
\begin{equation} \label{dabc1characterization.tm.pf.eq10} m a = (d_1+d_2+1)c_1-c_0+(d_1+1)(b-a).
\end{equation}
Then the condition \eqref{cnecessary.eq3} follows from  \eqref{dabc1characterization.tm.pf.eq7}, \eqref{dabc1characterization.tm.pf.eq8}
and  \eqref{dabc1characterization.tm.pf.eq10}, and the definition of the integer $d_1$.

\bigskip

($\Longleftarrow$)\ Let  $c_1=c-c_0-\lfloor (c-c_0)/a\rfloor a$,
and let $d_1$ and $d_2$ be as in  \eqref{cnecessary.eq1} and
\eqref{cnecessary.eq2}. Then $c_1\ne 0$  by $a/b\not\in \Q$, and
$-a <-c_0+b-a\le (d_1+d_2+1) c_1- c_0+(d_1+1)(b-a) \le (d_1+d_2+1)
c_1 <(d_1+d_2+1) a$ by \eqref{cnecessary.eq1} and
\eqref{cnecessary.eq2}. Also from \eqref{cnecessary.eq1} and
\eqref{cnecessary.eq2}, we see that $(d_1+d_2+1) c_1-
c_0+(d_1+1)(b-a)\in (d_1+d_2+1)\lfloor c/b\rfloor b-c+\lfloor
c/b\rfloor b+(d_1+1)b +a\Z=a\Z$. Thus
\begin{equation} \label{dabc1characterization.tm.pf.eq14}
(d_1+d_2+1) c_1- c_0+(d_1+1)(b-a)=m a\end{equation}
for some integer  $0\le m\le d_1+d_2$.
Set
\begin{equation}\label{dabc1characterization.tm.pf.eq14+}
\tilde\theta_{a,b,c}=
c_1 -m(b-a).\end{equation}
Then
\begin{equation}\label{dabc1characterization.tm.pf.eq15}
(d_1+d_2+1) \tilde \theta_{a,b,c}= c_0-(d_1+1)(b-a)+m(a-(d_1+d_2+1)(b-a))
\end{equation}
by \eqref{dabc1characterization.tm.pf.eq14}.
This together with $ 0\le m\le d_1+d_2$ and $0<c_0-(d_1+1)(b-a)< a-(d_1+d_2+1)(b-a)$  implies that
\begin{equation} \label{dabc1characterization.tm.pf.eq16}
\tilde\theta_{a, b, c}\in (0, a-(d_1+d_2+1)(b-a)).
\end{equation}

We claim that
\begin{equation} \label{dabc1characterization.tm.pf.eq17}
(n-n')\tilde \theta_{a, b,c}\not\in
(a-(d_1+d_2+1)(b-a))\Z
\end{equation}
for all $1\le n\ne n'\le d_1+d_2+1$. Suppose on the contrary that
\eqref{dabc1characterization.tm.pf.eq17} are not true. Then
$k\tilde \theta_{a,b,c}=l (a-(d_1+d_2+1)(b-a))$
for some integers $l\in \Z$ and $ k\in [1, d_1+d_2]\cap \Z$. Then
$k(m-\lfloor c/b\rfloor)= l(d_1+d_2+1)$ and $k(m- \lfloor (\lfloor c/b\rfloor b/a)\rfloor)=l(d_1+d_2+2)$
by  the assumption $a/b\not\in \Q$. Thus
$l=k (\lfloor (\lfloor c/b\rfloor b/a)\rfloor-\lfloor c/b\rfloor)$, which is a contradiction as
$1\le l< k$ by \eqref{dabc1characterization.tm.pf.eq16}.

Denote ${\mathcal K}_{a, b,c}:=\{n\tilde \theta_{a, b,c}\}_{n=1}^{d_1+d_2+1}+(a-(d_1+d_2+1)(b-a))\Z$ and rewrite ${\mathcal K}_{a, b,c}$ as
$\{z_n\}_{n=1}^{d_1+d_2+1}+(a-(d_1+d_2+1)(b-a))\Z$ where
$0<z_1< z_2<\ldots<z_{d_1+d_2+1}<a-(d_1+d_2+1)(b-a)$.
The existence of such an increasing sequence $\{z_n\}_{n=1}^{d_1+d_2+1}$ follows from \eqref{dabc1characterization.tm.pf.eq17}.
 Given any $\delta\in (0, c_0-(d_1+1)(b-a))$ (respectively $\delta\in (c_0-(d_1+1)(b-a), a-(d_1+d_2+1)(b-a))$),  it follows from
  \eqref{dabc1characterization.tm.pf.eq15} and \eqref{dabc1characterization.tm.pf.eq16} that for  any integer $k\in [0, m]$
  (resp. $k\in [0, m-1]$)
there is one and only one
integer $n\in [1,  d_1+d_2+1]$ such that
$n\tilde \theta_{a, b,c}\in k(a-(d_1+d_2+1)(b-a))+[\delta, \delta+\tilde \theta_{a, b, c})$
and  for $k\in \Z\backslash [0, m]$ (resp. $k\in \Z\backslash [0, m-1]$), and
 there is no integer $n\not\in [1, d_1+d_2+1]$ such that
$n\tilde \theta_{a, b,c}\in k(a-(d_1+d_2+1)(b-a))+[\delta, \delta+\tilde \theta_{a, b, c})$.
%
The above observations together with  \eqref{dabc1characterization.tm.pf.eq16} and \eqref{dabc1characterization.tm.pf.eq17}
 imply that
\begin{equation}\label{dabc1characterization.tm.pf.eq19+} \#\big([\delta, \delta+\tilde \theta_{a, b, c})\cap
\big( \{z_k\}_{k=1}^{d_1+d_2+1}+\{0,a-(d_1+d_2+1)(b-a)\} \big)=m+1\end{equation}
 for $\delta\in (0, c_0-(d_1+1)(b-a))$,
 and
\begin{equation}\label{dabc1characterization.tm.pf.eq19++} \#\big([\delta, \delta+\tilde \theta_{a, b, c})\cap
\big( \{z_k\}_{k=1}^{d_1+d_2+1}+\{0,a-(d_1+d_2+1)(b-a)\} \big)=m\end{equation}
 for $\delta\in (c_0-(d_1+1)(b-a), a-(d_1+d_2+1)(b-a))$.

Now let us expand marks located at  $\{z_l\}_{l=1}^{d_1+d_2+1}+(a-(d_1+d_2+1)(b-a))\Z$ to holes
of length $b-a$ located at $\{y_l\}_{l=1}^{d_1+d_2+1}+a\Z$ on the real line by
\begin{equation}\label{dabc1characterization.tm.pf.eq212} y_l=z_l+(l-1)(b-a), 1\le l\le d_1+d_2+1.\end{equation}
Clearly  $0<y_1<y_2<\ldots<y_{d_1+d_2+1}<a$.
Now let us prove
\begin{equation}\label{dabc1characterization.tm.pf.eq222}
(R_{a,b,c})^n(c-c_0)+a\Z= y_{l(n)}+a\Z\quad {\rm for \ all}\ 0\le
n\le d_1+d_2,
\end{equation}
by induction on  $0\le n\le d_1+d_2$, where $l(n)\in [1, d_1+d_2+1]$ is the unique integer such that $z_{l(n)}\in
(n+1) \tilde \theta_{a, b,c}+(a-(d_1+d_2+1)(b-a))\Z$.
Applying \eqref{dabc1characterization.tm.pf.eq19+}
with sufficiently small $\delta$ and recalling $z_1>0$ proves  that
$z_{m+1}=\tilde\theta_{a, b,c}$.
Combining  \eqref{cnecessary.eq3} and \eqref{dabc1characterization.tm.pf.eq15} gives
\begin{equation}\label{dabc1characterization.tm.pf.eq20} z_{d_1+1}=c_0+a-b-d_1(b-a)= (d_1+d_2+1)\tilde \theta_{a, b, c}-m(a-(d_1+d_2+1)(b-a)).
\end{equation}
Thus
\begin{equation}\label{dabc1characterization.tm.pf.eq201}
y_{l(0)}=y_{m+1}=z_{m+1}+m(b-a)=\tilde \theta_{a, b,c}+m(b-a)=c_1\end{equation}
and
\begin{equation}\label{dabc1characterization.tm.pf.eq202}
y_{l(d_1+d_2)}= y_{d_1+1}=z_{d_1+1}+d_1(b-a)=c_0+a-b.
\end{equation}
The conclusion \eqref{dabc1characterization.tm.pf.eq222}  for $n=0$
follows from \eqref{dabc1characterization.tm.pf.eq201}.
Inductively we assume that \eqref{dabc1characterization.tm.pf.eq222} holds for $n=k\le d_1+d_2-1$.
Then $z_{l(k)}\ne c_0-(d_1+1) (b-a)$ by \eqref{dabc1characterization.tm.pf.eq15}, \eqref{dabc1characterization.tm.pf.eq17}
and the observation that $l(k)\ne d_1+d_2+1$.
If $z_{l(k)}<  c_0-(d_1+1) (b-a)$,
then $y_{l(k)}<c_0+a-b$  by \eqref{dabc1characterization.tm.pf.eq202}  and
\begin{eqnarray} (R_{a, b,c})^{k+1}(c-c_0)  & =  &  R_{a, b,c}((R_{a, b,c})^{k}(c-c_0))
 \in
 R_{a, b,c}(y_{l(k)})+a\Z\nonumber\\
 & =  & y_{l(k)}+\lfloor c/b\rfloor b +b +a\Z\nonumber\\
 & = &  z_{l(k)}+ \tilde\theta_{a, b, c} +(m+l(k))(b-a)+a\Z.
 \end{eqnarray}
 Note that $z_{l(k+1)}= z_{l(k)}+\tilde \theta_{a, b, c}$ or $z_{l(k+1)}=z_{k(l)}+\tilde \theta_{a, b, c}-(a-(d_1+d_2+1)(b-a))$.
 For the first case, $l(k+1)=l(k)+m+1$ as $[z_{l(k)}, z_{l(k)}+\tilde \theta_{a, b, c})\cap
 \big( \{z_k\}_{k=1}^{d_1+d_2+1}+\{0,a-(d_1+d_2+1)(b-a)\} \big)=m+1$ by \eqref{dabc1characterization.tm.pf.eq19+}
 and hence
\begin{eqnarray} (R_{a, b,c})^{k+1}(c-c_0) & \in &   z_{l(k)}+ \tilde\theta_{a, b, c} +(m+l(k))(b-a)+a\Z\nonumber\\
 & = & z_{l(k+1)}+(l(k+1)-1)(b-a)+a\Z= y_{l(k+1)}+a\Z.
\end{eqnarray}
Similarly for the second case, $l(k+1)=l(k)+m+1-(d_1+d_2+1)$ since
$\#\big([0, z_{l(k+1)})+(a-(d_1+d_2+1)(b-a))\cap
 \big( \{z_k\}_{k=1}^{d_1+d_2+1}+\{0,a-(d_1+d_2+1)(b-a)\} \big)\big)=
\#\big(\big( [0, z_{l(k)})\cap
 \{z_k\}_{k=1}^{d_1+d_2+1}\} \big)\cup \big( [z_{l(k)}, z_{l(k)}+\tilde \theta_{a, b, c})\cap
 \{z_k\}_{k=1}^{d_1+d_2+1}\} \big)\big)= l(k)-1+m+1=m+l(k)$ by \eqref{dabc1characterization.tm.pf.eq19+}.
Thus
 \begin{eqnarray} (R_{a, b,c})^{k+1}(c-c_0) & \in &   z_{l(k)}+ \tilde\theta_{a, b, c} +(m+l(k))(b-a)+a\Z\nonumber\\
 &= & z_{l(k+1)}+(a-(d_1+d_2+1)(b-a))\nonumber\\
 & &  +(l(k+1)+(d_1+d_2+1)-1)(b-a)+a\Z\nonumber\\
 & = &  y_{l(k+1)}+a\Z.
\end{eqnarray}
This shows that the inductive conclusion holds when $z_{l(k)}<c_0-(d_1+1)(b-a)$.
Similarly we can show that the inductive conclusion holds when $z_{l(k)}>c_0-(d_1+1) (b-a)$.

From \eqref{dabc1characterization.tm.pf.eq222}, we see that for any $0\le n\le d_1+d_2-1$,
$(R_{a, b,c})^n([c-c_0, c-c_0+b-a))+a\Z= [y_{l(n)}, y_{l(n)}+b-a)+a\Z$
is contained either in $[0, c_0+a-b)+a\Z$ or $[c_0, a)+a\Z$,
and $(R_{a, b,c})^D([c-c_0, c-c_0+b-a))+a\Z=[c_0+a-b, c_0)+a\Z$. Therefore
${\mathcal S}_{a, b, c}$ is the complement of $\cup_{n=0}^{d_1+d_2} ([y_{l(n)}, y_{l(n)}+b-a)+a\Z) $
and hence  it is not an empty set.
\end{proof}

\subsection{Proof of Theorem \ref{newmaintheorem4}}
 (XII):\quad
We observe that ${\mathcal G}(\chi_{[0,c)}, a\Z\times \Z/b)$ is not a Gabor frame if and only if ${\mathcal D}_{a, b, c}\ne \emptyset$
if and only if ${\mathcal S}_{a, b, c}\ne \emptyset$ and
\eqref{sabcstar.tm.eq1} does not hold
if and only if the triple $(a, b,c)$ satisfies \eqref{cnecessary.eq1}, \eqref{cnecessary.eq2}, \eqref{cnecessary.eq3}, and
$c-(\lfloor c/b\rfloor+1)(d_1+1)(b-a)- \lfloor c/b\rfloor (d_1+1)(b-a)\ne a$. In the above argument,
 the first equivalence holds by  Theorem \ref{framenullspace1.tm}, the second one follows from   \eqref{dabc2subsetdabc1}
  and Theorem  \ref{sabcstar.tm}, and
the last one is obtained from Theorem \ref{dabc1characterization.tm}
and the observation that \eqref{sabcstar.tm.eq1} holds if and only if
$c-(\lfloor c/b\rfloor+1)(d_1+1)(b-a)- \lfloor c/b\rfloor (d_2+1)(b-a)=a$
as  there are $d_1$ holes of length $b-a$
in ${\mathcal S}_{a, b, c}\cap [0, c_0+a-b)$ and
$d_2$ holes of length $b-a$
in ${\mathcal S}_{a, b, c}\cap [c_0, a)$ by Theorem \ref{dabc1holes.tm}.

\section{Maximal Invariant Sets with Rational Time-Frequency Lattice}
\label{rational.section}

In this section, we  prove Theorem \ref{newmaintheorem5}.
To do so,  let us consider  how to expand the line with marks 
for the case that $a/b=p/q$  for some coprime integers $p$ and $q$ and  $c/b\in
\Z/q$.   Unlike for the case that $b/a\not\in \Q$, we  insert gaps of either large or small sizes
 at each location of marks and also insert a gap at the origin, which makes
the augmentation operation rather delicate and complicated. But on the other hand, we get the help from the finite cyclic group structure for the marks
 stated in Theorem \ref{dabc1algebra.tm}.

\begin{thm}
\label{dabc1discretecharacterization.tm}
 Let  $(a,b,c)$ be a triple of positive numbers satisfying
$a<b<c, b-a<c_0:=c-\lfloor c/b\rfloor b<a, 0<c_1:=\lfloor c/b\rfloor b-\lfloor (\lfloor c/b\rfloor b/a)\rfloor a<2a-b,
  \lfloor c/b\rfloor\ge
2, a/b=p/q$  for some coprime integers $p$ and $q$, and  $c/b\in
\Z/q$. Then
 ${\mathcal S}_{a, b, c}\ne \emptyset$ if and only if the triple $(a, b,  c)$  of positive numbers is  one of the following three types:
\begin{itemize}
\item  [{1)}] $c_0<{\rm gcd}(a, c_1)$.

\item[{2)}] $b-c_0<{\rm gcd} (a, c_1+b)$.

\item[{3)}]
 There exist nonnegative integers $d_1, d_2, d_3, d_4$ such that
   \begin{equation}\label{cnecessarydiscrete.eq1}
   0< B_d:=a-(d_1+d_2+1)(b-a)\in N b\Z/q,
\end{equation}
\begin{equation}\label{cnecessarydiscrete.eq2}
 Nc_1+(d_1+d_3+1)(b-a)\in a \Z,
\end{equation}
\begin{equation}\label{cnecessarydiscrete.eq3}
 (d_1+d_2+1)(N c_1+(d_1+d_3+1)(b-a))-(d_1+d_3+1)a\in N a\Z,
\end{equation}
\begin{equation} \label{cnecessarydiscrete.eq5}
c_0=(d_1+1)(b-a)+(d_1+d_3+1)B_d/N+\delta
\end{equation}
for some $\delta\in (-\min(B_d/N, a-c_0), \min(B_d/N, c_0+b-a))$,
\begin{equation}\label{cnecessarydiscrete.eq6}
{\rm gcd}(Nc_1+(d_1+d_3+1)(b-a), Na)=a,
\end{equation}
and
\begin{equation}\label{cnecessarydiscrete.eq7}
 \#E_{a,b,c}^d= d_1,
\end{equation}
where $N=d_1+d_2+d_3+d_4+2$
 and $E_{a, b, c}^d$ is
defined as in \eqref{eabcd.def}.
\end{itemize}
\end{thm}

%
%
%

In  Theorem \ref{dabc1discretecharacterization.tm}, we insert a gap of
large size at the origin for the first two cases, while
a gap of small size is inserted at the origin for the third case.
For the first two cases,  no gaps of small size
have been inserted at any location of  marks and the size of gaps inserted is always $c_0$
for the first case and $b-c_0$ for the second case.
For the  third case, the nonnegative integer parameters $d_1, d_2$
are indeed the numbers of gaps of  size $b-a+|\delta|$ inserted in $[0, c_0+a-b)$ and  $[c_0, a)$ respectively, and
  the nonnegative integer parameters $d_3, d_4$ are the numbers of gaps of  size $|\delta|$ inserted in $[0, c_0+a-b)$
and $[c_0, a)$, excluding
the one inserted at the origin, respectively.

In next two subsections, we give the proofs of  Theorems \ref{dabc1discretecharacterization.tm} and \ref{newmaintheorem5}.

\subsection{Nontrivial maximal invariant sets with rational time-frequency lattices}
\label{proofoftheorem3.10.section}
In this subsection, we prove Theorem \ref{dabc1discretecharacterization.tm}.
The necessity of Theorem \ref{dabc1discretecharacterization.tm} follows essentially from Theorems \ref{dabc1discreteholes.tm} and \ref{dabc1algebra.tm}.
We examine  five cases to verify the sufficiency. For the case 1) $c_0<{\rm gcd}(c_1, a)$, we show that
$[c_0, {\rm gcd}(c_1, a))+{\rm gcd}(c_1, a) \Z$ is an invariant set under the transformation $R_{a, b, c}$
and it has empty intersection with  black holes of transformations $R_{a, b, c}$
and $\tilde R_{a, b, c}$. This together with the maximality
of the set ${\mathcal S}_{a, b,c}$ in Lemma \ref{dabc0rationalmax.prop} implies that ${\mathcal S}_{a, b, c}\ne \emptyset$.
Similarly for the case 2) $b-c_0<{\rm gcd}(a, c_1+b)$, we  verify that $[0,{\rm gcd}(a, c_1+b)-b+c_0) + {\rm gcd}(a, c_1+b)\Z$
is invariant under the transformation $R_{a, b, c}$
and it has empty intersection with  black holes of  transformations $R_{a, b, c}$ and $\tilde R_{a, b, c}$.
For the case 3), we start from putting marks at $h\Z$ and insert gaps
$\frac{b-a+|\delta|}{2} ({\rm sgn}(\delta+b/(2q))+[-1, 1])$
at marks located at $lmh+NhZ, 1\le l\le d_1+d_2+1$, and
$\frac{|\delta|}{2} ({\rm sgn}(\delta+b/(2q))+[-1, 1])$
at other marked locations,
where $N=d_1+d_2+d_3+d_4+2$, $h= (a-(d_1+d_2+1)(b-a))/N-|\delta|$ and $m=(Nc_1+(d_1+d_3+1)(b-a))/a$. We then show that the gaps just inserted
form a set that is invariant under the transformation $R_{a, b, c}$ and that contains
 black holes of the transformations $\tilde R_{a, b, c}$ and $R_{a, b, c}$.

\begin{proof}[Proof of Theorem \ref{dabc1discretecharacterization.tm}]

($\Longrightarrow$)\ Let $N_1,  N_2, \delta, \delta', h$ be as in Theorem \ref{dabc1discreteholes.tm}.

\bigskip
{\bf Case 1}:\ {\em  $\delta=c_0+a-b$ and $\delta'=0$.}

In this case, $N_2=0$  and
$(R_{a, b, c})^n([c_1, c_1+c_0)+a\Z), 0\le n\le N_1$, are mutually disjoint  gap with
$(R_{a, b, c})^{N_1}([c_1, c_1+c_0)+a\Z)=[0, c_0)+a\Z$ by  Theorem \ref{dabc1discreteholes.tm}. Therefore $N_1\ne 0$ as $c_1>0$.
Observe that
$$(R_{a, b, c})^n([c_1, c_1+c_0)+a\Z)=[c_1, c_1+c_0)+n(c_1-a)+a\Z, 0\le n\le N_1$$
because $-a<c_1-a<0$ and
$(R_{a, b, c})^n([c_1, c_1+c_0)+a\Z)\subset [c_0, a)+a\Z$ for all $0\le n\le N_1-1$.
Replacing $n$ by $N_1$  in the above  equality and recalling that
$(R_{a, b, c})^{N_1}([c_1, c_1+c_0)+a\Z)=[0, c_0)+a\Z$ gives
$$c_1+N_1(c_1-a)\in a\Z,$$
which implies that
\begin{equation}\label{dabc1discretecharacterization.tm.pf.case-1.eq2}
c_1=-N_1(c_1-a)+ka\end{equation}
for some integer $k$.
Write $c_1/b=r/q$ and let $m={\rm gcd}(p, r)$.
Then it follows from \eqref{dabc1discretecharacterization.tm.pf.case-1.eq2}
that $N_1+1\in  p\Z/m$.
This together with  mutual disjointness of the gaps  $[c_1, c_1+c_0)+n(c_1-a)+a\Z, 0\le n\le N_1$,
implies that
\begin{equation}\label{dabc1discretecharacterization.tm.pf.case-1.eq3} N_1+1=p/m,\end{equation}
 as otherwise $p/m\le N_1$ and  $[c_1, c_1+c_0)+a\Z=[c_1, c_1+c_0)+(p/m)(c_1-a)+a\Z$.
Observe that
  $$\cup_{n=0}^{N_1} (n(c_1-a)+a\Z)=\cup_{n=0}^{p/m-1} (n(c_1-a)+a\Z)=\{0, mb/q, \ldots, (p-m)b/q\}+a\Z.$$
 Therefore the
 mutual disjointness of the gaps  $[c_1, c_1+c_0)+n(c_1-a)+a\Z, 0\le n\le N_1$,
 becomes $c_0\le mb/q={\rm gcd}(c_1, a)$.
We notice that $\cup_{n=0}^{N_1}[c_1, c_1+c_0)+n(c_1-a)+a\Z=\cup_{n=0}^{p/m-1} ([0,mb/q)+mnb/q+a\Z)=\R$ if
$c_0=mb/q$. This proves the desired first condition $c_0<{\rm gcd}(c_1, a)$  in Theorem
\ref{dabc1discretecharacterization.tm}.

\bigskip
{\bf Case 2}:\ {\em  $\delta=0$ and $\delta'=c_0-a$.}

In this case, $N_2=0$  and
$(R_{a, b, c})^n([c_1+c_0-a, c_1+b-a)+a\Z), 0\le n\le N_1$, are mutually disjoint gap with
$(R_{a, b, c})^{N_1}([c_1+c_0-a, c_1+b-a)+a\Z)=[c_0+a-b, a)+a\Z$ by  Theorem \ref{dabc1discreteholes.tm}.
Therefore $N_1\ge 1$ as $c_1<2a-b$. 
Observe that
$$(R_{a, b, c})^n([c_1+c_0-a, c_1+b-a)+a\Z)=[c_1+c_0-a, c_1+b-a)+n(c_1+b-a)+a\Z$$
for all $0\le n\le N_1$,
because $0<c_1+b-a<a$ and
$(R_{a, b, c})^n([c_1+c_0-a, c_1+b-a)+a\Z)\subset [0, c_0+a-b)+a\Z$ for all $0\le n\le N_1-1$.
Replacing $n$ by $N_1$  in the above  equality and recalling that
$(R_{a, b, c})^{N_1}([c_1+c_0-a, c_1+b-a)+a\Z)=[c_0+a-b, a)+a\Z$ give
\begin{equation}\label{dabc1discretecharacterization.tm.pf.case0.eq2}
(N_1+1)(c_1+b-a)\in a\Z.\end{equation}
This together with mutual disjointness of $[c_1+c_0-a, c_1+b-a)+n(c_1+b-a)+a\Z, 0\le n\le N_1$, implies that
$$N_1+1=a/{\rm gcd}(c_1+b, a),$$
as otherwise $ a/{\rm gcd}(c_1+b, a)\le N_1$ and
$[c_1+c_0-a, c_1+b-a)+a\Z=[c_1+c_0-a, c_1+b-a)+ (a/{\rm gcd}(c_1+b, a)) (c_1+b-a)+a\Z$, which is a contradiction.
Therefore mutual disjointness of the gaps
$(R_{a, b, c})^n([c_1+c_0-a, c_1+b-a)+a\Z), 0\le n\le N_1$,
becomes  mutual disjointness of the gaps
$[c_1+c_0-a, c_1+b-a)+i {\rm gcd}(c_1+b, a)+a\Z, 0\le i\le a/{\rm gcd}(c_1+b, a)-1$,
which holds if and only if
$b-c_0\le {\rm gcd}(c_1+b, a)$.
Also we notice that
\begin{eqnarray*}  & & \cup_{n=0}^{N_1} (R_{a, b, c})^n([c_1+c_0-a, c_1+b-a)+a\Z)\\
& = &
\cup_{i=0}^{a/{\rm gcd}(c_1+b, a)-1}([c_1+c_0-a, c_1+b-a)+i {\rm gcd}(c_1+b, a)+a\Z)\\
&= & c_1+b-a+[-{\rm gcd}(c_1+b, a), 0)+ {\rm gcd}(c_1+b, a)\Z=\R
\end{eqnarray*}
if $b-c_0={\rm gcd}(c_1+b, a)$, which contradicts to ${\mathcal S}_{a, b, c}\ne \emptyset$.
This leads to the desired second condition $b-c_0< {\rm gcd}(c_1+b, a)$ in Theorem \ref{dabc1discretecharacterization.tm}.

\bigskip
{\bf Case 3}:\ {\em  $0<\delta<c_0+a-b$ and $\delta'=0$.}

By Theorem \ref{dabc1discreteholes.tm}, $N_1\ge 0$ and $N_2\ge 1$.
Denote by $d_1, d_2$
 the number of  big gaps $(R_{a, b, c})^n(c-c_0+[0, b-a+\delta)), 0\le n\le N_1-1$,
of length $b-a+\delta$ contained in $[0, c_0+a-b-\delta)+a\Z$ and in
$[c_0, a)+a\Z$ respectively, and similarly denote  by $d_3$ and $d_4$
the number of small gaps
 $(R_{a, b, c})^m([c_0+\lfloor c/b\rfloor b-\delta, c_0+\lfloor c/b\rfloor b)), 0\le m\le N_2-1$,
of length $\delta$ contained in $[\delta, c_0+a-b-\delta)+a\Z$  and in
$[c_0, a)+a\Z$ respectively.
 Now let us verify that
\eqref{cnecessarydiscrete.eq1}--\eqref{cnecessarydiscrete.eq7} hold for the above nonnegative integer parameters $d_1, d_2, d_3, d_4$.

\smallskip

{\em Proof of \eqref{cnecessarydiscrete.eq1}}. \quad
By Theorem \ref{dabc1discreteholes.tm},  for every $0\le
n\le N_1-1$, the gap $(R_{a, b, c})^n(c-c_0+[0, b-a+\delta)+a\Z)$
is either contained in $[0, c_0+a-b)+a\Z$ or $[c_0, a)+a\Z$. Hence
\begin{equation}\label{dabc1discretecharacterization.lem.pf.eq1} N_1=d_1+d_2.
\end{equation}
 Similarly
\begin{equation}\label{dabc1discretecharacterization.lem.pf.eq2}  N_2-1=d_3+d_4
\end{equation}
as
for any  $0\le m\le N_2-1$, the gap
 $(R_{a, b, c})^m([c_0+\lfloor c/b\rfloor b-\delta, c_0+\lfloor c/b\rfloor b))=
(R_{a, b, c})^{m+1}([c_0+a-b-\delta, c_0)\backslash [c_0+a-b, c_0)+a\Z)$
 of length $\delta$ is contained either in  $[0, c_0+a-b-\delta)+a\Z$ and in
 $[c_0, a)+a\Z$.
Combining \eqref{dabc1discreteholes.tm.eq1}, \eqref{dabc1discreteholes.tm.eq8},
\eqref{dabc1discreteholes.tm.eq9},
 \eqref{dabc1discretecharacterization.lem.pf.eq1},    and
\eqref{dabc1discretecharacterization.lem.pf.eq2},
we obtain that there are $(d_1+d_2+1)$ gaps of length $b-a+\delta$ and
$(d_3+d_4+1)$ gaps of length $\delta$, and
$N:=(d_1+d_2+d_3+d_4+2)$ intervals of length $h$ on one period $[0, a)$. Therefore
\begin{equation}\label{deltb+h} 0<a-(d_1+d_2+1)(b-a) =  N (h+\delta) \in  N b\Z/q.
\end{equation}
This proves    \eqref{cnecessarydiscrete.eq1}.

 {\em Proof of \eqref{cnecessarydiscrete.eq2}}. \quad  By
\eqref{dabc1discreteholes.tm.eq3},
\eqref{dabc1discreteholes.tm.eq5} and the definition of
nonnegative integers $d_i, 1\le i\le 4$, we obtain that
\begin{equation}\label{dabc1discretecharacterization.lem.pf.eq3}
c-c_0+b-a+\delta+ d_1 (\lfloor c/b\rfloor+1)b+d_2 \lfloor
c/b\rfloor b\in c_0+ a\Z,
\end{equation}
and
\begin{equation}\label{dabc1discretecharacterization.lem.pf.eq4}
c_0+\lfloor c/b\rfloor b -\delta + d_3 (\lfloor
c/b\rfloor+1)b+d_4\lfloor c/b\rfloor b\in a\Z.
\end{equation}
Adding \eqref{dabc1discretecharacterization.lem.pf.eq3} and
\eqref{dabc1discretecharacterization.lem.pf.eq4} leads to
 $c+(d_1+d_3+1) (\lfloor c/b\rfloor+1)b+
(d_2+d_4) \lfloor c/b\rfloor  b\in c_0+a\Z$. Then
 $Nc_1+(d_1+d_3+1)(b-a)\in a\Z$ and
 \eqref{cnecessarydiscrete.eq2} is true.

{\em Proof of \eqref{cnecessarydiscrete.eq5}}.\quad
 By Theorem \ref{dabc1discreteholes.tm} and the definition of the integers $d_1$ and $d_3$,
the interval $[0, c_0+a-b-\delta)$  is covered by $d_1$ gaps of length $b-a+\delta$,
$d_3+1$ gaps of length $\delta$, and $d_1+d_3+1$ intervals of
length $h$. This together with \eqref{deltb+h} leads to
\begin{eqnarray} \label{dabc1discretecharacterization.lem.pf.eq5}
c_0+a-b-\delta
&= & d_1(b-a+\delta)+ (d_3+1)\delta+ (d_1+d_3+1) h\nonumber\\
&= & d_1(b-a)+(d_1+d_3+1)B_d/N.
\end{eqnarray}
This proves  \eqref{cnecessarydiscrete.eq5}.

 {\em Proof of
\eqref{cnecessarydiscrete.eq3}}. \quad Substituting the expression  in \eqref{dabc1discretecharacterization.lem.pf.eq5} into
\eqref{dabc1discretecharacterization.lem.pf.eq3}, we obtain that
\begin{eqnarray}
a\Z &\ni & c-c_0+b-a+\delta-c_0+ d_1 (\lfloor c/b\rfloor+1)b+d_2 \lfloor c/b\rfloor b-d_1a\nonumber\\
& = & d_1 (\lfloor c/b\rfloor+1)b+d_2 \lfloor c/b\rfloor b+\lfloor
c/b\rfloor
b+b-(d_1+1)a\nonumber\\
& & -(d_1+1)(b-a)-(d_1+d_3+1)B_d/N
\nonumber\\
& = & (d_1+d_2+1) \lfloor c/b\rfloor b-(d_1+d_3+1) B_d/N.
\end{eqnarray}
Multiplying $N$ at both sides of the above equation leads to the
desired inclusion \eqref{cnecessarydiscrete.eq3}.

{\em Proof of \eqref{cnecessarydiscrete.eq6}}.\quad  Define
\begin{equation}\label{cnecessarydiscrete.eq6.pf.case1.eq1} m=(N c_1+(d_1+d_3+1)(b-a))/a. \end{equation}
Then $m$ is a positive integer in $[1, N-1]$ by
\eqref{cnecessarydiscrete.eq2} and
the observation that
 $0<N c_1+(d_1+d_3+1)(b-a)< N (2a-b)+(d_1+d_3+1)(b-a)=Na-(d_2+d_4+1)(b-a)
\le Na$. Let $Y_{a, b, c}$ be as in Theorem
\ref{dabc1algebra.tm} and let $m_1$ be the nonnegative
integer in $[0, N-1]$ such that $Y_{a, b, c}(c_1+b-a)\in
m_1h+ Y_{a, b, c}(a) \Z$. We claim the following:
\begin{equation}\label{cnecessarydiscrete.eq6.pf.case1.eq3} m_1=m.
\end{equation}
 Recall that $(R_{a, b, c})^{N_1}([c_1,
c_1+b-a+\delta)+a\Z)=[c_0+a-b-\delta, c_0)+a\Z$ by Theorem
\ref{dabc1discreteholes.tm}, and that there are $d_1+d_3+1$ gaps in
the interval $[0, c_0+a-b-\delta)$. This together with Theorem
\ref{dabc1algebra.tm} that
\begin{equation} \label{cnecessarydiscrete.eq6.pf.case1.eq4}
(d_1+d_2+1) m_1h -(d_1+d_3+1)h\in Y_{a, b, c}(a) \Z=Nh \Z.
\end{equation}
Then the number of gaps of length $b-a+\delta$ contained $[0, c_1)$
is $((d_1+d_2+1) m_1h -(d_1+d_3+1)h)/Y_{a, b, c}(a)= ((d_1+d_2+1)
m_1-(d_1+d_3+1))/N$. This implies that  there are $m_1$ gaps
 contained in $[0, c_1)$ with $((d_1+d_2+1)
m_1-(d_1+d_3+1))/N$ of them are gaps of  length $b-a+\delta$.
Hence
\begin{eqnarray*}
c_1 &= & m_1 h+ \big(m_1-\frac{(d_1+d_2+1) m_1-(d_1+d_3+1)}{N}\big) \delta\\
& & +\frac{(d_1+d_2+1) m_1-(d_1+d_3+1)}{N}(b-a+\delta)\nonumber\\
& = &
(m_1 a-(d_1+d_3+1)(b-a))/N.
\end{eqnarray*}
This together with \eqref{cnecessarydiscrete.eq6.pf.case1.eq1}
proves \eqref{cnecessarydiscrete.eq6.pf.case1.eq3}.

We return to the proof of \eqref{cnecessarydiscrete.eq6}.
The above claim \eqref{cnecessarydiscrete.eq6.pf.case1.eq3}, together with
 \eqref{dabc1discretecharacterization.lem.pf.eq1},
\eqref{dabc1discretecharacterization.lem.pf.eq2},  and the observation
that ${\mathcal K}_{a, b, c}$ is a cyclic group generated by $Y_{a, b, c}(c_1+b-a)$  and  has order $N_1+N_2+1=N$
by Theorem \ref{dabc1algebra.tm}  and
Remark \ref{dabc1discretealgebraic.tm.rem},  proves   \eqref{cnecessarydiscrete.eq6}.

{\em Proof of \eqref{cnecessarydiscrete.eq7}}.\quad  Applying
\eqref{cnecessarydiscrete.eq6.pf.case1.eq1} and
\eqref{cnecessarydiscrete.eq6.pf.case1.eq3}, and recalling that $[c_1, c_1+b-a)$ is a hole in the complement
of the maximal invariant set ${\mathcal S}_{a, b, c}$,
we have that
\begin{equation}
Y_{a, b, c}(c_1)-mh\in N h\Z.
\end{equation}
Then $n\in E_{a, b, c}^d$ if and only if $(R_{a, b, c})^n [c_1,
c_1+b-a+\delta)$ is a big gap contained in $[0, c_0+a-b)+a\Z$. This
implies that the cardinality of  the set $E_{a, b, c}^d$ is equal to $d_1$
from the definition of the nonnegative integer $d_1$.

This completes the proof of the necessity for the case that
$\delta\in (0, c_0+a-b)$ and $\delta'=0$.

\bigskip
{\bf Case 4}:\ {\em  $\delta'\in (c_0-a, 0)$ and $\delta=0$.}

By Theorem \ref{dabc1discreteholes.tm}, $N_1\ge 0$ and $N_2\ge 1$.
Denote by $d_1, d_2$
 the numbers of  big gaps $(R_{a, b, c})^n(c-c_0+[\delta', b-a)), 0\le n\le N_1-1$,
of length $b-a-\delta'$ contained in $[0, c_0+a-b)+a\Z$ and in
$[c_0-\delta', a+\delta')+a\Z$ respectively, and similarly denote  by $d_3$ and $d_4$
the numbers of small gaps
 $(R_{a, b, c})^m([c_0+\lfloor c/b\rfloor b, c_0+\lfloor c/b\rfloor b-\delta')), 0\le m\le N_2-1$,
of length $-\delta'$ contained in $[0, c_0+a-b)+a\Z$  and in
$[c_0-\delta', a+\delta')+a\Z$ respectively.
 We may
follow the argument for the first case 
 and prove the desired properties
\eqref{cnecessarydiscrete.eq1}--\eqref{cnecessarydiscrete.eq7}
with the above nonnegative integers $d_1, d_2, d_3$ and $d_4$.

\smallskip

{\bf Case 5}:\ $\delta=\delta'=0$.

By Theorem \ref{dabc1discreteholes.tm}, $N_1\ge 0$ and $N_2\ge 1$.
Denote by $d_1, d_2$
 the numbers of  gaps $(R_{a, b, c})^n(c-c_0+[0, b-a)), 0\le n\le N_1-1$,
of length $b-a$ contained in $[0, c_0+a-b)+a\Z$ and in
$[c_0, a)+a\Z$ respectively, and similarly denote  by $d_3$ and $d_4$
the numbers of   $(R_{a, b, c})^m(c_0), 1\le m\le N_2$, contained in
$(0, c_0+a-b)+a\Z$ or $[c_0, a)+a\Z$  respectively.
Then     we may follow
the  argument for the first case line by line
and establish
\eqref{cnecessarydiscrete.eq1}--\eqref{cnecessarydiscrete.eq7}
with the above nonnegative integer parameters $d_1, d_2, d_3$ and $d_4$.

This completes the proof of the necessity for the case that
$\delta=\delta'=0$ and also the proof of the necessity.

\bigskip

($\Longleftarrow$)\ We examine five cases to prove the sufficiency.

{\bf Case 1}: \ $c_0<{\rm gcd}(c_1, a)$.

Let $N+1=a/{\rm gcd}(c_1, a)$ for some nonnegative integer $N$ and define
$T=\big(\cup_{n=0}^N [c_0, {\rm gcd}(c_1, a))+n(a-c_1)\big)+a\Z$.
Then
\begin{eqnarray}
T&=&\big(\cup_{i=0}^N [c_0, {\rm gcd}(c_1, a))+i {\rm gcd}(c_1, a)\big)+a\Z\nonumber\\
&= &
[c_0, {\rm gcd}(c_1,a))+ {\rm gcd}(c_1,a)\Z,
\end{eqnarray}
and
$T$ has empty intersection with  black holes of
the transformations $R_{a, b, c}$ and $\tilde R_{a, b, c}$, since
$T\cap [0, c_0)=T\cap [c_1, c_0+c_1)=\emptyset$,
and for any $t\in T$,
$$
R_{a, b, c}(t)  =  t+c_1 \in [c_0, {\rm gcd}(c_1, a))+ c_1+ {\rm gcd}(c_1, a)\Z =T.
$$
Therefore $T\subset {\mathcal S}_{a, b, c}$ (in fact $T={\mathcal S}_{a, b, c}$)
as ${\mathcal S}_{a, b, c}$ is the maximal invariant set that has empty intersection with the black hole of
the transformation $R_{a, b, c}$  by Lemma \ref{dabc0rationalmax.prop}.
Thus ${\mathcal S}_{a, b, c}$ is not an empty set as
the restriction of the set $T$ on $[0, a)$ consists of $a/{\rm gcd}(c_1, a)$ intervals of length ${\rm gcd}(c_1, a)-c_0>0$.

{\bf Case 2}: \ $b-c_0<{\rm gcd}(a, c_1+b)$.

Write $a/{\rm gcd}(a, c_1+b)=N+1$ for some nonnegative integer $N$ and define
$T'=\big(\cup_{i=0}^N [0, {\rm gcd}(a, c_1+b)-b+c_0)+i(c_1+b-a)\big)+a\Z=
[0, {\rm gcd}(a, c_1+b)-b+c_0)+ {\rm gcd}(a, c_1+b)\Z$. Then
$T'$ has empty intersection with  black holes of
the transformations $R_{a, b, c}$ and $\tilde R_{a, b, c}$, since
$T'\cap [c_1, c_1+b-a)=T'\cap [c_0+a-b, a)=\emptyset$,
and for any $t\in T$,
\begin{eqnarray*}
R_{a, b, c}(t) & = &  t+c_1+b 
 \in
([0, {\rm gcd}(a, c_1+b)-b+c_0)\\
& & \quad + c_1+b+{\rm gcd}(a, c_1+b) \Z=T'.
\end{eqnarray*}
Therefore $T'\subset {\mathcal S}_{a, b, c}$     (in fact $T'={\mathcal S}_{a, b, c}$)
as ${\mathcal S}_{a, b, c}$ is the maximal invariant set that has empty intersection with the black holes of
the transformations $R_{a, b, c}$ and $\tilde R_{a, b, c}$ by Lemma \ref{dabc0rationalmax.prop}.
Thus ${\mathcal S}_{a, b, c}$ is not an empty set as
the restriction of the set $T'$ on $[0, a)$ consists of $N+1$ intervals of length ${\rm gcd}(a, c_1+b)-b+c_0>0$.

{\bf Case 3}: \ There exist nonnegative integers $d_1, d_2, d_3, d_4$
and $\delta\in (0, \min(B_d/N, c_0+a-b))$ satisfying
\eqref{cnecessarydiscrete.eq1}--\eqref{cnecessarydiscrete.eq7}.

In this case, we  set
\begin{equation}\label{hdef}
h=(a-(d_1+d_2+1)(b-a))/{N}-\delta,\end{equation}
\begin{equation} \label{mdef}
m=( N c_1+(d_1+d_3+1)(b-a))/a,
\end{equation}
and
\begin{equation}\label{tildemdef}
\tilde m= ((d_1+d_2+1) m-(d_1+d_3+1))/N.
\end{equation}
 Then
\begin{equation} \label{hproperty1}
0<h\in b\Z/q
\end{equation}
by \eqref{cnecessarydiscrete.eq1} and
\eqref{cnecessarydiscrete.eq5}; $m$ is a positive integer no
larger than $N-1$, i.e.,
\begin{equation}\label{mproperty}
m\in \Z\cap [1, N-1]
\end{equation}
as $ 0<Nc_1/a\le m<( N (2a-b)+(d_1+d_3+1)(b-a))/a 
< N$; and
$\tilde m$ is a nonnegative integer no larger than $m$,
\begin{equation}\label{tildemproperty} \tilde m\in [0, m]\cap
\Z\end{equation} by \eqref{cnecessarydiscrete.eq3}.
 Moreover,
\begin{eqnarray}\label{c1property} & & m\frac{a-(d_1+d_2+1)(b-a)}{N}+ \tilde m
(b-a)\nonumber \\
& = & \frac{m}{N} a -\frac{d_1+d_3+1}{N}(b-a)=c_1.
\end{eqnarray}
 In order to expand the real line $\R$ with marks at $h\Z$ to create an
invariant set under the transformation $R_{a,b,c}$,
 we insert gaps $[0, b-a+\delta)$ located at $lmh+Nh\Z, 1\le
l\le d_1+d_2+1$, and gaps $[0, \delta)$ otherwise. Recall from
\eqref{cnecessarydiscrete.eq6} that $(l-l')mh\not\in Nh\Z$ for all
$1\le l\ne l'\le N$. Therefore
 we  have inserted $d_1+d_2+1$ gaps $[0, b-a+\delta)$ and
$d_3+d_4+1$ gaps $[0, \delta)$ on the interval $[0, Nh)$. Thus
after performing the above expansion,  the interval $[0, Nh)$ with marks on $[0, Nh)\cap h\Z$ becomes the interval $[0,
Nh+(d_1+d_2+1)(b-a+\delta)+(d_3+d_4+1)\delta)=[0, a)$ with gaps  $[y_i, y_i+h_i),0 \le i\le N-1$, where $0=y_0\le y_1 \le \ldots
\le y_{D}$ and $h_i\in \{b-a+\delta, \delta\}, 0\le i\le N-1$. Now we
want to prove that
\begin{equation}\label{ym=c1}
y_m=c_1.
\end{equation}
For that purpose, we need the following claim:

\begin{claim}\label{dabc1discretecharacterization.lem.pf.claim1}
 For $s\in [0, N-1]\cap \Z$, the cardinality of the set $\{l\in [1, d_1+d_2+1]\cap \Z|\ lmh\in sh+[0, mh)+Nh\Z\}$ is equal to
$\tilde m+1$ if $1\le s\le d_1+d_3+1$, and  $\tilde m$ otherwise.
\end{claim}

\begin{proof}
For any $i\in \Z$, let $k_i=\lfloor (iN+m+s-1)/m\rfloor$
the unique integer such that $k_i mh\in [sh, sh+mh)+ iNh$.
Therefore $1\le k_i\le d_1+d_2+1$ if and only if
$m\le iN+m+s-1\le (d_1+d_2+1)m+m-1$ if and only if
 $1-s\le iN\le (d_1+d_2+1)m-s= N\tilde m+ (d_1+d_3+1-s)$.
Therefore
\begin{eqnarray*}
& & \#\{l\in [1, d_1+d_2+1]\cap \Z|\ lmh\in sh+[0, mh)+Nh\Z\}
\\
& = & \sum_{i\in \Z} \#\{l\in [1, d_1+d_2+1]\cap \Z|\ lmh\in sh+[0, mh)+iN h\}\\
& = & \sum_{i\in \Z} \#\{k_i\in [1, d_1+d_2+1]\cap \Z\}\\
& = & \#([(1-s)/N, \tilde m+(d_1+d_3+1-s)/N]\cap \Z).
 \end{eqnarray*}
Counting the number of integers in the interval $[(1-s)/N, \tilde m+(d_1+d_3+1-s)/N]$ 
 proves the claim.
%
%
%
\end{proof}

We return to the proof of the equality \eqref{ym=c1}. By Claim \ref{dabc1discretecharacterization.lem.pf.claim1},
we have inserted $\tilde m$ interval of length $b-a+\delta$
and $m-\tilde m$ interval of length $\delta$ in the marked interval
$[0, mh)$. So after performing the expansion,  the mark located at $mh$ on
the line becomes the gap located at $mh+(m-\tilde m)\delta+\tilde
m (b-a+\delta)$, which is equal to $c_1$ by \eqref{c1property}.
This completes the proof of the equality \eqref{ym=c1}.

\smallskip
Next we show that
\begin{equation}\label{yd1d2=c0}
y_{d_1+d_3+1}= c_0+a-b-\delta.
\end{equation}
By \eqref{cnecessarydiscrete.eq7}, we have inserted $d_1$ gaps of
length $b-a+\delta$ and $(d_1+d_3+1)-d_1$ intervals of length
$\delta$ in the marked interval $[0, (d_1+d_3+1)h)$. Therefore the mark
located at $(d_1+d_3+1)h$ becomes
$(d_1+d_3+1)h+d_1(b-a+\delta)+(d_3+1)\delta=c_0+a-b-\delta$ after
inserting gaps, where the last equality follows from
\eqref{cnecessarydiscrete.eq5}. Hence
 \eqref{yd1d2=c0} follows.

\smallskip
Then we  prove by induction on $0\le k\le N-1$ that
\begin{equation}\label{biggap.eqk}
(R_{a, b, c})^{k} (c-c_0)+a\Z=y_{l(k)}+a\Z, \ 0\le k\le d_1+d_2,
\end{equation}
and
\begin{equation}\label{smallgap.eqk}
(R_{a, b, c})^{m} (c_0+a-b-\delta)+a\Z=y_{l(m+d_1+d_2)}+a\Z,\
1\le m\le d_3+d_4+1,
\end{equation}
where $l(k)\in (k+1)m+N\Z$.
We remark that $l(0)=m, l(d_1+d_2+d_3+d_4+1)=l(N-1)=0$ and $l(d_1+d_2)=d_1+d_3+1$, where the last equality follows from \eqref{cnecessarydiscrete.eq3}.

\begin{proof}[Proof of \eqref{biggap.eqk} and \eqref{smallgap.eqk}]
The conclusion \eqref{biggap.eqk} for $k=0$ from \eqref{ym=c1} and
the observation that $l(0)=m$. Inductively, we assume that the
conclusion \eqref{biggap.eqk} holds for some $0\le k\le
d_1+d_2-1$.  Then $l(k)\ne 0, d_1+d_3+1$ as $l(d_1+d_2)=d_1+d_3+1$
and $l(N-1)=0$.  If $0<l(k)<d_1+d_3+1$, then
\begin{eqnarray}\label{biggap.eqk.pf.eq2}
y_{l(k+1)} & = &
 y_{l(k)}+ (\tilde m+1)(b-a+\delta)+ (m-\tilde m-1) \delta+ mh \nonumber \\
& = &  y_{l(k)}+c_1+b-a
\end{eqnarray}
if $l(k+1)-l(k)=m$,
and
\begin{eqnarray} \label{biggap.eqk.pf.eq3}
y_{l(k+1)}+a & = & y_{l(k)}+ (\tilde m+1)(b-a+\delta)+ (m-\tilde m-1) \delta+ mh \nonumber \\
& = &  y_{l(k)}+c_1+b-a
\end{eqnarray}
if $l(k+1)-l(k)=m-N$, where \eqref{biggap.eqk.pf.eq2} and
\eqref{biggap.eqk.pf.eq3} hold
as we have inserted $\tilde m+1$ gaps of size $b-a+\delta$
and $m-(\tilde m+1)$ gaps of size $\delta$ on $[l(k)h, (l(k)+m)h)$ by
Claim \ref{dabc1discretecharacterization.lem.pf.claim1}.
Also we obtain from \eqref{yd1d2=c0} that  $y_{l(k)}\in [0, c_0+a-b-\delta)$ when $0<l(k)<d_1+d_3+1$,
which together with the inductive hypothesis implies that
\begin{eqnarray}\label{biggap.eqk.pf.eq1}
(R_{a, b, c})^{k+1} (c-c_0)+a\Z & = & R_{a, b ,c} (y_{l(k)})+a\Z\nonumber\\
&= &
y_{l(k)}+c_1+b-a+a\Z.
\end{eqnarray}
Combining  \eqref{biggap.eqk.pf.eq2}, \eqref{biggap.eqk.pf.eq3} and \eqref{biggap.eqk.pf.eq1} leads to
\begin{equation}\label{biggap.eqk.pf.eq4}
(R_{a, b, c})^{k+1} (c-c_0)+a\Z = y_{l(k+1)}+a\Z
\end{equation}
if $0<l(k)<d_1+d_3+1$.

Similarly if $d_1+d_3+1<l(k)\le N-1$,  we have that
\begin{eqnarray}\label{biggap.eqk.pf.eq5}
y_{l(k+1)}-y_{l(k)}\in  c_1+a\Z
\end{eqnarray}
because  we have inserted $\tilde m$ gaps of size $b-a+\delta$
and $m-\tilde m$ gaps of size $\delta$ on $[l(k)h, (l(k)+m)h)$ by
Claim \ref{dabc1discretecharacterization.lem.pf.claim1}; and
\begin{equation}\label{biggap.eqk.pf.eq6}
(R_{a, b, c})^{k+1} (c-c_0)+a\Z
=
y_{l(k)}+c_1+a\Z,
\end{equation}
since $y_{l(k)}\in [c_0, a)$ by \eqref{yd1d2=c0}.
Combining \eqref{biggap.eqk.pf.eq5} and \eqref{biggap.eqk.pf.eq6} yields
\begin{equation}\label{biggap.eqk.pf.eq7}
(R_{a, b, c})^{k+1} (c-c_0)+a\Z = y_{l(k+1)}+a\Z
\end{equation}
if $d_1+d_3+1<l(k)\le N-1$. Therefore we can proceed our inductive proof
by \eqref{biggap.eqk.pf.eq4} and \eqref{biggap.eqk.pf.eq7}. This completes the proof of the equalities in
\eqref{biggap.eqk}.

Notice that
$y_{l(d_1+d_2)}= y_{d_1+d_3+1}= c_0+a-b-\delta$ by \eqref{yd1d2=c0}. Hence
\begin{eqnarray}\label{kd1d2.eq}
 & \!&\! (R_{a, b, c})^{m} (c_0+a-b-\delta)+a\Z= (R_{a, b, c})^{m}(y_{l(d_1+d_2)})+a\Z\nonumber\\
\quad &=\! &\!
(R_{a, b, c})^{m+d_1+d_2}(y_{l(0)})+a\Z=  (R_{a, b, c})^{m+d_1+d_2}(c-c_0)+a\Z\end{eqnarray}
for all $1\le m\le d_3+d_4+1$. Then we can follow the argument to prove \eqref{biggap.eqk}
to  show that \eqref{smallgap.eqk} holds.
\end{proof}

Finally from \eqref{biggap.eqk} and \eqref{smallgap.eqk}
the mutually disjoint gaps we have inserted are
$(R_{a, b, c})^k(c-c_0)+[0, b-a+\delta)+a\Z, 0\le k\le d_1+d_2$, and
$(R_{a, b, c})^m(c_0+a-b-\delta)+[0, \delta)+a\Z, 1\le m\le d_3+d_4+1$.
 Moreover
$(R_{a, b, c})^{d_1+d_2}(c-c_0)+[0, b-a+\delta)+a\Z=[c_0+a-b-\delta, c_0)+a\Z$
by \eqref{yd1d2=c0} and
$l(d_1+d_2)=d_1+d_3+1$; and
$(R_{a, b, c})^{d_3+d_4+1}(c_0+a-b-\delta)+[0, \delta)+a\Z=(R_{a, b, c})^{N-1}(c_0+a-b-\delta)+[0, \delta)+a\Z
=[0, \delta)+a\Z$
by \eqref{kd1d2.eq} and $l(N-1)=0$.
Notice that the union of the above gaps is invariant under the transformation
$R_{a, b, c}$
and contains the black holes of the transformations  $R_{a, b, c}$ and $\tilde R_{a, b, c}$. Therefore
its complement is the set ${\mathcal S}_{a, b ,c}$ by Lemma \ref{dabc0rationalmax.prop}, whose restriction on $[0, a)$
has Lebesgue measure $Nh$. Thus the conclusion that ${\mathcal S}_{a, b, c}\ne \emptyset$  is established for this case.

\bigskip

{\bf Case 4}: \ There exist nonnegative integers $d_1, d_2, d_3, d_4$
and $\delta\in (-\min(B_d/N, a-c_0), 0)$ satisfying
\eqref{cnecessarydiscrete.eq1}--\eqref{cnecessarydiscrete.eq7}.

  Set
$h=(a-(d_1+d_2+1)(b-a))/{N}+\delta$ and $m=( N c_1+(d_1+d_3+1)(b-a))/a$. We expand
the real line $\R$ with marks at $h\Z$  by inserting gaps $[\delta+a-b, 0)$ located at $lmh+Nh\Z, 1\le
l\le d_1+d_2+1$, and gaps $[\delta, 0)$ otherwise.
After performing the above augmentation operation,  the interval $[0, Nh)$
with marks $[0, Nh)\cap h\Z$ becomes the interval 
$[0, a)$ with gaps   $[y_i+h_i, y_i),0 \le i\le N_1$, where $0< y_1 \le \ldots
\le y_{N}=a$ and $h_i\in \{\delta+a-b, \delta\}, 1\le i\le N$.
We follow the argument used in  Case 3 to show that $y_m=c_1+b-a, y_{d_1+d_3+1}=c_0-\delta$
and
for  $0\le k\le N-1$,
\begin{equation}\label{biggap.eqk.case4}
(R_{a, b, c})^{k} (c-c_0+b)+a\Z=y_{l(k)}+a\Z, \ 0\le k\le d_1+d_2,
\end{equation}
and
\begin{equation}\label{smallgap.eqk.case4}
(R_{a, b, c})^{m} (c_0-\delta)+a\Z=y_{l(k)}+a\Z,\
1\le m\le d_3+d_4+1,
\end{equation}
where $l(k)\in (k+1)m+N\Z$.
Therefore
 ${\mathcal S}_{a, b ,c}$ is the complement of  $\big(\cup_{n=0}^{d_1+d_2} [y_{l(k)}+a-b+\delta, y_{l(k)})+a\Z\big)\cup
\big(\cup_{m=1}^{d_3+d_4+1}[y_{l(k)}+\delta, y_{l(k)})+a\Z\big)$,  whose restriction on $[0, a)$ has Lebesgue measure $Nh>0$.
This prove the sufficiency   for Case 4.

\bigskip

{\bf Case 5}: \ There exist nonnegative integers $d_1, d_2, d_3, d_4$
and $\delta=0$ satisfying
\eqref{cnecessarydiscrete.eq1}--\eqref{cnecessarydiscrete.eq7}.

In this case, we set $h=(a-(d_1+d_2+1)(b-a))/{N}$ and $m=( N c_1+(d_1+d_3+1)(b-a))/a$, and
 expand the real line $\R$ with marks at $h\Z$  by
 inserting gaps $[0, b-a)$ located at $lmh+Nh\Z, 1\le
l\le d_1+d_2+1$, and gaps of zero length otherwise, c.f. the fourth subfigure of Figure \ref{holesremoval.fig} in Section \ref{maintheorem.section}. Then
after performing the above operation, 
the interval $[0, Nh)$ becomes the interval $[0, a)$ with gaps $[y_i, y_i+h_i),0 \le i\le N-1$, where $0=y_0\le y_1 \le \ldots
\le y_{N-1}$ and $h_i\in \{b-a, 0\}, 0\le i\le N-1$.
We follow the argument in Case 3
to show that $y_m=c_1,
y_{d_1+d_3}= c_0+a-b$ and  by induction on $0\le k\le N-1$ that
\begin{equation}\label{biggap.eqk.case5}
(R_{a, b, c})^{k} (c-c_0)+a\Z=y_{l(k)}+a\Z, \ 0\le k\le d_1+d_2,
\end{equation}
and
\begin{equation}\label{smallgap.eqk.case5}
(R_{a, b, c})^{m} (c_0)+a\Z=y_{l(m+d_1+d_2)}+a\Z,\
1\le m\le d_3+d_4+1,
\end{equation}
where $l(k)\in (k+1)m+N\Z$. Thus the union of gaps of size $b-a$ is  $\cup_{n=0}^{d_1+d_2} (R_{a, b, c})^n ([c-c_0, c-c_0+b-a)+a\Z$ with
$(R_{a, b, c})^{d_1+d_2}( [c-c_0, c-c_0+b-a)+a\Z)=[c_0+a-b, c_0)+a\Z$. Therefore
${\mathcal S}_{a, b, c}$ is the complement of the above union of finite gaps and the sufficiency in the fifth case follows.
\end{proof}


\subsection{Proof of Theorem \ref{newmaintheorem5}}
%

 (XIII): \quad
By Theorem \ref{framenullspace1.tm}, ${\mathcal G}(\chi_{[0,c)}, a\Z\times \Z/b)$ is not a Gabor frame if and only if
${\mathcal D}_{a, b, c}\ne\emptyset$, which in turn becomes
${\mathcal S}_{a, b, c}\ne \emptyset$ and $(\lfloor c/b\rfloor +1) |{\mathcal S}_{a, b, c}\cap [0, c_0+a-b)|+
\lfloor c/b\rfloor  |{\mathcal S}_{a, b, c}\cap [c_0, a)|\ne a$ by
\eqref{dabc2subsetdabc1} and Theorem \ref{sabcstar.tm}.
For the case that the triple  $(a, b, c)$ satisfies the first condition in Theorem \ref{dabc1discretecharacterization.tm},
it follows from the argument used in the proof of Theorem \ref{dabc1discretecharacterization.tm}
 that ${\mathcal S}_{a, b, c}\cap [0, c_0+a-b)=\emptyset$ and
${\mathcal S}_{a, b, c}\cap [c_0, a)= \cup_{i=0}^N [c_0, {\rm gcd}(a, c_1))+i{\rm gcd}(a, c_1)$, where $N+1=a/{\rm gcd}(a,c_1)$.
Hence \eqref{sabcstar.tm.eq1} holds if and only if $(N+1) \lfloor c/b\rfloor ({\rm gcd}(a, c_1)-c_0)=a$ if and only if
$ \lfloor c/b\rfloor ({\rm gcd}(a, c_1)-c_0)= {\rm gcd}(a, c_1)$.
For the case that  the triple  $(a, b, c)$ satisfies the second condition in Theorem \ref{dabc1discretecharacterization.tm},
 ${\mathcal S}_{a, b, c}\cap [c_0, a)=\emptyset$ and
${\mathcal S}_{a, b, c}\cap [0, c_0+a-b)= \cup_{i=0}^N [0, {\rm gcd}(c_1+b, a)-b+c_0)+i{\rm gcd}(c_1+b, a)$, where $N=a/{\rm gcd}(c_1+b, a)-1$. Hence
 \eqref{sabcstar.tm.eq1} holds if and only  if $(N+1) (\lfloor c/b\rfloor+1) ({\rm gcd}(c_1+b, a)-b+c_0)=a$ if and only if
$ (\lfloor c/b\rfloor+1)( {\rm gcd}(c_1+b, a)-b+c_0) = {\rm gcd}(c_1+b, a)$.
For the case that the triple $(a, b, c)$ satisfies the third condition in Theorem
\ref{dabc1discretecharacterization.tm}, there are $d_1+d_3+1$ intervals of length $h$ contained in $[0, c_0+b-a)$
and $d_2+d_4+1$ intervals of length $h$ contained in $[c_0, a)$, where $h+|\delta|=B_d/N$. Therefore
\eqref{sabcstar.tm.eq1} holds if and only if
$(\lfloor c/b\rfloor +1) (d_1+d_3+1) h + \lfloor c/b\rfloor (d_2+d_4+1) h= a$ if and only if $h=a/(\lfloor c/b\rfloor N+(d_1+d_3+1))$
if and only if $a/(\lfloor c/b\rfloor N+(d_1+d_3+1))+|\delta|=B_d/N$.
Therefore the conclusion (XIII) holds by Theorem \ref{dabc1discretecharacterization.tm}.

\bigskip
(XIV):\quad This conclusion is given in \cite[Section 3.3.6.1]{janssen03}.

\begin{appendix}

\section{Algorithm}
\label{algorithm.appendix}

In this appendix, we provide a finite-step process to
verify whether the  Gabor system  ${\mathcal G}_(\chi_{[0,c)}, a\Z\times \Z/b)$
is a  Gabor frame for any given triple of $(a, b, c)$ of positive numbers.

By Theorems \ref{newmaintheorem1}--\ref{newmaintheorem5},
it suffices to consider triples satisfying  assumptions in  Conclusion (XII) of Theorem \ref{newmaintheorem4} and
in Conclusion (XIII)  of Theorem \ref{newmaintheorem5}.
For a triple $(a, b, c)$ of positive numbers
 satisfying  assumptions in Conclusion (XII) of Theorem \ref{newmaintheorem4}, 
we let  $A_0=[c_1, c_1+b-a)$ (the black hole of the piecewise linear transformation
   $\tilde R_{a, b, c}$ in \eqref{tilderabcnewplus.def})
 and $S_0=[0,a)\backslash A_0$, and define holes $A_k:=[\alpha_k, \beta_k)\subset [0, a)$ and  invariant sets $S_k\subset [0, a), k\ge 1$, iteratively
$$A_k=\left\{\begin{array}{l}  R_{a, b, c}(\alpha_{k-1})-\lfloor(R_{a, b, c}(\alpha_{k-1})/a)\rfloor a+[0, b-a)  \\
  \hfill  {\rm if}\ A_{k-1}\subset [0, c_0+a-b)\  {\rm or } \ A_{k-1}\subset [c_0, a),\\
A_{k-1}  \hfill   {\rm if}\ A_{k-1}= [c_0+a-b, c_0),
\\  \! [0, a) \hfill  {\rm otherwise},
\end{array}
\right.
$$
and
$S_k=S_{k-1}\backslash A_k, k\ge 1$.
By Theorem  \ref{dabc1holes.tm}, there exists a nonnegative integer $L\le a/(b-a)$ such that
$A_k$ is invariant for $k\ge L$, which implies that
 $S_L$ is the maximal invariant set ${\mathcal S}_{a, b, c}$ in \eqref{dabc1.def} and $S_k=S_L$ for all $k\ge L$.
Thus by  Theorems \ref{framenullspace1.tm} and \ref{sabcstar.tm},
${\mathcal G}(\chi_{[0,c)}, a\Z\times \Z/b)$ is a Gabor frame if and only if either $S_L=\emptyset$ or
$(\lfloor c/b\rfloor+1) |S_L\cap[0, c_0+a-b)|+ \lfloor c/b\rfloor  |S_L\cap[c_0, a)|=a$.

For a triple $(a, b, c)$ of positive numbers that satisfies assumptions in Conclusion (XIII) of Theorem \ref{newmaintheorem5},
we  let  $B_0=[c_1, c_1+b-a)$
 and $T_0=[0,a)\backslash B_0$, and define holes
 $B_k:=[\gamma_k, \delta_k)\subset [0, a)$ and  invariant sets $T_k\subset [0, a), k\ge 1$, iteratively
$$B_k:= 
\left\{\begin{array}{l} R_{a, b, c}(\gamma_{k-1})- \lfloor( R_{a, b, c}(\gamma_{k-1})/a)\rfloor a+[0, \delta_{k-1}-\gamma_{k-1})\\
 \hfill   {\rm if}\ 0\le \gamma_{k-1}<\delta_{k-1}\le c_0+a-b, \\
 R_{a, b, c}(\gamma_{k-1})- \lfloor(R_{a, b, c}(\gamma_{k-1})/a)\rfloor a+[0, c_0+a-b-\gamma_{k-1})\\
\hfill \qquad   {\rm if}\ 0\le \gamma_{k-1}<c_0+a-b<\delta_{k-1}\le c_0, \\
  B_{k-1} \hfill   {\rm if}\ c_0+a-b\le \gamma_{k-1}<\delta_{k-1}\le c_0,\\
  c- \lfloor c/a\rfloor a +[0, \delta_{k-1}-c_0)\\
\hfill
 {\rm if} \  c_0+b-a\le \gamma_{k-1}<c_0<\delta_{k-1}\le a,
 \\
 R_{a, b, c}(\gamma_{k-1})-  \lfloor( R_{a, b, c}(\gamma_{k-1})/a)\rfloor a +[0, \delta_{k-1}-\gamma_{k-1})\\
 \hfill \quad    {\rm if}\ c_0\le \gamma_{k-1}<\delta_{k-1}\le a,\\
 \! [0, a)\quad  \hfill  {\rm otherwise},
\end{array}
\right.
$$
and
$T_k=T_{k-1}\backslash B_k,  k\ge 1$. 
By \eqref{sabcblackholes} and Theorem   \ref{dabc1discreteholes.tm}, there exists a nonnegative integer $L$ such that
$B_k\subset [c_0+a-b, c_0)$ or $B_k=[0,a)$ for all $k\ge L$, which implies that
 $T_L$ is the maximal invariant set ${\mathcal S}_{a, b, c}$ in \eqref{dabc1.def}.
Therefore by Proposition \ref{dabcsabc.thm} and  Theorems \ref{framenullspace1.tm} and \ref{sabcstar.tm},
${\mathcal G}(\chi_{[0,c)}, a\Z\times \Z/b)$ is a Gabor frame if and only if either $T_L=\emptyset$ or
$(\lfloor c/b\rfloor+1) |T_L\cap[0, c_0+a-b)|+ \lfloor c/b\rfloor  |T_L\cap[c_0, a)|= a$.

\section{Sampling signals in a shift-invariant space} 
\label{sampling.appendix}

An interesting problem in sampling is to identify generators  $\phi$ and sampling-shift lattices $a\Z\times b\Z$
such that any signal $f$ in the shift-invariant space
 \begin{equation}\label{appendixB.sis.def}
V_2(\phi, b\Z):=\Big\{\sum_{\lambda\in b\Z} d(\lambda) \phi(t-\lambda): \ \sum_{\lambda\in b\Z} |d(\lambda)|^2<\infty\Big\}
\end{equation}
 can be stably recovered from its equally-spaced
samples $f(t_0+\mu), \mu\in a\Z$, for any initial sampling position $t_0$, that is,
there exist positive constants $A$ and $B$ such that
\begin{equation}\label{appendix.stability.def}
A\|f\|_2\le (\sum_{\mu\in a\Z} |f(t_0+\mu)|^2\big)^{1/2}\le B \|f\|_2
\end{equation}
for all $f\in V_2(\phi, b\Z)$ and $t_0\in \R$.
The range of sampling-shift parameters $a$ and $b$ is fully known only for few generators $\phi$. For instance,
the classical Whittaker-Shannon-Kotel'nikov sampling theorem states that
\eqref{appendix.stability.def} hold for signals bandlimited to $[-\sigma, \sigma]$ if and only if $a\le b=\pi/\sigma$. For
 the uniform spline generator $\underbrace{\chi_{[0, b)}*\cdots* \chi_{[0, b)}}_{n \ {\rm times}}$
 obtained by convoluting the characteristic function on $[0, b)$ for $n$ times,
 \eqref{appendix.stability.def} hold if and only if $a<b$ \cite{akramgrochenig00, sunaicm10}.
In this appendix, we consider the above range problem for the generator $\chi_I$, the characteristic function on an interval $I$.
Our results indicate that it is almost equivalent to the $abc$-problem for Gabor systems, and hence geometry of the range of
sampling-shift  parameters could be arbitrary complicated.

\smallskip

We say that $\{\phi(\cdot-\lambda):\ \lambda\in b\Z\}$ is a {\em Riesz basis} for the shift-invariant space $V_2(\phi, b\Z)$
if there exist positive constants $A$ and $B$ such that
\begin{equation}
A\big(\sum_{\lambda\in b\Z} |d(\lambda)|^2\big)^{1/2}\le \big\|\sum_{\lambda\in b\Z} d(\lambda) \phi(\cdot-\lambda)\big\|_2\le B
\big(\sum_{\lambda\in b\Z} |d(\lambda)|^2\big)^{1/2}
\end{equation}
for all square-summable sequences $(d(\lambda))_{\lambda\in b\Z}$.
For an interval $I=[d, c+d)$,
$\{\chi_I(\cdot-\lambda): \lambda\in b\Z\}$ is  a Riesz basis for the shift-invariant space $V_2(\chi_I, b\Z)$ except that $2\le c/b\in \Z$.
Therefore except that $2\le c/b\in \Z$,
any signal $f$ in $V_2(\chi_I, b\Z)$ can be stably recovered from its equally-spaced
samples $f(t_0+\mu), \mu\in a\Z$, for any initial sampling position $t_0$ if and only if
infinite matrices
${\bf M}_{a, b, c}(t), t\in \R$, in \eqref{infinitematrix.def} have the uniform stability property \eqref{framestabilitymatrix.tm.eq1}, c.f.
\cite{akramgrochenig01, uieee00,  walter92}. This together with  the characterization of  frame property of the Gabor system
${\mathcal G}(\chi_I, a\Z\times \Z/b)$ in \cite{RS97} leads to the following equivalence between our sampling problem associated with
the box generator $\chi_I$
and the $abc$-problem for Gabor systems.

\begin{prop}\label{sampling.prop} Let  $a, b>0$ and $I$ be an interval with length $c>0$. Except that $2\le c/b\in \Z$,
the following two statements are equivalent.
\begin{itemize}
\item [{(i)}]
Any signal $f$ in the shift-invariant space $V_2(\chi_I, b\Z)$ can be stably recovered from  equally-spaced samples
$f(t_0+\mu), \mu\in a\Z$, for any initial sampling position $t_0\in \R$.
\item [{(ii)}]
 ${\mathcal G}(\chi_I, a\Z\times \Z/b)$ is a Gabor frame for $L^2(\R)$.
 \end{itemize}
\end{prop}

If $I=[d, c+d)$ with $2\le c/b\in \Z$, then the shift-invariant space $V_2(\chi_I, b\Z)$ is not closed in $L^2(\R)$, but its closure is
the shift-invariant space generated by $\chi_{I'}$ where $I'=[d, d+b)$.
Therefore for the case that  $I=[d, c+d)$ with $2\le c/b\in \Z$,
any signal $f$ in $V_2(\chi_I, b\Z)$ can be stably recovered from  equally-spaced samples
$f(t_0+\mu), \mu\in a\Z$, for any initial sampling position $t_0\in \R$ if and only if
any signal $f$ in $V_2(\chi_{[d, b+d)}, b\Z)$ can be stably recovered from  equally-spaced samples
$f(t_0+\mu), \mu\in a\Z$, for any initial sampling position $t_0\in \R$  if and only if $a\le b$.
This together with Theorems \ref{newmaintheorem1}--\ref{newmaintheorem5} and Proposition \ref{sampling.prop}
provides the full  classification of sampling-shift lattices $a\Z\times b\Z$ such that
any signal $f$ in $V_2(\chi_I, b\Z)$ can be stably recovered from  equally-spaced samples
$f(t_0+\mu), \mu\in a\Z$, for any initial sampling position $t_0\in \R$.

\smallskip

We remark that two statements in Proposition \ref{sampling.prop} are not equivalent for the case that $2\le c/b\in \Z$ and $a\le b$. The reason is that
under that assumption on the triple $(a, b, c)$,
${\mathcal G}(\chi_{I}, a\Z\times \Z/b)$ is not a  Gabor frame by Theorems \ref{newmaintheorem1} and \ref{newmaintheorem2}, while
any signal $f$ in $V_2(\chi_I, b\Z)$ can be stably recovered from  equally-spaced samples
$f(t_0+\mu), \mu\in a\Z$, for any initial sampling position $t_0\in \R$.

Oversampling, i.e., $a<b$, helps for  perfect reconstruction of band-limited signals and spline signals  from their equally-spaced samples
\cite{akramgrochenig00, akramgrochenig01, uieee00}. Our results
 indicate that oversampling does not always implies the stability of  sampling and reconstruction process for
  signals in the shift-invariant space $V_2(\chi_I, b\Z)$.

\section{Non-ergodicity of a piecewise linear transformation}
\label{ergodic.appendix}

 Recall that the piecewise linear transformation $R_{a, b, c}$ is non-expansive and non-measure-preserving  map on the real line,
 In this appendix, we study non-ergodicity associated with this new map on the line. Particularly, we establish the equivalence between the empty set property for
 the maximal invariant set ${\mathcal S}_{a, b,c}$  and non-ergodicity of the piecewise linear transformation $R_{a, b, c}$.
 The reader may refer to \cite{walter92book} for ergodic theory of various dynamic systems.

\begin{thm} \label{ergodic.appendix.tm}
Let   $(a,b,c)$ be a triple of positive numbers satisfying
$a<b<c, b-a<c_0:=c-\lfloor c/b\rfloor b<a, 0<c_1:=\lfloor c/b\rfloor b-\lfloor (\lfloor c/b\rfloor b/a)\rfloor a<2a-b$ and $\lfloor c/b\rfloor\ge
2$, and let $R_{a, b, c}$ and ${\mathcal S}_{a, b,c}$   be the piecewise linear transformation in \eqref{rabcnewplus.def}
and its maximal invariant set in \eqref{dabc1.def} respectively.
Then ${\mathcal S}_{a, b, c}=\emptyset$ if and only if for any $t\in \R$  there exists $t_0\in [c_0+a-b, c_0)$ such that
\begin{equation}\label{ergodic.appendix.tm.eq1}
\lim_{n\to \infty} \frac{ \sum_{k=0}^{n-1} f( (R_{a, b, c})^k(t))}{n}= f(t_0)
\end{equation}
for all  continuous periodic functions $f$ with period $a$.
\end{thm}

\begin{proof} (i)$\Longrightarrow$(ii)\quad  Clearly it suffices to prove  that for any $t\in \R$ there exists a nonnegative integer $N$  such that
\begin{equation}\label{ergodic.appendix.tm.pf.eq1}
(R_{a, b, c})^N(t)\in [c_0+a-b, c_0)+a\Z.
\end{equation}
Suppose, on the contrary, that such a nonnegative integer $N$ does not exist. Then
$(R_{a, b, c})^n(t)\not\in [c_0+a-b, c_0)+a\Z$ for all $n\ge 0$.
Define ${\bf x}=({\bf x}_t(\lambda))_{\lambda\in b\Z}$ by
 ${\bf x}_t(\lambda)=1$ if $\lambda=(R_{a, b, c})^n(t)-t$ for some  nonnegative integer $n$,
 and ${\bf x}_t(\lambda)=0$ otherwise.
 Then ${\bf x}_t\in {\mathcal B}_b$ as $(R_{a, b, c})^n(t)-t\in b\Z$.
 Following the argument in  Proposition \ref{dabc1pointcharacterization.prop}, we have that
\begin{equation} \label{ergodic.appendix.tm.pf.eq2}
{\bf M}_{a, b, c}(t) {\bf x}_t(\mu)=1\quad {\rm  for\ all} \ \  0\le \mu\in a\Z.\end{equation}
Similar to the index $Q_{a, b, c}$ in \eqref{qabc.def},  we define
$$\tilde Q_{a, b, c}:= 
\sup_{t\in \R} \sup_{{\bf x}\in {\mathcal B}_b} \tilde Q_{a, b, c}(t, {\bf x})$$
where
$\tilde K(t, {\bf  x})=\{\mu\in a\Z:  {\bf M}_{a, b, c}(t) {\bf x}(\mu)=1\}$
and
$$\tilde Q(t, {\bf x})=\left\{\begin{array}{ll} 0 & {\rm if} \ \tilde K(t, {\bf x})=\emptyset\\
\sup\{n\in N| \ [\mu, \mu+na)\subset \tilde K(t, {\bf x}) & \\
\quad \quad \ {\rm for \ some} \ \mu\in a\Z\} & {\rm otherwise}.
\end{array}\right.$$
Following the argument in Lemma \ref{dabc2toqabc.lem}, we can establish the following equivalence:
\begin{equation} \label{ergodic.appendix.tm.pf.eq3}
{\mathcal S}_{a, b, c}=\emptyset\ {\rm if\ and \ only\ if} \ \tilde Q_{a, b, c}<\infty.
\end{equation}
Therefore combining \eqref{ergodic.appendix.tm.pf.eq2} and \eqref{ergodic.appendix.tm.pf.eq3}
leads to the conclusion that ${\mathcal S}_{a, b, c}\ne \emptyset$, which is a contradiction.

\smallskip

(ii)$\Longrightarrow$(i)
We examine two cases  $a/b\in \Q$ and $a/b\not\in \Q$  to verify the empty-set property for ${\mathcal S}_{a, b, c}$.

\smallskip
{\bf Case 1}: $a/b\in \Q$.

Suppose, on the contrary, that ${\mathcal S}_{a, b, c}\ne \emptyset$.
Write $a/b=p/q$ for some coprime integers $p$ and $q$, take $t\in {\mathcal S}_{a, b, c}$, and define
$t_k=(R_{a, b, c})^k (t), k\ge 0$.
Then $t_k \not\in [c_0+a-b, c_0)+a\Z$ for all $k\ge 0$ by Proposition \ref{dabc1pointcharacterization.prop}.
Observe that for any $k\ge 0$, there exists an integer $0\le r_k<p$ such that
$t_k-t-\frac{r_k}{q}b\in a\Z$ as $(t_k-t)/b\in \Z$ and $a=(p/q)b$.
Therefore there exist two integers $n_1<n_2$ such that $t_{n_1}-t_{n_2}\in a\Z$, which together with the one-to-one correspondence of the
transformation $R_{a, b, c}$ on ${\mathcal S}_{a, b, c}$ implies that $t_{n_2-n_1}-t_0\in a\Z$.
Thus
\begin{equation}
\lim_{n\to \infty} \frac{\sum_{k=0}^{n-1} f((R_{a, b, c})^k(t_{n_1}))}{n}=
\frac{\sum_{k=0}^{n_2-n_1-1} f((R_{a, b, c})^k(t_0))}{n_2-n_1},\end{equation}
and the limit in
\eqref{ergodic.appendix.tm.eq1} does not hold for any continuous periodic function $f$
 that it is positive on $[c_0+a-b-\epsilon, c_0)$ and take zero value on $[0, c_0+a-b-\epsilon)$ and $[c_0, a)$,
 where $\epsilon>0$ is sufficiently small. This is a contradiction.

\smallskip
{\bf Case 2}: $a/b\not\in \Q$.

Take $t_0\in \R$,  and let $t_\infty\in [c_0+a-b, c_0)$ such that
\begin{equation}\label{ergodic.appendix.tm.pf.eq4}
\lim_{n\to \infty} \frac{\sum_{k=0}^{n-1} f((R_{a, b, c})^k(t_0))}{n}=f(t_\infty)
\end{equation}
for all continuous  periodic function $f$ with period $a$.
We exam three subcases  to verify that $t_0\not\in {\mathcal S}_{a, b, c}$.

\smallskip
{\bf  Case 2a}: \ $t_\infty\in (c_0+a-b, c_0)$.

Take  a continuous periodic function
$f$ such that it is positive on $(c_0+a-b, c_0)$ and take zero value on $[0, c_0+a-b)$ and $[c_0, a)$.
Then it follows that \eqref{ergodic.appendix.tm.pf.eq4} that $f((R_{a, b, c})^{k_0}(t_0))>0$ for  some $k_0\ge 0$, which implies that
$(R_{a, b, c})^{k_0}(t_0)$ belongs to the black hole $[c_0+a-b, c_0)+a\Z$ of the transformation $R_{a, b, c}$.
Hence $t_0\not\in {\mathcal S}_{a, b, c}$ because ${\mathcal S}_{a, b, c}$ is invariant on the transformation $R_{a, b, c}$, and
${\mathcal S}_{a, b, c}$ has empty intersection with its black hole $[c_0+a-b, c_0)+a\Z$. 

\smallskip
{\bf Case 2b}: \ $t_\infty=c_0+a-b$ and $(R_{a, b, c})^{k_0}(t_0)\in [c_0+a-b+a, c_0)+a\Z$ for some nonnegative integer $k_0$.

In this case,  $t_0\not\in {\mathcal S}_{a, b, c}$ as
${\mathcal S}_{a, b, c}$ is invariant on the transformation $R_{a, b, c}$ and $[c_0+a-b, c_0)\cap {\mathcal S}_{a, b, c}=\emptyset$.

\smallskip
{\em Case 2c}: \ $t_\infty=c_0+a-b$ and $t_k:=(R_{a, b, c})^{k}(t_0)\not\in [c_0+a-b, c_0)+a\Z$ for all nonnegative $k\ge 0$.

Take a sufficiently small positive number $\epsilon>0$ and  a periodic function $f_\epsilon(t)$  whose restriction on $[0, a)$ is given by
$\max(0, 1-|t-c_0+a-b|/\epsilon)$. The function $f_\epsilon$ is the hat function supported in $[-\epsilon, \epsilon]+c_0+a-b+a\Z$.
By \eqref{ergodic.appendix.tm.pf.eq4},
for any given $L$ there exists an integer $k$ such that $t_{k+l}\in (c_0+b-a-\epsilon, c_0+b-a)+a\Z$
for all $0\le l\le L$, as otherwise
$$\lim_{n\to \infty} \frac{\sum_{k=0}^{n-1} f((R_{a, b, c})^k(t))}{n} \le \frac{L-1}{L}\ne 1=f(t_\infty).$$
As $t_{k+l}\in (c_0+b-a-\epsilon, c_0+b-a)+a\Z$ for all $0\le l \le L$, we have
that
$t_{k+l}= (R_{a, b, c})^l(t_k)= t_k+ l (\lfloor c/b\rfloor+1) b, 0\le l\le L$. Thus
$l (\lfloor c/b\rfloor+1) b\in (-\epsilon, \epsilon)+a\Z$ for all $0\le l\le L$, which is a contradiction
as $L$ can be chosen sufficiently large and $a/b\not\in \Q$.
\end{proof}

\end{appendix}

\bigskip
{\bf Acknowledgement}:\
The authors would like to thank Professors Akram Aldroubi, Hans Feichtinger, Deguang Han and Charles Micchelli for their remarks and suggestions.
The  project  is partially  supported by the National Science Foundation of China (No. 10871180),
NSFC-NSF (No. 10911120394), and the National Science Foundation (DMS-1109063).

\end{document}